\documentclass[11pt,a4paper]{article}
\usepackage{mymacros}

\newenvironment{ack}{\section*{Acknowledgements}}{}
\newenvironment{funding}{}{}

\usepackage[titles]{tocloft}
\newcommand{\GFF}{\mathrm{GFF}}
\newcommand{\SG}{\mathrm{SG}}
\newcommand{\FF}{\mathrm{FF}}
\usepackage{slashed}
\newcommand{\Dirac}{\slashed{\partial}}
\newcommand{\wick}[1]{\mathopen{:}#1\mathclose{:}}

\newcommand{\PP}{\bm{P}}

\newcommand{\EE}{\bm{E}}

\newcommand{\nmu}{{\mu}}
\newcommand{\nz}{{z}}
\newcommand{\nzeta}{{\zeta}}

\title{The Coleman correspondence at the free fermion point}
\author{Roland Bauerschmidt\footnote{University of Cambridge, Statistical Laboratory, DPMMS. E-mail: {\tt rb812@cam.ac.uk}.}
  \and Christian Webb\footnote{University of Helsinki, Department of Mathematics and Statistics. E-mail: {\tt christian.webb@helsinki.fi}.}}
\date{\vspace*{-2em}}

\begin{document}
\maketitle
\begin{abstract}
We prove that the truncated correlation functions of the charge and gradient fields associated with the massless sine-Gordon model on $\R^2$
with $\beta=4\pi$ exist for all coupling constants and 
are equal to those of the chiral densities and vector current of free massive Dirac fermions.
This is an instance of Coleman's prediction that  the massless sine-Gordon model and the massive Thirring model
are equivalent (in the above sense of correlation functions).
Our main novelty is that we prove this correspondence
starting from the Euclidean path integral
in the non-perturbative regime of the infinite volume models.
We use this correspondence to show that the correlation functions of the massless sine-Gordon model with $\beta=4\pi$ decay exponentially
and that the corresponding probabilistic field is localized.
\end{abstract}



\section{Introduction and main results}\label{sec:intro}

Statistical and quantum field theory in two (Euclidean) dimensions is rich and special in various ways.
This manifests itself, for example, through the existence of the powerful
theory of conformal field theory (CFT), 
the possibility of quasiparticles which are neither bosons nor fermions but instead have anyonic particle statistics,
or the perhaps surprising possibility of equivalence of fermionic and bosonic field theories---known as bosonization.
The two-dimensional setting also provides one of the main test cases for the understanding of strongly interacting field theories.
The massless sine-Gordon model is a principal example of a
\emph{non-conformal} perturbation of a CFT in two dimensions.
Despite the absence of conformal symmetry, there is a detailed
but almost entirely conjectural description of many of its physical features, not accessible by perturbation theory,
including the prediction of a mass gap for all coupling constants.
These features pertain to the infinite volume theory. In this paper, we study the arguably most fundamental (and simplest)
instance of this---the Coleman correspondence at the free fermion point, which we prove
starting from the path integral formulation
in the \emph{non-perturbative} regime of the infinite volume models
and for all coupling constants.

\subsection{Coleman correspondence}\label{sec:intro-correspondence}

The Coleman correspondence is a prototype for bosonization \cite{PhysRevD.11.2088}.
It states that the massless sine-Gordon model with parameters $(\beta,z)$
and the massive Thirring model with parameters $(g,\mu)$ are equivalent  in the sense of correlation functions
when $(\beta,z)$ and $(g,\mu)$ are appropriately related.
This is an instance of bosonization because
the sine-Gordon model is a bosonic quantum field theory
while the massive Thirring model is a fermionic quantum field theory.
The equivalence is especially striking when $\beta=4\pi$, which corresponds 
to parameters of the massive Thirring model for which it becomes non-interacting
(free massive Dirac fermions), while the sine-Gordon model is interacting (non-Gaussian).

Previous mathematical results have established variants of the Coleman correspondence
in the \emph{perturbative} regime, i.e.,
for small coupling constants and with finite volume interaction term
(or with external mass term),
see \cite{MR446210,MR0443693,MR1672504,MR2461991} and Section~\ref{sec:literature}.
In this article, we prove that, for $\beta=4\pi$, the Coleman correspondence holds
in the \emph{non-perturbative} regime of the infinite volume models (without external mass term).
Unlike the previous results, our proof has thus strong implications for the
massless sine-Gordon model with $\beta=4\pi$: 
we show exponential decay of correlations for all $z\neq 0$ and that the field is probabilistically localized---results
that are non-perturbative in the coupling constant (and false for $z=0$).

Before stating our results,
we first introduce the sine-Gordon model and free Dirac fermions
(both in their Euclidean versions).
The \emph{massless sine-Gordon model}
with coupling constants $\beta\in (0,8\pi)$ and $z\in \R$ is
defined in terms of the limit $\epsilon\to 0$, $m\to 0$, $L\to\infty$ of
the probability measures 
\begin{equation} \label{e:SGdef}
  \nu^{\SG(\beta,z|\epsilon,m,L)}(d\varphi) \propto 
  \exp\qa{2\nz \int_{\Lambda_L}\epsilon^{-\beta/4\pi} \cos(\sqrt{\beta}\varphi(x)) \, dx} \nu^{\GFF(\epsilon,m)} (d\varphi),
\end{equation}
where
$\Lambda_L = \{ x\in \R^2: |x| \leq L \}$ is the Euclidean disk of radius $L$, and
$\nu^{\GFF(\epsilon,m)}$ is the Gaussian free field (GFF) on $\R^2$  with mass $m>0$  regularised at scale $\epsilon>0$.
Here the precise choice of the regularisation of the GFF is not important,
but to be concrete, we choose $\nu^{\GFF(\epsilon,m)}$ as the Gaussian measure supported on $C^\infty(\R^2)$ with covariance kernel
\begin{equation}
   \label{e:cov}
   \int_{\epsilon^2}^\infty ds\, e^{-s(-\Delta +m^2)}(x,y).
\end{equation}
We denote the expectation with respect to the measure $\nu^{\SG(\beta,z|\epsilon,m,L)}$ by $\avg{\cdot}_{\SG(\beta,z|\epsilon,m,L)}$.
The gradient correlation functions are 
the moments of $\partial\varphi$ and $\bar\partial\varphi$
in the limit $\epsilon \to 0$, $m\to 0$, $L\to \infty$.
The charge correlation functions are the limits (when they exist) of linear combinations of
expectations of products of
\begin{equation}
  \wick{e^{\pm i\sqrt{\beta} \varphi(x)}}_\epsilon := \epsilon^{-\beta/4\pi} e^{\pm i\sqrt{\beta} \varphi(x)}
  \label{e:wick}
\end{equation}
or its smeared version, defined for $f \in L^\infty(\R^2)$ with compact support by
\begin{equation}
  \wick{e^{\pm i\sqrt{\beta} \varphi}}_\epsilon(f) := \epsilon^{-\beta/4\pi} \int_{\R^2} dx\, f(x) e^{\pm i\sqrt{\beta} \varphi(x)}.
  \label{e:wick-f}
\end{equation}
The relevant linear combinations are the truncated correlation functions
(or, joint cumulants). 
For example, for $\beta \geq 4\pi$,
the charge one-point function $\avg{ \wick{e^{i\sqrt{\beta} \varphi}}_\epsilon(f) }_{\SG(\beta,z|\epsilon,m,L)}$
diverges when $z \neq 0$ and $\int_{\R^2} f\, dx \neq 0$ (see Proposition~\ref{prop:onepoint}),
but we will see that the truncated charge two-point function, defined by
\begin{align}
  & 
  \avga{ \wick{e^{i\sqrt{\beta} \varphi}}_\epsilon(f_1)\; \wick{e^{-i\sqrt{\beta} \varphi}}_\epsilon(f_2) }^T_{\SG(\beta,z|\epsilon,m,L)}
  \nnb & 
  =
  \avga{ \wick{e^{i\sqrt{\beta} \varphi}}_\epsilon(f_1) \wick{e^{-i\sqrt{\beta} \varphi}}_\epsilon(f_2) }_{\SG(\beta,z|\epsilon,m,L)}
  \nnb &\qquad 
    -
  \avga{ \wick{e^{i\sqrt{\beta} \varphi}}_\epsilon(f_1)}_{\SG(\beta,z|\epsilon,m,L)}
  \avga{ \wick{e^{-i\sqrt{\beta} \varphi}}_\epsilon(f_2)}_{\SG(\beta,z|\epsilon,m,L)},
\end{align}
has a non-trivial limit for test functions $f_1$ and $f_2$ with disjoint supports.
In general, the truncated correlation function of observables $O_1, \dots, O_n$ can be defined inductively by
\begin{equation} \label{e:truncated}
  \avg{O_1 \cdots O_n}^T
  =
  \avg{O_1 \cdots O_n} - \sum_{P \in \mathfrak{P}_n} \prod_j \avg{\prod_{i\in P_j} O_i}^T,
  \qquad \avg{O_i}^T = \avg{O_i},
\end{equation}
where the sum is over proper partitions $P=(P_j) \in \mathfrak{P}_n$ of $[n]= \{1,\dots, n\}$.
Note that $\avg{O_1 \cdots O_n}^T$ does not only depend on the product of the $O_i$,
and a more precise notation would be $\avg{O_1 ; \dots ; O_n}^T$.
However, the formal product notation without semicolons is standard and more convenient for our purposes.
Equivalent to \eqref{e:truncated}, the truncated correlations
are Taylor coefficients of the logarithm of the joint moment generating function of $O_1, \dots, O_n$
if it exists, see Appendix~\ref{app:ferm}. 

Free fermions are defined in terms of their correlation kernel.
The correlation kernel of \emph{free Dirac fermions} of mass $\mu\in \R$
is the fundamental solution of the massive Dirac operator on $\R^2$ for which we
use the representation
\begin{equation} \label{e:Dirac}
  \Dirac+\nmu
  = \begin{pmatrix} \nmu & 2 \bar\partial \\ 2 \partial & \nmu \end{pmatrix},
\end{equation}
where $\partial=\frac12(-i\partial_0+\partial_1)$ and $\bar\partial = \frac12(i\partial_0+\partial_1)$
and we identify $x =(x_0,x_1) \in \R^2$ with $ix_0+x_1 \in \C$.
In terms of the modified Bessel function $K_0$, this fundamental solution is explicitly given by
\begin{equation} \label{e:Smu}
  S(x,y)
  = -\frac{1}{2\pi}
  \begin{pmatrix}
    -\nmu K_0(|\mu||x-y|) & 2\bar\partial_x K_0(|\mu||x-y|)
    \\
    2\partial_x K_0(|\mu||x-y|) & -\nmu K_0(|\mu||x-y|)
  \end{pmatrix}
  \sim
  \frac{1}{2\pi} \begin{pmatrix} 0 & 1/(\bar x-\bar y) \\ 1/(x-y) & 0 \end{pmatrix},
\end{equation}
where $\sim$ holds as $\mu\to 0$;
see Section~\ref{sec:ferm}.
For distinct points
$x_1,\dots, x_n,y_1,\dots, y_n\in\R^2$,
and any
$\alpha_1,\dots,\alpha_n,\beta_1,\dots,\beta_n \in \{1,2\}$,
we then denote the correlation functions of free Dirac fermions by
\begin{equation} \label{e:fermdef}
  \avga{\prod_{i=1}^n \bar\psi_{\alpha_i}(x_i)\psi_{\beta_i}(y_i)}_{\FF(\mu)}
  = \det(S_{\alpha_i\beta_j}(x_i,y_j))_{i,j=1}^n.
\end{equation}
The right-hand side is regarded as the definition of the left-hand side.
In Appendix~\ref{app:ferm}, some standard operational tools for free fermions that we will use later are collected.
Because $S$ is singular, the correlations of $\bar\psi_{\alpha_i}(x_i)\psi_{\beta_i}(x_i)$ are not defined, in general,
but for distinct points $x_1,\dots, x_n$ with $n>1$, the truncated correlations of
$\bar\psi_{\alpha_i}\psi_{\beta_i}(x_i)$ formally make sense and are given by
\begin{equation} \label{e:fermtruncdef}
  \avga{\prod_{i=1}^n \bar\psi_{\alpha_i}\psi_{\beta_i}(x_i)}^T_{\FF(\mu)}
  = (-1)^{n+1} \sum_\pi \prod_{i=1}^{n} S_{\alpha_{\pi^i(1)}\beta_{\pi^{i+1}(1)}}(x_{\pi^i(1)},x_{\pi^{i+1}(1)})
\end{equation}
where the sum is over cyclic permutations $\pi$ on $[n]=\{1,\dots,n\}$.
For our purposes,
we again regard the 
right-hand side of \eqref{e:fermtruncdef} as the definition of the 
left-hand side of \eqref{e:fermtruncdef}.
(As explained in Appendix~\ref{app:ferm},
we note that if $S$ was not singular, then
\eqref{e:fermtruncdef} would be an identity that follows from \eqref{e:truncated}
and \eqref{e:fermdef} without restriction to distinct points. Alternatively one could thus
define the truncated correlations as limits of 
regularized correlations and arrive at the same result as our definition.)
Finally, for any $f_1, \dots, f_n: \R^2 \to \R$ such that the following integrand is (absolutely) integrable,
we will write
\begin{equation} \label{e:truncated-smeared}
  \avga{\prod_{i=1}^n \bar\psi_{\alpha_i}\psi_{\beta_i}(f_i)}^T_{\FF(\mu)}
  = \int dx_1\, \cdots dx_n \, f_1(x_1)\cdots f_n(x_n)\, \avga{\prod_{i=1}^n \bar\psi_{\alpha_i}\psi_{\beta_i}(x_i)}^T_{\FF(\mu)}.
\end{equation}
Since $S(x,y)$ has singularity $O(1/|x-y|)$, for $n \geq 3$, the above integrand is integrable
for all bounded $f_i$ with compact support. For $n =2$, this is true
for $f_1$ and $f_2$ with disjoint compact support.

\medskip

For $\beta=4\pi$, the Coleman correspondence is the following theorem, our first main result.

\begin{theorem} \label{thm:correspondence}
  Let $\beta=4\pi$ and $z \in \R$. Then
  the limit $\epsilon\to 0$, $m\to 0$, $L\to \infty$
  of the truncated correlation functions of
  $\partial\varphi, \bar\partial\varphi, \wick{e^{+i\sqrt{\beta}\varphi}},\wick{e^{-i\sqrt{\beta}\varphi}}$
  of the sine-Gordon model exist (under the restrictions below),
  and they are equal to the correlation functions of free massive Dirac fermions
  with mass $\mu=Az$ (the constant $A$ is defined below):
  for $n+n'+q+q' \geq 2$
  and all test functions $f_1^+, \dots, f_{n}^+, f^-_1, \dots, f^-_{n'}   \in L^\infty(\R^2)$ and $g_1^+, \dots, g_q^+, g^-_{1}, \dots, g^-_{q'} \in  C^\infty_c(\R^2)$,
  all with disjoint compact supports,
  \begin{multline} \label{e:thm-correspondence}
    \lim_{L\to\infty}\lim_{m\to 0}\lim_{\epsilon\to 0}
    \\ 
    \avga{\prod_{k=1}^{n} \wick{e^{+i \sqrt{4\pi} \varphi}}_\epsilon(f_k^+)
      \prod_{k'=1}^{n'}\wick{e^{-i \sqrt{4\pi} \varphi}}_\epsilon(f_{k'}^-)
      \prod_{j=1}^{q} (-i \partial \varphi(g_j^+))
      \prod_{j'=1}^{q'} (+i\bar\partial \varphi(g_{j'}^-))}_{\SG(4\pi,z|\epsilon,m,L)}^T
    \\
    = A^{n+n'} B^{q+q'}
    \avga{
      \prod_{k=1}^{n} {\bar\psi_1 \psi_1}(f_k^+)
      \prod_{k'=1}^{n'} {\bar\psi_2 \psi_2}(f_{k'}^-)
      \prod_{j=1}^{q} {\bar\psi_2\psi_1}(g_j^+)
      \prod_{j'=1}^{q'} {\bar\psi_1 \psi_2}(g_{j'}^-)
    }_{\FF(\mu)}^T,
  \end{multline}
  where $A=4\pi  e^{-\gamma/2}$ and $B=\sqrt{\pi}$ (and where $\gamma$ is the Euler--Mascheroni constant).

  Moreover, for $n+n'+q+q'\geq 3$ and $n+n'=q+q'=1$, the statement is true without the disjoint support assumption.
\end{theorem}

We emphasize that the right-hand side is the explicit polynomial in $S_{\alpha\beta}(x,y)$
given by \eqref{e:fermtruncdef} which is integrated over the test functions as in \eqref{e:truncated-smeared}.
To lighten notation, we will write the limit on left-hand side of \eqref{e:thm-correspondence} as
\begin{equation} \label{e:sg-truncated}
    \avga{\prod_{k=1}^{n} \wick{e^{+i \sqrt{4\pi} \varphi}}(f_k^+)
      \prod_{k'=1}^{n'}\wick{e^{-i \sqrt{4\pi} \varphi}}(f_{k'}^-)
    \prod_{j=1}^{q} (-i\partial \varphi(g_j^+))
    \prod_{j'=1}^{q'} (+i\bar\partial \varphi(g_{j'}^-))}_{\SG(4\pi,z)}^T.
\end{equation}

By choosing $n+n'=0$ and $q+q'=2$, respectively $n+n'=2$ and $q+q'=0$,
the gradient and charge two-point functions 
of the sine-Gordon model are in particular given, 
for test functions $f_1$ and $f_2$ with disjoint support, by:
\begin{align} \label{e:twopoint1}
  & 
  \avg{\partial\varphi(f_1) \partial\varphi(f_2)}_{\SG(4\pi,z)}
    \nnb 
  &\qquad = -\frac{B^2}{\pi^2} \int dx_1 \, dx_2\, f_1(x_1) f_2(x_2) \,
    (\partial_{x_1} K_0(A|z||x_1-x_2|))^2
    ,
  \\
  \label{e:twopoint2}
  & 
  \avg{\partial\varphi(f_1) \bar\partial\varphi(f_2)}_{\SG(4\pi,z)}
    \nnb 
  &\qquad = -\frac{B^2A^2z^2}{4\pi^2} \int dx_1 \, dx_2\, f_1(x_1) f_2(x_2)\,
    (K_0(A|z||x_1-x_2|))^2
    ,
\intertext{and}
  \label{e:twopoint3}
  & 
    \avg{\wick{e^{i \sqrt{4\pi} \varphi}}(f_1) \, \wick{e^{-i\sqrt{4\pi} \varphi}}(f_2) }_{\SG(4\pi,z)}^T
    \nnb 
  &\qquad= \frac{A^2}{\pi^2} \int dx_1 \, dx_2\, f_1(x_1) f_2(x_2)\,
    |\partial_{x_1} K_0(A|z||x_1-x_2|)|^2
    ,
  \\
  \label{e:twopoint4}
  & 
    \avg{\wick{e^{i \sqrt{4\pi} \varphi}}(f_1) \,\wick{e^{i\sqrt{4\pi} \varphi}}(f_2) }_{\SG(4\pi,z)}^T
  \nnb 
  &\qquad = -\frac{A^4z^2}{4\pi^2}
    \int dx_1 \, dx_2\, f_1(x_1) f_2(x_2)\,
    (K_0(A|z||x_1-x_2|))^2
    .
\end{align}
Indeed, for example, by  \eqref{e:fermtruncdef} and \eqref{e:Smu},
\begin{align}
  \avg{{\bar\psi_2\psi_1}(x){\bar\psi_2\psi_1}(y)}_{\FF(\mu)}^T
  = -S_{21}(x,y)S_{21}(y,x)
  = \frac{1}{(2\pi)^2}(2\partial_x K_0(|\mu| |x-y|))^2
\end{align}
so that \eqref{e:thm-correspondence} gives \eqref{e:twopoint1},
noting that for the gradient two-point function we can drop the truncation of the expectation since $\avg{\partial \varphi(f_i)}_{\SG(\beta,z)}=0$ by symmetry.
The equalities \eqref{e:twopoint2}--\eqref{e:twopoint4} are analogous.

Note that the right-hand side of \eqref{e:twopoint3} is not integrable for overlapping test functions,
explaining the restriction in the $(n+n',q+q')=(2,0)$ case in Theorem~\ref{thm:correspondence}.
For the gradient two-point functions, i.e., the case $(n+n',q+q')=(0,2)$,
the statement can be extended to test functions with overlapping support,
but the singular integrals on the right-hand side of \eqref{e:twopoint1} and \eqref{e:twopoint2} then require a more careful interpretation, as in the following theorem.
Similarly as before, we write
\begin{align}
  \avg{\partial\varphi(f_1) \partial\varphi(f_2)}_{\SG(4\pi,z)}
  &:=\lim_{L\to\infty}\lim_{m\to0}\lim_{\epsilon\to 0}\avg{\partial\varphi(f_1) \partial\varphi(f_2)}_{\SG(4\pi,z|\epsilon,m,L)}
\end{align}
when the limits exist, and similarly for its complex conjugate and $\avg{\partial\varphi(f_1) \bar\partial\varphi(f_2)}_{\SG(4\pi,z)}$.

\begin{theorem} \label{thm:twopoint}
  Let $\beta=4\pi$ and $z \in \R$. Then for $f_1,f_2 \in C_c^\infty(\R^2)$,
  \begin{align} \label{e:twopoint1-overlap}
  \avg{\partial\varphi(f_1) \partial\varphi(f_2)}_{\SG(4\pi,z)}
  &= -\frac{B^2}{\pi^2} \mathrm{p.v.}\int
    dx_1 \, dx_2\, f_1(x_1) f_2(x_2) \,
    (\partial_{x_1} K_0(A|z||x_1-x_2|))^2
    ,
  \\
  \label{e:twopoint2-overlap}
    \avg{\partial\varphi(f_1) \bar\partial\varphi(f_2)}_{\SG(4\pi,z)}
    &= -\frac{B^2A^2z^2}{4\pi^2} \int dx_1 \, dx_2\, f_1(x_1) f_2(x_2)\,
      (K_0(A|z||x_1-x_2|))^2
      \nnb
    &\qquad\qquad
      +\frac{1}{4} \int dx\, f_1(x) f_2(x) 
      ,
\end{align}
where $\mathrm{p.v.}\int$ denotes the Cauchy principal value integral: $\lim_{\delta\to 0}\int_{|x_1-x_2|\geq \delta}$.

In particular, the limits defining the left-hand sides exist.
\end{theorem}

In particular,
since the modified Bessel function $K_0$ and its derivative decay exponentially,
the massless sine-Gordon correlation functions decay exponentially whenever $z\neq 0$, when $\beta=4\pi$.
It is conjectured (but in general open) that the massless sine-Gordon model has exponential decay of correlations
for all $\beta \in (0,8\pi)$ and $z\in\R \setminus \{0\}$, with an explicit  conjectured relation between the rate of exponential decay (mass) and the parameters of the sine-Gordon model
\cite{10.1142/S0217751X9500053X}.
For further discussion of this problem, see also the last paragraph of \cite[p.717]{MR2461991}
and Section~\ref{sec:literature} below.

The exponential decay is, of course, in contrast to the well-known situation for the GFF (i.e., the case $z=0$).
It is an elementary computation that GFF
correlations decay polynomially:
\begin{align}
  \avg{\partial \varphi(x)\partial\varphi(y)}_{\GFF}
  &= \frac{-1}{4\pi (x-y)^2},
  \\
  \avg{\partial \varphi(x)\bar\partial\varphi(y)}_{\GFF}
  &= 0,
  \\
  \avg{\wick{e^{i \sqrt{4\pi} \varphi(x)}} \, \wick{e^{-i\sqrt{4\pi} \varphi(y)}} }_{\GFF}
  &= 
  \frac{4e^{-\gamma}}{|x-y|^2},
  \\
  \avg{\wick{e^{i \sqrt{4\pi} \varphi(x)}} \, \wick{e^{i\sqrt{4\pi} \varphi(y)}} }_{\GFF}
  &= 0,
\end{align}
and that the one-point functions exist and vanish; see, for example, the computations in Section~\ref{sec:massless-gff}.
The free field correlations $\avg{\cdot}_{\GFF}$ are defined as in \eqref{e:sg-truncated} with $z=0$.

While the above results are for the charge and gradient correlation functions,
as a consequence we can also construct the (probabilistic) massless sine-Gordon field itself when $\beta=4\pi$ and $z \neq 0$.
Note that the assumption $z\neq 0$ is essential as the massless GFF on $\R^2$
only exists up to an additive constant -- not in the sense of the following theorem.

\begin{theorem} \label{thm:field}
  Let $\beta =4\pi$ and $z\in \R$, $z \neq 0$.
  Then there exists a probability measure on $\mathcal{S}'(\R^2)$
  (not restricted to test functions with mean $0$)
  whose expectation we denote by $\avg{\cdot}_{\SG(4\pi,z)}$
  with the following properties.
  For any $f,g \in C_c^\infty(\R^2)$ with $\int dx\, f = 0 = \int dx\, g$,
  \begin{align}
    \avg{e^{i\varphi(f)}}_{\SG(4\pi,z)} &=
    \lim_{L\to\infty}\lim_{m\to 0} \lim_{\epsilon\to 0} \avg{e^{i\varphi(f)}}_{\SG(4\pi,z|\epsilon,m,L)},\\
    \avga{\varphi(f)\varphi(g)}_{\SG(4\pi,z)}&=\lim_{L\to\infty}\lim_{m\to 0}\lim_{\epsilon\to 0}\avga{\varphi(f)\varphi(g)}_{\SG(4\pi,z|\epsilon,m,L)}.
  \end{align}
  For $f,g\in C_c^\infty(\R^2)$,   one has $\avg{\varphi(f)}_{\SG(4\pi,z)} =0$ and the two-point function is given by
  \begin{equation} \label{e:twopoint-fourier}
    \avg{\varphi(f)\varphi(g)}_{\SG(4\pi,z)} = \int_{\R^2} \frac{dp}{(2\pi)^2} \, \hat f(p) \hat g(-p) \, \hat C_{Az}(p),
  \end{equation}
  where $\hat f(p)=\int_{\R^2}dx\, f(x)e^{-ip\cdot x}$ is the Fourier transform of $f$,
  \begin{equation}
  	\hat C_\mu(p)=\mu^{-2}F(|p|/\mu), \qquad where \qquad F(x)=\frac{1}{x^2}-4\frac{\mathrm{arsinh}(x/2)}{x^3\sqrt{4+x^2}},
  	\label{e:2ptsinh}
  \end{equation}
  and $\mathrm{arsinh}(x)=\log(x+\sqrt{x^2+1})$ is the inverse hyperbolic sine.

  In particular, the above massless sine-Gordon field on $\R^2$ is localized and has exponential decay of correlation:
  for any $f \in C_c^\infty(\R^2)$, 
  \begin{equation} \label{e:varbd}
    \sup_{x,y \in \R^2}\avg{(\varphi(f_x)-\varphi(f_y))^2}_{\SG(4\pi,z)}  < \infty,
  \end{equation}
  and
  \begin{equation}
    \label{e:twopoint-expdecay}
    \avg{\varphi(f_x)\varphi(f_y)}_{\SG(4\pi,z)} \qquad \text{decays exponentially as $|x-y|\to\infty$,}
  \end{equation}
  where $f_x(y)= f(y-x)$ denotes the translation of $f$ to $x \in \R^2$.
\end{theorem}

Finally, we comment on the exclusion of the one-point functions in Theorem~\ref{thm:correspondence}.
While the charge one-point functions vanish
in the massless free field case, the following proposition shows that they typically
diverge when $z\neq 0$.

\begin{proposition} \label{prop:onepoint}
 Let $\beta=4\pi$ and $z \in \R$. For $f\in L^\infty(\R^2)$ with compact support, the charge one-point functions satisfy
\begin{equation}
  \lim_{L\to\infty}\lim_{m\to 0}\lim_{\epsilon\to 0}\left[\tfrac{1}{\log \epsilon^{-1}}\langle \wick{e^{\pm i\sqrt{\beta}\varphi}}_\epsilon(f)\rangle_{\SG({4\pi},z|\epsilon,m,L)}\right]=2\pi \nz e^{-\gamma}\int_{\R^2} dx\, f(x),
\end{equation}
while the gradient one-point functions vanish (by symmetry):
\begin{equation}
  \avg{\partial \varphi(f)}_{\SG({4\pi},z)}=\avg{\bar\partial \varphi(f)}_{\SG({4\pi},z)}=0.
\end{equation}
\end{proposition}

The above divergence of the charge one-point functions is shown in Theorem~\ref{thm:cf} item (iv),
in fact more generally for all $\beta \in [4\pi,6\pi)$.
As a consequence of this and of the existence of the truncated charge correlation functions,
none of the untruncated charge correlation functions involving a test function with
$\int_{\R^2} f\, dx \neq 0$ converge as $\epsilon \to 0$.
On the other hand, since the gradient one-point functions exist,
the existence of the truncated gradient correlation functions also implies that of
the untruncated gradient correlation functions
\begin{equation}
  \avga{\prod_{j=1}^q \partial \varphi(g_j^+) \prod_{j'=1}^{q'} \bar\partial\varphi(g_{j'}^-)}_{\SG(4\pi,z)},
\end{equation}
with explicit expressions given  by inverting \eqref{e:truncated}.

Before discussing consequences of Theorems~\ref{thm:correspondence}--\ref{thm:twopoint}
and our more general analysis in their proofs, we remark on the physical interpretation of the fermionic side of the Coleman correspondence.

\begin{remark}
  The Coleman correspondence can be written in terms of Dirac matrices $\gamma^\mu$
  satisfying $\gamma^\mu\gamma^\nu + \gamma^\nu\gamma^\mu = 2\delta_{\mu\nu} \mathbf 1$.
  In the representation we have chosen, these are
  \begin{equation} \label{e:gamma} 
    {\bf 1} = \begin{pmatrix} 1 & 0 \\ 0 & 1 \end{pmatrix},
    \qquad
    \gamma^0 = \begin{pmatrix} 0 & i \\ -i & 0 \end{pmatrix},
    \qquad
    \gamma^1 = \begin{pmatrix} 0 & 1 \\ 1 & 0 \end{pmatrix},
    \qquad
    \gamma^5 =
    \begin{pmatrix} 1 & 0 \\ 0 & -1 \end{pmatrix}.
  \end{equation}
  Thus $\gamma^0 = -\sigma_2$, $\gamma^1 = \sigma_1$, $\gamma^5 = \sigma_3$, where the $\sigma_i$ are the Pauli matrices.
  In terms of these, the Dirac operator can be written as
  \begin{align}
    \Dirac = \gamma^0\partial_0 + \gamma^1 \partial_1
    = \begin{pmatrix}
      0 & i\partial_0 + \partial_1 \\
      -i\partial_0 + \partial_1 & 0
    \end{pmatrix}
   = \begin{pmatrix}
      0 & 2\bar\partial \\
      2\partial & 0
    \end{pmatrix}.
    \label{e:slash}
  \end{align}
  The Coleman correspondence can then be regarded as the following equivalence of the fields:
  \begin{align}
    \label{e:corr1}
    \wick{e^{i\sqrt{{4\pi}}\varphi}}
    &\quad\leftrightarrow\quad
      A\bar\psi_1\psi_1
      =
      \frac{A}{2} {\bar\psi ({\bf 1}+\gamma_5) \psi},
    \\
        \label{e:corr2}
    \wick{e^{-i\sqrt{{4\pi}}\varphi}}
    &\quad\leftrightarrow\quad
      A\bar\psi_2\psi_2
      =
      \frac{A}{2} {\bar\psi ({\bf 1}-\gamma_5) \psi},
    \\
        \label{e:corr3}
    -i \partial \varphi
    &\quad\leftrightarrow\quad
      B\bar\psi_2\psi_1
      =
      \frac{B}{2} {\bar\psi (i\gamma^0+\gamma^1) \psi},
    \\
    \label{e:corr4}
    +i \bar\partial \varphi
    &\quad\leftrightarrow\quad
      B\bar\psi_1\psi_2
      =
      \frac{B}{2} {\bar\psi (-i\gamma^0+\gamma^1) \psi}.
  \end{align}  
  The right-hand sides of \eqref{e:corr1}--\eqref{e:corr2} have the interpretation of being the chiral densities associated with the
  spinor field $\psi$,
  and the right-hand sides of \eqref{e:corr3}--\eqref{e:corr4} that of the vector current $\bar\psi\gamma^\mu\psi$
  (written in complex coordinates); see, for example, \cite[Section~3]{MR937363}.
\end{remark}

\subsection{Further results}
\label{sec:intro-gradient}

Our estimates for the sine-Gordon model together with the correlation inequalities from \cite{MR496191}
also imply the following results for the infinite volume limit for $\beta \in (0,6\pi)$.

The first theorem is for the infinite volume limit of the massless sine-Gordon field modulo constants (the `gradient field').
Let $\mathcal S'(\R^2)/\text{constants}$ denote the topological dual of the (closed) subspace of integral-$0$ functions of the
Schwartz space $\mathcal S(\R^2)$. 

\begin{theorem} \label{thm:gradientlim}
  Let $\beta \in (0,6\pi)$ and $z\in \R$.
  Then for any $f \in C_c^\infty(\R^2)$ with $\int f \, dx =0$,
  the limit
  \begin{equation}
    \avg{e^{i\varphi(f)}}_{\SG(\beta,z)} :=
    \lim_{L\to\infty}\lim_{m\to 0} \lim_{\epsilon\to 0} \avg{e^{i\varphi(f)}}_{\SG(\beta,z|\epsilon,m,L)}
  \end{equation}
  exists, and extends to the characteristic functional of a probability measure on
  the space 
  $\mathcal S'(\R^2)/\text{constants}$
  whose expectation we denote by $\avg{\cdot}_{\SG(\beta,z)}$. This measure 
  is invariant under Euclidean transformations and satisfies
  \begin{equation} \label{e:corrbd1}
    \avg{e^{i\varphi(f)}}_{\SG(\beta,z)} \geq
    e^{-\frac12 (f,(-\Delta)^{-1}f)}
    ,
    \qquad
    \avg{\varphi(f)^2}_{\SG(\beta,z)} \leq
    (f,(-\Delta)^{-1}f)
    .
  \end{equation}
\end{theorem}

For $m>0$ fixed and $\nz>0$, we similarly obtain the existence of the infinite volume of limit of the massive sine-Gordon field.

\begin{theorem} \label{thm:fieldlim}
  Let $\beta \in (0,6\pi)$, $m,\nz>0$.
  For any $f \in C_c^\infty(\R^2)$ (not assuming $\int dx\, f  =0$), the limit
  \begin{equation}
    \avg{e^{i\varphi(f)}}_{\SG(\beta,z|m)} :=
    \lim_{L\to\infty}\lim_{\epsilon\to 0} \avg{e^{i\varphi(f)}}_{\SG(\beta,z|\epsilon,m,L)}
  \end{equation}
  exists, extends to the characteristic functional of a probability measure
  on $\mathcal S'(\R^2)$
  whose expectation we denote by $\avg{\cdot}_{\SG(\beta,z|m)}$.
  This measure is invariant under Euclidean transformations and satisfies 
  \begin{alignat}{2}
    \label{e:corrbd2}
    \avg{e^{i\varphi(f)}}_{\SG(\beta,z|m)} &\geq
    e^{-\frac12 (f,(-\Delta+m^2)^{-1}f)},
    &\qquad
    \avg{\varphi(f)^2}_{\SG(\beta,z|m)} &\leq
    (f, (-\Delta+m^2)^{-1}f)
    ,
    \intertext{and}
    \label{e:corrbd3}
    \avg{e^{i\varphi(f)}}_{\SG(\beta,z|m)} &\geq \avg{e^{i\varphi(f)}}_{\SG(\beta,z)},
    &\qquad
    \avg{\varphi(f)^2}_{\SG(\beta,z|m)} &\leq \avg{\varphi(f)^2}_{\SG(\beta,z)},
  \end{alignat}
  where the right-hand sides of the last two bounds are as in Theorem~\ref{thm:gradientlim} and hold if $\int dx\, f = 0$.
\end{theorem}   

For $\beta=4\pi$, we can then deduce using the localization bound \eqref{e:varbd}
that the $m\to 0$ limit can be taken after the infinite volume limit,
which means that
the formal $\varphi\mapsto \varphi+\frac{2\pi}{\sqrt{\beta}}\Z$-symmetry of the massless sine-Gordon model is spontaneously broken in the infinite volume limit.

\begin{corollary} \label{cor:zle0}
  Let $\beta=4\pi$ and $\nz>0$. Then for any $f \in C_c^\infty(\R^2)$ (not assuming $\int dx\, f  =0$),
  \begin{equation} \label{e:zle0}
    \avg{\varphi(f)^2}_{\SG(4\pi,z|0^+)} :=
    \lim_{m\to 0}\lim_{L\to\infty} \avg{\varphi(f)^2}_{\SG(4\pi,z|m,L)}
    \leq \int_{\R^2}\frac{dp}{(2\pi)^2} \, |\hat f(p)|^2 \hat C_{Az}(p) 
    < \infty,
  \end{equation}
  where $\hat C_{\mu}(p)$ is as in Theorem~\ref{thm:twopoint}.
  
  Moreover, the limit
  $\avg{\cdot}_{\SG(4\pi,z|0^+)} := \lim_{m\to 0} \lim_{L\to\infty} \avg{\cdot}_{\SG(4\pi,z|m,L)}$ exists
  in the sense of characteristic functionals and defines
   a probability measure on $\mathcal S'(\R^2)$ (not dividing out constants).
\end{corollary}

We expect that $\avg{\cdot}_{\SG(4\pi,z|0^+)}$ is the same as $\avg{\cdot}_{\SG(4\pi,z)}$
    but our arguments do not imply this.

\subsection{Heuristics and previous results}
\label{sec:literature}

The \emph{formal} equivalence of the massless sine-Gordon model and the massive Thirring model was observed by Coleman in \cite{PhysRevD.11.2088}.
The massive Thirring model with parameters $(g,\mu)$ is formally given by a fermionic path integral with ``density''
\begin{equation} \label{e:thirring}
  \exp\qa{-\int_{\R^2} dx\, \pa{ \psi \Dirac \bar\psi +\nmu (\psi_1\bar\psi_1+\psi_2\bar\psi_2) - 2g \psi_1\bar\psi_2 \psi_2\bar\psi_1} }.
\end{equation}
Coleman observed that, order by order in a formal expansion,
the massless sine-Gordon model with parameters $(\beta,z)$ is related to the massive Thirring model
if the parameters of the two models are related by
\begin{equation} \label{e:Coleman-gmu}
  g= B^2(1-\frac{4\pi}{\beta}), \qquad \mu=Az.
\end{equation}
Heuristically  this  prediction is not difficult to understand from the type of
massless Gaussian free field and massless free fermion computations we derive in Section~\ref{sec:massless};
these are versions of the identifications \eqref{e:corr1}--\eqref{e:corr4} in the elementary situation of the massless Gaussian free field and massless free fermions.
Indeed, after rescaling $\varphi$ by $\sqrt{4\pi/\beta}$, the measure of sine-Gordon model with parameters $(\beta,z)$ has formal density
\begin{align}
  &\exp\qa{-\int_{\R^2}dx\, \pa{ \frac{8\pi}{\beta} (\partial\varphi) (\bar\partial\varphi) - \nz (\wick{e^{i\sqrt{4\pi}\varphi}}+\wick{e^{-i\sqrt{4\pi}\varphi}}) } }
  \\
  &=\exp\left[-\int_{\R^2} dx\, 2 (\partial\varphi) (\bar\partial\varphi)
    \right.\nnb &\qquad\qquad \left. 
    +
    \int_{\R^2} dx\, \pa{ 2 (1-\frac{4\pi}{\beta}) (-i\partial\varphi) (+i\bar\partial\varphi) + \nz (\wick{e^{i\sqrt{4\pi}\varphi}}+\wick{e^{-i\sqrt{4\pi}\varphi}}) } \right].
    \nonumber
\end{align}
Thus relative to the massless free fermion ``measure'' respectively the massless Gaussian free field measure,
formally, the massive Thirring model and the sine-Gorden model are weighted by
\begin{gather}
  \label{e:thirring2}
  \exp\qa{\int_{\R^2} dx\, \pa{2g \bar\psi_1\psi_2 \bar\psi_2\psi_1 +\nmu (\bar\psi_1\psi_1+\bar\psi_2\psi_2)} },
  \\
  \exp\qa{\int_{\R^2} dx\, \pa{ 2 (1-\frac{4\pi}{\beta}) (-i\partial\varphi) (+i\bar\partial\varphi) +\nz (\wick{e^{i\sqrt{4\pi}\varphi}}+\wick{e^{-i\sqrt{4\pi}\varphi}}) }},
\end{gather}
and one can see the order by order correspondence of the models with parameters \eqref{e:Coleman-gmu},
using the equivalence of the correlations of \eqref{e:corr1}--\eqref{e:corr4} with respect to the noninteracting measures.
To directly apply these identities, note that we changed the order of the Grassmann variables
in \eqref{e:thirring2} compared to \eqref{e:thirring} explaining the change of the sign of the quadratic term.

Mathematically, this formal argument is however far from a proof.
To start with, the probabilistic or analytic existence of the massless sine-Gordon
model and the massive Thirring model is a nontrivial problem.
Both the ultraviolet (short-distance) and infrared (long-distance) behavior of both models
cause significant difficulties, while both regimes need to be handled to establish the Coleman correspondence
for the infinite volume models. We summarize the most relevant previous results on these problems now.

Concerning the ultraviolet stability of the sine-Gordon model, we note that
various constructions of the finite volume sine-Gordon model exist under different assumptions
(see in particular \cite{MR0443693,MR0421409,
  MR649810,MR849210,
  MR914427,
  MR1777310,MR1240586,
  MR4010781,
  1903.01394,
  2009.09664}),
but that
none of these covers all $\beta \in (0,8\pi)$ and all $z\in \R$.
For the ultraviolet construction of the Thirring model,
for $|g|$ small,
we refer in particular to \cite{MR2308750} which considers the massive case
using previous results on the massless case
including \cite{MR2070092,MR1947693}; see also the further references given therein.
In preparation for later discussion, we stress that it is a technically important ingredient of
these analyses that the finite volume regularisations of
these models are defined on a torus with spatially constant mass term.

Concerning the infrared behavior of the massless sine-Gordon and the massive Thirring models,
the following previous results are particularly relevant.
For the sine-Gordon model with $\beta>0$ small, exponential decay
of the charge correlation functions was proved for the model constructed with
Dirichlet boundary conditions \cite{MR923850}; see also the discussion on Debye screening further below.
For the massive Thirring model, stretched exponential correlation decay was proved for $|g|$ small
(corresponding to $|\beta-4\pi|$ small on the sine-Gordon side),
with antiperiodic boundary conditions \cite{MR2308750}.
For a potential application of these results to a  proof of the Coleman correspondence
in the regimes they apply to
(and thus to transfer the results from one side to the other as we do for $\beta=4\pi$ in this article),
we  emphasize the boundary conditions the former results are proved for.
Indeed,  the generalization of the argument we use for $\beta=4\pi$
(which is in line with Coleman's original proposal) would require
`free' boundary versions of the former results,
by which we mean that the sine-Gordon model is defined in terms of the
infinite volume free field but with finite volume interaction,
and that the massive Thirring model is defined with infinite volume quartic
interaction term but with finite volume mass term, all with uniform dependence on the volume.
We expect such estimates are true, but due to lack of translation invariance,
they are significantly more difficult to obtain and pose interesting problems for future works.

Concerning the Coleman correpondence, i.e., the equivalence of both models, 
we mention that in view of the restrictions in the domains of construction of the two models,
previous results are
restricted to models with finite volume sine-Gordon interaction or with fixed external `bare' mass $m>0$.
In particular, for avoidance of doubt, we stress that the Coleman correspondence
for the \emph{massless} sine-Gordon model in \emph{infinite volume} remains open for $\beta\neq 4\pi$.
In the presence of a bare mass or finite volume interaction, the relevant previous results are as follows.
For $\beta<4\pi$, a variant of the Coleman correspondence between the \emph{massive} sine-Gordon model (i.e., $m>0$ fixed) and
the massive Thirring--Schwinger model (QED$_2$) was proved in \cite{MR0443693}
for $z/m^2$ is sufficiently small; see also \cite{MR0523016} for a review.
Also for $\beta<4\pi$, but now with finite volume interaction instead of with an external mass term,
a version of the Coleman correspondence was shown in \cite{MR446210}.
In the same regime,  $\beta<4\pi$ and finite volume sine-Gordon interaction, a construction of 
Haag--Kastler nets of the sine-Gordon model with finite volume interaction
in Lorentzian signature was carried out in \cite{MR4236060},
and a version of the Coleman correspondence was verified in this setting of algebaric QFT.
For $\beta=4\pi$ but with finite volume interaction,
a version of the Coleman correspondence applying to the sine-Gordon model
with small coupling constant $z$ (depending on the volume) was proved in \cite{MR1672504}.
Finally, for $\beta$ in a neighborhood of $4\pi$, but again in finite volume and all coupling constants small depending on the volume,
the Coleman correspondence was proved in \cite{MR2461991}.

The integrability of two-dimensional conformal field theories is celebrated and well known.
That non-conformal  perturbations of conformal field theories are in some cases expected to remain integrable
is perhaps more surprising. The sine-Gordon and massive Thirring models are such examples,
and our result confirms the most fundamental (and arguably simplest) instance of this integrability.
In the physics literature, many other exact results have been predicted by employing various techniques.
For example, at the free fermion point $\beta=4\pi$,
exact expressions for the fractional charge two-point functions, i.e.,
$\avg{\wick{e^{i\sqrt{\alpha_1}\varphi(x_1)}}\wick{e^{i\sqrt{\alpha_2}\varphi(x_2)}}}_{\SG(4\pi,z)}$ with $\alpha_1,\alpha_2 \in (0,4\pi)$,
were derived in \cite{MR1297289},
the mass was determined for general $\beta$ by using a mapping to the continuum limit of an inhomogeneous six-vertex model and the Bethe ansatz
in \cite{MR1488586,10.1142/S0217751X9500053X},
and exact expressions for the fractional charge one-point function $\avg{\wick{e^{i\sqrt{\alpha}\varphi(x)}}}_{\SG(\beta,z)}$ for $\alpha \in (0,4\pi)$
and general $\beta$ were derived in \cite{MR1453266} by extrapolation of exact results for $\beta=4\pi$ and in the
asymptotics $\beta \to 0$.
Further references are \cite{MR1672130,MR1850795,PhysRevD.11.3424},
and for a review, see also \cite{MR2021021}.
All of these integrability results in infinite volume remain conjectural (except for our results at $\beta=4\pi$).
In finite volume, we mention the 
rigorous connection to the XOR Ising model at $\beta=2\pi$ proved in \cite{MR4149524}.

It is also well known that the sine-Gordon model is exactly related to the classical two-dimensional (two-component) Coulomb gas.
For this, we refer in particular to \cite{MR0434278} and also \cite{MR496191} where, using this relation,
many fundamental properties of the Coulomb gas have been derived when $\beta<4\pi$
including existence of the pressure and correlation functions,
the exact equation of state for the pressure,
and the exact scaling behaviour in $z$ of the rate of exponential decay of correlations assuming its existence in a suitable sense.
The latter exponential decay of correlations is in general open.
For the related three-dimensional Coulomb gas, exponential decay (Debye screening)
was proved for $\beta>0$ and $z$ both small in \cite{MR574172}.
The methods have also been partially extended to the two-dimensional Coulomb gas in \cite{MR923850}.
This latter result is incomplete in the sense that it requires  small coupling constants and
more signficiantly that it relies on Dirichlet boundary conditions.
On the other hand, the relation between the sine-Gordon model and the Thirring model only holds
for `free' boundary conditions in finite volume in the previously discussed sense.
Thus the proof of Debye screening of the two-dimensional
Coulomb gas with free boundary condition (and its equivalence with Dirichlet boundary conditions)
remains an interesting problem. For related results in the three-dimensional setting, see also \cite{MR818828}.
Correlation inequalities for the Coulomb gas and the sine-Gordon model as well as their applications are discussed
in \cite{MR0475380,MR0456220,MR496191}; we make some use of these in Section~\ref{sec:estimates}.
Assuming the validity of the Coleman correspondence at the free fermion point (which we here prove),
its implications for the Coulomb gas at $\beta=4\pi$ are discussed in \cite{MR923851};
see also \cite{MR799499}.

Next, we mention a few related bosonization results.
The concept of bosonization goes at least back to \cite{MR172638}; see also \cite{10.1142/2436} for a review.
In the free field case,
the boson--fermion correspondence has been extended by disorder operators \cite{MR937363}.
Second, while the bosonization relations in this paper rely essentially on the precise asymptotics of the correlations
in the continuum limit,
in the massless free field case, exact discrete versions have been found as well; see in particular \cite{MR3369909}.
Some bosonization results are also expected to extend to Riemann surfaces \cite{MR937363}.
For applications of bosonization of free fermions, see, for example,
\cite[Chapter 10.5]{MR1175177}.

Finally, let us emphasize that the massless sine-Gordon model is an essential example of a two-dimensional \emph{non-conformal} perturbation of a CFT.
For \emph{conformal} field theories, a lot of recent progress has been made, in particular
for the Ising model (see \cite{1307.4104,MR3305999,MR3296821} and references therein)
and for the Liouville CFT (see \cite{MR4060417,2005.11530} and references therein).
Moreover, we mention that models related to the \emph{massless} Thirring model have also been studied
in detail, in particular recently in the form of interacting dimers \cite{MR3606736,MR4121614}.


\subsection{Outline of the paper}
\label{sec:outline}

The paper is structured as follows.
 
In Section~\ref{sec:massless}, we derive the Coleman correspondence in the (noninteracting)
massless case $z=\mu=0$.
This analysis is elementary and the result is well known, but lacking a reference
providing exactly what we need later we include the short and instructive proofs.
This is also an opportunity for us to introduce notation as well as to
collect various estimates for Gaussian free fields and massless fermions for later use.

In Section~\ref{sec:estimates}, we state our estimates for the sine-Gordon model and free fermions,
and then prove our main theorems assuming these estimates.
As discussed already briefly in Section~\ref{sec:literature}, it is important
that these estimates apply to `free' boundary versions of both models.
The remainder of the paper is mainly devoted to proving these estimates.

In Sections~\ref{sec:renormpot} and~\ref{sec:sg-expansion}, we consider the sine-Gordon side.
In particular, we construct the renormalized potential for the regularized sine-Gordon model in Section~\ref{sec:renormpot},
and then use it, in Section~\ref{sec:sg-expansion},
to prove the analyticity of the partition function of the sine-Gordon model and the convergence of the correlations functions,
for any finite volume interaction.
The analysis in Section~\ref{sec:renormpot} extends the continuous renormalisation group approach of \cite{MR914427}
by allowing space dependent coupling constants and extraction of the precise estimates needed subsequently; similar results could presumably be obtained using the related methods of \cite{MR814849,MR2461991}.
The analyticity and convergence results of Section~\ref{sec:sg-expansion} rely on the 
combination of the expansions for the renormalized potential up a finite scale at which they converge
with qualitative bounds and concentration estimates for
Gaussian measures, which provide sufficient control in the regime where the expansions fail to converge.

In Section~\ref{sec:ferm}, we prove the corresponding results on the free fermion side.
Our main work here goes into the analysis of the Green's function of the Dirac operator with finite volume mass term.
Due to lack of the maximum principle or a random walk representation for the Dirac operator
as well as lack of translation invariance, we rely on a series construction by expansion in a carefully chosen basis.
This basis is related to eigenfunctions of the Laplacian on the disk
and the spherical geometry is convenient here,
but we expect that analogous results hold in more general geometry.

In Appendix~\ref{app:ferm}, we collect a few (well known) operational tools for cumulants and free fermions
that we use in various places throughout the paper.

\subsection{Notation}

We will write $f\in L_c^\infty(\R^2)$ if $f$ is compactly supported and essentially bounded.
We  write similarly $f \in L^\infty_c(\R^2 \times \{\pm 1\})$ if $f(x,\pm 1)$ is compactly supported and essentially bounded.
We often write $\xi = (x,\sigma) \in \R^2 \times \{\pm 1\}$ and
\begin{equation}
  \int d\xi \, f(\xi) \equiv \int_{\R^2 \times \{-1,1\}} d\xi \, f(\xi) \equiv \sum_{\sigma\in\{-1,1\}} \int_{\R^2} dx\, f(x,\sigma).
\end{equation}
Throughout the paper, $|\cdot|$ denotes the Euclidean norm, and we will often make use of the identification of $\R^2$ and $\C$.
More precisely, we will denote the components of a point $x \in \R^2$ by $x=(x_0,x_1)$ and its identification
with an element in $\C$
by $x_1+ix_0$. We will also repeatedly write $[n]:=\{1,2,\dots,n\}$. We write $A\subset B$ to indicate that $A$ is any subset of $B$ (no need to be proper).

\section{Free field estimates and bosonization of massless fermions}
\label{sec:massless}

A well-understood (but essential) step in the proof of Theorem~\ref{thm:correspondence} is to verify \eqref{e:thm-correspondence} when $z=\mu=0$.
Results of this flavor exist in the literature, see \cite[Section~3]{MR937363} or \cite[Section~2.2]{MR1672504}, for example, 
but since neither of these references provides the exact statements that we need,
we will give a derivation in our set-up in this section.
Along the way we  will also  collect estimates for the correlations of the free field that we require 
for the proof of the Coleman correspondence with $z\neq 0$.

\subsection{Fermionic side: massless free fermion correlations}
\label{sec:massless-ferm}

We start with computation of the correlation functions of free massless Dirac fermions
whose correlation kernel $S$ is given by \eqref{e:Smu} with $\mu=0$, i.e.,
\begin{equation} \label{e:S0-bis-sec2}
  S(x,y)
  =\frac{1}{2\pi}\begin{pmatrix}
    0 & 1/(\bar x-\bar y) \\
    1/(x-y) & 0
  \end{pmatrix}.
\end{equation}
In this section, the fermionic correlation functions are then defined by
\begin{equation} \label{e:fermdef-bis}
  \avga{\prod_{i=1}^n \bar\psi_{\alpha_i}(x_i)\psi_{\beta_i}(y_i)}_{\FF(0)}
  = \det(S_{\alpha_i\beta_j}(x_i,y_j))_{i,j=1}^n
\end{equation}
whenever the determinant on the right-hand side is well-defined, i.e.,
for all $i,j \in [n]$, either $x_i \neq y_j$ or $\alpha_i=\beta_j$.
The Coleman correspondence is in terms of truncated correlation functions, and importantly, we shall require the setting where $x_i=y_i$. These truncated correlation functions are defined by 
\begin{equation} \label{e:fermtruncdef-bis}
  \avga{\prod_{i=1}^n \bar\psi_{\alpha_i}\psi_{\beta_i}(x_i)}^T_{\FF(0)}
  = (-1)^{n+1} \sum_\pi \prod_{i=1}^{n} S_{\alpha_{\pi^i(1)}\beta_{\pi^{i+1}(1)}}(x_{\pi^i(1)},x_{\pi^{i+1}(1)}),
\end{equation}
where the sum is over cyclic permutations $\pi$,
whenever the right-hand side is well-defined, i.e., 
for all $i,j \in [n]$, either $x_i \neq x_j$ or $\alpha_i =\beta_j$ and $\alpha_j=\beta_i$.
These definitions are consistent with \eqref{e:fermdef} and \eqref{e:fermtruncdef} but slightly more general.
(This generality is required for the proof of Theorem~\ref{thm:correspondence}.)

We will need various identities for these determinants defining our correlation functions.
These identities are conveniently seen in the representation of these determinants in terms of Grassmann integrals.
We discuss the details of this representation and prove the required (well-known) properties in Appendix~\ref{app:ferm}. The connection between our discussion there and that here is that to study the correlation functions $\avg{\prod_{1\leq i\leq n} \bar\psi_{\alpha_i}(x_i)\psi_{\beta_i}(y_i)}_{\FF(0)}$ the matrix $(K_{ij})$ in Lemma~\ref{le:fcor} can be defined to be $S_{\alpha_i\beta_j}(x_i,y_j)$ off the diagonal and on the diagonal to be a constant real number chosen so large that $K$ is invertible -- the exact value of this constant is irrelevant (see also Remark~\ref{rk:KS}).
This definition and Lemma~\ref{le:fcor} then allow us to deduce that the properties of Lemma~\ref{le:fcor2} hold also for the correlation functions we are considering here. Based on this representation, we can also use Lemma~\ref{le:fcor3}, to see that \eqref{e:fermtruncdef-bis} is also consistent with \eqref{e:truncated},
i.e.,
\begin{equation}
  \avga{\bar\psi_{\alpha_i}\psi_{\beta_i}(x_i)}^T_{\FF(0)}
  = \avga{\bar\psi_{\alpha_i}\psi_{\beta_i}(x_i)}_{\FF(0)} = 0 \qquad (\text{assuming } \alpha_i=\beta_i),
\end{equation}
and, for $n \geq 2$,
\begin{equation}
  \avga{\prod_{i=1}^n \bar\psi_{\alpha_i}\psi_{\beta_i}(x_i)
  }^T_{\FF(0)}
  =
    \avga{\prod_{i=1}^n \bar\psi_{\alpha_i}\psi_{\beta_i}(x_i)}_{\FF(0)} - \sum_{P \in \mathfrak{P}_n} \prod_j \avga{\prod_{i\in P_j} \bar\psi_{\alpha_i}\psi_{\beta_i}(x_i)}^T_{\FF(0)}, 
     \label{e:truncated-bis}
\end{equation}
when the right-hand sides exist. Thus when the untruncated correlation functions exist, they determine the truncated ones by  \eqref{e:truncated-bis}.
In view of this fact, the next lemma determines the truncated correlation functions
\begin{equation}
  \avga{
    \prod_{k=1}^n \bar\psi_1\psi_1(x_k) \prod_{k'=1}^{n'} \bar\psi_2\psi_2(y_{k'})
  }_{\FF(0)}^T
\end{equation}
when $x_k \neq y_{k'}$ for all $k$ and $k'$.

\begin{lemma}
  For any $x_1, \dots, x_n, y_1,\dots, y_{n'}$ in $\R^2$ with $x_k \neq y_{k'}$ for all $k \in [n]$ and $k' \in [n']$,
  \begin{equation}
    \label{e:massless-ferm1}
    \avga{
      \prod_{k=1}^n \bar\psi_1\psi_1(x_k) \prod_{k'=1}^{n'} \bar\psi_2\psi_2(y_{k'})
    }_{\FF(0)}
    =\mathbf 1_{n=n'} \frac{1}{(2\pi)^{2n}} \absa{\det\left(\frac{1}{x_k-y_{k'}}\right)_{k,k'=1}^n}^2
    .
  \end{equation}
\end{lemma}

\begin{proof}
  First consider $n \neq n'$. Then every term in the expansion of the determinant \eqref{e:fermdef-bis}
  that defines the left-hand side
  must contain a factor $S_{11}$ or $S_{22}$ and hence vanish.
  Let us thus assume now that $n=n'$. Then, by anticommutativity (see \eqref{e:fermantisym}),
  \begin{multline}
    \avga{\prod_{k=1}^n \bar\psi_1(x_k)\psi_1(x_k) \prod_{k=1}^n \bar\psi_2(y_{k})\psi_2(y_{k})}_{\FF(0)}
    \\ 
    =
      (-1)^n \avga{\prod_{k=1}^n \bar\psi_1(x_k)\psi_2(y_k) \prod_{k=1}^n \bar\psi_2(y_k)\psi_1(x_k)}_{\FF(0)}.
  \end{multline}
  Since $S_{11}=S_{22}=0$ the right-hand side factorizes (see \eqref{e:fermfactorize}), and by \eqref{e:fermdef-bis}
  it is hence equal to
  \begin{multline}
    (-1)^n\avga{\prod_{k=1}^n \bar\psi_1(x_k)\psi_2(y_k)}_{\FF(0)} \avga{\prod_{k=1}^n\bar\psi_2(y_k) \psi_1(x_k)}_{\FF(0)}
    \\
      = \frac{(-1)^n}{(2\pi)^{2n}} \det\left(\frac{1}{\bar x_k-\bar y_{k'}}\right)_{k,k'=1}^{n} \det\left(\frac{1}{y_k-x_{k'}}\right)_{k,k'=1}^{n}
  \end{multline}
  which gives the right-hand side of the claim.
\end{proof}

The next two lemmas then allow computing all truncated correlation functions involving also the factors $\bar\psi_2\psi_1$ and $\bar\psi_1\psi_2$.

\begin{lemma} \label{lem:fermtruncid}
  For $n+n'+q+q'\geq 2$
  and any distinct $x_1, \dots, x_n$, $y_1, \dots, y_{n'}$, $z_1, \dots, z_q$, $w_1, \dots, w_{q'}$, $z$, $w$ in $\R^2$,
  the following identities hold: 
  \begin{align} \label{e:fermtruncid1}
    &\avga{
      {\bar\psi_2\psi_1}(z)
      \prod_{j=1}^{q} {\bar\psi_2\psi_1}(z_j)
      \prod_{{j'}=1}^{q'} {\bar\psi_1\psi_2}(w_{j'})
      \prod_{{k}=1}^n \bar\psi_1\psi_1(x_{k})
      \prod_{{k'}=1}^{n'} \bar\psi_2\psi_2(y_{k'})
    }^T_{\FF(0)}
    \nnb
    &= \frac{\mathbf 1_{n=n'}}{2\pi} \sum_{{i}=1}^n
    \pa{\frac{1}{x_i-z}-\frac{1}{y_i-z}}
    \nnb &\qquad \times 
    \avga{
      \prod_{j=1}^{q} {\bar\psi_2\psi_1}(z_j)
      \prod_{{j'}=1}^{q'} {\bar\psi_1\psi_2}(w_{j'})
      \prod_{{k}=1}^n \bar\psi_1\psi_1(x_{k})
      \prod_{{k'}=1}^n \bar\psi_2\psi_2(y_{k'})
    }^T_{\FF(0)}
    ,
  \end{align}
  and
  \begin{align} \label{e:fermtruncid2}
    &\avga{
      {\bar\psi_1\psi_2}(w)
      \prod_{j=1}^{q} {\bar\psi_2\psi_1}(z_j)
      \prod_{{j'}=1}^{q'} {\bar\psi_1\psi_2}(w_{j'})
      \prod_{{k}=1}^n \bar\psi_1\psi_1(x_{k})
      \prod_{{k'}=1}^{n'} \bar\psi_2\psi_2(y_{k'})
    }^T_{\FF(0)}
    \nnb
    &= - 
    \frac{\mathbf 1_{n=n'}}{2\pi} \sum_{{i}=1}^n
    \pa{\frac{1}{\bar x_i-\bar w}-\frac{1}{\bar y_i-\bar w}}
    \nnb & \qquad \times 
    \avga{
      \prod_{j=1}^{q} {\bar\psi_2\psi_1}(z_j)
      \prod_{{j'}=1}^{q'} {\bar\psi_1\psi_2}(w_{j'})
      \prod_{{k}=1}^n \bar\psi_1\psi_1(x_{k})
      \prod_{{k'}=1}^n \bar\psi_2\psi_2(y_{k'})
    }^T_{\FF(0)}
    .
  \end{align}
  The right-hand sides are interpreted as $0$ when $n=n'=0$.
\end{lemma}

\begin{proof}
  Since the proofs of \eqref{e:fermtruncid1} and \eqref{e:fermtruncid2} are analogous, we only consider  \eqref{e:fermtruncid1}.
  By \eqref{e:fermtruncdef-bis}, when $n+n'+q+q'\geq 2$, we have
  \begin{multline} \label{e:fermcoleman-pf1}
    \avga{
	\prod_{j=1}^{q} {\bar\psi_2\psi_1}(z_j)
	\prod_{{j'}=1}^{q'} {\bar\psi_1\psi_2}(w_{j'})
	\prod_{{k}=1}^n \bar\psi_1\psi_1(x_{k})
	\prod_{{k'}=1}^{n'} \bar\psi_2\psi_2(y_{k'})
    }^T_{\FF(0)}
    \\
    =
    (-1)^{n+n'+q+q'+1} \sum_{\pi \in C_{n+n'+q+q'}} \prod_{i=1}^{n+n'+q+q'}
    S_{\alpha_{\pi^i(1)}\beta_{\pi^{i+1}(1)}}(u_{\pi^i(1)},u_{\pi^{i+1}(1)})
  \end{multline}
  where we have defined
  \begin{equation}
    (\alpha_i,\beta_i,u_i) =\begin{cases}
      (1,1,x_i) & (1 \leq i \leq n)\\
      (2,2,y_{i-n}) & (n < i \leq n+n')\\
      (2,1, z_{i-n-n'}) & (n+n' < i \leq n+n'+q)\\
      (1,2, w_{i-n-n'-q}) & (n+n'+{q}< i \leq n+n'+q+q').
    \end{cases}
  \end{equation}
  By \eqref{e:S0-bis-sec2}, all terms that contain a factor $S_{11}$ or $S_{22}$ vanish.
  Therefore it is necessary that the number of factors of $\bar\psi_1$ equals that of $\psi_2$,
  which implies that $n=n'$ if \eqref{e:fermcoleman-pf1} is nonzero which we thus assume from now on.
  The truncated correlation function
  \begin{equation} \label{e:fermcoleman-pf2}
    \avga{
    {\bar\psi_2\psi_1}(z)
    \prod_{j=1}^{q} {\bar\psi_2\psi_1}(z_j)
    \prod_{{j'}=1}^{q'} {\bar\psi_1\psi_2}(w_{j'})
    \prod_{{k}=1}^n \bar\psi_1\psi_1(x_{k})
    \prod_{{k'}=1}^{n'} \bar\psi_2\psi_2(y_{k'})
}^T_{\FF(0)}
  \end{equation}
  is given by replacing one of the factors $S_{\alpha\beta'}(u,u')$ in \eqref{e:fermcoleman-pf1} by
  \begin{equation}
    -S_{\alpha 1}(u,z) S_{2 \beta'}(z,u')
  \end{equation}
  and then summing over the choice of which factor gets replaced.
  Using again $S_{11}=S_{22}=0$,
  the last term vanishes unless $(\alpha,\beta')= (2,1)$, and in this case,
  \begin{align}
    -S_{21}(u,z) S_{21}(z,u')
    =
    \frac{1}{(2\pi)^2(u-z)(u' - z)}
    &=
    \frac{1}{(2\pi)^2} \pa{\frac{1}{u- z} - \frac{1}{u' - z}} \frac{1}{u'-u}
    \nnb
    &=
      \frac{1}{2\pi}\pa{\frac{1}{u' -z}-\frac{1}{u-z}}
      S_{21}(u,u').
  \end{align}
  Thus the replacement of the factor $S_{21}(u,u')$ is equivalent to multiplying it by
  \begin{equation}
    \frac{1}{2\pi}\pa{\frac{1}{u' -z}-\frac{1}{u-z}}.
  \end{equation}
  The possibilities for $(u,u')$ that are compatible with the constraint $(\alpha,\beta')=(2,1)$ are
  \begin{equation}
    (u,u') = (y_i,x_j), \quad (u,u')=(y_i,z_k), \quad  (u,u')=(z_k,x_j) ,  \quad  (u,u')=(z_k,z_l),
  \end{equation}
  for some $i,j \in [n]$ and $k,l \in [q]$ with $k \neq l$.
  In these cases we obtain factors of, respectively,
  \begin{equation}
    \begin{gathered}
      \frac{1}{2\pi}\pa{ \frac{1}{x_j-z}-\frac{1}{y_i-z}  }
      ,\qquad
      \frac{1}{2\pi}    \pa{  \frac{1}{z_l-z} -\frac{1}{z_k-z}}
      ,\\
      \frac{1}{2\pi}\pa{ \frac{1}{z_k-z} -\frac{1}{y_i-z}}
      ,\qquad
      \frac{1}{2\pi}    \pa{ \frac{1}{x_j-z} - \frac{1}{z_k-z}}
      .
    \end{gathered}
  \end{equation}
  In the sum over cycles in \eqref{e:fermcoleman-pf1}, we may restrict to cycles which give a nonvanishing contribution,
  and we will do this in the following.
  Then by symmetry, given any pair $(i,j) \in [n]^2$, the proportion $r$ of such cycles giving the factor
  $S_{21}(y_i,x_j)$ is independent of $(i,j)$;
  given any pair $(i,k) \in [n] \times [q]$, the proportion $s$ of such cycles giving the factor $S_{21}(y_i,z_k)$
  is independent of $(i,k)$ and the same as the proportion of cycles giving the factor $S_{21}(z_k,x_i)$;
  and given any pair $(k,l) \in [q]^2$ with $k \neq l$ the proportion $t$ of cycles giving the factor $S_{21}(z_k,z_l)$
  is independent of $(k,l)$.
  Therefore \eqref{e:fermcoleman-pf2} is obtained from \eqref{e:fermcoleman-pf1} by multiplication with
  $1/2\pi$ and 
  \begin{multline}
    r \sum_{i,j} \pa{ \frac{1}{x_j-z} - \frac{1}{y_i-z}}
    +s\sum_{i,k} \pa{ \frac{1}{z_k-z}  - \frac{1}{y_i-z} }
    +s\sum_{i,k}\pa{ \frac{1}{x_i-z} - \frac{1}{z_k-z}}
    \\
    +t \sum_{k,l} \pa{\frac{1}{z_l-z}- \frac{1}{z_k-z}  }
    = (rn+sq) \sum_{i}    \pa{\frac{1}{x_i-z} -\frac{1}{y_i-z}}
    .
  \end{multline}
  Since for any cycle $\pi$ with nonvanishing contribution,
  each of the points $y_i$ must appear once as the first argument of $S_{\alpha\beta}$ (and then necessarily $\alpha=2$)
  and each $x_i$ once as the second argument of $S_{\alpha\beta}$ (and then necessily $\beta=1$)
  in the product in \eqref{e:fermcoleman-pf1}, we also see that $rn+sq=1$.
  Thus we have recovered \eqref{e:fermtruncid1} in the case $n+n'\neq 0$.
  
  For the case $n=n'=0$, the same argument shows that the only possibility for $u,u'$ is now $(u,u')=(z_k,z_l)$,
  and as before, this gives a zero contribution since the sum $\sum_{k,l}((z_k-z)^{-1}-(z_l-z)^{-1})$ vanishes. This concludes the proof.
\end{proof}

\begin{lemma}
  For $q+q' \geq 2$  and any distinct $z_1, \dots, z_q, w_1, \dots, w_{q'}$ in $\R^2$,
  \begin{equation} \label{e:fermtruncid3}
    \avga{
      \prod_{j=1}^{q} {\bar\psi_2\psi_1}(z_j)
      \prod_{j'=1}^{q'} {\bar\psi_1\psi_2}(w_{j'})
    }^T_{\FF(0)}
    =
    \begin{cases}
          \frac{1}{(2\pi)^2(z_1-z_2)^2} & (q=2,q'=0)\\
          \frac{1}{(2\pi)^2(\bar w_1 -\bar w_2)^2} & (q=0, q'=2)
      \\
      0 & \text{else}.
    \end{cases}
  \end{equation}
\end{lemma}

\begin{proof}
  Lemma~\ref{lem:fermtruncid} implies that the left-hand side is $0$ when $q+q' > 2$.
  In the case $(q,q')=(1,1)$, any of the products in \eqref{e:fermtruncdef-bis} must contain
  factors $S_{11}$ or $S_{22}$, and thus vanish as well. In the case $(q,q')=(2,0)$, by \eqref{e:fermtruncdef-bis}, we get
  \begin{equation}
    - S_{21}(z_1,z_2)S_{21}(z_2,z_1) = - \frac{1}{(2\pi)^2} \frac{1}{z_1-z_2} \frac{1}{z_2-z_1}
    = \frac{1}{(2\pi)^2}\frac{1}{(z_1-z_2)^2}.
  \end{equation}
  The case $(q,q')=(0,2)$ is analogous.
\end{proof}

\subsection{Bosonic side: free field correlations}
\label{sec:massless-gff}

For the computation of the free field correlations,
we first recall that, for $\epsilon>0$ and $m>0$, our regularized GFF is the centered Gaussian field with covariance
\begin{equation}
  \int_{\epsilon^2}^\infty ds \, e^{-s(-\Delta +m^2)}(x,y)
  = \int_{\epsilon^2}^{\infty} ds\, \frac{e^{-\frac{|x-y|^2}{4s}}}{4\pi s} e^{-m^2 s}
  .
  \label{e:cov2}
\end{equation}
We write $\nu^{\GFF(\epsilon,m)}$ for the (centered) Gaussian measure with this covariance.
It is a basic fact that this measure is supported on smooth functions
and that the covariance of the derivatives of the field is given by the derivatives of the covariance,
see e.g. \cite[Appendix B]{MR3339158}.
We also recall the definition
\begin{align}
  \wick{e^{\pm i \sqrt{\beta} \varphi(x)}}_\epsilon :=
  \epsilon^{-\beta/4\pi} e^{\pm i \sqrt{\beta} \varphi(x)}.
\end{align}
Our goal is to compute the truncated correlation functions
\begin{align} \label{e:gffcfdef}
  &\avga{
    \prod_{j=1}^{q} \partial \varphi(z_j)
    \prod_{j'=1}^{q'} \bar\partial \varphi(w_{j'})
    \prod_{k=1}^n \wick{e^{+i\sqrt{\beta}\varphi(x_{k})}}
    \prod_{k'=1}^{n'} \wick{e^{-i\sqrt{\beta}\varphi(y_{k'})}}
  }_{\GFF}^T
  \\ \notag
  &:=
  \lim_{m\to 0}\lim_{\epsilon\to 0}
   \avga{
     \prod_{j=1}^{q} \partial \varphi(z_j)
    \prod_{j'=1}^{q'} \bar\partial \varphi(w_{j'})
    \prod_{k=1}^n \wick{e^{+i\sqrt{\beta}\varphi(x_k)}}_\epsilon
    \prod_{k'=1}^{n'} \wick{e^{-i\sqrt{\beta}\varphi(y_{k'})}}_\epsilon
  }_{\GFF(\epsilon,m)}^T
\end{align}
as well as smeared versions of them.

The following estimates for the covariance of $\varphi$ and its derivatives will be useful.
(As before, $\gamma$ is the Euler--Mascheroni constant.)

\begin{lemma} \label{lem:cov}
  Uniformly on compact subsets of $x \neq y \in \R^2$, as $\epsilon \to 0$,
  \begin{align}
    \avg{\varphi(x)^2}_{\GFF(\epsilon,m)}
    + \frac{1}{2\pi} \log \epsilon
    &\to -\frac{1}{2\pi} \log m - \frac{\gamma}{4\pi},\label{e:var}
    \\
    \avg{\varphi(x)\varphi(y)}_{\GFF(\epsilon,m)}
    &\to -\frac{1}{2\pi} \log m - \frac{1}{2\pi} \log \frac{|x-y|}{2} - \frac{\gamma}{2\pi} + O(m|x-y|).\label{e:covest}
  \end{align}
  Moreover, uniformly on compact sets of $x\neq y$, as $\epsilon\to 0$ and then $m\to 0$,
  \begin{equation} \label{e:cderiv1lim}
    - \avg{\partial\varphi(x)\varphi(y)}_{\GFF(\epsilon,m)}
    =  \avg{\varphi(x)\partial\varphi(y)}_{\GFF(\epsilon,m)}
    \to \frac{1}{4\pi} \frac{1}{x-y},
  \end{equation}
  and
  \begin{equation} \label{e:cderivlim}
    \avg{\partial\varphi(x)\partial\varphi(y)}_{\GFF(\epsilon,m)}
    \to -\frac{1}{4\pi} \frac{1}{(x-y)^2},
    \qquad
    \avg{\partial\varphi(x)\bar\partial\varphi(y)}_{\GFF(\epsilon,m)}
    \to 0.
  \end{equation}
  Moreover, for any $g \in L^\infty_c(\R^2)$, uniformly in compact subsets of $u \in \R^2$,
  \begin{equation} \label{e:phigradphi}
    \avga{\varphi(u)\partial \varphi(g)}_{\GFF(\epsilon,m)}
    \to -\int_{\R^2}dx g(x)\frac{1}{4\pi}\frac{1}{x-u},
  \end{equation}
  and for all $f,g \in L^\infty_c(\R^2)$ with disjoint supports,
  \begin{equation}
  	\label{e:ddsmear}
    \avga{\partial \varphi(f)\partial \varphi(g)}_{\GFF(\epsilon,m)} \to
    -\frac{1}{4\pi}\int_{\R^2} dx\, dy\, f(x)\, g(y) \, \frac{1}{(x-y)^2}
    .
  \end{equation}
  Finally, for any $f \in L^\infty_c(\R^2)$ with $\int f\, dx=0$, uniformly on compact subsets of $x\in \R^2$,
  \begin{equation}
    \avg{\varphi(x)\varphi(f)}_{\GFF(\epsilon,m)} \to -\int_{\R^2} dy\, \frac{1}{2\pi} \log|x-y|\, f(y) .
    \label{e:csmear}
  \end{equation}
  The limits above also exist when $\epsilon\to 0$ with $m>0$ fixed and have the same local uniformity.
\end{lemma}

\begin{proof}[Proof]
  The estimates here are largely routine, so we sketch the main ideas and leave the full details to the reader.
  Let us consider separately pointwise estimates and smeared estimates. 
	 
 \medskip 
	 
 \underline{Pointwise estimates:} 
 For \eqref{e:var}, we note that by definition, 
  \begin{align}
    \avga{\varphi(x)^2}_{\GFF(\epsilon,m)}
    = \int_{\epsilon^2}^\infty dt\, \frac{e^{-m^2 t}}{4\pi t}
    & 
      = \int_{m^2 \epsilon^2}^\infty dt\, \frac{e^{-t}}{4\pi t}
    \nnb  
    &=
      \frac{1}{4\pi}\Gamma(0,m^2\epsilon^2)
      \nnb
    &=
    -\frac{1}{2\pi}\log(m\epsilon) - \frac{\gamma}{4\pi} + O(\epsilon^2 m^2),
  \end{align}
  where $\Gamma$ is the incomplete Gamma function and we used its well-known asymptotics.
  Similarly, for \eqref{e:covest}, we note that as $\epsilon \to 0$,
  \begin{align}
    \avga{\varphi(x)\varphi(y)}_{\GFF(\epsilon,m)}
    & 
    \to \int_0^\infty dt\, \frac{e^{-|x-y|^2/4t}}{4\pi t}e^{-m^2 t}
    \nnb 
    &=  \int_0^\infty dt\, \frac{e^{-(m |x-y|/2)(t+1/t)}}{4\pi t} 
      \nnb
    &=  \frac{1}{2\pi} K_0(m |x-y|)
      \nnb
    &= -\frac{1}{2\pi} \log \frac{m|x-y|}{2} - \frac{\gamma}{2\pi} + O(m|x-y|),
  \end{align}
  where $K_0$ the modified Bessel function of the second kind and we used its well-known asymptotics.
  The proofs of \eqref{e:cderiv1lim} and \eqref{e:cderivlim} are similar,
  and make use of standard asymptotics of Bessel functions -- we omit further details.
  
  \medskip 
  
  \underline{Smeared estimates:} Consider next \eqref{e:phigradphi}.
  For $g\in L_c^\infty(\R^2)$ and $u \in \R^2$, we have
  \begin{align}
    \avga{\varphi(u)\nabla \varphi(g)}_{\GFF(\epsilon,m)}&=\int_{\R^2}dx \, g(x)\int_{\epsilon^2}^{\infty}ds\left(-\frac{x-u}{2s}\right)\frac{e^{-\frac{|x-u|^2}{4s}}}{4\pi s}e^{-m^2s}\nnb
                                                         &=-\int_{\R^2}dy \, ye^{-|y|^2/4}\int_{\epsilon^2}^{\infty}ds \frac{e^{-m^2 s}}{8\pi \sqrt{s}} g(u+\sqrt{s}y).
  \end{align}
  Thus if $u\in \R^2$ is in some fixed compact set, say a disc of radius $r_1$ and we choose that $r_2>0$ is such that $\mathrm{supp}(g)\subset B(0,r_2)$ (where both $r_1,r_2$ are fixed in $\epsilon,m$), then one readily checks via the triangle inequality that $|g(u+\sqrt{s}y)|\leq \|g\|_{L^\infty(\R^2)}\mathbf 1\{s\leq (r_1+r_2)^2/|y|^2\}$. Applying this type of bound in the above integral representation, it follows that as $\epsilon,m\to 0$, $\avga{\varphi(u)\nabla \varphi(g)}_{\GFF(\epsilon,m)}$	converges  uniformly in $u$ in a fixed compact set.
  On the other hand, this type of estimate can readily be used to justify the use of the dominated convergence theorem so using \eqref{e:cderiv1lim}, we see that in fact as $\epsilon,m\to 0$,
  \begin{equation} \label{e:phigradphi_old}
    \avga{\varphi(u)\nabla \varphi(g)}_{\GFF(\epsilon,m)}\to -\int_{\R^2}dx g(x)\frac{1}{2\pi}\frac{x-u}{|x-u|^2},
  \end{equation}
  and that this is a locally bounded function of $u$.
  
  The bound \eqref{e:ddsmear} follows directly from \eqref{e:cderivlim}, while  \eqref{e:csmear} follows from \eqref{e:covest} through similar estimates as above (and making use of our assumption that $\int f=0$). This concludes our proof.
\end{proof}        

Next, we record a basic estimate for the charge correlation functions.

\begin{lemma} 
  \label{le:GFFcharge}
  For any $\beta>0$ and any distinct $x_1, \dots, x_n, y_1, \dots, y_{n'}$ in $\R^2$, where $n+n' \geq 1$,
  the limits
  \begin{align}
    &\avga{
    \prod_{k=1}^n \wick{e^{+i\sqrt{\beta}\varphi(x_k)}}
    \prod_{k'=1}^{n'} \wick{e^{-i\sqrt{\beta}\varphi(y_{k'})}}
      }_{\GFF(m)}
      \nnb
    &=\lim_{\epsilon \to 0}
    \avga{
    \prod_{k=1}^n \wick{e^{+i\sqrt{\beta}\varphi(x_k)}}_{\epsilon}
    \prod_{k'=1}^{n'} \wick{e^{-i\sqrt{\beta}\varphi(y_{k'})}}_{\epsilon}
    }_{\GFF(\epsilon,m)}
    \\
    &\avga{
    \prod_{k=1}^n \wick{e^{+i\sqrt{\beta}\varphi(x_k)}}
    \prod_{k'=1}^{n'} \wick{e^{-i\sqrt{\beta}\varphi(y_{k'})}}
      }_{\GFF}
      \nnb
    &=\lim_{m \to 0} \lim_{\epsilon \to 0}
    \avga{
    \prod_{k=1}^n \wick{e^{+i\sqrt{\beta}\varphi(x_k)}}_{\epsilon}
    \prod_{k'=1}^{n'} \wick{e^{-i\sqrt{\beta}\varphi(y_{k'})}}_{\epsilon}
  }_{\GFF(\epsilon,m)}
\end{align}
exist, and
  \begin{align}
    \label{e:gff-ee}
    &\avga{
    \prod_{k=1}^n \wick{e^{+i\sqrt{\beta}\varphi(x_k)}}
    \prod_{k'=1}^{n'} \wick{e^{-i\sqrt{\beta}\varphi(y_{k'})}}
      }_{\GFF}
      \nnb
      &=
      \mathbf 1_{n=n'} 
      (4e^{-\gamma})^{\beta n/4\pi}
      \frac{\prod_{i<j} |x_i-x_j|^{\beta/2\pi} |y_i-y_j|^{\beta/2\pi}}{\prod_{i,j}|x_i-y_j|^{\beta/2\pi}}
    \end{align}
    where the empty product $\prod_{i<j}$ is interpreted as $1$ if $n=n'=1$.
\end{lemma}

\begin{proof}
  Since $\varphi$ is Gaussian under $\nu^{\GFF(\epsilon,m)}$ with covariance $c(x,y) = \avg{\varphi(x)\varphi(y)}_{\GFF(\epsilon,m)}$,
  \begin{align}
    &\avga{
    \prod_{k=1}^n \wick{e^{+i\sqrt{\beta}\varphi(x_k)}}
    \prod_{k'=1}^{n'} \wick{e^{-i\sqrt{\beta}\varphi(y_{k'})}}
      }_{\GFF(\epsilon,m)}
      \nnb
    & \qquad = 
      \epsilon^{-(n+n') (\beta/4\pi)}e^{-\frac{\beta}{2}\qa{ \sum_{i,j=1}^{n} c(x_i,x_j) + \sum_{i,j=1}^{n'}c(y_i,y_j) -
      2 \sum_{i=1}^n\sum_{j=1}^{n'} c(x_i,y_j) }}
    \nnb
    &
      \qquad =
      (\epsilon^{-1/2\pi} e^{-c(0,0)})^{\beta (n+n')/2}
      e^{-\beta\qa{ \sum_{i<j} c(x_i,x_j) + \sum_{i<j} c(y_i,y_j) -
      \sum_{i,j} c(x_i,y_j) }}
      .
  \end{align}
  By Lemma~\ref{lem:cov}, the limits $\epsilon\to 0$ and $m\to 0$ both exist, and the $\epsilon,m\to 0$ limit is given by
  \begin{multline}
    \lim_{m\downarrow 0}
    m^{(\beta/4\pi)(n-n')^2}
    e^{(\beta \gamma/4\pi) (n+n')/2} (2^{\beta/2\pi} e^{- \gamma \beta /2\pi})^{n}
    \frac{\prod_{i<j} |x_i-x_j|^{\beta/2\pi} |y_i-y_j|^{\beta/2\pi}}{\prod_{i,j}|x_i-y_j|^{\beta/2\pi}}
    \\
    = \mathbf 1_{n=n'} (4e^{-\gamma})^{\beta n/4\pi}
      \frac{\prod_{i<j} |x_i-x_j|^{\beta/2\pi} |y_i-y_j|^{\beta/2\pi}}{\prod_{i,j}|x_i-y_j|^{\beta/2\pi}}
    \end{multline}
    as claimed.
\end{proof}

By definition,
the truncated correlation functions of $\wick{e^{\pm i\sqrt{\beta}\varphi}}$
are determined by  \eqref{e:gff-ee} and \eqref{e:truncated}.
The next two lemmas give the general truncated correlations also
involving factors $\partial \varphi$ or $\bar\partial \varphi$.

\begin{lemma}\label{le:GFFmixed}
  Let $\beta>0$.
  For $n \geq 1$, $q,q'\geq 0$, $x_1,...,x_n,z_1,...,z_q,w_1,...,w_{q'}\in \R^2$ distinct, and $\sigma_1,...,\sigma_n\in \{-1,1\}$,
  the limits
  \begin{align}
    &\avga{\prod_{k=1}^n \wick{e^{i\sqrt{\beta}\sigma_k\varphi(x_k)}}\prod_{j=1}^{q} \partial \varphi(z_j)\prod_{j'=1}^{q'} \bar \partial \varphi(w_{j'})}_{\GFF(m)}^T
    \nnb
    &\qquad
      = \lim_{\epsilon \to 0}
      \avga{\prod_{k=1}^n \wick{e^{i\sqrt{\beta}\sigma_k\varphi(x_k)}}_{\epsilon}\prod_{j=1}^{q} \partial \varphi(z_j)\prod_{j'=1}^{q'} \bar \partial \varphi(w_{j'})}_{\GFF(\epsilon,m)}^T
    \\
        &\avga{\prod_{k=1}^n \wick{e^{i\sqrt{\beta}\sigma_k\varphi(x_k)}}\prod_{j=1}^{q} \partial \varphi(z_j)\prod_{j'=1}^{q'} \bar \partial \varphi(w_{j'})}_{\GFF}^T
    \nnb
    &\qquad= \lim_{m\to 0}\lim_{\epsilon \to 0}
        \avga{\prod_{k=1}^n \wick{e^{i\sqrt{\beta}\sigma_k\varphi(x_k)}}_{\epsilon}\prod_{j=1}^{q} \partial \varphi(z_j)\prod_{j'=1}^{q'} \bar \partial \varphi(w_{j'})}_{\GFF(\epsilon,m)}^T
      \end{align}
      exist uniformly on compact subsets of $u_i\neq u_j$ for $i\neq j$ (where the $u_i$ are an enumeration of the points $x_k,z_j,w_{j'}$), and
  we have
  \begin{multline} \label{e:GFFmixed}
    \avga{\prod_{k=1}^n \wick{e^{i\sqrt{\beta}\sigma_k\varphi(x_k)}}\prod_{j=1}^{q} \partial \varphi(z_j)\prod_{j'=1}^{q'} \bar \partial \varphi(w_{j'})}_{\GFF}^T
    \\
    =      \avga{\prod_{k=1}^n \wick{e^{i\sqrt{\beta}\sigma_k\varphi(x_k)}}}_\GFF^T \prod_{j=1}^q\left(i\frac{\sqrt{\beta}}{4\pi}\sum_{k=1}^n\frac{\sigma_k}{x_k-z_j}\right) \prod_{j'=1}^{q'}\left(i\frac{\sqrt{\beta}}{4\pi}\sum_{k=1}^n\frac{\sigma_k}{\bar x_k-\bar w_{j'}}\right)
    .
  \end{multline}
\end{lemma}

\begin{proof}
  By Lemma~\ref{lem:cumulants},
  when $\epsilon,m>0$, the truncated correlation functions are given by
  \begin{align} \label{e:gff-cumulantslogexp}
    &\avga{\prod_{k=1}^n \wick{e^{i\sqrt{\beta}\sigma_k\varphi(x_k)}}_\epsilon\prod_{j=1}^q \partial \varphi(z_j)\prod_{j'=1}^{q'} \bar \partial \varphi(w_{j'})}_{\GFF(\epsilon,m)}^T
    \\ 
    &=\prod_{k=1}^n \left. \frac{\partial}{\partial \mu_k}\right|_{\mu_k=0}
    \prod_{j=1}^{q} \left.\frac{\partial}{\partial \nu_j}\right|_{\nu_j=0}
    \prod_{j'=1}^{q'} \left.\frac{\partial}{\partial \eta_{j'}}\right|_{\eta_{j'}=0}
    \nnb & \qquad
    \log \avga{\exp\left[\sum_{k=1}^n\mu_k\wick{e^{i\sqrt{\beta}\sigma_k\varphi(x_k)}}_\epsilon
        +\sum_{j=1}^q\nu_j\partial \varphi(z_j)
           +\sum_{j'=1}^{q'} \eta_{j'} \bar\partial \varphi(w_{j'})\right]}_{\GFF(\epsilon,m)}.
           \nonumber
  \end{align}
  We would like to use the Girsanov--Cameron--Martin theorem
  to get rid of the $\partial \varphi$ and $\bar \partial \varphi$ terms at the expense of replacing $\mu_k$ by something which depends on $\nu_j$, $\eta_{j'}$, $z_j$, and $w_{j'}$ as well.
  We need to be slightly careful here as $\partial \varphi$ and $\bar \partial \varphi$ are complex valued, and Girsanov's theorem holds a priori only for real-valued Gaussian random variables.
  
  To justify the use of Girsanov's theorem in our setting, assume we have some real-valued Gaussian random variables $X_1,...,X_N$ and (possibly complex) constants $\gamma_1,...,\gamma_N$. Then by a routine combination of  the dominated convergence theorem (to justify continuity), Fubini's theorem, and Morera's theorem, one finds that 
  \begin{equation}
  (\lambda_1,...,\lambda_N)\mapsto \avga{e^{\sum_{j=1}^N \gamma_j e^{i X_j}} e^{\sum_{j=1}^N \lambda_j X_j}}
  \end{equation}
  is an entire function.
  Then by an elementary version of Girsanov's theorem for finite dimensional Gaussian vectors
  (which is just completion of the square and change of variables),
  we find for real $\lambda_i$ that 
  \begin{equation}
  \avga{e^{\sum_{j=1}^N \gamma_j e^{i X_j}} e^{\sum_{j=1}^N \lambda_j X_j}}=\avga{e^{\sum_{j=1}^N \gamma_j e^{i X_j+i \sum_{k=1}^N \lambda_k\avga{X_jX_k}}}}e^{\frac{1}{2}\avga{\left(\sum_{j=1}^N\lambda_j X_j\right)^2}}.
  \end{equation}
  Using a similar argument as before, one checks that this also defines an entire function of the $\lambda_i$, so as these entire functions agree on real values, they must be the same:
  \begin{equation}
  \avga{e^{\sum_{j=1}^N \gamma_j e^{i X_j}} e^{\sum_{j=1}^N \lambda_j X_j}}=\avga{e^{\sum_{j=1}^N \gamma_j e^{i X_j+i\sum_{k=1}^N \lambda_k\avga{X_jX_k}}}}e^{\frac{1}{2}\avga{\left(\sum_{j=1}^N\lambda_j X_j\right)^2}},
  \end{equation}
  also for complex $\lambda_i$.
  
  Applying this to our setting (taking $X_k$ to consist of $\varphi(x_k)$ and the real and imaginary parts of $\partial \varphi(z_j)$ and $\bar \partial \varphi(w_{j'})$ -- the values of $\gamma_i$ and $\lambda_i$ are chosen accordingly), we see that the expectation on the right-hand side of \eqref{e:gff-cumulantslogexp} equals (when the expectation is non-zero -- this is true at least for small enough parameter values, and in the end, we evaluate derivatives at zero)
  \begin{align}
    &\left\langle \exp\left[\sum_{k=1}^n\mu_k\wick{e^{i\sqrt{\beta}\sigma_k\varphi(x_k)}}_\epsilon
      \right.\right.\nnb & \qquad\qquad \times \left.\left. 
      e^{i\sqrt{\beta}\sigma_k\left(\sum_{j=1}^q\nu_j\avg{\varphi(x_k)\partial \varphi(z_j)}_{\GFF(\epsilon,m)}
        +\sum_{j'=1}^{q'}\eta_{j'} \avg{\varphi(x_k)\bar\partial \varphi(w_{j'})}_{\GFF(\epsilon,m)}\right)}\right] \right\rangle_{\GFF(\epsilon,m)}
        \nnb
     &\times \exp\qa{\frac{1}{2}\avga{\left(\sum_{j=1}^{q} \nu_j \partial \varphi(z_j)+\sum_{j'=1}^{q'}\eta_{j'}\bar \partial \varphi(w_{j'})\right)^2}_{\GFF(\epsilon,m)}}.
  \end{align}
  The last term does not contribute when we take derivatives with respect to $\mu_k$ so we can ignore it.
  Therefore, using the last identity and rewriting the result in terms of the truncated charge correlations
  given by \eqref{e:gff-cumulantslogexp} with $q=q'=0$,
  \begin{align}
    &\prod_{k=1}^n \left. \frac{\partial}{\partial \mu_k}\right|_{\mu_k=0}\log \avga{\exp\left[\sum_{k=1}^n\mu_k\wick{e^{i\sqrt{\beta}\sigma_k\varphi(x_k)}}_\epsilon+\sum_{j=1}^q\nu_j\partial \varphi(z_j)+\sum_{j'=1}^{q'} \eta_{j'} \bar\partial \varphi(w_{j'})\right]}_{\GFF(\epsilon,m)}
    \nnb
    &\qquad =\avga{\prod_{k=1}^n \wick{e^{i\sqrt{\beta}\sigma_k \varphi(x_k)}}_\epsilon}_{\GFF(\epsilon,m)}^T
    \\\notag
    &\qquad \qquad \times \prod_{k=1}^n e^{i\sqrt{\beta}\sigma_k\left(\sum_{j=1}^q\nu_j\avga{\varphi(x_k)\partial \varphi(z_j)}_{\GFF(\epsilon,m)}
      +\sum_{j'=1}^{q'}\eta_{j'} \avga{\varphi(x_k)\bar\partial \varphi(w_{j'})}_{\GFF(\epsilon,m)}\right)}.
  \end{align}
  Thus, carrying out the $\nu_j$ and $\eta_{j'}$ differentiations,  we obtain
  \begin{align}
    &\avga{\prod_{k=1}^n \wick{e^{i\sqrt{\beta}\sigma_k\varphi(x_k)}}_\epsilon
      \prod_{j=1}^q \partial \varphi(z_j)\prod_{j'=1}^{q'} \bar \partial \varphi(w_{j'})}_{\GFF(\epsilon,m)}^T
      \notag \\
    &\qquad =\avga{\prod_{k=1}^n \wick{e^{i\sqrt{\beta}\sigma_k \varphi(x_k)}}_\epsilon}_{\GFF(\epsilon,m)}^T
      \times \prod_{j=1}^q \left(i\sqrt{\beta}\sum_{k=1}^{n}\sigma_k \avga{\varphi(x_k)\partial \varphi(z_j)}_{\GFF(\epsilon,m)}\right)\notag \\
    &\qquad \qquad  \times \prod_{j'=1}^{q'} \left(i\sqrt{\beta}\sum_{k=1}^n\sigma_k \avga{\varphi(x_k)\bar\partial \varphi(w_{j'})}_{\GFF(\epsilon,m)}\right). \label{e:truncem}
  \end{align}
  Using the covariance estimate \eqref{e:cderiv1lim} (and its complex conjugate version),
  we obtain \eqref{e:GFFmixed} by taking $\epsilon\to 0$ and $m \to 0$.
\end{proof}

\begin{lemma}
  For $q+q' \geq 1$ and $z_1, \dots, z_q, w_1, \dots, w_{q'} \in \R^2$ distinct,
  the limits
    \begin{align}
      \avga{\prod_{j=1}^q \partial \varphi(z_j)\prod_{j'=1}^{q'} \bar \partial \varphi(w_{j'})}_{\GFF(m)}^T
      &= \lim_{\epsilon\to 0}
        \avga{\prod_{j=1}^q \partial \varphi(z_j)\prod_{j'=1}^{q'} \bar \partial \varphi(w_{j'})}_{\GFF(\epsilon,m)}^T
        \\
      \avga{\prod_{j=1}^q \partial \varphi(z_j)\prod_{j'=1}^{q'} \bar \partial \varphi(w_{j'})}_{\GFF}^T
      &= \lim_{m\to 0}\lim_{\epsilon\to 0}
      \avga{\prod_{j=1}^q \partial \varphi(z_j)\prod_{j'=1}^{q'} \bar \partial \varphi(w_{j'})}_{\GFF(\epsilon,m)}^T
    \end{align}
  exist, and
  \begin{equation} \label{e:GFFgrad}
    \avga{\prod_{j=1}^q \partial \varphi(z_j)\prod_{j'=1}^{q'} \bar \partial \varphi(w_{j'})}_{\GFF}^T\\
    =\begin{cases}
      -\frac{1}{4\pi}\frac{1}{(z_1-z_2)^2} & (q=2,q'=0)\\
      -\frac{1}{4\pi}\frac{1}{(\bar w_1-\bar w_2)^2} & (q=0, q'=2)\\
      0 & \text{else}.
    \end{cases}
  \end{equation}
   \label{le:gradcor}
\end{lemma}

\begin{proof}
  Since $\partial\varphi(z_j)$ and $\bar\partial\varphi(w_{j'})$ are Gaussian variables,
  only the second order cumulants (truncated correlation functions) are non-wero and given by the covariance
  \begin{align}
    \avga{\prod_{j=1}^q \partial \varphi(z_j)\prod_{j'=1}^{q'} \bar \partial \varphi(w_{j'})}_{\GFF(\epsilon,m)}^T
    =\begin{cases}
      \avga{ \partial \varphi(z_1) \partial \varphi(z_2)}_{\GFF(\epsilon,m)} & (q=2,q'=0)\\
      \avga{ \bar\partial \varphi(w_1) \bar\partial \varphi(w_2)}_{\GFF(\epsilon,m)} & (q=0,q'=2)\\
      \avga{ \partial \varphi(z_1) \bar\partial \varphi(w_1)}_{\GFF(\epsilon,m)} & (q=1,q'=1)\\
      0 & \text{else}.
    \end{cases}\label{e:truncgrad}
  \end{align}
  Their limits as $\epsilon \to 0$ and $m\to 0$ are given by \eqref{e:cderivlim} (and its complex conjugate version).
\end{proof}
 
We are ultimately interested in smeared correlation functions, and there is some care to be taken on the diagonal of the pointwise correlation functions. The following result describes what happens with the truncated charge correlation functions.

\begin{lemma} \label{lem:gffintegrability1}
  For $\beta \in (0,6\pi)$ and $n \neq 2$,
  the truncated charge correlations are in $L^1_{\mathrm{loc}}((\R^2)^{n})$.
  Namely, for any $\sigma_1,...,\sigma_n\in \{-1,1\}$ and 
  any compact set $K\subset (\R^2)^{n}$, 
  \begin{equation}
    \int_{K}dx_1\cdots dx_n\left|\avga{\prod_{k=1}^n\wick{e^{i\sqrt{\beta}\sigma_k\varphi(x_k)}}}_\GFF^T \right|<\infty.
  \end{equation}
  Moreover, if $f_1,...,f_n\in L^\infty_c(\R^2\times\{-1,1\})$, then for $n\neq 2$,
  \begin{align}
          &\lim_{m\to 0}\lim_{\epsilon\to 0}\int_{(\R^2\times \{-1,1\})^n}d\xi_1\cdots d\xi_n \, f_1(\xi_1)\cdots f_n(\xi_n)\, \avga{\prod_{k=1}^n\wick{e^{i\sqrt{\beta}\sigma_k\varphi(x_{k})}}_\epsilon}_{\GFF(\epsilon,m)}^T\nnb
        &\qquad  =\int_{(\R^2\times \{-1,1\})^n}d\xi_1\cdots d\xi_n \, f_1(\xi_1)\cdots f_n(\xi_n) \, \avga{\prod_{k=1}^n\wick{e^{i\sqrt{\beta}\sigma_{k}\varphi(x_{k})}}}_{\GFF}^T.
  \end{align}
  If $K$ and the set $\{x_k = x_{k'} \text{ for some $k \neq k'$}\}$ are disjoint and if
  the $f_k$ have disjoint supports, the statements also hold for $n=2$.
\end{lemma}

The proof of this lemma is not completely straightforward from the direct definition of the truncated charge correlation functions.
For example, in \cite[Lemma~3]{MR1672504}, the analogous statement is only shown for $\beta<4\pi$ (and the need for the statement at $\beta=4\pi$
is circumvented there by defining the sine-Gordon model with $\beta=4\pi$ in terms of the limit $\beta \uparrow 4\pi$).
For us, Lemma~\ref{lem:gffintegrability1} follows immediately as a by-product of our later analysis, and we thus postpone its proof
to Section~\ref{sec:pfconv}.

For the gradient fields, we have the following smeared analogue of Lemma~\ref{le:gradcor}.

\begin{lemma}
	\label{le:gradint}
    For $q,q'\geq 0$ with $q+q'\geq 1$ and $g_1,...,g_q,h_1,...,h_{q'}\in C_c^\infty(\R^2)$, the limits
    \begin{align}
      \avga{\prod_{j=1}^q\partial \varphi(g_j)\prod_{j'=1}^{q'}\bar \partial \varphi(h_{j'})}_{\GFF(m)}^T
      &:= \lim_{\epsilon\to 0}\avga{\prod_{j=1}^q\partial \varphi(g_j)\prod_{j'=1}^{q'}\bar \partial \varphi(h_{j'})}_{\GFF(\epsilon,m)}^T
      \\
            \avga{\prod_{j=1}^q\partial \varphi(g_j)\prod_{j'=1}^{q'}\bar \partial \varphi(h_{j'})}_{\GFF}^T
      &:=\lim_{m\to 0}\lim_{\epsilon\to 0}\avga{\prod_{j=1}^q\partial \varphi(g_j)\prod_{j'=1}^{q'}\bar \partial \varphi(h_{j'})}_{\GFF(\epsilon,m)}^T
    \end{align}
    exist, and
    \begin{multline}
      \avga{\prod_{j=1}^q\partial \varphi(g_j)\prod_{j'=1}^{q'}\bar \partial \varphi(h_{j'})}_{\GFF}^T
      \\ 
      =\begin{cases}
       \frac{1}{2\pi}\int_{\R^2\times \R^2}dx\,dy\,\partial g_1(x)\partial g_2(y) \log |x-y|^{-1} & (q=2, q'=0)\\
    	\frac{1}{2\pi}\int_{\R^2\times \R^2}dx\,dy\,\bar \partial h_1(x)\bar \partial h_2(y) \log |x-y|^{-1} & (q=0, q'=2)\\
    	\frac{1}{4}\int_{\R^2}dx\, g_1(x)h_1(x) & (q=q'=1)\\
    	0 & \text{else}	.
    	\end{cases}
      \end{multline}
      For $(q,q')=(2,0),(0,2)$, the right-hand sides are also equal to the Cauchy principal value integrals
        \begin{equation} \label{e:gffpv}
          \frac{-1}{4\pi} \mathrm{p.v.}\int dx\, dy\, \frac{g_1(x)g_2(y)}{(x-y)^2},
          \qquad
          \frac{-1}{4\pi} \mathrm{p.v.}\int dx\, dy\, \frac{h_1(x)h_2(y)}{(\bar x-\bar y)^2}.
        \end{equation}   
\end{lemma}
      
\begin{proof}
	The fact that the truncated correlation function vanishes for $q+q'=1$ or $q+q'\geq3$ follows from the fact that we are dealing with centered Gaussian random variables.
	
	We thus need to only focus on the $q+q'=2$ case. The $q=2,q'=0$ and $q'=2,q=0$ cases follow readily from \eqref{e:csmear} (note that $\int \partial g_i=\int\bar \partial h_j=0$). For the $q=q'=1$ case, we find again from \eqref{e:csmear} and integrating by parts that 
	\begin{align}
          \avga{\partial \varphi(g_1)\bar \partial \varphi(h_1)}_{\GFF(\epsilon,m)}
          &\to \frac{1}{2\pi}\int_{\R^2}dx\,dy\, \partial g_1(x)\bar \partial h_1(y)\log |x-y|^{-1}\nnb
            &=\frac{1}{4\pi}\int_{\R^2}dx\,dy\, \partial g_1(x)h_1(y)\frac{1}{\bar y-\bar x},
	\end{align}
        from which the claim follows after noting that $\partial_x\frac{1}{\pi(\bar x-\bar y)}=\delta(x-y)$.
        For smooth test functions, it is well known that \eqref{e:gffpv} follows by integration by parts.
\end{proof}

With Lemma~\ref{lem:gffintegrability1} and Lemma~\ref{le:gradint} in hand, we can describe the smeared free field correlation functions in the generality we need them.

\begin{lemma} \label{lem:gffintegrability}
  Let $\beta \in (0,6\pi)$, $n=1$ or $n\geq 3$, and $q,q'\geq 0$.
  If $f_1,...,f_n\in L_c^\infty(\R^2\times \{-1,1\})$
  and $g_1,...,g_q,h_1,...,h_{q'}\in C_c^\infty(\R^2)$,
  then 
  \begin{align}
      (\xi_1,...,\xi_n,z_1,...,z_q,w_1,...,w_{q'})&\mapsto f_1(\xi_1)\cdots f_n(\xi_n)g_1(z_1)\cdots g_{q}(z_q)h_1(w_1)\cdots h_{q'}(w_{q'})
      \nnb
    & \qquad \times \avga{\prod_{k=1}^n \wick{e^{i\sqrt{\beta}\sigma_k\varphi(x_k)}}\prod_{j=1}^q \partial \varphi(z_j)\prod_{j'=1}^{q'} \bar \partial \varphi(w_{j'})}_{\GFF}^T\nnb
    &\in L^1((\R^2\times \{-1,1\})^{n}\times (\R^2)^{q+q'}),
  \end{align}
  and
  \begin{align}
    &\lim_{m\to 0} \lim_{\epsilon \to 0}\avga{\prod_{k=1}^n \wick{e^{i\sqrt{\beta}\sigma_k\varphi}}(f_k)\prod_{j=1}^q \partial \varphi(g_j)\prod_{j'=1}^{q'} \bar \partial \varphi(h_{j'})}_{\GFF(\epsilon,m)}^T
      \nnb&=
      \int \prod_{i=1}^n d\xi_i  \, f_i(\xi_i) \prod_{j=1}^q dz_j \, g_j(z_j)\prod_{j'=1}^{q'} dw_{j'} \,h_{j'}(w_{j'})\, 
      \nnb&\qquad\qquad
      \times \avga{\prod_{k=1}^n \wick{e^{i\sqrt{\beta}\sigma_k\varphi(x_k)}}\prod_{j=1}^{q} \partial \varphi(z_j)\prod_{j'=1}^{q'} \bar \partial \varphi(w_{j'})}_{\GFF}^T
      .
  \end{align}
  Moreover, if the $f_j$ have disjoint supports or if $q+q'\geq 1$, the claims hold also for $n=2$.
\end{lemma}
  
\begin{proof}  
  In the case that $n \neq 2$ or that the $f_k$ have disjoint supports,
  the claim follows immediately from  Lemma \ref{le:GFFmixed} and Lemma~\ref{lem:gffintegrability1}.
  Thus the only slightly delicate case is the claim that the supports of $f_1$ and $f_2$ need not be disjoint for $n=2$ if $q+q'\geq 1$. For this, note first that if $\sigma_1=\sigma_2$, then the charge correlation function vanishes and there is nothing to prove.
  For $\sigma_1\neq \sigma_2$, let us only prove that the limiting quantity is integrable -- justifying convergence can be readily deduced with a similar argument.
  By \eqref{e:gff-ee}, the truncated charge two-point function function is proportional to
  \begin{equation}
    \frac{1}{|x_1-x_2|^{\beta/2\pi}},
  \end{equation}  
  and, by Lemma~\ref{le:GFFmixed}, the correlation function in the claim is thus proportional to
  \begin{equation}
    \frac{1}{|x_1-x_2|^{\beta/2\pi}}
    \prod_{j=1}^q\left(i\frac{\sqrt{\beta}}{4\pi}\left(\frac{1}{x_1-z_j}-\frac{1}{x_2-z_j}\right)\right)
    \prod_{j'=1}^{q'}\left(i\frac{\sqrt{\beta}}{4\pi}\left(\frac{1}{\bar x_1-\bar w_{j'}}-\frac{1}{\bar x_2-\bar w_{j'}}\right)\right)
    .
  \end{equation}
  It thus suffices to show that
  \begin{multline}
    \frac{1}{|x_1-x_2|^{\beta/2\pi}}
    \prod_{j=1}^q \left|\frac{1}{x_1-z_j}-\frac{1}{x_2-z_j}\right|
    \prod_{j'=1}^{q'} \left|\frac{1}{x_1-w_{j'}}-\frac{1}{x_2-w_{j'}}\right|
    \\\leq
    |x_1-x_2|^{-\beta/2\pi+{q+q'}}
    \prod_{j=1}^q\frac{1}{|z_j-x_1||z_j-x_2|}
    \prod_{j'=1}^{q'}\frac{1}{|w_{j'}-x_1||w_{j'}-x_2|}
  \end{multline}
  is locally integrable.
  One readily checks that since we are integrating over given compact sets, each $z_j$-integral gives a bound of the form $1+|\log |x_1-x_2||$
  and analogously for the $w_{j'}$ integrals. Thus it suffices to check the local integrability of
  \begin{equation}
    (1+|\log |x_1-x_2||)^{q+q'}|x_1-x_2|^{-\beta/2\pi+q+q'}.
  \end{equation}
  As we are in two dimensions, this certainly holds for $\beta<6\pi$ when $q+q' \geq 1$.
\end{proof}

\subsection{Bosonization in the massless case}

That the Coleman correspondence \eqref{e:thm-correspondence} holds in the non-interacting case $z=\mu=0$
follows
by matching the above computations of the correlation functions of massless free fermions and
of the massless Gaussian free field, together
with the following well-known identity for Cauchy--Vandermonde matrices:
\begin{equation}
  \label{e:detid1}
  \det\pa{\frac{1}{x_i-y_j}}_{i,j=1}^n
  =
  \frac{\prod_{1 \leq i<i'\leq {n}} (x_i-x_{i'}) \prod_{1 \leq j< j' \leq n} (y_j - y_{j'})}
  {\prod_{1 \leq i \leq {n}} \prod_{1 \leq j \leq n} (x_i-y_j)}
  .
\end{equation}

This allows us to prove the Coleman correspondence in the case $\mu=z=0$.

\begin{corollary} \label{cor:massless-corr}
  Let $\beta=4\pi$, $z=\mu=0$. For $n,n',q,q'\geq 0$ with $n+n'+q+q'=1$ or $n+n'+q+q'\geq 3$, $f_1^+,...,f_n^+,f_1^-,...,f_{n'}^{-}\in L_c^\infty(\R^2)$, and $g_1^+,...,g_q^+,g_1^-,...,g_{q'}^{-}\in C_c^\infty(\R^2)$,
  the identity \eqref{e:thm-correspondence} holds:
  \begin{multline} \label{e:correspondence-GFF}
  \avga{\prod_{k=1}^{n} \wick{e^{+i \sqrt{4\pi} \varphi}}(f_k^+)
    \prod_{k'=1}^{n'}\wick{e^{-i \sqrt{4\pi} \varphi}}(f_{k'}^-)
    \prod_{j=1}^{q} (-i \partial \varphi(g_j^+))
    \prod_{j'=1}^{q'} (+i\bar\partial \varphi(g_{j'}^-))}_{\GFF}^T
    \\
  = A^{n+n'} B^{q+q'}
   \avga{
    \prod_{k=1}^{n} {\bar\psi_1 \psi_1}(f_k^+)
    \prod_{k'=1}^{n'} {\bar\psi_2 \psi_2}(f_{k'}^-)
    \prod_{j=1}^{q} {\bar\psi_2\psi_1}(g_j^+)
    \prod_{j'=1}^{q'} {\bar\psi_1 \psi_2}(g_{j'}^-)
  }_{\FF(0)}^T,
\end{multline}
where $A$ and $B$ are as in Theorem~\ref{thm:correspondence}.

Moreover, if $n+n'+q+q'=2$, we have the following statements: (i) for $n+n'=2$, $q+q'=0$, the claim holds if $f_i^\pm$ have disjoint supports, (ii) if $n+n'=1$ and $q+q'=1$, the claim holds in the same generality as for $n+n'+q+q'\geq 3$ (both sides vanish), and (iii) if $n+n'=0$ and $q+q'=2$, the claim holds either if $g_j^\pm$ are disjoint supports, or if the right hand side is understood as that given by Lemma~\ref{le:gradint}.
\end{corollary}

\begin{proof}
  Let $q=q'=0$. Then
  applying \eqref{e:gff-ee} with $\beta=4\pi$,
  the determinant identity \eqref{e:detid1}, and finally \eqref{e:massless-ferm1},
  we find that for any distinct points,
  \begin{align}
    &\avga{
    \prod_{k=1}^n \wick{e^{+i\sqrt{\beta}\varphi(x_k)}}
    \prod_{k'=1}^{n'} \wick{e^{-i\sqrt{\beta}\varphi(y_{k'})}}
      }_{\GFF}
      \nnb 
    &=
      \mathbf 1_{n=n'} 
      (4e^{-\gamma})^{\beta n/4\pi}
      \frac{\prod_{i<j} |x_i-x_j|^{\beta/2\pi} |y_i-y_j|^{\beta/2\pi}}{\prod_{i,j}|x_i-y_j|^{\beta/2\pi}}
    \nnb
    &=
      \mathbf 1_{n=n'} 
      (4e^{-\gamma})^{n}
      \frac{\prod_{i<j} |x_i-x_j|^{2} |y_i-y_j|^{2}}{\prod_{i,j}|x_i-y_j|^{2}}
    \\
    &=
      (4\pi  e^{-\gamma/2})^{n+n'}
      \avga{
      \prod_{k=1}^n \bar\psi_1\psi_1(x_k) \prod_{k'=1}^{n'} \bar\psi_2\psi_2(y_{k'})
      }_{\FF(0)}
      \nonumber
      .
  \end{align}
  Using this, if $q+q'>0$ then \eqref{e:fermtruncid1}--\eqref{e:fermtruncid2} 
  and \eqref{e:fermtruncid3} for the fermionic side
  respectively \eqref{e:GFFmixed} for the bosonic side
  imply that, for any distinct points,
\begin{align} \label{e:GFF-correspondence}
  &\avga{\prod_{k=1}^{n} \wick{e^{+i \sqrt{4\pi} \varphi(x_k)}}
    \prod_{k'=1}^{n'}\wick{e^{-i \sqrt{4\pi} \varphi(y_{k'})}}
    \prod_{j=1}^{q} (-i \partial \varphi(z_j))
    \prod_{j'=1}^{q'} (+i\bar\partial \varphi(w_{j'}))}_{\GFF}^T
  \\ & \nonumber
  = (4\pi  e^{-\gamma/2})^{n+n'}
   \sqrt{\pi}^{q+q'}
   \avga{
    \prod_{k=1}^{n} {\bar\psi_1 \psi_1}(x_k)
    \prod_{k'=1}^{n'} {\bar\psi_2 \psi_2}(y_{k'})
    \prod_{j=1}^{q} {\bar\psi_2\psi_1}(z_j)
    \prod_{j'=1}^{q'} {\bar\psi_1 \psi_2}(w_{j'})
  }_{\FF(0)}^T.
\end{align}
The claim (along with the relevant restrictions for the $n+n'+q+q'=2$-case) now follows from Lemma \ref{le:gradint} (possibly using integration by parts) and Lemma \ref{lem:gffintegrability}.
\end{proof}

\section{Estimates for the sine-Gordon model and free fermions; proof of main theorems}
\label{sec:estimates}

In this section, we record our main estimates for sine-Gordon correlation functions as well as those for free fermions with a finite volume mass term. The proofs of these estimates are presented in the remainder of the paper.
Assuming these estimates, we then give our proofs of the theorems of Section~\ref{sec:intro}
in this section.
The intuition for Theorems~\ref{thm:correspondence}--\ref{thm:twopoint} is as outlined in Section~\ref{sec:literature}.
Namely, in view of the Coleman correspondence when $z=\mu=0$, i.e., Corollary~\ref{cor:massless-corr},
the sine-Gordon measure which is formally obtained from the GFF measure by weighting it by
\begin{equation}\label{e:sg-weight}
  e^{2\nmu\int dx\, \wick{\cos\sqrt{4\pi}\varphi(x)}\mathbf 1_{\Lambda_L}(x)}
\end{equation}
 should correspond to the massless free fermion ``Grassmann measure'' weighted by
\begin{equation} \label{e:ff-weight}
  e^{A\nz\int dx\, (\bar\psi_1 \psi_1(x)+\bar\psi_2 \psi_2(x))\mathbf 1_{\Lambda_L}(x)}.
\end{equation}
Our estimates stated in this section provide the required analyticity and convergence to make this correspondence rigorously.
Our main innovation here is that our estimates hold for all $z$ in a complex neighborhood of the entire real axis
(not just a neighborhood of the origin) and for all $L>0$,
and that we control the infinite volume limit $L\to\infty$.
The main analyticity results for the sine-Gordon model stated in this section do not cause additional difficulties
for general $\beta \in (0,6\pi)$,
so we state them in this generality.
Together with well-known correlation inequalities they then imply the remaining results
stated in Section~\ref{sec:intro}.

\subsection{The sine-Gordon model and estimates for its correlation functions}
\label{sec:intro-sg}

To state our estimates for the sine-Gordon model,
we begin with the precise definition of our regularization of the continuum, finite volume, massless sine-Gordon model.

For $\epsilon,m>0$, we define the probability measure $\nu^{\GFF(\epsilon,m)}$ of the regularized GFF
as in Section~\ref{sec:massless-gff} and recall that $\nu^{\GFF(\epsilon,m)}$ is supported on $C^\infty(\R^2)$.
We then take as a regularization of the sine-Gordon model the probability measure
\begin{equation}
\nu^{\SG(\beta,z|\epsilon,m,\Lambda)}(d\varphi) = 
\frac{1}{Z(\beta,z|\epsilon,m,\Lambda)}
\exp\qa{2\nz \int_{\Lambda} dx\, \epsilon^{-\beta/4\pi} \cos(\sqrt{\beta}\varphi)} \nu^{\GFF(\epsilon,m)} (d\varphi),
\label{e:sgdef}
\end{equation}
where $\Lambda \subset \R^2$ is a compact set, 
$\beta\in(0,6\pi)$, $z\in \R$, and $Z$ is the partition function -- a normalization constant.
We will also write $\avga{\cdot}_{\GFF(\epsilon,m)}$ for integration with respect to $\nu^{\GFF(\epsilon,m)}$ and $\avga{\cdot}_{\SG(\beta,z|\epsilon,m,\Lambda)}$ for integration with respect to $\nu^{\SG(\beta,z|\epsilon,m,\Lambda)}$.
For $\Lambda=\Lambda_L = \{x\in \R^2 : |x| \leq L\}$ we of course recover our definition
of $\avg{\cdot}_{\SG(\beta,z|\epsilon,m,L)}$ in \eqref{e:SGdef}, but we allow more general $\Lambda$ here because this
allows us to obtain the Euclidean invariance of the infinite volume limits in Theorems~\ref{thm:gradientlim} and~\ref{thm:fieldlim}.

Let us comment briefly on some of the restrictions we have imposed here.
As mentioned earlier, the continuum sine-Gordon model is interesting for $\beta\in (0,8\pi)$. While we are mainly interested in proving the Coleman correspondence for $\beta=4\pi$, the sine-Gordon estimates we prove hold for all $\beta\in(0,6\pi)$, so we present the results in this generality. The regime $\beta\in[6\pi,8\pi)$ is also interesting, but would require finer estimates. 
For $\beta \in (0,4\pi)$, the sine-Gordon measure is absolutely continuous with respect to the GFF when $\Lambda$ is compact.
The free fermion point, $\beta=4\pi$, is precisely where this fails.

We now state our main result about 
the sine-Gordon correlation functions that are important for the Coleman correspondence.

\begin{theorem}\label{thm:cf}
  For $\beta\in (0,6\pi)$, $z\in \R$, and $\Lambda \subset \R^2$ compact,
  $n,q,q'\geq 0$ and 
  $f_1,...,f_n\in L_c^\infty(\R^2)$, $g_1,...,g_q,h_1,...,h_{q'}\in C_c^\infty(\R^2)$
  and $\sigma_1,...,\sigma_{n}\in \{-1,1\}$,
  \begin{enumerate}
  \item If either $(n,q+q')\neq (1,0)$ and $(n,q+q')\neq (2,0)$ or if $f_1,f_2$ have disjoint supports, the limit
    \begin{align}
      &\avga{\prod_{{k}=1}^{n} \wick{e^{i\sqrt{\beta}\sigma_{k}\varphi}}(f_{k}) \prod_{{j}=1}^{q}\partial \varphi(g_{j})\prod_{{j'}=1}^{q'} \bar{\partial} \varphi(h_{j'})}_{\SG(\beta,z|\Lambda)}^T\\
      &\notag\qquad :=\lim_{m\to 0}\lim_{\epsilon\to 0}\avga{\prod_{{k}=1}^{n} \wick{e^{i\sqrt{\beta}\sigma_{k}\varphi}}_\epsilon(f_{k}) \prod_{{j}=1}^{q}\partial \varphi(g_{j})\prod_{{j'}=1}^{q'} \bar{\partial} \varphi(h_{j'})}_{\SG(\beta,z|\epsilon,m,\Lambda)}^T
    \end{align}
    exists and is finite.

  \item Under the assumptions of item (i), the function 
    \begin{equation}
      z\mapsto \avga{\prod_{{k}=1}^{n} \wick{e^{i\sqrt{\beta}\sigma_{k}\varphi}}(f_{k}) \prod_{{j}=1}^{q}\partial \varphi(g_{j})\prod_{{j'}=1}^{q'} \bar{\partial} \varphi(h_{j'})}_{\SG(\beta,z|\Lambda)}^T
    \end{equation}
    has an analytic continuation into a $\Lambda$-dependent neighborhood of the real axis. Moreover, it is even in $z$ when ${n}=0$.
    
  \item Under the assumptions of item (i), for any $l\geq 0$,
    \begin{align}
      &\left.\frac{d^{l}}{dz^{l}}\right|_{z=0}\avga{\prod_{{k}=1}^{n} \wick{e^{i\sqrt{\beta}\sigma_{k}\varphi}}(f_{k}) \prod_{{j}=1}^{q}\partial \varphi(g_{j})\prod_{{j'}=1}^{q'} \bar{\partial} \varphi(h_{j'})}_{\SG(\beta,z|\Lambda)}^T\\
      &= 
        \avga{\prod_{{k}=1}^{n} \wick{e^{i\sqrt{\beta}\sigma_{k}\varphi}}(f_{k}) \prod_{{j}=1}^{q}\partial \varphi(g_{j})\prod_{{j'}=1}^{q'} \bar{\partial} \varphi(h_{j'}) \left(\wick{e^{i\sqrt{\beta}\varphi}(\mathbf 1_{\Lambda})}+\wick{e^{-i\sqrt{\beta}\varphi}(\mathbf 1_{\Lambda})}\right)^{l}}_{\GFF}^T.
        \nonumber
    \end{align}
  \item 	For any $f\in L_c^\infty(\R^2)$ with support in $\Lambda$, we have for $\beta\in (4\pi,6\pi)$
    \begin{multline}
      \lim_{m\to 0}\lim_{\epsilon\to 0}\left[\epsilon^{\frac{\beta}{2\pi}-2}\langle \wick{e^{\pm i\sqrt{\beta}\varphi}}_\epsilon(f)\rangle_{\SG(\beta,z|\epsilon,m,\Lambda)}\right]\\
      = 2\pi\nz e^{-\frac{\gamma \beta}{4\pi}}\int_{\Lambda}dx \, f(x)\int_0^\infty dr\, r^{-\frac{\beta}{2\pi}+1}e^{-\frac{\beta}{4\pi}\Gamma(0,r^2)},
    \end{multline}
    where $\Gamma$ is the incomplete gamma function, and for $\beta=4\pi$,
    \begin{equation}
      \lim_{m\to 0}\lim_{\epsilon\to 0}\left[\tfrac{1}{\log \epsilon ^{-1}}\langle \wick{e^{\pm i\sqrt{\beta}\varphi}}_\epsilon(f)\rangle_{\SG(\beta,z|\epsilon,m,\Lambda)}\right]=2\pi \nz e^{-\gamma}\int_{\R^2} dx\, f(x)
      .
    \end{equation}
  \end{enumerate}
\end{theorem}

By essentially the same proof, we also obtain the following existence of the $\varphi$ field.

\begin{theorem} \label{thm:sgphi}
  Let $\beta\in (0,6\pi)$, $z\in \R$, $m\in (0,\infty)$, and $\Lambda\subset \R^2$ compact.
  Then for any  $f \in C_c^\infty(\R^2)$ and $w\in \C$, the limit
  \begin{equation}
    \avg{e^{w\varphi(f)}}_{\SG(\beta,z|m,\Lambda)} = \lim_{\epsilon \to 0} \avg{e^{w\varphi(f)}}_{\SG(\beta,z|\epsilon,m,\Lambda)}
  \end{equation}
  exists and is entire in $w$. 
  If also $\int f \, dx = 0$, then the limit
  \begin{equation}
    \avg{e^{w\varphi(f)}}_{\SG(\beta,z|\Lambda)} = \lim_{m\to 0}\lim_{\epsilon \to 0} \avg{e^{w\varphi(f)}}_{\SG(\beta,z|\epsilon,m,\Lambda)}
  \end{equation}
  also exists and is an even function of $z$ and an entire function of $w$.
\end{theorem}

Before turning to fermions, we comment here on a few facts the reader might want to keep in mind concerning these theorems.
		
First of all, we 
recall from \eqref{e:truncated} and the discussion following it
that the product notation in the truncated correlation functions 
$\langle \prod_{i=1}^n X_i\rangle^T$ means 
$\langle X_1;X_2;\dots;X_n\rangle^T$,
and that
correspondingly, in item (iii), terms involving powers 
should be interpreted as
\begin{equation}
  \avga{\left(\prod_{i=1}^n X_i\right)Y^l}^T=\avga{X_1;\dots;X_n;Y;\dots;Y}^T,
\end{equation}
where there are $l$ copies of $Y$.
 
Next we mention
that by Lemma \ref{lem:gffintegrability} and our assumptions on $n,q,q'$
the derivatives in item (iii) are indeed finite as they should be.

Finally we mention that in the literature, there certainly exist some results that are
similar to parts of this theorem -- see in particular \cite{MR1672504,MR2461991,MR1777310,1903.01394}.
What we believe is truly new, and critical to our proof of the Coleman correspondence, is that we are able to treat all values of $z\in \R$ and prove analyticity in a neighborhood of the real axis -- not just in a neighborhood of the origin.

\bigskip

We now turn to describing what we need to know about free massive fermions with a finite volume mass term.

\subsection{Free fermion estimates}

As discussed at the beginning of Section~\ref{sec:estimates}, we will establish the equivalence of the sine-Gordon measure
with finite volume interaction with that of Dirac fermions with a finite volume mass term.
We will here choose $\Lambda = \Lambda_L$ to be a disk of radius $L>0$ centered at the origin.
We again take the pragmatic approach of defining the free fermion model with a finite volume mass term,
formally represented by the fermionic path integral with weight \eqref{e:ff-weight}, 
directly through its correlation functions. 
Namely, given a corresponding propagator $S_{\mu \mathbf 1_{\Lambda_L}}$ 
constructed in  Theorem~\ref{thm:ferm} below,
the correlation functions are defined by formulas like \eqref{e:fermdef} and \eqref{e:fermtruncdef}
but now with $S_{\mu \mathbf 1_{\Lambda_{L}}}$ instead of $S$.
In particular, given $n\geq 3$, $f_1,...,f_n\in L_c^\infty(\R^2)$ 
and $\alpha_1,\beta_1,...,\alpha_n,\beta_n\in \{1,2\}$,
the smeared truncated correlation functions are defined by
\begin{multline} \label{e:truncated-muLambda}
  \avga{\prod_{i=1}^n \bar\psi_{\alpha_i}\psi_{\beta_i}(f_i)}^T_{\FF(\mu \mathbf 1_{\Lambda_L})}
  \\ 
  := (-1)^{n+1} \sum_\pi \int_{(\R^2)^n} \prod_{i=1}^n f_{i}(x_i)\prod_{i=1}^{n} S_{\mu\mathbf 1_{\Lambda_L}; \alpha_{\pi^i(1)}\beta_{\pi^{i+1}(1)}}(x_{\pi^i(1)},x_{\pi^{i+1}(1)}),
\end{multline}
where we sum over cyclic permutations $\pi$ -- as we see in our proof of Theorem \ref{thm:ferm}, this is finite for $n\geq 3$.
For $n=2$, the same definition applies to $f_1$ and $f_2$ with disjoint compact supports.
For $n=2$ and $f_1$ and $f_2$ with overlapping supports the above integral is no longer necessarily finite and
we will instead consider the two-point function
with the singularity subtracted, i.e., 
\begin{multline} \label{e:trunc2ptsub}
  \int dx_1\, dx_2\, f_1(x_1) f_2(x_2)
  \\   \times 
  \pa{-S_{\mu\mathbf{1}_{\Lambda_L};\alpha_1\beta_2}(x_1,x_2)S_{\mu\mathbf{1}_{\Lambda_L};\alpha_2\beta_1}(x_2,x_1)
    + S_{0;\alpha_1\beta_2}(x_1,x_2)S_{0;\alpha_2\beta_1}(x_2,x_1)},
\end{multline}
with $S_0$ is given by the right-hand side of \eqref{e:S0-bis-sec2}. This is
formally equal to 
\begin{equation}
  \avg{\bar\psi_{\alpha_1}\psi_{\beta_1}(f_1) \bar\psi_{\alpha_2}\psi_{\beta_2}(f_2)}_{\FF(\mu\mathbf 1_\Lambda)}^T
  -
  \avg{\bar\psi_{\alpha_1}\psi_{\beta_1}(f_1) \bar\psi_{\alpha_2}\psi_{\beta_2}(f_2)}_{\FF(0)}^T
  .
\end{equation}      
Therefore the existence of the propagator $S_{\mu \mathbf 1_{\Lambda_{L}}}$
and some of its basic properties are our main result
concerning such models -- this is summarized in the following theorem. 
Here recall our definition of the Dirac operator $\Dirac$ from \eqref{e:Dirac}.

\begin{theorem} \label{thm:ferm}
  For each $\mu\in \R$ and $L>0$, the Dirac operator with finite volume mass term, $i\Dirac+\nmu \mathbf 1_{\Lambda_L}$,
  where $\Lambda_L = \{x\in \R^2 : |x|\leq L\}$,
  has a fundamental solution $S_{\mu \mathbf 1_{\Lambda_L}}(x,y)$, $x\neq y$,
  with values in $\C^{2\times 2}$, namely 
  \begin{equation}
    (i\Dirac_x+\nmu \mathbf 1_{\Lambda_L}(x))S_{\mu \mathbf 1_{\Lambda_L}}(x,y)=\delta(x-y) \qquad \text{and} \qquad \lim_{x\to\infty}S_{\mu\mathbf 1_{\Lambda_L}}(x,y)=0,
  \end{equation}	   
  such that given $n \geq 3$, $f_1,...,f_n\in L_c^\infty(\Lambda_L)$ 
  and $\alpha_1,\beta_1,...,\alpha_n,\beta_n\in \{1,2\}$,
  the smeared truncated correlation functions \eqref{e:truncated-muLambda}
  satisfy the following properties:
  \begin{enumerate}
  \item The function
    \begin{equation}
      \mu \mapsto \avga{\prod_{i=1}^n \bar\psi_{\alpha_i}\psi_{\beta_i}(f_i)}^T_{\FF(\mu \mathbf 1_{\Lambda_L})}
    \end{equation}
    has an analytic continuation into an $L$-dependent neighborhood of the real axis.
    In particular, the smeared truncated correlation function is finite.
    (For $\mu=0$, $S_{\mu \mathbf 1_\Lambda}=S_0$.)
    
  \item For $l\geq 1$,
    \begin{multline}
      \left.\frac{d^{l}}{d\mu^{l}}\right|_{\mu=0} \avga{\prod_{i=1}^n \bar\psi_{\alpha_i}\psi_{\beta_i}(f_i)}^T_{\FF(\mu \mathbf 1_{\Lambda_L})}
      \\ 
      =
      \avga{\prod_{i=1}^n \bar\psi_{\alpha_i}\psi_{\beta_i}(f_i)\big(\bar\psi_{1}\psi_{1}(\mathbf 1_{\Lambda_L})+\bar\psi_{2}\psi_{2}(\mathbf 1_{\Lambda_L})\big)^{l}}^T_{\FF(0)}.
    \end{multline}
  \item
    For any $\mu\in \R$, as $L\to\infty$, 
    \begin{equation}
      \avga{\prod_{i=1}^n \bar\psi_{\alpha_i}\psi_{\beta_i}(f_i)}^T_{\FF(\mu \mathbf 1_{\Lambda_L})}\rightarrow \avga{\prod_{i=1}^n \bar\psi_{\alpha_i}\psi_{\beta_i}(f_i)}^T_{\FF(\mu)}.
    \end{equation}
    On the right hand side, the correlation functions with index $\FF(\mu)$ are defined by 
    the propagator \eqref{e:Smu} of Dirac fermions with infinite volume mass term $\mu$.
  \end{enumerate}
  For $n = 2$, the same statements remain true if $f_1$ and $f_2$ have disjoint compact supports,
  or if the truncated two-point function is replaced by \eqref{e:trunc2ptsub} in (i) and on the left-hand sides of (ii) and (iii),
  and analogously on the right-hand side of (iii).
\end{theorem}

We again comment on some issues regarding this theorem. 

First of all, one could readily formulate a non-smeared version of this result as well, but for the proof of Theorem~\ref{thm:correspondence}, the smeared versions of the correlation functions are the relevant ones.

Secondly, in item (ii), the correct way to understand the term on the right hand side is that one expands the power, uses multilinearity, and \eqref{e:fermdef}. Moreover, the fact that the right hand side is finite 
follows from the last statement in Corollary~\ref{cor:massless-corr}.
	
Finally, we mention that given that this theorem is essentially about controlling a finite volume approximation to massive free fermions,
we expect that at least parts of this result is well known to some experts. Unfortunately we were not unable to find a suitable reference for the results we need.
Our proof makes use of the convenient domain of a disk, but we
expect that the result also holds for much more general domains.

\subsection{Proof of Theorems~\ref{thm:correspondence} and \ref{thm:twopoint}}

Assuming Theorems~\ref{thm:cf} and~\ref{thm:ferm},
we are now in a position to prove the Coleman correspondence at $\beta=4\pi$.

\begin{proof}[Proof of Theorem~\ref{thm:correspondence}]
	It suffices to show that
	the correlation functions of the massless sine-Gordon model with an interaction term supported in $\Lambda_L$
	at $\beta=4\pi$
	and those of free Dirac fermions with a mass term supported in $\Lambda_L$ agree for all $z \in \R$ (and corresponding $\mu =Az$)
	and all $L<\infty$.
	Indeed, by Theorem \ref{thm:ferm} item (iii),
	the smeared truncated correlation functions of free Dirac fermions with a $\Lambda_L$ mass term
	converge as $L\to\infty$ to their infinite volume versions, 
	and hence, the identification in finite volume implies that
	the sine-Gordon correlation functions converge to the same limit.
	
	For the equivalence in finite volume, let us write $O_{k}^B$ for one of the quantities $\wick{e^{\pm i \sqrt{4\pi}\varphi}}$, $\partial \varphi$, or $\bar \partial \varphi$ on the sine-Gordon side, and write $O_{k}^F$ for the corresponding one on the fermionic side -- the correspondence being the one given by the statement of Theorem \ref{thm:correspondence}. 
        Thus $\wick{e^{i\sqrt{4\pi}\varphi}}$ corresponds to $A\bar\psi_1\psi_1$,
        $\wick{e^{-i\sqrt{4\pi}\varphi}}$ corresponds to $A\bar\psi_2\psi_2$,
        $-i\partial \varphi$ corresponds to $B\bar\psi_2\psi_1$,
        and $+i\bar\partial \varphi$ corresponds to $B\bar\psi_1\psi_2$.
        We also let $f_j$ be compactly supported and either essentially bounded or smooth (depending on whether it is a charge or gradient observable that is acting on it)
        that $O_{k}^B$ and $O_{k}^F$ act on.
	
        Let us first focus on the case where 
        $n+n'+q+q'\geq 3$ and let us not assume that the supports of the test functions are disjoint.
        To see that the truncated correlation functions agree for all $z \in \R$ when $L<\infty$,
	we use that both are analytic in $z$ respectively $\mu$ in a complex neighbourhood of the real axis,
	by Theorem~\ref{thm:cf} item (ii) and Theorem~\ref{thm:ferm} item (i).
	By unique analytic continuation, it therefore suffices to verify that they agree in a complex
	neighbourhood of $z=\mu=0$. This in turn holds if the truncated correlation functions agree at
	$z=0$ and all $z$-derivatives at $z=0$ agree.
	That they agree for $z=0$ is Corollary~\ref{cor:massless-corr}.
	On the fermionic side, the $\mu$-derivatives at $\mu=0$ are given
	by Theorem \ref{thm:ferm} item (ii) as
	\begin{align} \label{e:ferm-muderiv-bis}
	&\frac{d^{l}}{dz^{l}}
	\avga{\prod_{{k}=1}^{n+n'+q+q'} O_{k}^F(f_{k})}^T_{\FF(Az \mathbf 1_{\Lambda_{L}})} \Bigg|_{z=0}
        \nnb
	&\qquad = A^l
	\avga{\prod_{{k}=1}^{n+n'+q+q'} O_{k}^F(f_{k}) \Big(\bar\psi_1\psi_1(\mathbf 1_{\Lambda_{L}})+\bar\psi_2\psi_2(\mathbf 1_{\Lambda_{L}})\Big)^{l}}^T_{\FF(0)}
	.
	\end{align}
	On the sine-Gordon side, the ${z}$-derivatives at ${z}=0$ are given by
	Theorem \ref{thm:cf} item (iii) as
	\begin{align}
        &\frac{d^{l}}{dz^{l}}
	\avga{\prod_{{k}=1}^{n+n'+q+q'} O_{k}^B(f_{k})}^T_{\SG(\beta,z |{\Lambda_{L}})}\Bigg|_{z=0}
        \nnb
        &\qquad =
	\avga{\prod_{{k}=1}^{n+n'+q+q'} O_{k}^B(f_{k})\, \Big(\wick{e^{i\sqrt{\beta}\varphi}(\mathbf 1_{{\Lambda_{L}}})}+\wick{e^{-i\sqrt{\beta}\varphi}(\mathbf 1_{{\Lambda_{L}}})}\Big)^{l}}_{\GFF}^T.
	\end{align}
	That these are equal when $\beta=4\pi$ again follows from Corollary \ref{cor:massless-corr}.

        The same argument is valid for 
        $n+n'=q+q'=1$. Moreover, if we further assume that $f_k$ have disjoint supports, then the same argument works also for general $n+n'+q+q'=2$.
\end{proof}

The remaining $(n+n',q+q')=(0,2)$ case with overlapping test functions, i.e., Theorem~\ref{thm:twopoint},
works similarly, as follows.

\begin{proof}[Proof of Theorem~\ref{thm:twopoint}]
By Theorem~\ref{thm:cf}, the relevant finite volume $\epsilon,m\to 0$ limits exist on the sine-Gordon side. Moreover, by Lemma \ref{le:gradint}, it suffices to show that
\begin{align} \label{e:twopoint1-pf}
  &\avg{\partial\varphi(f_1) \partial\varphi(f_2)}_{\SG(4\pi,z)}
    - \avg{\partial\varphi(f_1) \partial\varphi(f_2)}_{\GFF}
    \nnb
  & 
    = -\frac{B^2}{\pi^2} \int dx_1 \, dx_2\, f_1(x_1) f_2(x_2) \,
    \pa{(\partial_{x_1} K_0(A|z||x_1-x_2|))^2-\frac{1}{4(x_1-x_2)^2}}
    ,
  \\
  \label{e:twopoint2-pf}
  &\avg{\partial\varphi(f_1) \bar\partial\varphi(f_2)}_{\SG(4\pi,z)}
    -
    \avg{\partial\varphi(f_1) \bar\partial\varphi(f_2)}_{\GFF}
    \nnb
  & 
    = -\frac{B^2A^2 z^2}{4\pi^2} \int dx_1 \, dx_2\, f_1(x_1) f_2(x_2)\,
    (K_0(A|z||x_1-x_2|))^2
    .
\end{align}
The claim then follows from the result about the GFF two-point function from
Lemma~\ref{le:gradint}.
  
The proof of \eqref{e:twopoint1-pf} and \eqref{e:twopoint2-pf} is analogous to that of Theorem~\ref{thm:correspondence},
as follows.
To be concrete, we focus on the proof of  \eqref{e:twopoint1-pf}; the other one is analogous.
By Theorem \ref{thm:ferm}, item (iii), it suffices to show that
\begin{multline}
  \avg{\partial\varphi(f_1) \partial\varphi(f_2)}_{\SG(4\pi,z|\Lambda_L)}
        - \avg{\partial\varphi(f_1) \partial\varphi(f_2)}_{\GFF}
  = B^2 \int dx_1 \, dx_2\, f_1(x_1) f_2(x_2) \,
  \\ \times 
  \pa{-S_{Az\mathbf 1_\Lambda;21}(x_1,x_2)S_{Az\mathbf 1_\Lambda;21}(x_2,x_1)+ S_{0;21}(x_1,x_2)S_{0;21}(x_1,x_2) }.
\end{multline}
For $z=0$, this claim is trivial as both sides vanish then.
Theorems~\ref{thm:cf} and and the special $n=2$ case of Theorem~\ref{thm:ferm} now again imply that both sides are
analytic in $z$ and that their derivatives are identical, using  Corollary \ref{cor:massless-corr}.
\end{proof}

\subsection{Proof of Theorems~\ref{thm:gradientlim} and \ref{thm:fieldlim}}

For the proofs of the results stated in Section~\ref{sec:intro-gradient},
we need the following correlation inequalities from \cite{MR496191}.

First note that the $\varphi\mapsto -\varphi$ symmetry of the measure implies
$\avg{e^{i\varphi(f)}}_{\SG(\beta,z|\epsilon,m,\Lambda)} = \avg{\cos(\varphi(f))}_{\SG(\beta,z|\epsilon,m,\Lambda)}$.
For $\nz>0$, it then follows from \cite[Corollary~3.2]{MR496191} that,
as a function of $m>0$ and $\nz>0$ and the set $\Lambda$,
\begin{equation} \label{e:monotone}
  \avg{e^{i\varphi(f)}}_{\SG(\beta,z|\epsilon,m,\Lambda)} \text{ is increasing, and}
  \qquad
  \avg{\varphi(g)^2}_{\SG(\beta,z|\epsilon,m,\Lambda)}  \text{ is decreasing}.
\end{equation}
Indeed, by rescaling $\varphi$ by $\sqrt{\beta}$, in the notation of \cite[Section~3]{MR496191}, one has
\begin{equation}
  \avg{F(\varphi/\sqrt{\beta})}_{\SG(\beta,z | \epsilon,m,\Lambda)}
  = \avg{F(\varphi)}_{C,\rho},
\end{equation}
where
\begin{equation}
  C = \beta \int_{\epsilon^2}^\infty dt\, e^{t\Delta-tm^2},
  \qquad \rho(dx) = 2\nz \epsilon^{-\beta/4\pi} \mathbf 1_{\Lambda}(x) \, dx,
\end{equation}
and \cite[Corollary~3.2]{MR496191} states that if $\rho_1 \leq \rho_2$ and $C_2 \leq C_1$ then
\begin{equation}
  \avg{\cos(\varphi(g))}_{C_2,\rho_2} \geq \avg{\cos(\varphi(g))}_{C_1,\rho_1}, \qquad
  \avg{\varphi(g)^2}_{C_2,\rho_2} \leq \avg{\varphi(g)^2}_{C_1,\rho_1}.
\end{equation}
The monotonicity \eqref{e:monotone} is immediate from this.

As a particular case of \eqref{e:monotone} we get the following infrared bound: 
for any $f\in C_c^\infty(\R^2)$, $m,\nz>0$ and $\Lambda$, we have
\begin{equation}
  \avga{\varphi(f)^2}_{\SG(\beta,z|\epsilon,m,L)}\leq \avga{\varphi(f)^2}_{\GFF(\epsilon,m)}.
  \label{e:GFF2ptbound}
\end{equation} 

\begin{proof}[Proof of Theorems~\ref{thm:gradientlim} and \ref{thm:fieldlim}]
	Since the proofs of Theorems~\ref{thm:gradientlim} and~\ref{thm:fieldlim} are essentially identical,
	we focus on the first theorem and leave the modifications for the second theorem to the reader.
	
	By Theorem~\ref{thm:sgphi}, for any $f\in C_c^\infty(\R^2)$ with $\int f\, dx=0$, the limit
	\begin{equation}
		\avg{e^{i\varphi(f)}}_{\SG(\beta,z|\Lambda)}:=
		\lim_{m\to 0} \lim_{\epsilon\to 0}\avg{e^{i\varphi(f)}}_{\SG(\beta,z|\epsilon,m,\Lambda)}
	\end{equation}
	exists and is invariant under $z \mapsto -z$. Thus without loss of generality we can and will assume $\nz>0$.
	By \eqref{e:monotone}, it follows that $\avg{e^{i\varphi(f)}}_{\SG(\beta,z|\Lambda)}$ is monotone in $\Lambda$,
	and thus converges as $\Lambda \uparrow \R^2$ to a limit which we denote by $\avg{e^{i\varphi(f)}}_{\SG(\beta,z)}$.
	
	The limit is trivially bounded above by $1$ and the map $f\mapsto \avg{e^{i\varphi(f)}}_{\SG(\beta,z)}$ satisfies
	the following continuity estimate: for any $g \in C_c^\infty(\R^2)$ with $\int dx\, g = 0$,
	\begin{equation} \label{e:laplacecont}
		|\avg{e^{i\varphi(f+g)}}_{\SG(\beta,z)} - \avg{e^{i\varphi(f)}}_{\SG(\beta,z)}|
		\leq \frac12 \avg{\varphi(g)^2}_{\SG(\beta,z)} \leq \frac12 (g,(-\Delta)^{-1}g).
	\end{equation}
	Indeed, for $\epsilon,m>0$ and $\Lambda$ finite,
	the analogue of the first inequality is immediate,
	and the second inequality follows from \eqref{e:monotone}. The claimed inequality then follows by taking the
	limits in $\epsilon,m,\Lambda$.
	
	In particular, if functions $g_k \in C_c^\infty(\R^2)$ with $\int dx \, g_k =0$ converge to $0$ in the topology of $\mathcal S(\R^2)$,
	the right-hand side of \eqref{e:laplacecont} converges to $0$. Since $C_c^\infty(\R^2)$ is dense in $\mathcal S(\R^2)$
	(and likewise for the subspaces of functions which integrate to $0$),
	it follows that $\avg{e^{i\varphi(f)}}_{\SG(\beta,z)}$ extends to a continuous
	functional on $\mathcal S'(\R^2)/\text{constants}$
        (the topological dual space of the closed subspace of integral-$0$ functions in $\mathcal S(\R^2)$).
	Minlos's theorem then implies that $\avg{e^{i\varphi(f)}}_{\SG(\beta,z)}$ is the characteristic functional
	of a probability measure on $\mathcal S'(\R^2)/\text{constants}$.
	
	That the limit is Euclidean invariant is a standard argument that
	follows from the Euclidean invariance of the GFF and the monotonicity of
	$\avg{e^{i\varphi(f)}}_{\SG(\beta,z|\Lambda)}$ in $\Lambda$ for any increasing family of sets,
	see, e.g., \cite[Section~VIII.6]{MR0489552}.
	
	Finally, the bounds \eqref{e:corrbd1}--\eqref{e:corrbd3} are immediate from the monotonicity of \eqref{e:monotone}.
\end{proof}

\subsection{Proof of Theorem~\ref{thm:field} and Corollary~\ref{cor:zle0}}

\begin{proof}[Proof of Theorem~\ref{thm:field}]
  The main step of the proof will be to show that \eqref{e:twopoint-fourier} holds
  for functions with integral 0, i.e.,
  for all $f_1,f_2 \in \mathcal{S}(\R^2)$ with $\int dx \, f_i = 0$,
    \begin{equation}
    \avg{\varphi(f_1)\varphi(f_2)}_{\SG(4\pi,z)} =
    \lim_{L\to\infty}\lim_{m\to 0} \lim_{\epsilon\to 0} \avg{\varphi(f_1)\varphi(f_2)}_{\SG(4\pi,z|\epsilon,m,L)}.
  \end{equation}
  To this end,
let us first further assume that there are $g_i,h_i \in C_c^\infty(\R^2)$ such that
\begin{equation} \label{e:fgh}
  f_i = \partial g_i + \bar\partial h_i.
\end{equation}
In this case, we note that, by integrating by parts,
\begin{align}	
  \avga{\varphi(f_1)\varphi(f_2)}_{\SG(4\pi,z|\epsilon,m,L)}
  &=\avga{\partial\varphi(g_1)\partial\varphi(g_2)}_{\SG(4\pi,z|\epsilon,m,L)}
    \nnb &\quad 
    +\avga{\bar\partial\varphi(h_1)\bar\partial\varphi(h_2)}_{\SG(4\pi,z|\epsilon,m,L)}
    \nnb
  &\quad +\avga{\partial\varphi(g_1)\bar\partial\varphi(h_2)}_{\SG(4\pi,z|\epsilon,m,L)}
    \nnb &\quad 
    +\avga{\bar\partial\varphi(h_1)\partial\varphi(g_2)}_{\SG(4\pi,z|\epsilon,m,L)},
\end{align} 
and we find that, by Theorem \ref{thm:twopoint}, the $\epsilon,m\to 0$, $L\to\infty$ limits exist and
\begin{align}
  \avg{\varphi(f_1)\varphi(f_2)}_{\SG(4\pi,z)}
  &=
    \avg{\partial\varphi(g_1)\partial\varphi(g_2)}_{\SG(4\pi,z)}
    \nnb &\quad 
  +
  \avg{\bar\partial\varphi(h_1)\bar\partial\varphi(h_2)}_{\SG(4\pi,z)}
  \nnb
  &\quad+
    \avg{\partial\varphi(g_1)\bar\partial\varphi(h_2)}_{\SG(4\pi,z)}
    \nnb &\quad 
  +
    \avg{\bar\partial \varphi(h_1)\partial\varphi(g_2)}_{\SG(4\pi,z)}
    .
\end{align}
To express the right-hand side as in \eqref{e:twopoint-fourier}--\eqref{e:2ptsinh},
let us first look at the $g_1,g_2$-term -- the remaining terms are similar.
Recalling that $\frac{1}{2\pi}K_0(A|z||x-y|)$ is the covariance of the massive free field, we have the following Fourier space representation of $K_0$:
\begin{equation}
  K_0(A|z||x-y|)=\int_{\R^2} \frac{dp}{2\pi}\, \frac{e^{-ip\cdot(x-y)}}{|p|^2+A^2|z|^2},
\end{equation}
where the integral is understood either in principal value sense or in the sense of distributions. Thus with the convention $\hat f(p)=\int_{\R^2}f(x)e^{-ip\cdot x}dp$ for the Fourier transform and $2\hat \partial = i\bar p$, Theorem~\ref{thm:twopoint} and a routine calculation shows that (with integrals understood in a principal value sense)
\begin{align}
	  &\avg{\partial\varphi(g_1)\partial\varphi(g_2)}_{\SG(4\pi,z)}
          \nnb
          &=\frac{1}{16\pi^3}\int_{\R^2\times \R^2}dp_1\, dp_2\, \hat g_1(p_1+p_2)\hat g_2(-p_1-p_2)\frac{\bar p_1\bar p_2}{(|p_1|^2+A^2|z|^2)(|p_2|^2+A^2|z|^2)}\notag\\
	  &=\int_{\R^2} \frac{dp}{(2\pi)^2}\,\hat g_1(p)\hat g_2(-p)\int_{\R^2}\frac{dq}{4\pi} \frac{\bar q(\bar p-\bar q)}{(|q|^2+A^2|z|^2)(|p-q|^2+A^2|z|^2)}\\
	  &\notag =\int_{\R^2} \frac{dp}{(2\pi)^2}\,\widehat {\partial g_1}(p)\widehat{ \partial g_2}(-p)\frac{1}{\pi\bar p^2}\int_{\R^2}dq \frac{\bar q(\bar p-\bar q)}{(|q|^2+A^2|z|^2)(|p-q|^2+A^2|z|^2)}\\
	  &=:\int_{\R^2} \frac{dp}{(2\pi)^2}\,\widehat {\partial g_1}(p)\widehat{ \partial g_2}(-p)\hat C_{A|z|}(p).\notag 
\end{align}
The $\hat C$ on the right-hand side can be computed as follows.
Going into polar coordinates, scaling the radial variable, and translating the angular variable shows that
\begin{equation}
	\hat C_{A|z|}(p)=\frac{2}{|p|^2}\int_0^\infty dr\, \frac{r}{r^2+\mu_p^2} \int_0^{2\pi}\frac{dt}{2\pi}\frac{e^{-it}r(1-re^{-it})}{1+r^2-2r\cos t+\mu_p^2}
\end{equation}
where $\mu_p=A|z|/|p|$.
To evaluate the $t$-integral through the residue theorem, we note that out of the two poles for $\eta =e^{-it}$,
\begin{equation}
  \eta=\frac{1+\mu_p^2+r^2\pm \sqrt{(1+\mu_p^2+r^2)^2-4r^2}}{2r},
\end{equation}
only the minus-one is inside the unit disk, and we thus find 
\begin{multline}
  \hat C_{A|z|}(p)=\frac{2}{|p|^2}\int_0^\infty dr \, \frac{r}{r^2+\mu_p^2}\oint_{|\eta|=1}\frac{d\eta}{2\pi i \eta}\frac{\eta r(1-r\eta)}{1+r^2-r(\eta+\eta^{-1})+\mu_p^2} \\ 
  =\frac{2}{|p|^2}\int_0^\infty dr \, \frac{r}{r^2+\mu_p^2}\frac{-\mu_p^2+r^2-(\mu_p^2+r^2)^2+(\mu_p^2+r^2)\sqrt{(1+\mu_p^2+r^2)^2-4r^2}}{2\sqrt{(1+\mu_p^2+r^2)^2-4r^2}}.
\end{multline}
A straightforward (but slightly tedious) calculation shows that the last integrand can be written as
\begin{multline}
  \frac{1}{4}\partial_r \Bigg(r^2-\sqrt{(1+\mu_p^2)^2+2(\mu_p^2-1)r^2+r^4}
    -\frac{2\mu_p^2\log(r^2+\mu_p^2)}{\sqrt{1+4\mu_p^2}}
    \\ 
    +\frac{2\mu_p^2\log(1+3\mu_p^2-r^2+\sqrt{1+4\mu_p^2}\sqrt{(1+\mu_p^2)^2+2(\mu_p^2-1)r^2+r^4})}{\sqrt{1+4\mu_p^2}}\Bigg)
    \\
    =\frac{r}{r^2+\mu_p^2}\frac{-\mu_p^2+r^2-(\mu_p^2+r^2)^2+(\mu_p^2+r^2)\sqrt{(1+\mu_p^2+r^2)^2-4r^2}}{2\sqrt{(1+\mu_p^2+r^2)^2-4r^2}},
\end{multline}
from which we see after another slightly tedious calculation that $\hat C_{A|z|}(p)$ equals
\begin{align}
  \frac{1}{ |p|^2}\Bigg(1+\frac{\mu_p^2}{\sqrt{1+4\mu_p^2}}\left[\log \mu_p^2+\log(\sqrt{1+4\mu_p^2}-1)-\log(1+3\mu_p^2+(1+\mu_p^2)\sqrt{1+4\mu_p^2})\right]\Bigg).
\end{align}
Finally, an elementary calculation shows that 
\begin{equation}
	\frac{x^3+3x+(x^2+1)\sqrt{x^2+4}}{\sqrt{x^2+4}-x}=\left(\frac{x}{2}+\sqrt{\frac{x^2}{4}+1}\right)^{4},
\end{equation}
from which we can deduce \eqref{e:2ptsinh} with another routine calculation.

We see in particular from this that $\hat C_{A|z|}$ is bounded for $|z|>0$.
A similar calculation shows that 
\begin{align}
	  \avg{\partial\varphi(g_1)\bar\partial\varphi(h_2)}_{\SG(4\pi,z)}=\int_{\R^2}\frac{dp}{(2\pi)^2}\, \widehat {\partial g_1}(p)\widehat{ \bar\partial h_2}(-p)\hat C_{A|z|}(p),
\end{align}
with the same $\hat C_{A|z|}$. Thus taking complex conjugates of these identities, we find that for our $f_i$ given by $f_i=\partial g_i+\bar \partial h_i$, 
\begin{equation}
	\avga{\varphi(f_1)\varphi(f_2)}_{\SG(4\pi,z)}=\int_{\R^2} \frac{dp}{(2\pi)^2} \, \hat f_1(p)\hat f_2(-p)\hat C_{A|z|}(p),
	\label{e:target}
\end{equation}
which is precisely the claim for the $f_i$ which can be represented this way.

Finally, to extend the statement to arbitrary $f_i \in C_c^\infty(\R^2)$ or $f_i \in \mathcal{S}(\R^2)$ with $\int dx\, f_i=0$, we note that
such $f_i$ can be written as in \eqref{e:fgh} but with $g_i$ and $h_i$ in $\mathcal{S}(\R^2)$,
by Taylor expanding $\hat f_i$.
Thus it remains to extend our argument to Schwartz functions. For this, given $f_i\in \mathcal S(\R^2)$ satisfying $\int f_i=0$, let $g_i,h_i\in \mathcal S(\R^2)$ be such that we have the representation \eqref{e:fgh}. Let us take $\chi\in C_c^\infty(\R^2)$ non-negative, bounded by 1, supported in $\Lambda_{2R}=\{x\in \R^2: |x|<2R\}$, equal to one in $\Lambda_R$, and with gradient bounded as a function of $R$. Then write $g_i=\chi g_i+(1-\chi)g_i$ and similarly for $h_i$. We then have 
\begin{equation}
	\avga{\varphi(f_1)\varphi(f_2)}_{\SG(4\pi,z|\epsilon,m,L)}=\Sigma_{1}(R|\epsilon,m,L)+\Sigma_{2}(R|\epsilon,m,L),
\end{equation}
where in $\Sigma_1$, we have kept only the $\chi g_i,\chi h_i$-terms, while in $\Sigma_2$ we have at least one $(1-\chi)g_i$ or $(1-\chi)h_i$-term.

Using the initial part of this proof and a routine dominated convergence argument, we see that when we let $\epsilon\to 0$, $m\to 0$, $L\to\infty$, and finally $R\to\infty$, $\Sigma_1$ converges to our target -- namely \eqref{e:target} (which is perfectly well defined for $f_i\in \mathcal S(\R^2)$). Thus we need to show that $\Sigma_2$ tends to zero in the same limit. For this, using \eqref{e:GFF2ptbound} and routine Cauchy-Schwarz arguments shows that (in the $\epsilon\to 0$, $m\to 0$, $L\to\infty$-limit) we end up estimating e.g. quantities of the form 
\begin{equation}
	\int_{\R^2\times \R^2} dx\, dy\, |\nabla ((1-\chi(x))g_1(x))||\nabla ((1-\chi(y))g_1(y))| |\log |x-y||.
\end{equation}
By dominated convergence, this tends to zero as $R\to\infty$, and one finds that $\Sigma_2$ tends to zero in our limit. This shows that \eqref{e:target} is true also for $f_1,f_2\in \mathcal S(\R^2)$ satisfying $\int f_i=0$.

The localization bound  \eqref{e:varbd} now follows easily by observing  that, $\int du\, (f_x(u)-f_y(u))=0$  so
by \eqref{e:twopoint-fourier} for integral $0$ test functions (for which we have now established \eqref{e:twopoint-fourier}),
the left-hand side of \eqref{e:varbd} is given by
\begin{equation}
  \sup_{x \in \R^2}      \int_{\R^2} \frac{dp}{(2\pi)^2}\, |\hat f(p)|^2 (2-2\cos(p\cdot x)) \hat C_{Az}(p).
\end{equation}
This is uniformly bounded since $\hat C_\mu(p)$ is bounded for $\mu \neq 0$.

Finally, we construct the required probability measure $\avg{\cdot}_{\SG(4\pi,z)}$ on $\mathcal{S}'(\R^2)$.
In Theorem~\ref{thm:gradientlim}, we have already constructed such a measure on $\mathcal{S}'(\R^2)/\text{constants}$,
i.e., for test functions $f \in \mathcal{S}(\R^2)$ with $\int dx\, f = 0$.
Using the uniform bound on $\hat C_\mu$ for $\mu \neq 0$ we can extend this measure to all test functions in $\mathcal{S}(\R^2)$ as follows.
Let $\gamma_N(x) = (2\pi N)^{-1}e^{-|x|^2/(2N)}$ be the density of the two-dimensional Gaussian probability measure of variance $N$
and Fourier transform $\hat \gamma_N(p) = e^{-\frac{1}{2}N|p|^2}$.
For any $f \in \mathcal{S}(\R^2)$, the function $f-\hat f(0)\gamma_N \in \mathcal S(\R^2)$ then has integral $0$, and
\begin{equation}
  \avg{e^{i\varphi(f-\hat f(0)\gamma_N)}}_{\SG(4\pi,z)}
\end{equation}
is well defined by Theorem~\ref{thm:gradientlim}.
For $\mu \neq 0$, we will show that it is a Cauchy sequence in $N$,
as a consequence of the boundedness of $\hat C_\mu$. Indeed,
\begin{align}
  &\absa{\avg{e^{i\varphi(f-\hat f(0)\gamma_N)}}_{\SG(4\pi,z)}
  -
  \avg{e^{i\varphi(f-\hat f(0)\gamma_M)}}_{\SG(4\pi,z)}}
    \nnb
  &\leq \frac{|\hat f(0)|^2}{2} \avg{ \varphi(\gamma_N-\gamma_M)^2}_{\SG(4\pi,z)}
    \nnb
  &= \frac{|\hat f(0)|^2}{2} \int_{\R^2}  \frac{dp}{(2\pi)^2} \, \absa{e^{-\frac{1}{2}N|p|^2} - e^{-\frac{1}{2}M|p|^2}}^2 \hat C_{A|z|}(p) \to 0,
\end{align}
as $N,M\to\infty$. For $f\in \mathcal{S}(\R^2)$, we may thus define
\begin{equation}
  \avg{e^{i\varphi(f)}}_{\SG(4\pi,z)} = \lim_{N\to\infty}   \avg{e^{i\varphi(f - \hat f(0)\gamma_N)}}_{\SG(4\pi,z)}.
\end{equation}
That this is indeed the characteristic functional of a probability measure on $\mathcal{S}'(\R^2)$
again follows from Minlos' theorem and the continuity of $f \mapsto \avg{e^{i\varphi(f)}}_{\SG(4\pi,z)}$ which follows from
the boundedness of $\hat C_\mu$ by an argument analogous to the above Cauchy sequence argument.
This argument also shows that the covariance is given by
\begin{equation}
  \avg{\varphi(f_1)\varphi(f_2)}_{\SG(4\pi,z)} = \int_{\R^2} \frac{dp}{(2\pi)^2} \, \hat f_1(p) \hat f_2(-p) \hat C_{A|z|}(p).
\end{equation}
The exponential decay \eqref{e:twopoint-expdecay} now follows (see e.g. \cite[Theorem IX.14]{MR0493420}) from the fact that $\hat C_\mu(p)$ is uniformly bounded and that, as one readily checks from \eqref{e:2ptsinh}, it has an analytic continuation
into
a strip $|\mathrm{Im}(p_0)|,|\mathrm{Im}(p_1)|<\eta$ for some $\eta>0$ (proportional to $|\mu|$).
\end{proof}

For the proof of Corollary~\ref{cor:zle0}, we need the following observation from \cite{MR496191} adapted to our setting.

\begin{lemma}
	Let $\beta \in (0,6\pi)$ and $m,\nz>0$. Then for any $f \in C^\infty_c(\R^2)$,
	with $f_x(y)= f(y-x)$,
	\begin{equation} \label{e:RL}
		\avg{\varphi(f)\varphi(f_x)}_{\SG(\beta,z|m)} \to 0 \qquad (|x|\to\infty).
	\end{equation}
\end{lemma}

\begin{proof}
	The argument is as in the proof of \cite[Theorem~4.4]{MR496191}.
	Indeed, by Theorem~\ref{thm:fieldlim},
	the measure $\avg{\cdot}_{\SG(\beta,z|m)}$ is translation invariant and satisfies, for any $f \in C_c^\infty(\R^2)$,
	\begin{equation}
		\avg{\varphi(f)^2}_{\SG(\beta,z|m)} \leq (f,(-\Delta+m^2)^{-1}f) = \int_{\R^2} \frac{dp}{(2\pi)^2} \frac{|\hat f(p)|^2}{|p|^2+m^2} .
	\end{equation}
	Therefore $C_f(x) = \avg{\varphi(f)\varphi(f_x)}_{\SG(\beta,z|m)}$ 
	satisfies 
	\begin{equation}
		0 \leq \hat C_f(p) \leq \frac{|\hat f(p)|^2}{|p|^2+m^2} \in L^1(\R^2)
	\end{equation}
	in the distributional sense. Indeed, this follows from
		\begin{equation}
			\avg{(\varphi*f)(g)^2}_{\SG(\beta,z)} = \int_{\R^2 \times \R^2}dx\, dy\, g(x) C_f(x-y) g(y) 
			= \int_{\R^2}\frac{dp}{(2\pi)^2}\, |\hat g(p)|^2 \hat C_f (p).
	\end{equation}
	Thus the Riemann--Lebesgue lemma implies that
	\begin{equation}
		\avg{\varphi(f)\varphi(f_x)} = \int_{\R^2} \frac{dp}{(2\pi)^2}\, e^{i p\cdot x} \hat C_f(p) \to 0 \qquad (|x|\to \infty)
	\end{equation}
	as claimed.
\end{proof}

\begin{proof}[Proof of Corollary~\ref{cor:zle0}]
	By \eqref{e:RL}, for $m>0$,
	\begin{equation}
		\avg{\varphi(f)^2}_{\SG(\beta,z|m)}
		= \frac12 \lim_{|x|\to\infty}   \avg{(\varphi(f)-\varphi(f_x))^2}_{\SG(\beta,z|m)}.
	\end{equation}
	By monotonicity in $m$ and $L$ due to \eqref{e:monotone}, the limits $m\to 0$ and $L\to \infty$ exist in both orders and,
	if $\int dx\, f  = 0$,
	\begin{equation}
		\sup_{m>0}\lim_{L\to\infty} \avg{\varphi(f)^2}_{\SG(\beta,z|m,\Lambda_L)}
		\leq
		\lim_{L\to\infty} \sup_{m> 0} \avg{\varphi(f)^2}_{\SG(\beta,z|m,\Lambda_L)}
		= \avg{\varphi(f)^2}_{\SG(\beta,z)}
		.
	\end{equation}
	In conclusion, we get that
	\begin{align}
		\sup_{m > 0} \avg{\varphi(f)^2}_{\SG(\beta,z|m)}
		&=
		\frac12 \sup_{m>0}\lim_{|x|\to\infty}   \avg{(\varphi(f)-\varphi(f_x))^2)}_{\SG(\beta,z|m)}
		\nnb
		&\leq
		\frac12 \limsup_{|x|\to\infty} \sup_{m>0}\avg{(\varphi(f)-\varphi(f_x))^2)}_{\SG(\beta,z|m)}
		\nnb
		&\leq
		\frac12 \limsup_{|x|\to\infty} \avg{(\varphi(f)-\varphi(f_x))^2)}_{\SG(\beta,z)}.
	\end{align}
	For $\beta=4\pi$, the right-hand side is finite by Theorem \ref{thm:field}.
	In fact, since $|\hat f|^2\hat C_\mu$
	is integrable for $\mu \neq 0$, by the Riemann-Lebesgue lemma, it is equal to
	\begin{equation}
		\limsup_{x\to\infty} \int_{\R^2}\frac{dp}{(2\pi)^2}\, |\hat f(p)|^2 (1-\cos(p\cdot x)) \hat C_{Az}(p)
		=
                 \int_{\R^2} \frac{dp}{(2\pi)^2}\, |\hat f(p)|^2 \hat C_{Az}(p) 
	\end{equation}
	as claimed.
        
        The proof of the existence of the infinite volume measure as $m\to 0$ is now exactly as in the proof of
        Theorems~\ref{thm:gradientlim} and \ref{thm:fieldlim}, only using the now proved bound \eqref{e:zle0}
        instead of the last bound in \eqref{e:laplacecont} for the continuity of the characteristic functional.
\end{proof}

\section{The sine-Gordon model: the renormalized potential}
\label{sec:renormpot}

One of our main tools in the proof of Theorem~\ref{thm:cf} are estimates for a renormalized version
of the sine-Gordon potential, and we turn to studying it now.
For $\varphi\in C_b(\R^2)$ and $\zeta\in L^\infty_c(\R^2\times \{-1,1\},\C)$,
we define
\begin{equation}
  v_0(\zeta,\varphi|\epsilon) = \epsilon^{-\frac{\beta}{4\pi}}
  \int_{\R^2\times \{-1,1\}} d\xi \, \zeta(\xi)e^{i\sqrt{\beta}\sigma \varphi(x)}, 
  \label{e:micropotdef}
\end{equation}
which we refer to as the microscopic (sine-Gordon) potential.
In terms of this microscopic potential, we introduce the following
generalized partition function
that can be seen as a generating function for charge correlation functions:
\begin{equation} \label{e:Zmicropot}
Z(\zeta|\epsilon,m)=\avga{e^{-v_0(\zeta,\varphi|\epsilon)}}_{\GFF(\epsilon,m)},
\end{equation}
where the GFF expectation is over $\varphi$.
For $\nzeta=-\nz\mathbf 1_{\Lambda}$ this would just be the partition function of the (regularized) sine-Gordon model.

Our analysis of the generating function $Z(\zeta|\epsilon,m)$ relies on a convenient decomposition of the regularized free field $\GFF(\epsilon,m)$. More precisely, we define for any $t,m>0$ and $x,y\in \R^2$ with $x\neq y$ 
\begin{equation}
	c_t^{m^2}(x-y):=\int_0^t ds \, \dot c_s^{m^2}(x-y):=\int_0^t ds \, e^{-m^2 s}\frac{e^{-\frac{|x-y|^2}{4s}}}{4\pi s} .
	\label{e:tcov}
\end{equation} 
For any $t>\epsilon^2$, note that
$c_t^{m^2}-c_{\epsilon^2}^{m^2}$ and $c_\infty^{m^2}-c_t^{m^2}$ are 
covariances, so
the fact that the sum of two independent Gaussian processes is a Gaussian process whose
covariance is the sum of the covariances of the two processes
implies that we can in fact write \eqref{e:Zmicropot} as
\begin{equation}
  Z(\zeta|\epsilon,m) = 
  \avga{e^{-v_t(\zeta,\varphi|\epsilon,m)}}_{\GFF(\sqrt{t},m)}
  \label{e:Zvt}
\end{equation}
where we have defined the renormalized potential $v_t$ by
\begin{equation}
  e^{-v_t(\zeta,\varphi|\epsilon,m)} = 
  \EE_{c_t^{m^2}-c_{\epsilon^2}^{m^2}}\pa{e^{-v_0(\zeta,\varphi+\eta|\epsilon)}},
  \label{e:pfrpot}
\end{equation}
and have written $\EE_{c_t^{m^2}-c_{\epsilon^2}^{m^2}}$ for the expectation with respect to the law of the Gaussian process with covariance $c_t^{m^2}-c_{\epsilon^2}^{m^2}$ and the last integral is over $\eta$.

The analysis of the $\epsilon,m\to 0$ behavior of the generating 
function $Z(\zeta|\epsilon,m)$ can thus be rephrased in terms of
$\epsilon,m\to 0$ asymptotics of the renormalized potential $v_t(\zeta,\cdot|\epsilon,m)$.
Note that as $\zeta$ is complex, only $e^{-v_t(\zeta,\varphi|\epsilon,m)}$ is 
a priori well defined, but we will see in this section that for any given $\zeta \in L^\infty_c(\R^2\times \{-1,1\})$ and $t$ small enough,
its logarithm $v_t(\zeta,\varphi|\epsilon,m)$ is also well-defined. 
Moreover, the goal of this section is to prove bounds for $v_t(\zeta,\varphi|\epsilon,m)$ that are uniform in $\epsilon>0$ and $m>0$.
Our analysis follows the approach of \cite{MR914427} as presented in \cite[Section 3]{MR4303014},
but it permits space-dependent coupling constants and we also work directly in the continuum.
As discussed in Section~\ref{sec:outline},
we expect that similar results could be obtained by using the methods
of \cite{MR814849,MR2461991}.
The $\epsilon \to 0$ and $m\to 0$ limits will be studied in Section~\ref{sec:sg-expansion}.

To control $v_t$ we will show in this section that the following expansion is convergent and agrees with $v_t(\varphi,\zeta|\epsilon,m)$
for $\zeta \in L^\infty_c(\R^2\times \{-1,1\})$ and suitable $t$:
\begin{equation}
   \sum_{n=1}^\infty \frac{1}{n!}\int_{(\R^2\times \{-1,1\})^n}d\xi_1\cdots d\xi_n\,\zeta(\xi_1)\cdots \zeta(\xi_n)
   \widetilde v_t^{n}(\xi_1,...,\xi_n|\epsilon,m)e^{i\sqrt{\beta}\sum_{j=1}^n \sigma_j\varphi(x_j)}
   \label{e:yukawa}
\end{equation}
where the coefficients $\widetilde v_t^{n}$ are determined recursively as follows.
For $t>\epsilon^2$ and $\xi\in \R^2\times \{-1,1\}$, we set
\begin{equation}
  \widetilde v_t^{1}(\xi|\epsilon,m) = 
  e^{-\frac{\beta}{2}(\int_{\epsilon^2}^{t}ds \, \dot c_s^{m^2}(0)+\frac{1}{4\pi}\log \epsilon^2)},
  \label{e:v1def}
\end{equation}
and for $n\geq 2$ and $\xi_j=(x_j,\sigma_j)\in \R^2\times \{-1,1\}$, 
\begin{align}   \label{e:vndef}
\widetilde  v_t^{n}(\xi_1,...,\xi_n|\epsilon,m) &= 
  \frac{1}{2}\int_{\epsilon^2}^tds \sum_{I_1\dot{\cup}I_2=[n]}
  \sum_{i\in I_1,j\in I_2}\dot u_s^{m^2}(\xi_i,\xi_j)\widetilde v_s^{|I_1|}(\xi_{I_1}|\epsilon,m)\widetilde v_s^{|I_2|}(\xi_{I_2}|\epsilon,m)\notag \\
  & \qquad \qquad \times e^{-(w_t^{m^2}(\xi_1,...,\xi_n)-w_s^{m^2}(\xi_1,...,\xi_n))},
\end{align}
where
\begin{align}
  \dot{u}_s^{m^2}(\xi_1,\xi_2)
  &= 
  \beta \sigma_1\sigma_2 \dot{c}_s^{m^2}(x_1-x_2),
  \\
  w_t^{m^2}(\xi_1,\dots,\xi_n)-w_s^{m^2}(\xi_1,\dots,\xi_n)
  &= 
  \frac{1}{2}\sum_{i,j=1}^n \int_s^t dr\, \dot u_r^{m^2}(\xi_i,\xi_j).
  \label{e:uwdef}
\end{align}
We have also written $[n]=\{1,\dots,n\}$ and $I_1\dot \cup I_2=[n]$ to indicate that $I_1\cap I_2=\emptyset$ and $I_1\cup I_2=[n]$.
For controlling the expansion \eqref{e:yukawa}, 
we introduce the following norms for $f: (\R^2\times \{-1,1\})^n\to \C$:
\begin{equation}
  \|f\|_n=\begin{cases}
  \sup_{\xi\in \R^2\times \{-1,1\}}|f(\xi)|, & \text{if} \quad n=1,\\
  \sup_{\xi_1\in \R^2\times \{-1,1\}}\int_{(\R^2\times \{-1,1\})^{n-1}}d\xi_2\cdots d\xi_n|f(\xi_1,\dots,\xi_n)|, & \text{if} \quad n\geq 2.
  \end{cases}
  \label{e:norms}
\end{equation}

The goal of the rest of this section is to prove the following proposition.
In its statement, the condition $\beta<6\pi$ necessitates the exclusion of the $n=2$ term as
the analogous estimate fails when $\beta \geq 4\pi$, see also Remark~\ref{rk:conjecture8pi} below.
The $n=2$ term will be considered explicitly later.

\begin{proposition} \label{prop:mayer}
	For $\beta\in(0,6\pi)$, $t>0$, and $n\neq 2$, there exists functions $h_t^n:(\R^2\times \{-1,1\})^n\to[0,\infty]$ which are independent of $\epsilon,m$ and for  $0<\epsilon^2<t<m^{-2}$, one has
	\begin{equation}
	   |\widetilde v_t^{n}(\xi_1,\dots,\xi_n|\epsilon,m)|\leq h_t^{n}(\xi_1,\dots,\xi_n)
	  \label{e:ybd1}
	\end{equation}
	for all $\xi_1,\dots,\xi_n\in (\R^2\times \{-1,1\})^n$ and 
	\begin{equation}
	  \|h_t^{n}\|_n\leq n^{n-2}t^{-1}\left(C_\beta t^{1-\frac{\beta}{8\pi}}\right)^n
	  \label{e:ybd2}
	\end{equation}
	for some constant $C_\beta$ depending only on $\beta$.
	\label{pr:ybound}
\end{proposition}

\begin{remark}\label{rk:conjecture8pi}
It remains a conjecture \cite[p.672]{MR814849} that similar estimates remain valid
for all $\beta<8\pi$
when not only the $n=2$ term is excluded but when the first $n_0$ terms are excluded where
$n_0$ is the largest integer such that $2(n_0-1)-\beta n_0/4\pi \leq 0$.
(The results of \cite{MR1777310,MR849210} which do construct the (massive) sine-Gordon model for all
$\beta < 8\pi$ do not proceed by this expansion and instead rely on probabilistic estimates on large gradients,
thus leaving this stronger conjecture open.)
\end{remark}

Proposition~\ref{prop:mayer} 
allows us to identify the expansion \eqref{e:yukawa} with the renormalized potential as follows.

\begin{corollary}   \label{cor:yukawa}	
  For all $0<\epsilon^2<t<m^{-2}<\infty$, $\varphi\in C^\infty(\R^2)$,
  and $\zeta\in L^\infty_c(\R^2\times \{-1,1\},\C)$ satisfying
  \begin{equation} \label{e:zetasmall}
    \sup_{\xi\in \R^2\times \{-1,1\}}|\zeta(\xi)| 
    < \frac{1}{e C_\beta t^{1-\beta/8\pi}}
  \end{equation}
  where $C_\beta$ is the constant from Proposition~\ref{pr:ybound},
  the sums and integrals in the expansion \eqref{e:yukawa} converge absolutely and equal $v_t(\zeta,\varphi|\epsilon,m)$ defined
  in    \eqref{e:pfrpot}:
  \begin{multline}\label{e:vtexpan}
    v_t(\zeta,\varphi|\epsilon,m)= \sum_{n=1}^\infty \frac{1}{n!} \int_{(\R^2\times \{-1,1\})^n}d\xi_1\cdots d\xi_n\,
    \zeta(\xi_1)\cdots \zeta(\xi_n)
    \\\times  
    \widetilde v_t^{n}(\xi_1,\dots,\xi_n|\epsilon,m)e^{i\sqrt{\beta}\sum_{j=1}^n \sigma_j\varphi(x_j)}
    .
  \end{multline}  
\end{corollary}

\begin{proof}
Throughout the proof, we fix $\epsilon,m>0$ and $\zeta \in L^\infty(\R^2\times \{-1,1\})$ with support in a compact set $\Lambda \times \{-1,1\} \subset \R^2\times \{-1,1\}$,
and we will always assume that $t \in (\epsilon^2,t_0)$ where $t_0$ is the supremum over $t>\epsilon^2$ such that \eqref{e:zetasmall} holds.
Then, for $n \geq 3$,
\begin{multline} \label{e:vtsumbd}
  \frac{1}{n!}
  \int_{(\R^2\times \{-1,1\})^n}d\xi_1\cdots d\xi_n
  \absa{ \zeta(\xi_1)\cdots \zeta(\xi_n)
    \widetilde v_t^{n}(\xi_1,\dots,\xi_n|\epsilon,m)e^{i\sqrt{\beta}\sum_{j=1}^n \sigma_j\varphi(x_j)} } 
  \\
  \leq
  \frac{1}{n!}2|\Lambda|
  \|\zeta\|_{L^\infty(\R^2 \times \{\pm 1\})}^{n}
  \|h_t^n\|_n
  \leq
  2 |\Lambda|  n^{-2} t^{-1} (t/t_0)^{(1-\beta/8\pi)n}
\end{multline}
where we used $n^n/n! \leq e^n$.
The $n=1,2$ terms are  trivially bounded with $\epsilon,m$-dependent constants when $t > \epsilon^2$, uniformly in $\varphi: \R^2 \to \R$,
by the definitions \eqref{e:v1def}--\eqref{e:vndef}.
For $t<t_0$, it follows that the sum over $n$ in \eqref{e:yukawa} converges geometrically,
again uniformly in $\varphi: \R^2 \to \R$.
We denote this sum by $a_t(\varphi)$ and note that $a_t(\varphi)$ only depends on $\varphi|_\Lambda$,
so that we can consider $\varphi \mapsto a_t(\varphi)$ as a function $a_t: C(\Lambda) \to \R$.
We will denote the supremum norm on $C(\Lambda)$ by $\|\cdot\|$ below.
From the geometric convergence,
\begin{equation}
\left|e^{i\sqrt{\beta}\sum_{j=1}^n \sigma_j f(x_j)}-1-i\sqrt{\beta}\sum_{j=1}^n \sigma_jf(x_j)\right|\leq \frac{1}{2}\beta n^2 \|f\|^2,
\end{equation}
and similar estimates for higher derivatives, we then see that 
$a_t: C(\Lambda) \to \R$ is actually smooth, i.e., Frechet differentiable to any order, for $t \in (\epsilon^2,t_0)$.
Its first two derivatives are given by
\begin{align}
  D a_t(\varphi; f_1)
  &=
    i\sqrt{\beta}
    \sum_{n=1}^\infty \frac{1}{n!}{\int_{(\R^2\times \{-1,1\})^n}d\xi_1\cdots d\xi_n} \, \zeta(\xi_1)\cdots \zeta(\xi_n)
    \nnb
    &\qquad 
      \times \sum_{k=1}^n \sigma_k f_1(x_k)
      \widetilde v_t^{n}(\xi_1,\dots,\xi_n|\epsilon,m)e^{i\sqrt{\beta}\sum_{j=1}^n \sigma_j\varphi(x_j)} 
      \nnb
  &=: \int dx_1\, f_1(x_1) \nabla a_t(\varphi,x_1),
    \\
  D^2 a_t(\varphi; f_1, f_2)
  &=
    -\beta
    \sum_{n=1}^\infty \frac{1}{n!}{\int_{(\R^2\times \{-1,1\})^n}d\xi_1\cdots d\xi_n} \, \zeta(\xi_1)\cdots \zeta(\xi_n)
    \nnb
    &\qquad 
      \times \sum_{k=1}^n \sigma_k f_1(x_k)\sum_{l=1}^n \sigma_l f_2(x_l)
      \widetilde v_t^{n}(\xi_1,\dots,\xi_n|\epsilon,m)e^{i\sqrt{\beta}\sum_{j=1}^n \sigma_j\varphi(x_j)}.
  \nnb
    &=: \int dx_1 \, dx_2\, f_1(x_1)f_2(x_2) \He a_t(\varphi,x_1,x_2) ,
\end{align}
where $f_1, f_2 \in C(\Lambda)$.
As in \eqref{e:vtsumbd}, 
$\|\nabla a_t(\varphi, \cdot)\|_{L^\infty(\R^2)}$ and $\|\He a_t(\varphi,\cdot,\cdot)\|_{L^1L^\infty(\R^2\times \R^2)}$ are bounded independently of $\varphi$
and, since $\zeta$ has support in $\Lambda$, it is also clear that $\nabla a_t(\varphi,\cdot)$ has support in $\Lambda$
and that $\He a_t(\varphi,\cdot,\cdot)$ has support in $\Lambda^2$.
Defining
\begin{align} \label{e:Deltaa}
  \Delta_{\dot c_t} a_t(\varphi)
  &= \int_{\Lambda^2} dx_1 \, dx_2\, \dot c_t^{m^2}(x_1-x_2) \He a_t(\varphi,x_1,x_2)
  \\
  (\nabla a_t(\varphi), \dot c_t^{m^2} \nabla a_t(\varphi))
  &= \int_{\Lambda^2} dx_1 \, dx_2\, \dot c_t^{m^2}(x_1-x_2) \nabla a_t(\varphi,x_1) \nabla a_t(\varphi,x_2) 
    ,
\end{align}
it then follows from \eqref{e:vndef} that, for $t \in (\epsilon^2,t_0)$,
\begin{align}
  \ddp{}{t} a_t(\varphi) = \frac12 \Delta_{\dot c_t} a_t(\varphi) - \frac12 (\nabla a_t(\varphi), \dot c_t^{m^2} \nabla a_t(\varphi)).
\end{align}
Let $h_t(\varphi) = e^{-a_t(\varphi)}$. Then by the chain rule, $h_t$ is also twice Frechet differentiable with
(using similar notation as above)
\begin{align}
  \nabla h_t(\varphi,x_1) &= - \nabla a_t(\varphi,x_1) e^{-a_t(\varphi)}
  \\
  \He h_t(\varphi,x_1,x_2) &= \qa{- \He a_t(\varphi,x_1,x_2) + \nabla a_t(\varphi,x_1)\nabla a_t(\varphi,x_2)} e^{-a_t(\varphi)},
\end{align}
and hence 
\begin{equation}
  \frac12 \Delta_{\dot c_t} h_t(\varphi) - \ddp{}{t} h_t(\varphi)
  = \qa{-\frac12 \Delta_{\dot c_t} a_t(\varphi) + \frac12(\nabla a_t(\varphi), \dot c_t^{m^2} \nabla a_t(\varphi)) + \ddp{}{t} a_t(\varphi)} e^{-a_t(\varphi)} = 0.
\end{equation}
We will show that $e^{-v_t}$ satisfies this same heat equation (with the same initial data at $t=\epsilon^2$) and argue that the solution must be unique, so $v_t=a_t$, which will then yield the proof.

The Laplacian $\Delta_{\dot c_t}h_t$ can alternatively be expressed as follows. 
Since $\Lambda$ is bounded, assume that $\Lambda \subset [-L,L]^2$.
Let $\chi$ be a smooth function with $\chi(t) = 1$ for $t\leq 4L$ and $\chi(t)=0$ for $t \geq 8L$.
We then choose a torus $\Lambda'$ of period $16L$
and set $\dot c_t'(x) = \sum_{n\in\Z^2} \dot c_t^{m^2}(x+16Ln)\chi(|x+16Ln|)$.
Thus $\dot c_t'$ is a smooth (periodic) function on $\Lambda'$,
and we note that $\dot c_t'(x)= \dot c_t^{m^2}(x)$ if $|x| \leq 4L$.
Thus if we regard $\Lambda$ as a subset of $\Lambda'$ (by embedding it into a fundamental domain centered at $0$ in the obvious way),
we have $\dot c_t'(x-y) =  \dot c_t^{m^2}(x-y)$, for $x,y\in \Lambda$.
In particular, there are $\dot \lambda_{t,k} \geq 0$ decaying rapidly in $k$ for each $t>0$ such that
\begin{align}
  \dot c_t^{m^2}(x-y) = \dot c_t'(x-y) =  \sum_{k} \dot\lambda_{k,t} f_{k}(x) f_{k}(y), \qquad \text{for $x,y \in \Lambda$},
\end{align}
where $(f_k)$ is the real orthonormal Fourier basis of $L^2(\Lambda')$ consisting of $\sin$ and $\cos$ functions,
so in particular satisfying $\|f_k\| \leq C$.
For a general function $g \in C^2_b(C(\Lambda))$ and $t>0$ we can now define
\begin{equation}
  \Delta_{\dot c_t} g(\varphi) = \sum_{k} \dot\lambda_{k,t} D^2 g(\varphi; f_k,f_k).
\end{equation}
By Fubini (whose application is justified by rapid convergence of all sums and integrals),
this definition is consistent with \eqref{e:Deltaa}.

Let $\Pi_N \varphi$ be the $L^2(\Lambda')$ projection of $\varphi|_{\Lambda'}$ to Fourier modes $k \leq N$.
For any $N$, the above implies that $h_t^N(\varphi) = h_t(\Pi_N\varphi)$
satisfies the finite dimensional heat equation
\begin{equation}
  \partial_t h_t^N(\varphi)
  =
  \frac12 \sum_{k \leq N} \dot\lambda_{k,t} D^2 h_t^N(\varphi; f_k,f_k)
  ,\qquad h_{\epsilon^2}^N(\varphi) = h_{\epsilon^2}(\Pi_N\varphi).
\end{equation}
Next we will verify that $g_t(\varphi) = e^{-v_t(\varphi)}$
defined in \eqref{e:pfrpot} also satisfies the heat equation $\partial_tg_t = \frac12 \Delta_{\dot c_t}g_t$
with the same initial condition $g_{\epsilon^2}=h_{\epsilon^2}$.
To see this, first observe that the definition of $g_t(\varphi)$ in \eqref{e:pfrpot} only depends on $\eta|_\Lambda$.
The Gaussian field $\eta|_\Lambda$ has covariance $c_t^{m^2}-c_{\epsilon^2}^{m^2}|_{\Lambda\times \Lambda}$
and can be realized
in terms of independent standard Gaussian random variables $(X_k)_{k \in \N}$ as
\begin{equation}
  \eta|_{\Lambda} = \sum_{k} \sqrt{\lambda_{k,t}} X_k f_k|_{\Lambda}
\end{equation}
where $\dot\lambda_{k,t}$ above is the $t$-derivative of these $\lambda_{k,t}$. (This follows from the fact that $\lambda_{k,t} = (f_k, c'_t f_k)$
and the differentibility of $c_t'$ in $t$.)
From this representation we again see that $g_t^N(\varphi) = g_t(\Pi_N \varphi)$ satisfies
\begin{equation}
  \partial_t g_t^N(\varphi)
  =
  {\frac12}  \sum_{k \leq N} \dot\lambda_{k,t} D^2 g_t^N(\varphi; f_k,f_k)
  \qquad g_{\epsilon^2}^N(\varphi) = g_{\epsilon^2}(\Pi_N\varphi) = h_{\epsilon^2}(\Pi_N\varphi).
\end{equation}
By the standard uniqueness of bounded solutions to such equations
(finite dimensional heat equations), we conclude that $h_t^N(\varphi) = g_t^N(\varphi)$ for all $t\in(\epsilon^2,t_0)$ and $N \in \N$.
It remains to conclude that this implies that $g_t(\varphi)=h_t(\varphi)$ for all smooth $\varphi$.
Indeed, $\|\Pi_N\varphi-\varphi\|\to 0$ for any smooth $\varphi$ and since both $g_t$ and $h_t$ are continuous in $\varphi \in C(\Lambda)$,
thus $h_t^N(\varphi)= h_t(\Pi_N\varphi) \to h_t(\varphi)$ as $N\to\infty$ and analogously $g_t^N(\varphi) \to g_t(\varphi)$.
\end{proof}

\subsection{Covariance and (massive) heat kernel estimates}

For the proof of Proposition~\ref{pr:ybound}, we require some basic estimates 
for the covariance $c_t^{m^2}$ and the (massive) heat kernel $\dot c_t^{m^2}$.
We turn to recording these now.
The most basic estimate we shall have use for is just for $c_t^{m^2}(x)$.

\begin{lemma}\label{le:covest1}
	There exists a universal constant $C>0$ such that for $0<t<m^{-2}$ and $x\in \R^2$, we have the estimate
	\begin{equation}
	  \left|c_t^{m^2}(x)+\frac{1}{2\pi}\log\left(\frac{|x|}{\sqrt{t}}\wedge 1\right)\right|\leq C.
	  \label{e:covest1}
	\end{equation}
\end{lemma}
\begin{proof}
	Let us write 
	\begin{equation}
	  c_t^{m^2}(x) = \int_0^{t}ds\,\frac{e^{-m^2 s}}{4\pi s}e^{-\frac{|x|^2}{4s}}
	  = \int_0^{\frac{t}{|x|^2}}ds\,\frac{e^{-m^2 s |x|^2}}{4\pi s}e^{-\frac{1}{4s}}.
	\end{equation}
	For $\frac{|x|}{\sqrt{t}}\geq 1$, we see that 
	\begin{equation}
	  0\leq c_t^{m^2}(x)\leq \int_0^1ds\,\frac{e^{-\frac{1}{4s}}}{4\pi s}<\infty,
	\end{equation}
	so it is sufficient to focus on the regime $|x|<\sqrt{t}$.
	Here (using $|e^{-x}-1|\leq x$ for $x>0$)
	\begin{align}
	  \left| c_t^{m^2}(x)+\frac{1}{2\pi}\log \frac{|x|}{\sqrt{t}}\right|
          &  \leq \int_1^{\frac{t}{|x|^2}}ds\,\frac{|e^{-m^2 s |x|^2-\frac{1}{4s}}-1|}{4\pi s}+\int_0^1 ds\,\frac{e^{-\frac{1}{4s}}}{4\pi s} \nnb
	  & \leq \frac{m^2 t}{4\pi}+\frac{1}{16\pi}\int_1^\infty \frac{ds}{s^2}+\int_0^1ds\, \frac{e^{-\frac{1}{4s}}}{4\pi s}.
	\end{align}
	Recalling that we are assuming that $m^2 t\leq 1$, this concludes the proof.
\end{proof}

The next estimate is slightly more involved.
\begin{lemma}\label{le:covest2}
	For $\beta\in(0,6\pi)$, $\xi_1,\xi_2,\xi_3\in \R^2\times \{-1,1\}$, and $0<\epsilon^2\leq t\leq m^{-2}$ we have 
	\begin{equation}
	  \left|\dot u_t^{m^2}(\xi_1,\xi_3)+\dot u_t^{m^2}(\xi_2,\xi_3)\right|\left|1-e^{-\sigma_1\sigma_2\beta \int_{\epsilon^2}^{t} ds\,\dot c_s^{m^2}(x_1-x_2)}\right|\leq F_t(\xi_1,\xi_2,\xi_3)
	\end{equation}
	for some function $F_t:(\R^2\times \{-1,1\})^3\to [0,\infty]$ which is invariant under permutations of the coordinates, independent of $\epsilon,m$, and in the notation \eqref{e:norms}, satisfies 
	\begin{equation}
	\|F_t\|_3\leq C_\beta t
	\end{equation}
	for some constant $C_\beta$ depending only on $\beta$. 
\end{lemma}

\begin{proof}
	The proof is slightly lengthy and we split it into two parts.
	
	\medskip 
	
	\underline{Case 1: $\sigma_1\neq \sigma_2$:} Let us first consider $\sigma_1\neq \sigma_2$ and bound the quantity
	$|\dot u_t^{m^2}(\xi_1,\xi_3)+\dot u_t^{m^2}(\xi_2,\xi_3)|=\beta|\dot c_t^{m^2}(x_1-x_3)-\dot c_t^{m^2}(x_2-x_3)|$. 
	Let us write $[x_1,x_2]$ for the line segment from $x_1$ to $x_2$. 
	We then have by the mean value theorem (recalling that $t\leq m^{-2}$)
	\begin{align}
	\left|\dot c_t^{m^2}(x_1-x_3)-\dot c_t^{m^2}(x_2-x_3)\right| & =
	\frac{e^{-m^2t}}{4\pi t}
	\left|e^{-\frac{|x_1-x_3|^2}{4t}}-e^{-\frac{|x_2-x_3|^2}{4t}}\right|\notag\\
	&\leq \frac{1}{4\pi t}
	|x_1-x_2| \sup_{u\in[x_1,x_2]}
	\left|\nabla_u e^{-\frac{|u-x_3|^2}{4t}}\right|. \label{e:1diff}
	\end{align}
        To bound the gradient, we use that
        for any $\alpha>0$, 
	there exists $A(\alpha)$ (depending only on $\alpha$) such that 
	\begin{equation}
	\left|\nabla_u e^{-\frac{|u-x_3|^2}{4t}}\right|
	\leq \frac{|u-x_3|}{t} e^{-\frac{|u-x_3|^2}{4t}}
	\leq A(\alpha) t^{-1/2}e^{-\alpha\frac{|u-x_3|}{\sqrt{t}}}.
	\end{equation}
	We used here the estimate that there 
	exists a $A(\alpha)$ such that $xe^{-\frac{x^2}{4}}\leq A(\alpha)e^{-\alpha x}$ for all $x>0$. 
	
	From the triangle inequality we find that, for $u\in [x_1,x_2]$,
	\begin{align}
	-|u-x_3|\leq |x_1-u|-|x_1-x_3|\leq |x_1-x_2|-|x_1-x_3|.
	\label{e:triang}
	\end{align}
	This leads to the following bound: for any $\alpha>0$
	\begin{equation}
	\left|\dot c_t^{m^2}(x_1-x_3)-\dot c_t^{m^2}(x_2-x_3)\right| \leq A(\alpha)\frac{1}{4\pi t^{3/2}}|x_1-x_2| e^{\alpha \frac{|x_1-x_2|}{\sqrt{t}}} e^{-\alpha\frac{|x_1-x_3|}{\sqrt{t}}}.
	\label{e:diffbound}
	\end{equation}
	
	The second term in our statement we bound with the following estimate which is a consequence of Lemma \ref{le:covest1} (recall that $\int_{\epsilon^2}^t\dot c_r^{m^2}(x_1-x_2)dr\leq c_t^{m^2}(x_1-x_2)$):
	\begin{align}\label{e:expest}
	\left|1-e^{\beta \int_{\epsilon^2}^{t} dr\, \dot c_r^{m^2}(x_1-x_2)}\right|&\leq e^{\beta c_t^{m^2}(x_1-x_2)}-1\notag \\
	&= \int_0^t dr\, \beta \dot c_r^{m^2}(x_1-x_2)e^{\beta c_r^{m^2}(x_1-x_2)}\notag \\
	&\leq C\int_0^t \frac{dr}{r}e^{-\frac{|x_1-x_2|^2}{4r}} e^{-\frac{\beta}{2\pi}\log \left[\frac{|x_1-x_2|}{\sqrt{r}}\wedge 1\right]},
	\end{align}
	where the constant is universal.
	
	Combining our estimates, we see that for each $\alpha>0$, there 
	exist $A(\alpha),\widetilde A(\alpha)$ (depending only on $\alpha$ and possibly different from our previous $A(\alpha)$) such that 
	\begin{align}
	&\left|\dot u_{t}^{m^2}(\xi_1,\xi_3)+\dot u_{t}^{m^2}(\xi_2,\xi_3)\right|\left|1-e^{-\sigma_1\sigma_2\beta \int_{\epsilon^2}^{t}ds\, \dot c_s^{m^2}(x_1-x_2)}\right|\nnb
	&\qquad \leq A(\alpha)t^{-3/2}\int_0^t \frac{dr}{r}|x_1-x_2| e^{\alpha \frac{|x_1-x_2|}{\sqrt{t}}} e^{-\alpha\frac{|x_1-x_3|}{\sqrt{t}}}e^{-\frac{|x_1-x_2|^2}{4r}} e^{-\frac{\beta}{2\pi}\log \left[\frac{|x_1-x_2|}{\sqrt{r}}\wedge 1\right]}\nnb
	&\qquad \leq \widetilde A(\alpha)t^{-3/2}\int_0^t \frac{dr}{r}|x_1-x_2| e^{\alpha \frac{|x_1-x_2|}{\sqrt{t}}} e^{-\alpha\frac{|x_1-x_3|}{\sqrt{t}}}e^{-2\alpha\frac{|x_1-x_2|}{\sqrt{r}}} e^{-\frac{\beta}{2\pi}\log \left[\frac{|x_1-x_2|}{\sqrt{r}}\wedge 1\right]}\nnb
	&\qquad \leq \widetilde A(\alpha)t^{-3/2}\int_0^t\frac{dr}{r} |x_1-x_2| e^{-\alpha\frac{|x_1-x_3|}{\sqrt{t}}}e^{-\alpha\frac{|x_1-x_2|}{\sqrt{r}}} e^{-\frac{\beta}{2\pi}\log \left[\frac{|x_1-x_2|}{\sqrt{r}}\wedge 1\right]},
	\end{align}
	where we made use of the estimate that for some $A(\alpha)$, $e^{-x^2}\leq A(\alpha)e^{-4\alpha x}$ for $x>0$ and that for $r\leq t$, $e^{\alpha \frac{|x_1-x_2|}{\sqrt{t}}-\alpha \frac{|x_1-x_2|}{\sqrt{r}}}\leq 1$. To summarize, choosing $\alpha=1$ we have the bound 
	\begin{align}
		&\left|\dot u_{t}^{m^2}(\xi_1,\xi_3)+\dot u_{t}^{m^2}(\xi_2,\xi_3)\right|\left|1-e^{-\sigma_1\sigma_2\beta \int_{\epsilon^2}^{t}ds \,\dot c_s^{m^2}(x_1-x_2)}\right|\nnb
		&\qquad \leq Ct^{-3/2}\int_0^t \frac{dr}{r} |x_1-x_2| e^{-\frac{|x_1-x_3|}{\sqrt{t}}}e^{-\frac{|x_1-x_2|}{\sqrt{r}}} e^{-\frac{\beta}{2\pi}\log \left[\frac{|x_1-x_2|}{\sqrt{r}}\wedge 1\right]}
	\end{align}
	for some universal constant $C$ and we can then define $F_t$ (at least in the case $\sigma_1\neq \sigma_2$) to be the function obtained by symmetrizing the above function with respect to the variables $x_i$. Note in particular that this is independent of $\epsilon,m$.
	
	To control $\|F_t\|_3$, let us in all of our terms (coming from symmetrization) shift $x_2$ and $x_3$ by $x_1$ so we are left with the estimate 
	\begin{equation}
	\|F_t\mathbf{1}_{\sigma_1\neq \sigma_2}\|_3\leq Ct^{-3/2}\int_0^t \frac{dr}{\sqrt{r}}\int_{\R^2}dx \, e^{-\frac{|x|}{\sqrt{t}}}\int_{\R^2}dy \, \frac{|y|}{\sqrt{r}}e^{-\frac{|y|}{\sqrt{r}}}\left(\frac{|y|}{\sqrt{r}}\wedge 1\right)^{-\frac{\beta}{2\pi}}
	\end{equation}
	for some universal constant $C$. By a change of integration variables, the $x$-integral is some universal constant times $t$ while the $y$-integral is some constant depending on $\beta$ times $r$ (note that the singularity at the origin is integrable precisely for $\beta<6\pi$). Thus for some constant $C_\beta$ (depending only on $\beta$)
	\begin{equation}
	\|F_t \mathbf{1}_{\sigma_1\neq \sigma_2}\|_3\leq C_\beta t^{-1/2}\int_0^t dr\, r^{1/2}\leq C_\beta t
	\end{equation}
	which was the claim. 
	
	\medskip 
	
	\underline{Case 2: $\sigma_1=\sigma_2$:} For $\sigma_1=\sigma_2$, we simply write 
	\begin{equation}
	\left|\dot u_{t}^{m^2}(\xi_1,\xi_3)+\dot u_{t}^{m^2}(\xi_2,\xi_3)\right|\leq \frac{\beta}{4\pi t}\left(e^{-\frac{|x_1-x_3|^2}{4t}}+e^{-\frac{|x_2-x_3|^2}{4t}}\right), 
	\end{equation}
	while for the exponential, we have by Lemma \ref{le:covest1} (for some universal constant $C$)
	\begin{align}
          \left|1-e^{-\beta \int_{\epsilon^2}^{t}dr\,\dot c_r^{m^2}(x_1-x_2)}\right|&\leq 1-e^{-\beta c_t^{m^2}(x_1-x_2)}
                                                                                    \nnb
          &=\int_0^t dr\, \beta \dot c_r^{m^2}(x_1-x_2)e^{-\beta c_r^{m^2}(x_1-x_2)}
          \nnb
          &\leq C\int_0^t dr\, \frac{e^{-\frac{|x_1-x_2|^2}{4r}}}{r}e^{\frac{\beta}{2\pi}\log\left[\frac{|x_1-x_2|}{\sqrt{r}}\wedge 1\right]}.
	\end{align}
    Combining the estimates, we have (for some possibly different universal constant)
    \begin{align}
      &\left|\dot u_{t}^{m^2}(\xi_1,\xi_3)+\dot u_{t}^{m^2}(\xi_2,\xi_3)\right|\left|1-e^{-\beta \int_{\epsilon^2}^{t} dr\, \dot c_r^{m^2}(x_1-x_2)}\right|
      \nnb
    &\qquad \leq Ct^{-1}\left(e^{-\frac{|x_1-x_3|^2}{4t}}+e^{-\frac{|x_2-x_3|^2}{4t}}\right)\int_0^t dr\,  \frac{e^{-\frac{|x_1-x_2|^2}{4r}}}{r}e^{\frac{\beta}{2\pi}\log\left[\frac{|x_1-x_2|}{\sqrt{r}}\wedge 1\right]}.
    \end{align} 
    The relevant function $F_t$ is again obtained by symmetrizing with respect to $x_1,x_2,x_3$.
    
    To estimate the norm, we can again get rid of $x_1$ by a shift of the integration variables. One is left with the estimate 
    \begin{equation}
    \|F_t \mathbf{1}_{\sigma_1= \sigma_2}\|_3\leq C \int_{\R^2}dx\frac{e^{-\frac{|x|^2}{4t}}}{t}\int_0^t dr\int_{\R^2}\frac{dy}{r} e^{-\frac{|y|^2}{4r}}\left(\frac{|y|}{\sqrt{r}}\wedge 1\right)^{\frac{\beta}{2\pi}}\leq \widetilde C t
    \end{equation}
    now for universal constants $C,\widetilde C$. This concludes the proof.
\end{proof}

The final estimate we shall need involves four points. 

\begin{lemma}\label{le:covest3}
	For $\beta\in(0,6\pi)$, $0<\epsilon^2\leq t<m^{-2}$, and $\xi_1,\dots,\xi_4\in \R^2\times \{-1,1\}$, we have
	\begin{multline}
          \left|\sum_{i\in \{1,2\},j\in \{3,4\}}\dot u_t^{m^2}(\xi_i,\xi_j)\right|\left|1-e^{-\beta\sigma_1\sigma_2\int_{\epsilon^2}^tdr\,\dot c_r^{m^2}(x_1-x_2)}\right|\left|1-e^{-\beta\sigma_3\sigma_4\int_{\epsilon^2}^t dr\,\dot c_r^{m^2}(x_3-x_4)}\right|
          \\
          \leq G_t(\xi_1,\xi_2,\xi_3,\xi_4)
	\end{multline}
	for some function $G_t$ which is independent of $\epsilon,m$ and is symmetric in the arguments. Moreover, there exists a constant $C_\beta$ depending only on $\beta$ such that, in the notation \eqref{e:norms},
	\begin{equation}
	\|G_t\|_4\leq C_\beta t^2.
	\end{equation}
\end{lemma}
\begin{proof}
	The proof is very similar to that of Lemma \ref{le:covest2}. We again split it into two cases. 
	
	\medskip
	
	\underline{Case 1: $\sigma_1\neq \sigma_2$ and $\sigma_3\neq \sigma_4$:} Let us begin by considering the case $\sigma_1\neq \sigma_2$ and $\sigma_3\neq \sigma_4$. Arguing as in \eqref{e:1diff}, but noting that now we are dealing with a kind of second order difference, we find bounds in terms of the second order derivative of the heat kernel. Using again an elementary estimate bounding $x^2 e^{-x^2}+e^{-x^2}$ in terms of $e^{-\alpha x}$ times a constant depending only on $\alpha>0$, we find that for each $\alpha>0$ there exists a constant $A(\alpha)$ (depending only on $\alpha$) such that
	\begin{align}
	\left|\sum_{\substack{i\in \{1,2\}\\j\in \{3,4\}}}\dot u_t^{m^2}(\xi_i,\xi_j)\right|&=\left|\dot c_t^{m^2}(x_1-x_3)-\dot c_t^{m^2}(x_2-x_3)-\dot c_t^{m^2}(x_1-x_4)+\dot c_t^{m^2}(x_2-x_4)\right|\nnb
          &\leq A(\alpha)\frac{|x_1-x_2||x_3-x_4|}{t^2}\sup_{\substack{u\in[x_1,x_2]\\v\in[x_3,x_4]}}e^{-\alpha \frac{|u-v|}{\sqrt{t}}}. 
	\end{align}
	Instead of the bound \eqref{e:triang}, we now use the fact (again following from the triangle inequality) that 
	\begin{equation}
	-|u-v|\leq |x_1-u|+|x_3-v|-|x_1-x_3|\leq |x_1-x_2|+|x_3-x_4|-|x_1-x_3|
	\end{equation}
	which yields for our choice of $\sigma$'s that for each $\alpha>0$, there exists $A(\alpha)$ such that
	\begin{equation}
	\left|\sum_{i\in \{1,2\},j\in \{3,4\}}\dot u_t^{m^2}(\xi_i,\xi_j)\right|\leq t^{-2}A(\alpha)|x_1-x_2||x_3-x_4|e^{\alpha\frac{|x_1-x_2|}{\sqrt{t}}+\alpha\frac{|x_3-x_4|}{\sqrt{t}}-\alpha \frac{|x_1-x_3|}{\sqrt{t}}}.
	\end{equation}
	The exponentials we estimate as in \eqref{e:expest} and arguing as in the proof of Lemma \ref{le:covest2} (choosing $\alpha'$ and $\alpha$ in a similar way etc.), we arrive at the bound 
	\begin{align}
	&\left|\sum_{i\in \{1,2\},j\in \{3,4\}}\dot u_t^{m^2}(\xi_i,\xi_j)\right|\left|1-e^{-\beta\sigma_1\sigma_2\int_{\epsilon^2}^t dr \,\dot c_r^{m^2}(x_1-x_2)}\right|\left|1-e^{-\beta\sigma_3\sigma_4\int_{\epsilon^2}^t dr \,\dot c_r^{m^2}(x_3-x_4)}\right| \nnb
	&\qquad \leq C t^{-2}e^{-\frac{|x_1-x_3|}{\sqrt{t}}}\int_0^t \frac{dr}{r^{1/2}} \frac{|x_1-x_2|}{\sqrt{r}}e^{-\frac{|x_1-x_2|}{\sqrt{r}}}\left(\frac{|x_1-x_2|}{\sqrt{r}}\wedge 1\right)^{-\frac{\beta}{2\pi}} \nnb
	&\qquad \qquad \qquad \qquad \times \int_0^t\frac{ds}{\sqrt{s}} \frac{|x_3-x_4|}{\sqrt{s}}e^{-\frac{|x_3-x_4|}{\sqrt{s}}}\left(\frac{|x_3-x_4|}{\sqrt{s}}\wedge 1\right)^{-\frac{\beta}{2\pi}}. 
	\end{align}
	Symmetrizing with respect to the $x_i$ yields our function $G_t$. Its norm can be estimated with similar scaling arguments as in the proof of Lemma \ref{le:covest2} and we find 
	\begin{equation}
	\|G_t \mathbf 1_{\sigma_1\neq \sigma_2}  \mathbf 1_{\sigma_3 \neq \sigma_4}\|_4\leq C_\beta t^{-1}\left(\int_0^t dr\, r^{1/2}\right)^2\leq C_\beta t^2,
	\end{equation}
	which was the claim. 

	\medskip 

    \underline{Case 2: $\sigma_1=\sigma_2$ or $\sigma_3=\sigma_4$ or both:} let us assume (by symmetry) that $\sigma_3=\sigma_4$. We can then use Lemma \ref{le:covest2} (and the triangle inequality) to write 
    \begin{align}
    &\left|\sum_{i\in \{1,2\},j\in \{3,4\}}\dot u_t^{m^2}(\xi_i,\xi_j)\right|\left|1-e^{-\beta\sigma_1\sigma_2\int_{\epsilon^2}^t dr\, \dot c_r^{m^2}(x_1-x_2)}\right|\left|1-e^{-\beta\sigma_3\sigma_4\int_{\epsilon^2}^t dr \,\dot c_r^{m^2}(x_3-x_4) }\right| \nnb
    &\qquad \leq (F_t(\xi_1,\xi_2,\xi_3)+F_t(\xi_1,\xi_2,\xi_4))\left(1-e^{-\beta \int_0^t dr\, \dot c_r^{m^2}(x_3-x_4)}\right)\nnb
    &\qquad \leq \beta(F_t(\xi_1,\xi_2,\xi_3)+F_t(\xi_1,\xi_2,\xi_4)) \int_0^t \frac{dr}{4\pi r}e^{-\frac{|x_3-x_4|^2}{4r}}.
    \end{align}
    The claim follows now by symmetrizing and Lemma~\ref{le:covest2} which implies
    \begin{equation}
      \|G_t (1-\mathbf 1_{\sigma_1\neq \sigma_2}  \mathbf 1_{\sigma_3 \neq \sigma_4})\|_4 \leq
      C_\beta t \int_{\R^2} dx \int_0^t \frac{dr}{4\pi r}e^{-\frac{|x|^2}{4r}}
      \leq
      C_\beta t^2
    \end{equation}
    as needed.
\end{proof}

\subsection{Proof of Proposition~\ref{pr:ybound}}

We begin our proof of Proposition~\ref{pr:ybound} with the remark that
\eqref{e:v1def} implies (using $x>0$, $1-e^{-x}\leq x$ and that $m^2t\leq 1$)
\begin{equation}
0 \leq \widetilde v_t^{1}(\xi|\epsilon,m)=e^{-\frac{\beta}{2}(\int_{\epsilon^2}^{t} ds \,\frac{e^{-m^2s}-1}{4\pi s}+\frac{1}{4\pi}\log t)}\leq t^{-\frac{\beta}{8\pi}}e^{\frac{\beta}{8\pi}m^2 t}\leq C_\beta t^{-\frac{\beta}{8\pi}}=:h_t^1(\xi)
\label{e:v1bound}
\end{equation}
for a constant $C_\beta$ depending only on $\beta$. This verifies the bound in Proposition~\ref{pr:ybound} for $n=1$.
We will verify the claimed bound explicitly also for $n=3$ and $n=4$, but prove the rest of it by induction. We will make use of the following explicit form of the $n=2$ term:
\begin{align}
  \widetilde v_t^{2}(\xi_1,\xi_2|\epsilon,m)
  &=\beta \sigma_1\sigma_2\int_{\epsilon^2}^{t}ds\, \bigg(\dot c_s^{m^2}(x_1-x_2)e^{-\beta(\int_{\epsilon^2}^{s} dr \,\dot c_r^{m^2}(0)+\frac{1}{4\pi}\log \epsilon^2)}
    \nnb
  & \qquad\qquad\qquad\qquad\qquad e^{-\beta \int_s^t dr\, \dot c_r^{m^2}(0)-\beta \sigma_1\sigma_2\int_{s}^t dr\,\dot c_r^{m^2}(x_1-x_2)} \bigg)
    \notag \\
&=e^{-\beta (\int_{\epsilon^2}^{t} dr\,\dot c_r^{m^2}(0)+\frac{1}{4\pi}\log \epsilon^2)}\left(1-e^{-\beta \sigma_1\sigma_2\int_{\epsilon^2}^{t} ds\,\dot c_s^{m^2}(x_1-x_2)}\right). 
\label{e:v2expl}
\end{align}
Indeed, this equalitity follows from a straightforward calculation (using \eqref{e:v1def} and \eqref{e:vndef}).
This allows us to prove Proposition \ref{pr:ybound} in the special case $n=3$.

\begin{lemma}\label{le:v3bound}
	For $\beta\in(0,6\pi)$ and $t>0$, there exists a function $h_t^3$ which is independent of $\epsilon,m>0$ and for $0<\epsilon^2<t<m^{-2}$
	\begin{equation} 
	|\widetilde v_t^{3}(\xi_1,\xi_2,\xi_3|\epsilon,m)|\leq h_t^3(\xi_1,\xi_2,\xi_3)
	\end{equation}
	and, in the notation \eqref{e:norms}, $\|h_t^3\|_3\leq C_\beta t^{-1}(t^{1-\frac{\beta}{8\pi}})^3$ for a constant $C_\beta$ depending only on $\beta$.
\end{lemma}
\begin{proof}
  From the definitions of $\widetilde v_t^{n}$ in \eqref{e:v1def} and \eqref{e:vndef} and the expression for $\widetilde v_t^2$ from \eqref{e:v2expl},
  a straightforward calculation shows that, for any $\xi_1,\xi_2,\xi_3$,
	\begin{align}
          &\widetilde v_t^{3}(\xi_1,\xi_2,\xi_3|\epsilon,m)
          \nnb
          & = \beta \left(\widetilde v_t^{1}(\xi|\epsilon,m)\right)^3 \int_{\epsilon^2}^{t} ds\bigg[\left(\dot u_s^{m^2}(\xi_1,\xi_2)+\dot u_s^{m^2}(\xi_1,\xi_3)\right)\left(1-e^{-\sigma_2\sigma_3\beta\int_{\epsilon^2}^{s}dr\,\dot c_r^{m^2}(x_2-x_3)}\right)\nnb
	& \qquad +\left(\dot u_s^{m^2}(\xi_2,\xi_1)+\dot u_s^{m^2}(\xi_2,\xi_3)\right)\left(1-e^{-\sigma_1\sigma_3\beta\int_{\epsilon^2}^{s}dr\,\dot c_r^{m^2}(x_1-x_3)}\right)\nnb
	& \qquad + \left(\dot u_s^{m^2}(\xi_3,\xi_1)+\dot u_s^{m^2}(\xi_3,\xi_2)\right)\left(1-e^{-\sigma_1\sigma_2\beta\int_{\epsilon^2}^{s} dr\,\dot c_r^{m^2}(x_1-x_2)}\right)\bigg]\nnb
	& \qquad \qquad \times e^{-\beta\sum_{1\leq i<j\leq 3}\sigma_i \sigma_j\int_{s}^t dr\, \dot c_r^{m^2}(x_i-x_j)}.
	\end{align} 
        We claim that we have the following estimate: for $t>s$,
        for any $\sigma_1,\sigma_2,\sigma_3$ and $x_1,x_2,x_3$ and for some universal constant $C>0$,
	\begin{multline}
          \sigma_1\sigma_2 \int_s^t dr\,\dot c_r^{m^2}(x_1-x_2)+\sigma_1\sigma_3 \int_s^t dr\,\dot c_r^{m^2}(x_1-x_3)+\sigma_2\sigma_3 \int_s^t dr\,\dot c_r^{m^2}(x_2-x_3)
          \\ 
          \geq -\frac{1}{4\pi}\log \frac{t}{s}-C.
	\label{e:covineq}
	\end{multline}
        Indeed, the worst case scenario is when
        $\sigma_1=\sigma_2\neq \sigma_3$ and $|x_1-x_2|\geq |x_1-x_3|,|x_2-x_3|$ (or the same with a permutation of indices).
        By the triangle inequality, at least one of $|x_1-x_3|$ and $|x_2-x_3|$ is greater than  $\frac{1}{2}|x_1-x_2|$. Thus we have from Lemma \ref{le:covest1},
	\begin{align}
	&\sigma_1\sigma_2 \int_s^t dr\, \dot c_r^{m^2}(x_1-x_2)+\sigma_1\sigma_3 \int_s^t dr\,\dot c_r^{m^2}(x_1-x_3)+\sigma_2\sigma_3 \int_s^t dr \,\dot c_r^{m^2}(x_2-x_3)\nnb
	&\qquad \geq  \inf_{x\in \R^2}\left(\int_s^t dr\,(\dot c_r^{m^2}(x)-\dot c_r^{m^2}(x/2))\right)-\frac{1}{4\pi}\int_s^t \frac{dr}{r}\nnb
	&\qquad = \inf_{x\in \R^2}\left( -\frac{1}{2\pi}\log \left(\frac{|x|}{\sqrt{t}}\wedge 1\right)+\frac{1}{2\pi}\log \left(\frac{|x|}{2\sqrt{t}}\wedge 1\right)\right)-\frac{1}{4\pi}\log \frac{t}{s}+O(1),
	\end{align}
	where the implied constant is universal. Going through the various cases ($|x|<\sqrt{t}$, $\sqrt{t}\leq |x|<2\sqrt{t}$, and $|x|\geq 2\sqrt{t}$), one readily checks that the infimum is $-\frac{1}{2\pi}\log 2$ and we have the bound \eqref{e:covineq}. 
	
	Now making use of Lemma \ref{le:covest2}, \eqref{e:covineq}, and \eqref{e:v1bound}, we find that for a constant $\widetilde C_\beta$ (depending only on $\beta$)
	\begin{equation}
	|\widetilde v_t^{3}(\xi_1,\xi_2,\xi_3|\epsilon,m)|\leq \widetilde C_\beta t^{-\frac{3\beta}{8\pi}}\int_{0}^t ds\left(\frac{t}{s}\right)^{\frac{\beta}{4\pi}} F_s(\xi_1,\xi_2,\xi_3)=:h_t^3(\xi_1,\xi_2,\xi_3),
	\end{equation}
	which is independent of $\epsilon,m$ as required.
        Finally, by Lemma \ref{le:covest2}, there is another constant $C_\beta$ depending only on $\beta$ such that
	\begin{equation}
	\|h_t^3\|_3\leq C_\beta t^{-\frac{3\beta}{8\pi}}\int_0^t ds\, \left(\frac{t}{s}\right)^{\frac{\beta}{4\pi}}s \leq C_\beta t^{2-\frac{3\beta}{8\pi}},
	\end{equation}
    which was precisely the claim.
\end{proof}

We now turn to $\widetilde v_t^{4}$.

\begin{lemma}\label{le:v4bound}
	For $\beta\in(0,6\pi)$ and $t>0$ there exists a function $h_t^4$, which is independent of $m$ and $\epsilon$ such that for $0<\epsilon^2<t<m^{-2}$, 
	\begin{equation}
	|\widetilde v_t^{4}(\xi_1,\xi_2,\xi_3,\xi_4|\epsilon,m)|\leq h_t^4(\xi_1,\xi_2,\xi_3,\xi_4)
      \end{equation}
      and $\|h_t^4\|_4\leq C_\beta t^{-1}(t^{1-\frac{\beta}{8\pi}})^4$ for a constant $C_\beta$ depending only on $\beta$.
\end{lemma}
\begin{proof}
	We begin with the recursion \eqref{e:vndef}. We see that there are two types of contributions: either $|I_1|=|I_2|=2$ or $|I_1|,|I_2|\in\{1,3\}$ (with $|I_1|+|I_3|=4$). Let us consider the latter case first. Here we can use \eqref{e:v1bound} and Lemma \ref{le:v3bound} along with the remark that $w_t^{m^2}-w_s^{m^2}\geq 0$ (since $\dot c_r^{m^2}$ is a covariance), to get the simple upper bound 
	\begin{align}
          &\frac{1}{2}\Bigg|\sum_{\substack{I_1\dot \cup I_2=[4]\\|I_1|,|I_2|\neq 2}} \int_{\epsilon^2}^{t}ds\sum_{i\in I_1,j\in I_2}\dot u_s^{m^2}(\xi_i,\xi_j)\widetilde v_s^{|I_1|}(\xi_{I_1}|\epsilon,m)\widetilde v_s^{|I_2|}(\xi_{I_2}|\epsilon,m)
          \nnb & \qquad\qquad\qquad\qquad\qquad \times 
          e^{-(w_t^{m^2}(\xi_1,\dots,\xi_4)-w_s^{m^2}(\xi_1,\dots,\xi_4))} \Bigg|\nnb
	&\leq C_\beta \sum_{k=1}^4\int_0^t \frac{ds}{s}\sum_{l\neq k}e^{-\frac{|x_k-x_l|^2}{4s}} h_s^1(\xi_k) h_s^3(\xi_{[4]\setminus \{k\}})
	\end{align}
	which is the contribution to $h_t^4$ from the $|I_1|,|I_2|\neq 2$-case. Note that using \eqref{e:v1bound} and Lemma \ref{le:v3bound}, one can check readily that the $\|\cdot \|_4$-norm of this quantity is bounded by (for some constants $C_\beta,\widetilde C_\beta$ depending only on $\beta$)
	\begin{equation}
	C_\beta \int_0^t ds\, s^{-\frac{\beta}{8\pi}}s^{-1}s^{3(1-\frac{\beta}{8\pi})}\leq \widetilde{C}_\beta t^{-1}t^{4(1-\frac{\beta}{8\pi})}
	\end{equation}
	which is precisely of the required form (note that the integral here is convergent since $\beta<6\pi$).
	
    It remains to control the $|I_1|=|I_2|=2$-case. A typical term that one encounters in the sum is of the form 
	\begin{align}
	&\int_{\epsilon^2}^{t}ds(\dot u_s^{m^2}(\xi_1,\xi_3)+\dot u_s^{m^2}(\xi_2,\xi_3)+\dot u_s^{m^2}(\xi_1,\xi_4)+\dot u_s^{m^2}(\xi_2,\xi_4))e^{-2\beta (\int_{\epsilon^2}^{s} dr\,\dot c_r^{m^2}(0)+\frac{1}{4\pi}\log \epsilon^2)}\nnb
          &\quad \times \left(1-e^{-\sigma_1\sigma_2 \beta \int_{\epsilon^2}^s dr\,\dot c_r^{m^2}(x_1-x_2)}\right)\left(1-e^{-\sigma_3\sigma_4 \beta\int_{\epsilon^2}^s dr\,\dot c_r^{m^2}(x_3-x_4)}\right)
            \nnb &\quad \times 
                   e^{-\frac{\beta}{2}\sum_{i,j\in[4]}\sigma_i\sigma_j\int_{s}^t dr\, \dot c_r^{m^2}(x_i-x_j)}
                   .
	\end{align}
	The last exponential term can again be dropped by positive definiteness of $\dot c_r^{m^2}$, so using Lemma \ref{le:covest3} 
	and \eqref{e:v1bound}, we see that for some $C_\beta$ depending only on $\beta$, such terms can be bounded by 
	\begin{equation}
    C_\beta\int_0^t ds\, s^{-\frac{\beta}{2\pi}}G_s(\xi_1,\xi_2,\xi_3,\xi_4),
	\end{equation}
    where $G_s$ is as in Lemma \ref{le:covest3}. Summing over the other contributions shows that all of the $|I_1|=|I_2|$-terms can be bounded by such quantities. Combining this with the $|I_1|,|I_2|\neq 2$ case gives the definition of $h_t^4$. Moreover, we note from Lemma \ref{le:covest3} that 
    \begin{equation}
    	\int_0^t ds\, s^{-\frac{\beta}{2\pi}}\|G_s\|_4\leq C_\beta \int_0^t ds\, s^{2-\frac{\beta}{2\pi}}\leq \widetilde C_\beta t^{-1}t^{4(1-\frac{\beta}{8\pi})}
    \end{equation}
    for some constants $C_\beta,\widetilde C_\beta$ depending only on $\beta$. Again, $\beta<6\pi$ played an important role here. Combined with the estimate from the previous case, we see that $\|h_t^4\|_4\leq C_\beta t^{-1}t^{4(1-\frac{\beta}{8\pi})}$ as required. This concludes the proof.
\end{proof}

We turn now to the proof of the general case. 
\begin{proof}[Proof of Proposition \ref{pr:ybound}]
	   As mentioned already, the proof is by induction. For propagating the induction, we find it convenient to prove the claim in a slightly different form. More precisely, we will prove the existence of functions $h_t^n$ (independent of $\epsilon,m$) for which $|\widetilde v_t^{n}(\cdot|\epsilon,m)|\leq h_t^n$ and for some $C_\beta$ depending only on $\beta$ and some universal constant $C>0$
	   \begin{equation}
	   \|h_t^n\|_n\leq n^{n-2} t^{-1}C_\beta^{n-1}\left(Ct^{1-\frac{\beta}{8\pi}}\right)^n,
	   \label{e:nboundalt}
	   \end{equation}
	   which of course implies the claim (with a possibly different $C_\beta$).
           For $n=1$, \eqref{e:nboundalt} is \eqref{e:v1bound} and for $n=3$ and $n=4$, \eqref{e:nboundalt} is proved in Lemma \ref{le:v3bound} and Lemma \ref{le:v4bound}. Let us now as our induction hypothesis assume that for some $n\geq 5$, the estimate \eqref{e:nboundalt} holds for all $k\leq n-1$ with $k\neq 2$. As mentioned, this has been verified for $n=5$. To advance the induction, we plug the hypothesis into \eqref{e:vndef},
           and need to be slightly careful about the contributions from $|I_1|=2$ or $|I_2|=2$.

           Let us consider the terms in \eqref{e:vndef} with $|I_1| \neq 2$ and $|I_2| \neq 2$ first.
	   In \eqref{e:vndef}, it will be sufficient to just drop the $w_t^{m^2}-w_s^{m^2}$-term (which, as before, is allowed due to the positive definiteness of $\dot c_r^{m^2}$). Then one readily checks (from \eqref{e:vndef} and our induction hypothesis) that the $|I_1|,|I_2|\neq 2$-contribution can be bounded by 
	   \begin{equation}
	   h_t^{n,1}(\xi_1,\dots,\xi_n):=\frac{\beta}{8\pi}\sum_{\substack{I_1\dot \cup I_2=[n]\\|I_1|,|I_2|\neq 2}}\sum_{i\in I_1,j\in I_2}\int_0^t\frac{ds}{s}e^{-\frac{|x_i-x_j|^2}{4s}}h_s^{|I_1|}(\xi_{I_1})h_s^{|I_2|}(\xi_{I_2}).
	   \end{equation}
	   Note that this is indeed independent of $\epsilon,m$ as required.   
	   Using the fact that $\int_{\R^2} dx\,\frac{e^{-\frac{|x|^2}{4t}}}{4\pi t}=1$ and our induction hypothesis, we find for the norm of this the bound 
	   \begin{align}
	   \|h_t^{n,1}\|_n&\leq \frac{\beta}{2}\sum_{\substack{I_1\dot \cup I_2=[n]\\|I_1|,|I_2|\neq 2}}|I_1||I_2|\int_0^tds\|h_s^{|I_1|}\|_{|I_1|}\|h_s^{|I_2|}\|_{|I_2|}\nnb
	   &\leq \frac{\beta}{2}C_\beta^{n-2}C^n\sum_{\substack{I_1\dot \cup I_2=[n]\\|I_1|,|I_2|\neq 2}}|I_1|^{|I_1|-1}|I_2|^{|I_2|-1}\int_0^t ds\, s^{-2+n(1-\frac{\beta}{8\pi})}\nnb
	   &=\frac{\beta}{2}\frac{1}{-1+n(1-\frac{\beta}{8\pi})}C_\beta^{n-2}C^n t^{-1}t^{n(1-\frac{\beta}{8\pi})}\sum_{k=1}^{n-1}{n\choose k} k^{k-1}(n-k)^{n-k-1}\nnb
	   &=\frac{\beta}{2}\frac{2(n-1)}{-1+n(1-\frac{\beta}{8\pi})}n^{n-2} C_\beta^{n-2} t^{-1}\left(C t^{1-\frac{\beta}{8\pi}}\right)^n
       \end{align}
       where in the last equality we made use of the identity $\sum_{k=1}^{n-1}{n\choose k}k^{k-1}(n-k)^{n-k-1}=2(n-1)n^{n-2}$.
       (This identity has the following combinatorial interpretation. The number of trees on $[n]$ is $n^{n-2}$.
       Thus $2(n-1)n^{n-2}$ represents the number of trees on $[n]$ together with a choice of a directed edge.
       Such trees rooted by a directed edge can also be obtained by connected two disjoint vertex rooted trees
       with $k$ and $n-k$ vertices by an edge connecting their roots.)
       Now for $n\geq 5$ and $\beta\in(0,6\pi)$, $0\leq \frac{\beta}{2}\frac{2(n-1)}{n(1-\frac{\beta}{8\pi})-1}$ is bounded by a universal constant, so possibly increasing $C_\beta$ verifies that the bound \eqref{e:nboundalt} holds for the contribution coming from $|I_1|,|I_2|\neq 2$.
	   
	   Let us now turn to the case where $|I_1|=2$ or $|I_2|=2$. We again drop the $w_t^{m^2}-w_s^{m^2}$-term from the exponential by positive definiteness. In terms of the notation of Lemma \ref{le:covest2}, we find (using the lemma and \eqref{e:v2expl}) that the contribution from the $|I_1|=2$ or $|I_2|=2$ case can be bounded by  
	   \begin{align}
             &\widetilde C_\beta\sum_{1\leq a<b\leq n}\sum_{j\in [n]\setminus \{a,b\}}\int_0^t ds \left|\dot u_s^{m^2}(\xi_a,\xi_j)+\dot u_s^{m^2}(\xi_b,\xi_j)\right|h_s^{n-2}(\xi_{[n]\setminus \{a,b\}})
               \nnb &\qquad\qquad\qquad\qquad\qquad\qquad\qquad \times 
               s^{-\frac{\beta}{4\pi}}\left|1-e^{-\beta \sigma_a \sigma_b c_s^{m^2}(x_a-x_b)}\right|
               \nnb
             &\qquad \leq \sum_{1\leq a<b\leq n}\sum_{j\in [n]\setminus \{a,b\}}\int_0^t ds\,s^{-\frac{\beta}{4\pi}} F_s(\xi_a,\xi_b,\xi_j)h_s^{n-2}(\xi_{[n]\setminus \{a,b\}})
               \nnb &\qquad 
               =:h_t^{n,2}(\xi_1,\dots,\xi_n)
	   \end{align} 
	   for some constant $\widetilde C_\beta$ depending only on $\beta$.	   
	   For the norm of this, we readily find from Lemma \ref{le:covest2} and our induction hypothesis that (for some possibly different $\widetilde C_\beta$, still depending only on $\beta$)
	   \begin{align}
	   \|h_t^{n,2}\|_n & \leq \widetilde C_\beta C_\beta^{n-3} C^{n-2} n^2 (n-2) (n-2)^{n-4}\int_0^t ds\, s^{-\frac{\beta}{4\pi}} s^{(n-2)(1-\frac{\beta}{8\pi})}\nnb
	   &=\widetilde C_\beta C_\beta^{n-3}C^{n-2}n^{n-2}\frac{(n-2)}{1-\frac{\beta}{4\pi}+(n-2)(1-\frac{\beta}{8\pi})}t^{(n-2)(1-\frac{\beta}{8\pi})+1-\frac{\beta}{4\pi}}.
	   \end{align}
	   The ratio here is again bounded by a universal constant, so possibly increasing $C_\beta$ (to account for this universal constant and $\widetilde C_\beta$) then yields the bound we are after. 
	   
	   In particular, choosing $h_t^n=h_t^{n,1}+h_t^{n,2}$ gives the required function and concludes the proof.
\end{proof}

\section{The sine-Gordon model: the partition and correlation functions}
\label{sec:sg-expansion}

The goal of this section is to prove Theorem~\ref{thm:cf}, which is our main statement about the correlation functions of the sine-Gordon model.
As already suggested in the previous section, a central tool in our proof of Theorem~\ref{thm:cf} is a suitable generating function for the correlation functions. To reiterate, the generating function we consider is (as in \eqref{e:Zmicropot}) for $\zeta\in L_c^\infty(\R^2\times \{-1,1\})$ given by
\begin{align}
Z(\zeta|\epsilon,m) =
  Z(\beta,\zeta|\epsilon,m) = \avga{\exp\qa{
  - \int d\xi \, 
  \epsilon ^{-\beta/4\pi}\zeta(\xi)e^{i\sqrt{\beta}\sigma\varphi(x)}}}_{\GFF(\epsilon,m)}
\label{e:pfdef}
\end{align}
with $\xi=(x,\sigma)$ and $\int d\xi = \sum_{\sigma\in\{\pm 1\}} \int_{\R^2}dx$ as before.
Of course, $\zeta(\xi)=-\nz \mathbf 1_{\Lambda}(x)$ is admissible and $Z(\zeta|\epsilon,m)$ then reduces to the normalization constant
in \eqref{e:sgdef}.
In general, note that we allow complex valued functions $\zeta$,
and that $Z(\zeta|\epsilon,m)$ is then not necessarily a normalizing constant for a positive measure.
The purpose of introducing $Z(\zeta|\epsilon,m)$ is that by choosing $\zeta$ to depend on suitable external parameters, we can obtain 
(smeared) sine-Gordon correlation functions from logarithmic derivatives of $Z(\zeta|\epsilon,m)$ with respect to these parameters. Thus if we can control $Z(\zeta|\epsilon,m)$ in the $\epsilon,m\to 0$ limit, we can also control the correlation functions.

A significant part of our analysis will rely on properties of the free field correlation functions studied in Section~\ref{sec:massless-gff}.
Particularly important for us will be charge correlation functions. Their importance can be seen, for example, from the fact that since $\wick{e^{i\sqrt{\beta}\varphi(x)}}_\epsilon$ is a bounded random variable for any $\epsilon>0$  one finds (for more details, see Lemma \ref{le:entire1})
\begin{equation}
  Z(\zeta|\epsilon,m)=\sum_{k=0}^\infty \frac{(-1)^k}{k!}
  \int d\xi_1 \cdots d\xi_k \, 
  \zeta(\xi_1)\cdots \zeta(\xi_k)\avga{\prod_{j=1}^k \wick{e^{i\sqrt{\beta}\sigma_j\varphi(x_j)}}_\epsilon }_{\GFF(\epsilon,m)}
  .
\label{e:Zexpa}
\end{equation}
It turns out that for $\beta\geq 4\pi$
(so in particular for $\beta=4\pi$),
$Z(\zeta|\epsilon,m)$ does not converge as $\epsilon\to 0$.
Heuristic evidence for this can be seen from Lemma~\ref{le:GFFcharge}
combined with the expansion \eqref{e:Zexpa}: one expects to have a 
divergence already at order $k=2$ in the expansion since 
\begin{equation}
	\avga{\wick{e^{i\sqrt{\beta}\varphi(x)}}\wick{e^{-i\sqrt{\beta}\varphi(y)
	}}}_{\GFF}\propto |x-y|^{-\frac{\beta}{2\pi}}
\end{equation}
is not integrable for $\beta\geq 4\pi$. 
It turns out that for $\beta\in[4\pi,6\pi)$, this is in a sense the only 
type of divergence that occurs and a non-trivial limit 
can be obtained once $Z$ is multiplied by
an explicit counterterm. This counterterm and the limit theorem for the partition 
function are most conveniently expressed in terms of truncated free field correlation functions
 which we again recall from Section~\ref{sec:massless-gff}.

The counterterm is then defined as follows: for $\xi_1,\xi_2\in \R^2\times \{-1,1\}$ let 
\begin{equation}
A(\xi_1,\xi_2|\epsilon,m)
=\avga{\wick{e^{i\sqrt{\beta}\sigma_1\varphi(x_1)}}_\epsilon\wick{e^{i\sqrt{\beta}\sigma_2\varphi(x_2)}}_\epsilon}_{\GFF(\epsilon,m)}^T.
\label{e:ct}
\end{equation}
We then define our renormalized partition function as 
\begin{equation}
\mathcal Z(\zeta|\epsilon,m)
:=Z(\zeta|\epsilon,m)\exp\left[-\frac{1}{2}
  \int d\xi_1 \, d\xi_2 \, 
  \zeta(\xi_1)\zeta(\xi_2)A(\xi_1,\xi_2|\epsilon,m)\right]
\label{e:rpf}.
\end{equation}
It follows from Lemma \ref{le:GFFcharge} that 
\begin{equation}
\lim_{m\to 0}\lim_{\epsilon\to 0}A(\xi_1,\xi_2|\epsilon,m)=\delta_{\sigma_1+\sigma_2,0}e^{-\frac{\gamma\beta}{4\pi}}\left(\frac{2}{|x_1-x_2|}\right)^{\beta/2\pi},
\end{equation}
and since this is non-integrable for $\beta\geq 4\pi$, our counterterm 
at least has a chance to cure the divergence of the partition function.
This is indeed true, in that $\mathcal Z$ turns out to have a finite limit
for $\beta<6\pi$, and thus in particular for $\beta=4\pi$ which is the case we are interested in.
For $\beta \geq 6\pi$ further counterterms, which turn out to involve higher order truncated correlation functions, 
would be required, see \cite{MR649810,MR849210,MR1777310}.

Before stating our result about the convergence of $\mathcal Z(\zeta|\epsilon,m)$,
recall from Lemma~\ref{lem:gffintegrability1} that, while the truncated charge two-point function is not integrable,
all higher order charge correlation functions are integrable.
With this notation and fact in hand, we are in a position to state our main result about $\mathcal Z(\zeta|\epsilon,m)$.
For $\beta < 4\pi$ the conclusions also follow from \cite{MR0434278}, but our extension to
$\beta<6\pi$ (crucially including the free fermion point $\beta=4\pi$) relies on new ideas.
We prove this in Section~\ref{sec:pfconv} and then deduce Theorem~\ref{thm:cf} in Section~\ref{sec:cfconv}.

\begin{theorem}\label{th:main1}
	For $\beta\in (0,6\pi)$, 
	$m \in (0,\infty)$, and $\zeta\in L_c^\infty(\R^2\times \{-1,1\},\C)$ the following claims hold.
	\begin{enumerate}
        \item The limits
              \begin{equation}
                \mathcal Z(\zeta|m)=\lim_{\epsilon\to 0}\mathcal Z(\zeta|\epsilon,m), \qquad
                \mathcal Z(\zeta)=\lim_{m\to 0}\lim_{\epsilon\to 0}\mathcal Z(\zeta|\epsilon,m),
               \end{equation}		
		exist and are finite.
		\item The functions $z\mapsto \mathcal Z(z\zeta|m)$ and $z\mapsto \mathcal Z(z\zeta)$ are entire functions of $z\in \C$ and $\mathcal Z(z\zeta)=\mathcal Z(-z\zeta)$.
              \item If $\zeta(x,1)=\overline{\zeta(x,-1)}$ for almost all $x\in \R^2$,
                then $\mathcal Z(\zeta|m)>0$ and $\mathcal Z(\zeta)>0$. 
		\item Finally if $\zeta_\alpha\in L^\infty_c(\R^2\times \{-1,1\})$ depends on some complex parameters $\alpha\in \C^N$ and
                  $\zeta_{\alpha}(\cdot |\epsilon,m)\in L^\infty_c(\R^2\times \{-1,1\})$ depends also on $\epsilon,m>0$ and these complex parameters $\alpha$ in such a way that for some $K\subset \C^N$ compact
\begin{equation}
			\lim_{m\to 0}\limsup_{\epsilon\to 0}\sup_{\alpha\in K}\|\zeta_{\alpha}(\cdot |\epsilon,m)-\zeta_\alpha\|_{L^\infty(\R^2\times \{-1,1\})}=0,
\end{equation}		
		then 
\begin{equation}
	        \lim_{m\to 0}\limsup_{\epsilon\to 0}\left|\mathcal Z(\zeta_{\alpha}(\cdot|\epsilon,m)|\epsilon,m)-\mathcal Z(\zeta_\alpha)\right|=0
\end{equation}
        and the convergence is uniform in $\alpha\in K$.
        An analogous statement holds for $m \in (0,\infty)$ fixed.
	\end{enumerate}
\end{theorem} 

As a preliminary remark, we note that by rescaling space it suffices to prove the statements for fixed $m>0$ in this theorem only for $m \in (0,1)$; we will henceforth assume this.

Before we turn to the actual proofs, we need to recall some basic facts about regularity and extrema of Gaussian processes.

\subsection{Preliminaries -- regularity and extrema of Gaussian processes}

In this section, we record some basic facts we need to know about the regularization of the GFF to a scale $\sqrt{t}$ which is of order one, namely we look at the Gaussian process with law $\nu^{\GFF(\sqrt{t},m)}$ -- in particular in the $m\to 0$ limit. Given \eqref{e:Zvt}, this will be useful to control the renormalized partition function. The main fact we will prove in this section is the following. 

\begin{lemma} \label{le:grad}
For $0<t<m^{-2}$, $\Lambda \subset \R^2$ compact, and $p>0$, we have 
\begin{equation}
	  \avga{e^{p\|\nabla \varphi\|_{L^\infty(\Lambda)}}}_{\GFF(\sqrt{t},m)}\leq C_{p,t,\Lambda}
\end{equation}
for some constant $C_{p,t,\Lambda}<\infty$ which is independent of $m$. 
\end{lemma}

We will apply this estimate with $t>0$ fixed as in Corollary~\ref{cor:yukawa}.
Clearly, the constant $C_{p,t,\Lambda}$ must diverge as $t\to 0$ or $|\Lambda|\to\infty$
(as the limiting Gaussian free field is not differentiable as $t\to 0$ or bounded as $m\to  0$);
these divergences are  not important for our application.

First of all, using arguments based on Kolmogorov-Chentsov--type results (see e.g. \cite[Appendix B]{MR3339158}), one can check that the smoothness of $c_\infty^{m^2}-c_t^{m^2}$ (recall the notation \eqref{e:tcov}) implies that we can regard $\varphi$ as a smooth function, and $\nabla\varphi$ is a centered Gaussian process with covariance 
\begin{equation}
\avga{\partial_i \varphi(x)\partial_j\varphi(y)}_{\GFF(\sqrt{t},m)}=\int_t^\infty ds \,\frac{e^{-m^2 s}}{4\pi s}\left(\delta_{i,j}\frac{1}{2s}-\frac{(x_i-y_i)(x_j-y_j)}{4s^2}\right) e^{-\frac{|x-y|^2}{4s}}.
\label{e:gradcov}
\end{equation}

To estimate the exponential moments in Lemma \ref{le:grad}, we rely on two classical theorems about Gaussian processes. The first one is Dudley's theorem (see e.g. \cite[Theorem 1.3.3]{MR2319516}) which states that if for a centered real-valued Gaussian process $X$ on say a compact metric space $T$ we define a new (pseudo) metric by setting $d_X(t,s)=\sqrt{\EE[(X(t)-X(s))^2]}$, then 
\begin{equation}
\EE\left(\sup_{t\in T}X(t)\right)\leq C\int_0^\infty d\epsilon\,\sqrt{\log N_X(\epsilon)},
\label{e:dud}
\end{equation}
where $C$ is a universal constant, and $N_X(\epsilon)$ is the minimal number of (closed) $d_X$-radius $\epsilon$ balls required to cover $T$.

The second result we need is the Borell-TIS inequality (see e.g. \cite[Theorem 2.1.1]{MR2319516}), which states that in the same setting as Dudley's theorem, if $X$ is further assumed to be almost surely bounded on $T$, and if $\sigma_T^2:=\sup_{t\in T}\EE X(t)^2$, then for all $u>0$
\begin{equation}
\PP\left(\sup_{t\in T}X(t)-\EE\left[\sup_{t\in T}X(t)\right]>u\right)\leq e^{-\frac{u^2}{2\sigma_T^2}}.
\label{e:BTIS}
\end{equation}

With these tools, we can prove our claim about $\nabla \varphi$.

\begin{proof}[Proof of Lemma~\ref{le:grad}]
  First of all, we note that by a simple Cauchy-Schwarz argument, it is enough for us to prove the claim for $\|\partial_0\varphi\|_{L^\infty(\Lambda)}$ (or $\|\partial_1 \varphi\|_{L^\infty(\Lambda)}$ as they both have the same distribution) instead of $\|\nabla\varphi\|_{L^\infty(\Lambda)}$. Then noting that 
  \begin{align}
    \avga{e^{p\|\partial_0 \varphi\|_{L^{\infty}(\Lambda)}}}_{\GFF(\sqrt{t},m)}
    &\leq \avga{e^{p\sup_{x\in \Lambda}\partial_0 \varphi(x)}}_{\GFF(\sqrt{t},m)}+\avga{e^{p\sup_{x\in \Lambda}(-\partial_0 \varphi(x))}}_{\GFF(\sqrt{t},m)}\nnb
    &=2\avga{e^{p\sup_{x\in \Lambda}\partial_0 \varphi(x)}}_{\GFF(\sqrt{t},m)},
  \end{align}
  we see that is enough to consider only $\sup_{x\in \Lambda}\partial_0\varphi(x)$ instead of $\sup_{x\in \Lambda}|\partial_0\varphi(x)|$. 
  In this setup we can use Dudley's theorem and Borell--TIS.
  
  To apply Dudley's theorem, we note that 
  \begin{multline}
    d_{\partial_0\varphi}(x,y)^2=2\int_t^\infty ds \, \frac{e^{-m^2 s}}{8\pi s^2}\left(1-e^{-\frac{|x-y|^2}{4s}}\right)+2 (x_0-y_0)^2\int_t^\infty ds\,\frac{e^{-m^2s}}{16\pi s^3}e^{-\frac{|x-y|^2}{4s}}
    \\
    \leq C_t^2 |x-y|^2
  \end{multline}    
    for a constant $C_t$ independent of $m$. Thus we have $\{y\in \Lambda:d_{\varphi_1}(x,y)\leq \epsilon \}\supset \{y\in \Lambda: C_t|x-y|\leq \epsilon\}$. So the number of $d_{\partial_0\varphi}$-radius $\epsilon$ balls it takes to cover $\Lambda$ is less than the number of Euclidean radius $\epsilon/C_t$-balls it takes to cover $\Lambda$. It thus follows from Dudley's theorem, \eqref{e:dud}, that $\avga{\sup_{x\in \Lambda}\partial_0\varphi(x)}_{\GFF(\sqrt{t},m)}\leq \widetilde C_{t,\Lambda}$ for some constant $\widetilde C_{t,\Lambda}$ which is independent of $m$.
    
    Now $\sigma_{\Lambda}^2=\int_t^\infty ds\, \frac{e^{-m^2s}}{8\pi s^2}\leq \int_t^\infty \frac{ds}{8\pi s^2}=:\widehat C_t$. In particular, we have for say $u>2\widetilde C_{t,\Lambda}$
\begin{equation}
	    \frac{(u-\avga{\sup_{x\in \Lambda}\partial_0\varphi(x)}_{\GFF(\sqrt{t},m)})^2}{\sigma_{\Lambda}^2}\geq \frac{(u-\widetilde C_{t,\Lambda})^2}{\widehat C_t}\geq \frac{u^2}{4\widehat C_t}.
\end{equation}    
    Thus we find from the Borell-TIS inequality \eqref{e:BTIS} that, for $u>2\widetilde C_{t,\Lambda}$, 
    \begin{align}
    \nu^{\GFF(\sqrt{t},m)}\left(\sup_{x\in \Lambda}\partial_0\varphi(x)>u\right)\leq e^{-\frac{u^2}{4\widehat C_t}}
    \end{align}
    and 
    \begin{align}
    \avga{e^{p\sup_{x\in \Lambda}\partial_0 \varphi(x)}}_{\GFF(\sqrt{t},m)} &=\int_{\R} du \,p e^{pu}\nu ^{\GFF(\sqrt{t},m)}\left(\sup_{x\in \Lambda}\partial_0 \varphi(x)>u\right)\nnb
    &\leq \int_{-\infty}^{2\widetilde C_{t,\Lambda}} du\, pe^{pu}+\int_{2\widetilde C_{t,\Lambda}}^\infty du\, pe^{pu}e^{-\frac{u^2}{4\widehat C_t}},
    \end{align}
    which yields the desired claim for the exponential moments of $\sup_{x\in \Lambda}\partial_0\varphi(x)$ and by our preliminary considerations, also for $\|\nabla\varphi\|_{L^\infty(\Lambda)}$. This concludes the proof.
\end{proof}

\subsection{Uniform bounds for the renormalized partition function}

We are now ready to turn to analysis of the partition function.
We begin with the bounds for $\mathcal Z$ stated in the following proposition.
The estimate of item (ii) applies to uniformly small coupling constants $\zeta$ and is thus
a standard consequence of the expansion of the renormalized potential.
The estimates of items (i) and (iii) on the other hand apply to arbitrarily large $\zeta$
and make in addition use of the Gaussian concentration estimate of Lemma~\ref{le:grad}.

\begin{proposition} \label{pr:pfbound}
  	        Fix $\Lambda\subset \R^2$ compact and $\beta \in (0,6\pi)$.
  	        \begin{enumerate}
  	        	\item
  	        	For any fixed $M>0$,
  	         	
\begin{equation}
	  	        	\quad \sup_{\epsilon,m\in(0,1)}\sup_{\substack{\zeta \in L_c^\infty (\R^2\times \{-1,1\}): \\
  	        	\mathrm{supp}(\zeta(\cdot,\pm 1))\subset \Lambda,\\		
  	        	\|\zeta\|_{L^\infty (\R^2\times \{-1,1\})} \leq M}}\left|\mathcal Z(\zeta|\epsilon,m)\right|<\infty. 
\end{equation}  	        	
  	        	\item
  	        	There exists a $\delta=\delta_{\Lambda,\beta}>0$ independent of $\epsilon,m$ such that
\begin{equation}
	  	        	\inf_{\epsilon,m\in(0,1)}\inf_{\substack{\zeta \in L^\infty_c (\R^2\times \{-1,1\}):\\ \mathrm{supp}(\zeta(\cdot,\pm 1))\subset \Lambda, \\ \|\zeta\|_{L^\infty (\R^2\times \{-1,1\})} \leq \delta}}\left|\mathcal Z(\zeta|\epsilon,m)\right|>0.
\end{equation}

  	        	\item
  	        	For any fixed $M>0$ 
  	        	\begin{equation}
  	        	\inf_{\epsilon,m\in(0,1)}\inf_{\substack{\zeta\in L^\infty_c(\R^2\times \{-1,1\}):\\
  	        			\mathrm{supp}(\zeta(\cdot,\pm 1))\subset\Lambda,\\ \|\zeta\|_{L^\infty(\R^2\times \{-1,1\})}\leq M,\\\zeta(\cdot ,1)=\overline{\zeta(\cdot ,-1)}}} \mathcal Z(\zeta|\epsilon,m) > 0.
  	        	\end{equation}
  	        \end{enumerate}
\end{proposition}

Before we turn to the proof, we record a simple estimate that we will have use for in the proof and also later on. 

\begin{lemma}
     	\label{le:simple}
	For each $t>0$ and $\beta\in(0,6\pi)$, there exists a function $g_t\in L^1_{\mathrm{loc}}(\R^2\times \R^2)$ (namely for any compact $K\subset \R^2\times \R^2$, $\int_K dx\, dy\, |g_t(x,y)|<\infty$) which is independent of $\epsilon,m\in(0,1)$ such that for all $\xi_1,\xi_2\in \R^2\times \{-1,1\}$, with $A(\xi_1,\xi_2|\epsilon,m)$ from \eqref{e:ct} and $\widetilde v_t^2(\xi_1,\xi_2|\epsilon,m)$ from \eqref{e:v2expl},
	\begin{equation}
		|A(\xi_1,\xi_2|\epsilon,m)+\widetilde v_t^{2}(\xi_1,\xi_2|\epsilon,m)|\leq g_t(x_1,x_2).
	\end{equation}	
\end{lemma}

\begin{proof}
  By the definitions in \eqref{e:ct} and \eqref{e:v2expl}, we have 
	\begin{align}
          &A(\xi_1,\xi_2|\epsilon,m)+\widetilde v_t^{2}(\xi_1,\xi_2|\epsilon,m)
            \nnb & 
          =e^{-\beta \int_{\epsilon^2}^{1}ds\, \frac{e^{-m^2s}-1}{4\pi s}-\beta \int_1^\infty ds\, \frac{e^{-m^2 s}}{4\pi s}}\left(e^{-\beta \sigma_1\sigma_2\int_{\epsilon^2}^{\infty} ds\,\dot c_s^{m^2}(x_1-x_2)}-1\right)
          \nnb &\qquad 
	-e^{-\beta \int_{\epsilon^2}^{1}ds\,\frac{e^{-m^2s}-1}{4\pi s}-\beta \int_1^t ds\,\frac{e^{-m^2 s}}{4\pi s}}\left(e^{-\beta \sigma_1\sigma_2\int_{\epsilon^2}^{t} ds\,\dot c_s^{m^2}(x_1-x_2)}-1\right).
	\end{align}
	Note first of all, that the $\int_{\epsilon^2}^{1}$-integral is common to both terms and bounded for $m<1$, so that we can ignore it.
        Moreover, the contribution of the $1$'s is also uniformly bounded in $m$ (for any fixed $t$), so that we can also ignore them.
        Finally if $\sigma_1=\sigma_2$, then also the $e^{-\beta\sigma_1\sigma_2\int (\cdots)}$-terms are uniformly bounded, so there is nothing to prove in this case. The remaining question is to prove the required estimate for 
	\begin{multline}
          \left|e^{-\beta \int_1^\infty ds \,\frac{e^{-m^2 s}}{4\pi s}}e^{\beta \int_{\epsilon^2}^{\infty} ds\, \dot c_s^{m^2}(x_1-x_2)}-e^{-\beta \int_1^t ds \,\frac{e^{-m^2 s}}{4\pi s}}e^{\beta \int_{\epsilon^2}^{t} ds\, \dot c_s^{m^2}(x_1-x_2)}\right|\\
          =e^{\beta \int_{\epsilon^2}^{1} ds \,\dot c_s^{m^2}(x_1-x_2)}e^{\beta \int_1^t ds\, \frac{e^{-m^2s}}{4\pi s }\left(e^{-\frac{|x_1-x_2|^2}{4s}}-1\right)}
          \left|1-e^{\beta\int_t^\infty ds\,\frac{e^{-m^2s}}{4\pi s }\left(e^{-\frac{|x_1-x_2|^2}{4s}}-1\right)}\right|.
	\end{multline}
	Using repeatedly the estimate $|1-e^{-x}|\leq x$ for $x>0$ and Lemma \ref{le:covest1}, along with some elementary considerations, we find that, for some universal constant $C$,
	\begin{multline}
          \left|e^{-\beta \int_1^\infty ds\, \frac{e^{-m^2 s}}{4\pi s}}e^{\beta \int_{\epsilon^2}^{\infty} ds\,\dot c_s^{m^2}(x_1-x_2)}-e^{-\beta \int_1^t ds \,\frac{e^{-m^2 s}}{4\pi s}}e^{\beta \int_{\epsilon^2}^{t} ds\, \dot c_s^{m^2}(x_1-x_2)}\right|\\
          \leq C|x_1-x_2|^2\left(|x_1-x_2|\wedge 1\right)^{-\frac{\beta}{2\pi}}e^{\frac{\beta}{4\pi}\int_{\min(1,t)}^{\max(1,t)}\frac{ds}{s}}\int_t^\infty\frac{ds}{s^2}. 
	\end{multline}
	Note that this function is locally integrable for $\beta<6\pi$ (in fact for $\beta<8\pi$), so we are done.
\end{proof}

With this in hand, we can turn to the proof of Proposition~\ref{pr:pfbound}.

\begin{proof}[Proof of Proposition \ref{pr:pfbound} (i)]
  For $M>0$ and $\beta\in(0,6\pi)$, let us choose
  $t=t_{M,\beta}\in (0,m^{-2})$ independent of $\epsilon$
  (so that Lemma~\ref{le:grad} is applicable)
  but small enough
  that  Corollary~\ref{cor:yukawa} is applicable:
  for $C_\beta$ as in Proposition~\ref{pr:ybound}, we assume that
  \begin{equation}
    eC_\beta Mt^{1-\frac{\beta}{8\pi}}\leq \frac{1}{2}.
    \label{e:tchoice}
    \end{equation}
    With this choice of $t$, we use   \eqref{e:Zvt} and Corollary~\ref{cor:yukawa} to write for $\zeta\in L^\infty_c(\R^2\times \{-1,1\})$ satisfying $\mathrm{supp}(\zeta(\cdot,\pm 1))\subset\Lambda$, $\|\zeta\|_{L^\infty(\R^2\times \{-1,1\})}\leq M$,
    \begin{align}
      & 
        \mathcal{Z}(\zeta|\epsilon,m)
        \nnb 
      &=\avga{e^{-\sum_{n=1}^\infty \frac{1}{n!}\int_{(\R^2\times \{-1,1\})^n}d\xi_1\cdots d\xi_n\, \zeta(\xi_1)\cdots \zeta(\xi_n)\widetilde v_t^{n}(\xi_1,\dots,\xi_n|\epsilon,m)\left(e^{i\sqrt{\beta}\sum_{j=1}^n \sigma_j \varphi(x_j)}-\delta_{n,2}\right)}}_{\GFF(\sqrt{t},m)}\notag \\
      &\qquad \times e^{-\frac{1}{2}\int_{(\R^2\times \{-1,1\})^2}d\xi_1\,d\xi_2\,\zeta(\xi_1)\zeta(\xi_2)(A(\xi_1,\xi_2|\epsilon,m)+\widetilde v_t^{2}(\xi_1,\xi_2)|\epsilon,m))}
        .
    \label{e:pftrep}
    \end{align}
    The second factor in \eqref{e:pftrep} is bounded (uniformly in $\epsilon,m,\zeta$) by Lemma \ref{le:simple}, our assumption that $\zeta(\cdot,\pm1)$ has support in $\Lambda$, and our assumption $\|\zeta\|_{L^\infty(\R^2\times \{-1,1\})}\leq M$, so
    we only need to consider the first factor. For this,
    Proposition \ref{pr:ybound} yields 
    \begin{align}
    &\avga{\left|e^{-\sum_{n=1}^\infty \frac{1}{n!}{\int_{(\R^2\times \{-1,1\})^n}d\xi_1\cdots d\xi_n}\, \zeta(\xi_1)\cdots \zeta(\xi_n)\widetilde v_t^{n}(\xi_1,\dots,\xi_n|\epsilon,m)\left(e^{i\sqrt{\beta}\sum_{j=1}^n \sigma_j \varphi(x_j)}-\delta_{n,2}\right)}\right|}_{\GFF(\sqrt{t},m)}\nnb
      &\leq e^{\sum_{n\neq 2} 2|\Lambda|\frac{M^n}{n!}n^{n-2}(C_\beta t^{1-\frac{\beta}{8\pi}})^nt^{-1}}
      \nnb
      &\qquad \qquad
        \times \avga{e^{\frac{M^2}{2}\int_{(\Lambda\times \{-1,1\})^2}d\xi_1\, d\xi_2\,|\widetilde v_t^{2}(\xi_1,\xi_2|\epsilon,m)|\left|e^{i\sqrt{\beta}\sum_{j=1}^2 \sigma_j \varphi(x_j)}-1\right|}}_{\GFF(\sqrt{t},m)}.
    \end{align}
    Since $n^n/n! \leq e^n$,
    by our choice of $t$ in \eqref{e:tchoice},
    the $n\neq 2$ sum in the first term on the right-hand side above depends only on $\beta,\Lambda,M$
    (in particular, it does not depend on $\epsilon,m,\zeta$), so it only remains to control the $n=2$ contribution, i.e., the second term on the right-hand side. For this,
    we note from \eqref{e:v2expl} that the $\sigma_1=\sigma_2$-contribution is uniformly bounded by a quantity independent of $\epsilon,m,\zeta$, so again using \eqref{e:v2expl} and the elementary estimate $|e^{i\sqrt{\beta}(\varphi(x_1)-\varphi(x_2))}-1|\leq \sqrt{\beta}\|\nabla \varphi\|_{L^\infty(\Lambda)}|x_1-x_2|$ for $x_1,x_2\in \Lambda$, we see that, for some constant $C=C_{\beta,\Lambda,M}$  (in particular, independent of $\epsilon,m,\zeta$),
    \begin{align}
      &\avga{\left|e^{-\sum_{n=1}^\infty \frac{1}{n!}{\int_{(\R^2\times \{-1,1\})^n}d\xi_1\cdots d\xi_n}\, \zeta(\xi_1)\cdots \zeta(\xi_n)\widetilde v_t^{n}(\xi_1,\dots,\xi_n|\epsilon,m)\left(e^{i\sqrt{\beta}\sum_{j=1}^n \sigma_j \varphi(x_j)}-\delta_{n,2}\right)}\right|}_{\GFF(\sqrt{t},m)}\nnb
      &\leq C 
        \avga{e^{\sqrt{\beta}M^2\|\nabla \varphi\|_{L^\infty(\Lambda)}e^{-\beta (\int_{\epsilon^2}^{t} dr\, \dot c_r^{m^2}(0)+\frac{1}{4\pi}\log \epsilon^2)}\int_{\Lambda^2}dx_1\,dx_2\, |e^{\beta c_t^{m^2}(x_1-x_2)}-1||x_1-x_2|}}_{\GFF(\sqrt{t},m)}
      .
    \end{align}
    From Lemma~\ref{le:covest1}, we see that $\int_{\Lambda^2}dx_1\, dx_2\, |e^{\beta c_t^{m^2}(x_1-x_2)}-1||x_1-x_2|$ can be bounded by a quantity depending only on $\Lambda,t,\beta$, while on the other hand, recalling \eqref{e:v1bound} (and that we chose in addition to \eqref{e:tchoice} that $t\leq m^{-2}$)
    \begin{align}
    &e^{-\beta\left(\int_{\epsilon^2}^{t}dr \,\dot c_r^{m^2}(0)+\frac{1}{4\pi}\log \epsilon^2\right)}\leq \widetilde C_\beta t^{-(1-\frac{\beta}{8\pi})}
    \end{align}
    for a constant $\widetilde C_\beta$ depending only on $\beta$. In summary, 
    \begin{align}
      &\avga{\left|e^{-\sum_{n=1}^\infty \frac{1}{n!}{\int_{(\R^2\times \{-1,1\})^n}d\xi_1\cdots d\xi_n}\, \zeta(\xi_1)\cdots \zeta(\xi_n)\widetilde v_t^{n}(\xi_1,\dots,\xi_n|\epsilon,m)\left(e^{i\sqrt{\beta}\sum_{j=1}^n \sigma_j \varphi(x_j)}-\delta_{n,2}\right)}\right|}_{\GFF(\sqrt{t},m)}\nnb
      &\leq C \avga{e^{p\|\nabla\varphi\|_{L^\infty(\Lambda)}}}_{\GFF(\sqrt{t},m)},
    \end{align}
    for some constants $C=C_{\beta,\Lambda,M}$ and $p=p_{\beta,\Lambda,M}$ independent of $\epsilon,m,\zeta$.
    Thus Lemma~\ref{le:grad} shows that the expectation part of \eqref{e:pftrep} is bounded by a quantity independent of $m,\epsilon,\zeta$.
\end{proof}

\begin{proof}[Proof of Proposition \ref{pr:pfbound} (ii)]
        Consider first $\zeta\in L^\infty_c(\R^2\times \{-1,1\})$ with $\|\zeta\|_{L^\infty(\R^2\times \{-1,1\})}=1$ (say) and let us look at the function $z\mapsto \mathcal Z(z\zeta|\epsilon,m)$. Since
        for $z=0$, $\mathcal Z(z\zeta|\epsilon,m)=1$, it will turn out to be sufficient to bound the derivative of $\mathcal Z(z\zeta|\epsilon,m)$ uniformly in some neighborhood of the origin -- then taking a small enough neighborhood (independent of $\epsilon,m$), we can bound the distance to zero in this neighborhood. Translating this into a statement about $\zeta$ with small enough $L^\infty$-norm will follow from a scaling argument.
        
        For the derivative, note that for $\epsilon,m>0$, the relevant random variables are deterministically bounded, so $\mathcal Z(z\zeta|\epsilon,m)$ is an entire function. Thus for any compact $K'\subset B(0,R)$ and $z\in K'$, we have by Cauchy's integral formula,
          \begin{equation}
          	\frac{d}{dz}\mathcal Z(z\zeta|\epsilon,m)
          =\frac{2}{2\pi i}\oint_{|w|=R} dw\, \frac{\mathcal Z(w\zeta|\epsilon,m)}{(w-z)^2}.
          \end{equation}
	By Proposition \ref{pr:pfbound} (i), we can bound the numerator uniformly in $\epsilon,m,\zeta$ (recall that we normalized $\|\zeta\|_{L^\infty(\R^2\times \{-1,1\})}=1$) and we can assume $|w-z|$ to be uniformly bounded from below, so we see that also the derivative is uniformly bounded in compact subsets. Thus we see that there exists some $\delta>0$ (independent of $\epsilon,m,\zeta$) for which we have 
\begin{equation}		\inf_{\epsilon,m\in(0,1)}\inf_{|z|<\delta}\inf_{\substack{\zeta\in L_c^\infty(\R^2\times \{-1,1\})\\ \mathrm{supp}(\zeta(\cdot,\pm1))\subset\Lambda\\
			\|\zeta\|_{L^\infty(\R^2\times \{-1,1\})}=1}}|\mathcal Z(z\zeta|\epsilon,m)|>0.
\end{equation}
	By scaling, we note that this can be translated into precisely the claim of the proposition.
\end{proof}

\begin{proof}[Proof of Proposition \ref{pr:pfbound} (iii)] 
        Note from \eqref{e:yukawa} that under the condition $\zeta(x,1)=\overline{\zeta(x,-1)}$, the renormalized potential $v_t(\zeta,\cdot|\epsilon,m)$ is real. Thus making the same choices as in the proof of Proposition~\ref{pr:pfbound} (i), we can write 
	\begin{align}
          & 
          \mathcal Z(\zeta|\epsilon,m)
           \nnb 
          &\geq \avga{e^{-\sum_{n=1}^\infty \frac{1}{n!} 
            \int d\xi_1 \cdots d\xi_n \, 
            \left|\zeta(\xi_1)\cdots \zeta(\xi_n)\widetilde v_t^{n}(\xi_1,\dots,\xi_n|\epsilon,m)\left(e^{i\sqrt{\beta}\sum_{j=1}^n \sigma_j \varphi(x_j)}-\delta_{n,2}\right)\right|}}_{\GFF(\sqrt{t},m)}\notag \\
          &\qquad \times e^{-             \int d\xi_1 \, d\xi_2 \, 
            \left|\zeta(\xi_1)\zeta(\xi_2)(A(\xi_1,\xi_2|\epsilon,m)+\widetilde v_t^{2}(\xi_1,\xi_2|\epsilon,m))\right|}.
	\end{align}
	We can now argue exactly as in the proof of Proposition~\ref{pr:pfbound} (i): the sum of the $n\neq 2$-terms is deterministically bounded uniformly in $\epsilon,m,\zeta$. Similarly, Lemma~\ref{le:simple} lets us deduce that that the $A+\widetilde v_t^{2}(\cdot|\epsilon,m)$ term can be bounded from below by a uniform constant. It remains to argue that 
\begin{equation}
	          \avga{e^{-\frac{1}{2}            \int d\xi_1 \, d\xi_2 \, 
              \left|\zeta(\xi_1)\zeta(\xi_2)\widetilde v_t^{2}(\xi_1,\xi_2|\epsilon,m)\left(e^{i\sqrt{\beta}\sum_{j=1}^2 \sigma_j \varphi(x_j)}-1\right)\right|}}_{\GFF(\sqrt{t},m)}>C>0
\end{equation}
	for some constant $C$ independent of $\epsilon,m,\zeta$. As in the proof of Proposition~\ref{pr:pfbound}, this follows from Lemma~\ref{le:grad}, though now combined with Jensen's inequality (used in the form $\E \frac{1}{X}\geq \frac{1}{\E X}$ for a positive random variable $X$). This concludes the proof.
\end{proof}

Before turning to the proof of Theorem \ref{th:main1}, we introduce some notation and make some preliminary remarks about the renormalized partition function as an analytic function. 

\subsection{Expansion of the renormalized partition function}
\label{sec:cZexpan}

Our proof of convergence of the renormalized partition function and entirety of the limit will go through analyzing the series expansion of $z\mapsto \mathcal Z(z\zeta|\epsilon,m)$ and making use of the estimates from Proposition~\ref{pr:ybound} and Proposition~\ref{pr:pfbound}. For this purpose, we introduce some notation for the series expansion of $z\mapsto \mathcal Z(z\zeta|\epsilon,m)$. We formulate this as the following lemma. 

\newcommand{\Vsym}{\mathcal{\widetilde M}}

\begin{lemma} 	\label{le:entire1}
	For fixed $\epsilon,m>0$ and $\zeta\in L^\infty_c(\R^2\times \{-1,1\})$, the function $z\mapsto \mathcal Z(z\zeta|\epsilon,m)$ is entire and we have 
        \begin{equation}
		\mathcal Z(z\zeta|\epsilon,m)=\sum_{n=0}^\infty \frac{z^n}{n!}\mathcal M_n(\zeta|\epsilon,m),
              \end{equation}
	where
        \begin{equation}
          \mathcal M_n(\zeta|\epsilon,m)
          =
          \int_{(\R^2\times \{-1,1\})^n}d\xi_1\cdots d\xi_n \, \zeta(\xi_1)\cdots \zeta(\xi_n) \, \Vsym(\xi_1,\dots,\xi_n|\epsilon,m)
        \end{equation}
        with (recall the definition of $A$ from \eqref{e:ct})
        \begin{multline}          \label{e:Vdef}
          \Vsym(\xi_1,\dots,\xi_n|\epsilon,m)
          =
          \frac{1}{n!}\sum_{\tau\in S_n}
          \Bigg[ \sum_{j=0}^{\lfloor \frac{n}{2}\rfloor}\frac{n!}{j!(n-2j)!}(-1)^{n-j}2^{-j}
          \\
          \times \left(\prod_{l=1}^jA(\xi_{\tau_{2l-1}},\xi_{\tau_{2l}}|\epsilon,m)\right) \avga{\prod_{l'=2j+1}^n\wick{e^{i\sqrt{\beta}\sigma_{\tau_{l'}}\varphi(x_{{\tau_{l'}}})}}_\epsilon}_{\GFF(\epsilon,m)}\Bigg]
          ,
        \end{multline}
        where $S_n$ denotes the group of permutations of the set $[n]$.
	Moreover, for each $\delta,M>0$ and $\Lambda\subset \R^2$ compact, there exists a constant $C(\delta,M,\Lambda)$ independent of $\epsilon,m,\zeta,n$ such that 
	\begin{equation}
	\sup_{\epsilon,m\in(0,1)}\sup_{\substack{\zeta\in L_c^\infty(\R^2\times \{-1,1\})\\\mathrm{supp}(\zeta(\cdot,\pm 1))\subset\Lambda,\\ \|\zeta\|_{L^\infty(\R^2\times \{-1,1\})}\leq M}}\left|\mathcal M_n(\zeta|\epsilon,m)\right|\leq C(\delta,M,\Lambda)\delta^n n!.
	\label{e:Mbound}
	\end{equation}
\end{lemma}
\begin{proof}
    Let us recall from \eqref{e:pfdef} and \eqref{e:rpf} that 
    \begin{multline}
    \mathcal Z(z\zeta|\epsilon,m)=  \avga{
      e^{-z\int_{\R^2\times \{-1,1\}}d\xi \, \zeta(\xi)\wick{e^{i\sqrt{\beta}\sigma\varphi(x)}}_\epsilon}}_{\GFF(\epsilon,m)}
      \\
      \times e^{-\frac{z^2}{2} \int_{(\R^2\times \{-1,1\})^2}d\xi_1 d\xi_2 \, \zeta(\xi_1)\zeta(\xi_2)A(\xi_1,\xi_2|\epsilon,m)}. 
    \label{e:pfprod}
    \end{multline}
    As mentioned in the proof of Proposition \ref{pr:pfbound}, for each $\epsilon,m>0$, the expectation is an entire function of $z$ since 
    \begin{equation}
    	    \int_{\R^2\times \{-1,1\}}d\xi \, \zeta(\xi)\wick{e^{i\sqrt{\beta}\sigma\varphi(x)}}_\epsilon 
    \end{equation}
    is a bounded random variable. The second factor in \eqref{e:pfprod} is trivially an entire function of $z$ (for any fixed $\epsilon,m>0$). Thus we see that indeed $\mathcal Z(z\zeta|\epsilon,m)$ is entire.
    
    For the expansion coefficients, by series expanding both terms (and interchanging the order of summation/integration and expectation which is justified by the boundedness of the relevant random variables and a routine Fubini argument), we find 
    \begin{align}
      & 
        \mathcal Z(z\zeta|\epsilon,m)
        \nnb 
      &= \sum_{j=0}^\infty \frac{(-1)^jz^{2j}}{2^{j} j!}\int_{(\R^2\times \{-1,1\})^{2j}}d\xi_1\cdots d\xi_{2j} \, \zeta(\xi_1)\cdots \zeta(\xi_{2j}) \prod_{l=1}^{j}A(\xi_{2l-1},\xi_{2l}|\epsilon,m)\\
    &\quad \times\sum_{k=0}^\infty \frac{(-z)^{k}}{k!} \int_{(\R^2\times \{-1,1\})^{k}}d\xi_1'\cdots d\xi_{k}' \, \zeta(\xi_1')\cdots \zeta(\xi_k') \avga{\prod_{l'=1}^k\wick{e^{i\sqrt{\beta}\sigma_{l'}'\varphi(x_{l'}')}}_{\epsilon}}_{\GFF(\epsilon,m)}.\notag
    \end{align}
    The claim about the representation of $\mathcal M_n(\zeta|\epsilon,m)$ now follows by
    relabeling our integration variables (write $n=k+2j$ and $\xi_{l'}'=\xi_{2j+l'}$)
    and then symmetrizing in the arguments.
    
    Finally for the proof of the bound \eqref{e:Mbound}, note that as $\mathcal Z(z\zeta|\epsilon,m)$ is entire, Cauchy's integral formula implies that, for any $R>0$,
    \begin{equation}
    	\mathcal M_n(\zeta|\epsilon,m)=\frac{n!}{2\pi i}\oint_{|w|=R}dw \, \frac{\mathcal Z(w\zeta|\epsilon,m)}{w^{n+1}}.
    \end{equation}
    The claim now follows by choosing $R=\delta ^{-1}$ and (by Proposition \ref{pr:pfbound}) 
\begin{equation}
	    C(\delta,M,\Lambda)\geq \sup_{\epsilon,m\in (0,1)}\sup_{\substack{\zeta\in L_c^\infty(\R^2\times\{-1,1\}):\\ \mathrm{supp}(\zeta(\cdot,\pm1))\subset\Lambda,\\ \|\zeta\|_{L^\infty(\R^2\times \{-1,1\})}\leq \delta ^{-1}M}}|\mathcal Z(\zeta|\epsilon,m)|.
\end{equation}    
\end{proof}

In the course of our proof of convergence of $\mathcal Z(\zeta|\epsilon,m)$,
we will have use for an alternative representation for $\mathcal M_n(\zeta|\epsilon,m)$ in terms of the renormalized potential.
To control the kernel $\Vsym$ in terms of this alternative representation, 
we record the following simple fact. 

\begin{lemma}	\label{le:symker}
	For $\epsilon,m>0$, $\Vsym(\xi_1,\dots,\xi_n|\epsilon,m)$ is the unique continuous function of $\xi_1,\dots,\xi_n$ for which 
	\begin{align}
	&\frac{1}{n!}\frac{\partial}{\partial \delta_1}\cdots \frac{\partial}{\partial \delta_n}\mathcal M_n(\delta_1 f_1+\cdots +\delta_n f_n|\epsilon,m)\nnb
	&\qquad = \int_{(\R^2\times \{-1,1\})^{n}}d\xi_1\cdots d\xi_{n} \, f_1(\xi_1)\cdots f_n(\xi_n) \Vsym(\xi_1,\dots,\xi_n|\epsilon,m)
	\end{align}
	for all $f_1,\dots,f_n\in L^\infty_c(\R^2\times\{-1,1\})$.
\end{lemma}
\begin{proof}
        Uniqueness can be seen, for example, by choosing $f_i$ to be of the form 
\begin{equation}
		e^{2\pi i k_1 x_1/(2L)}e^{2\pi i k_2 x_2/(2L)}\mathbf 1\{|x_1|,|x_2|\leq L\}
\end{equation}
	in the $x$ variables and to be a Kronecker $\delta$-in the $\sigma$-variable. This shows that any two functions $\widetilde {\mathcal M}$ satisfying this relation have the same Fourier series in an arbitrary square, so by continuity they must be the same in this square. Since the square was arbitrary, they must be the same in all of $\R^2$.

	To see that $\Vsym$ actually satisfies this relation, write 
\begin{equation}
         \mathcal M_n\left(\left.\sum_{l=1}^n \delta_l f_l\right|\epsilon,m\right)
          = \int_{(\R^2\times \{-1,1\})^{n}}d\xi_1\cdots d\xi_{n}
          \prod_{j=1}^n \left(\sum_{l=1}^n \delta_l f_l(\xi_j)\right) \Vsym(\xi_1,\dots,\xi_n|\epsilon,m).
\end{equation} 
	Say the $\delta_n$-derivative can hit any of the $n$ terms in the $j$ product and produce a factor of $f_n(\xi_j)$ for some $j$. The $\delta_{n-1}$-derivative can hit any of the $n-1$ remaining terms in the $j$ product (and produce a factor of $f_{n-1}(\xi_{j'})$ for some $j'\neq j$) etc. We see that 
	\begin{align}
	&\frac{\partial}{\partial \delta_1}\cdots \frac{\partial}{\partial \delta_n}\mathcal M_n(\delta_1 f_1+\cdots +\delta_n f_n|\epsilon,m)\nnb
          &\qquad =\sum_{\tau \in S_n}\int_{(\R^2\times \{-1,1\})^{n}}d\xi_1\cdots d\xi_{n}\, f_1(\xi_{\tau_1})\cdots f_n(\xi_{\tau_n}) \Vsym(\xi_1,\dots,\xi_n|\epsilon,m)
          \nnb
	&\qquad = n!\int_{(\R^2\times \{-1,1\})^{n}}d\xi_1\cdots d\xi_{n} \, f_1(\xi_1)\cdots f_n(\xi_n)\Vsym(\xi_1,\dots,\xi_n|\epsilon,m),
	\end{align}
	where in the last step we simply relabeled our integration variables.
\end{proof}

As a final ingredient before we turn to proving convergence of the renormalized partition function, we will construct an integrable upper bound on $|\Vsym|$ by
representing $\mathcal M_n(\zeta|\epsilon,m)$ in terms of the renormalized partition function.
Before stating the bound,
we state the representation of $\mathcal M_n(\zeta|\epsilon,m)$ in terms of the renormalized potential $v_t(\zeta,\cdot|\epsilon,m)$.

\begin{lemma}
  Let $\beta\in(0,6\pi)$, let $\zeta\in L^\infty_c(\R^2\times \{-1,1\})$, and choose
  $t \in (\epsilon^2,m^{-2})$ as in Corollary~\ref{cor:yukawa}.
  Then $\mathcal M_n(\zeta|\epsilon,m)$ defined in Lemma~\ref{le:entire1} can be expressed as
	\begin{align}
          \mathcal M_n(\zeta|\epsilon,m)
          &=
            \int d\xi_1 \cdots d\xi_n \,
            \zeta(\xi_1)\cdots \zeta(\xi_n)\nnb
	&\quad \times \sum_{j=0}^{\lfloor \frac{n}{2}\rfloor}\frac{n!}{j!}2^{-j}(-1)^j\prod_{l=1}^j\left(A(\xi_{2l-1},\xi_{2l}|\epsilon,m)+\widetilde v_t^{2}(\xi_{2l-1},\xi_{2l}|\epsilon,m)\right)\nnb
          &\quad \times \sum_{k=1}^{n-2j}\frac{(-1)^{k}}{k!}\sum_{\substack{1\leq n_1,\dots,n_k\leq n-2j:\\n_1+\cdots +n_k=n-2j}}\frac{(n-2j)!}{n_1!\cdots n_k!}
          \nnb &\qquad\qquad \times 
          \prod_{l=1}^{k}\widetilde v_t^{n_l}(\xi_{2j+n_1+\dots+n_{l-1}+1},\dots,\xi_{2j+n_1+\dots+n_l}|\epsilon,m)\nnb
          &\qquad\qquad \times \avga{\prod_{l=1}^k \left(e^{i\sqrt{\beta}\sum_{p=2j+n_{1}+\dots+n_{l-1}+1}^{2j+n_1+\dots+n_l}\sigma_p \varphi(x_p)}-\delta_{n_l,2}\right)}_{\GFF(\sqrt{t},m)},
	\end{align}
	with the interpretation that if $n=2j$, then the $k$-sum equals one.
	\label{le:malt}
\end{lemma}
\begin{proof}
	From Corollary \ref{cor:yukawa}, we have (for our choice of $t$)
	\begin{equation}
          \mathcal Z(z\zeta|\epsilon,m) 	= F_1(z)F_2(z),
        \end{equation}
        where
        \begin{align}
          F_1(z)&:=e^{-\frac{z^2}{2}\int d\xi_1\,d\xi_2 \, \zeta(x_1)\zeta(x_2)(A(\xi_1,\xi_2|\epsilon,m)+\widetilde v_t^{2}(\xi_1,\xi_2|\epsilon,m))}\\
          F_2(z) &:=\avga{e^{-\sum_{n=1}^\infty \frac{z^n}{n!}\int d\xi_1 \cdots d\xi_n \, \zeta(\xi_1)\cdots \zeta(\xi_n)\widetilde v_t^{n}(\xi_1,\dots,\xi_n|\epsilon,m)(e^{i\sqrt{\beta}\sum_{j=1}^n \sigma_j \varphi(x_j)}-\delta_{n,2})}}_{\GFF(\sqrt{t},m)}.
	\end{align}
	Note that $F_1$ is entire and non-vanishing, so by Lemma \ref{le:entire1}, also $F_2$ is entire and we have for any fixed $R>0$
        \begin{equation}
          F_2^{(k)}(0)=\frac{k!}{2\pi i}\oint_{|z|=R} dz\, \frac{F_2(z)}{z^{k+1}}.
        \end{equation}
	By Fubini and the proof of Proposition \ref{pr:pfbound} (more precisely, Fubini is readily justified by controlling the $n\neq 2$ terms with Proposition \ref{pr:ybound} and the $n=2$ term with Lemma \ref{le:grad}), we find that if $R$ is chosen to satisfy \eqref{e:tchoice}, then
	\begin{multline}
	F_2^{(k)}(0)=\frac{k!}{2\pi i}\bigg\langle\oint_{|z|=R}\frac{dz}{z^{k+1}}\\
	\times  e^{-\sum_{n=1}^\infty \frac{z^n}{n!}\int_{(\R^2\times \{-1,1\})^{n}}d\xi_1\cdots d\xi_{n}\, \zeta(\xi_1)\cdots \zeta(\xi_n)\widetilde v_t^{n}(\xi_1,\dots,\xi_n|\epsilon,m)(e^{i\sqrt{\beta}\sum_{j=1}^n \sigma_j \varphi(x_j)}-\delta_{n,2})}\bigg\rangle_{\GFF(\sqrt{t},m)}.
	\end{multline}
	Moreover, again by Proposition \ref{pr:ybound}, the $z$-integrand is an entire (random) function and we find 
	\begin{align}
        &F_2^{(k)}(0)\nnb&=\avga{\left.\frac{d^k}{dz^k}\right|_{z=0}e^{-\sum_{n=1}^\infty \frac{z^n}{n!}\int_{(\R^2\times \{-1,1\})^{n}}d\xi_1\cdots d\xi_{n}\, \zeta(\xi_1)\cdots \zeta(\xi_n)\widetilde v_t^{n}(\xi_1,\dots,\xi_n|\epsilon,m)(e^{i\sqrt{\beta}\sum_{j=1}^n \sigma_j \varphi(x_j)}-\delta_{n,2})}}_{\GFF(\sqrt{t},m)}\nnb
	&=\sum_{l=1}^k\frac{(-1)^l}{l!}\sum_{\substack{1\leq n_1,\dots,n_l\leq k\\n_1+\cdots+n_l=k}}\frac{k!}{n_1!\cdots n_l!}\bigg\langle\prod_{l=1}^k \Bigg(\int_{(\R^2\times \{-1,1\})^{n_l}} d\xi_1\cdots d\xi_{n_l}\, \zeta(\xi_1)\cdots \zeta(\xi_{n_l})\nnb
	&\qquad \qquad\qquad\qquad \qquad \times \widetilde v_t^{n_l}(\xi_1,\dots,\xi_{n_l}|\epsilon,m)\left(e^{i\sqrt{\beta}\sum_{j=1}^{n_l}\sigma_j\varphi(x_j)}-\delta_{n_l,2}\right)\Bigg)\bigg\rangle_{\GFF(\sqrt{t},m)}
	\end{align}
	with the interpretation that if $k=0$, then $F_2^{(k)}(0)=1$.
	The claim now follows from noting that 
	\begin{equation}
		\mathcal M_n(\zeta|\epsilon,m)=\left.\frac{d^n}{dz^n}\right|_{z=0}F_1(z)F_2(z)
	\end{equation}
	and relabeling integration variables suitably.
\end{proof}

To conclude this section, we use this representation to prove that
$\Vsym(\xi_1,\dots,\xi_n|\epsilon,m)$
has an integrable upper bound which is independent of $\epsilon,m$ (allowing the use of dominated convergence).

\begin{lemma}
  	\label{le:Vub}
	For $\beta\in (0,6\pi)$, $n\geq 1$, and $t>0$ there exists a function $g_t\in L^1_{\mathrm{loc}}((\R^2\times \{-1,1\})^n)$, independent of $\epsilon,m$, such that for $0<\epsilon^2<t<m^{-2}$, 
	\begin{equation}
		|\Vsym(\xi_1,\dots,\xi_n|\epsilon,m)|\leq g_t(\xi_1,\dots,\xi_n).
	\end{equation}
\end{lemma}

\begin{proof}
        By Lemma~\ref{le:malt}, we see that $|\Vsym|$ can be bounded by sum of products of terms of the form 
\begin{gather}
		G_1(\widehat \xi_{1},\widehat\xi_{2}|\epsilon,m):=|A(\widehat \xi_{1},\widehat\xi_{2}|\epsilon,m)+\widetilde v_t^{2}(\widehat\xi_{1},\widehat\xi_{2}|\epsilon,m)|, 
                \\
	G_2(\xi'_{1},\dots,\xi'_{n'}|\epsilon,m):=|\widetilde v_t^{n'}(\xi'_{1},\dots,\xi'_{n'}|\epsilon,m)|
\end{gather}
with $2\neq n'\leq n$, and 
\begin{multline}
  G_3(\widetilde \xi_1,\dots,\widetilde \xi_{2k}|\epsilon,m)
  \\ 
  :=\avga{\prod_{l=1}^{k}\left|\widetilde v_t^{2}(\widetilde \xi_{2l-1},\widetilde \xi_{2l}|\epsilon,m)\left(e^{i\sqrt{\beta}(\widetilde\sigma_{2l-1}\varphi(\widetilde x_{2l-1})+\widetilde\sigma_{2l}\varphi(\widetilde x_{2l}))}-1\right) \right|}_{\GFF(\sqrt{t},m)}
    \end{multline}
    where we have written $\widehat \xi,\xi',\widetilde \xi$ to indicate that these variables are subsets (depending on the term in the sum) of the actual integration variables $\xi_i$ of $\Vsym$.
    Note that in each product of the $G_i$, the factors depend on disjoint sets of the $\xi_i$.
    Thus it is enough to prove the 
    corresponding integrability bounds for the $G_i$ terms separately.
	
    For $G_1$, this is simply Lemma \ref{le:simple}. For $G_2$, this follows from Proposition \ref{pr:ybound}.
    Finally, for $G_3$, we note that if we have $\widetilde \sigma_{2l-1}=\widetilde \sigma_{2l}$, we can bound $|e^{i\sqrt{\beta}(\widetilde\sigma_{2l-1}\varphi(\widetilde x_{2l-1})+\widetilde \sigma_{2l}\varphi(\widetilde x_{2l}))}-1|\leq 2$ and from \eqref{e:v2expl}, one readily checks the uniform bound
    \begin{equation}|\widetilde v_t^{2}(\widetilde \xi_{2l-1},\widetilde \xi_{2l}|\epsilon,m)|\leq e^{\beta \int_0^1 ds \,\frac{1-e^{-m^2 s}}{4\pi s}}\leq e^{\beta\int_0^1 ds\, \frac{1-e^{-s}}{4\pi s}}.
    \end{equation}
    It remains to control the quantities with $\widetilde \sigma_{2l-1}\neq \widetilde \sigma_{2l}$. For these we note (as in the proof of Proposition \ref{pr:pfbound}), that for $x_{2l-1},x_{2l}$ in a given compact set $\Lambda\subset \R^2$, we find from \eqref{e:v2expl}  that
	\begin{align}
	&\left|\widetilde v_t^{2}(\widetilde \xi_{2l-1},\widetilde \xi_{2l}|\epsilon,m)\left(e^{i\sqrt{\beta}(\varphi(\widetilde x_{2l-1})-\varphi(\widetilde x_{2l}))}-1\right) \right|\nnb
	&\quad\leq \sqrt{\beta}\|\nabla\varphi\|_{L^\infty(\Lambda)}|\widetilde x_{2l-1}-\widetilde x_{2l}|e^{\beta\int_0^1 ds\, \frac{1-e^{-s}}{4\pi s}} \left|e^{\beta c_t^{m^2}(\widetilde x_{2l-1}-\widetilde x_{2l})}-1\right|\nnb
	&\quad \leq \sqrt{\beta}\|\nabla\varphi\|_{L^\infty(\Lambda)}|\widetilde x_{2l-1}-\widetilde x_{2l}|e^{\beta\int_0^1 ds\, \frac{1-e^{-s}}{4\pi s}} \left(1+\left(\frac{|\widetilde x_{2l-1}-\widetilde x_{2l}|}{\sqrt{t}}\wedge 1\right)^{-\frac{\beta}{2\pi}}\right).
	\end{align} 
	
	By Lemma~\ref{le:grad}, arbitrary moments of $\|\nabla \varphi\|_{L^{\infty}(\Lambda)}$ under $\nu^{\GFF(\sqrt{t},m)}$ are bounded uniformly in $m$,
        so using that $|x-y|^{1-\frac{\beta}{2\pi}}$ is locally integrable (for $\beta<6\pi$), one gets a locally integrable upper bound which is independent of $\epsilon,m$ also for $G_3$. This concludes the proof.
\end{proof}

\subsection{Convergence of the renormalized partition function}
\label{sec:pfconv}

We now turn to the convergence of $\mathcal Z(\zeta|\epsilon,m)$ 
as $\epsilon,m\to 0$. 
With the uniform bounds from Section~\ref{sec:cZexpan} in place,
the main step is to show that
$\Vsym(\xi_1,\dots,\xi_n|\epsilon,m)$ has a pointwise limit as $\epsilon,m\to 0$.
This is the content of the following lemma. 
For $\xi_1, \dots, \xi_n \in \R \times \{\pm 1\}$ distinct, let
\begin{align} 
  & 
    \Vsym(\xi_1,\dots,\xi_n)
    \nnb 
  &=
    \frac{1}{n!}\sum_{\tau\in S_n}
    \sum_{j=0}^{\lfloor \frac{n}{2}\rfloor}\frac{n!}{j!(n-2j)!}(-1)^{n-j}2^{-j}\left(\prod_{l=1}^j\avga{\wick{e^{i\sqrt{\beta}{\sigma}_{\tau_{2l-1}}\varphi(x_{\tau_{2l-1}})}}\wick{e^{i\sqrt{\beta}{\sigma}_{\tau_{2l}}\varphi(x_{\tau_{2l}})}}}_{\GFF}\right)\notag\\
  &\qquad \times \avga{\prod_{l'=2j+1}^{n}\wick{e^{i\sqrt{\beta}{\sigma}_{\tau_{l'}}\varphi(x_{\tau_{l'}})}}}_{\GFF}\label{e:claim2}
\end{align}
and define $  \Vsym(\xi_1,\dots,\xi_n|m)$ analogously with $\GFF(m)$ instead of $\GFF$.

\begin{lemma}   	\label{le:Vsymlim}
  For $\xi_i\neq \xi_j$ for $i\neq j$, 
  \begin{align}
    \lim_{\epsilon\to 0}\Vsym(\xi_1,\dots,\xi_n|\epsilon,m)
    &=
      \Vsym(\xi_1,\dots,\xi_n|m),\\
    \lim_{m\to 0}\lim_{\epsilon\to 0}\Vsym(\xi_1,\dots,\xi_n|\epsilon,m)
    &=
      \Vsym(\xi_1,\dots,\xi_n).
  \end{align}
  Moreover, $\Vsym(\xi_1, \dots, \xi_n)=0$ if $n$ is odd.
\end{lemma}

\begin{proof}
  The convergence follows immediately from the definition of $\Vsym(\cdot |\epsilon,m)$  in \eqref{e:Vdef}
  and Lemma~\ref{le:GFFcharge} for the charge correlation functions of the GFF.

  That $\Vsym(\xi_1, \dots, \xi_n)=0$ if $n$ is odd
    follows from the case $n\neq n'$ in Lemma~\ref{le:GFFcharge},
    i.e., the fact that the massless GFF charge correlation functions vanish for nonneutral charge.
\end{proof}

With this result in hand, we can prove Theorem~\ref{th:main1}.

\begin{proof}[Proof of Theorem~\ref{th:main1}]
  We only consider the statements for $\epsilon\to 0$ and $m\to 0$;
  the ones with $m>0$ fixed are completely analogous.
Let us begin by defining our limit candidate. By Lemmas~\ref{le:Vub}--\ref{le:Vsymlim},
with $\Vsym$ defined in \eqref{e:claim2}, we see that 
	\begin{equation}\label{e:Mnzetadef}
          \lim_{m\to 0}\lim_{\epsilon\to 0}\mathcal M_n(\zeta|\epsilon,m)
          =
          \int d\xi_1 \cdots d\xi_n \,
          \zeta(\xi_1)\cdots \zeta(\xi_n)\Vsym(\xi_1,\dots,\xi_n)=:\mathcal M_n(\zeta)
	\end{equation}
	and that this quantity is finite for all $\zeta\in L_c^\infty(\R^2\times \{-1,1\})$. In fact, this argument also shows that if $\zeta_{\alpha}(\cdot |\epsilon,m)\to \zeta_\alpha$ in $L^\infty(\R^2\times \{-1,1\})$ and uniformly in $\alpha$ in some compact set, then also $\mathcal M_n(\zeta_{\alpha}(\cdot|\epsilon,m)|\epsilon,m)\to \mathcal M_n(\zeta_\alpha)$ -- uniformly in $\alpha$.
    Moreover, from \eqref{e:Mbound}, we see that if $\alpha$ is in some fixed compact set $K\subset \C^N$, then for each $\delta>0$, $|\mathcal M_n(\zeta_\alpha)|\leq C(\delta,K) \delta ^n n!$. In particular, 
	\begin{equation} \label{e:Zcalexpans}
	\mathcal Z(z\zeta_\alpha):=\sum_{n=0}^\infty \frac{z^n}{n!}\mathcal M_n(\zeta_\alpha)
	\end{equation}
	defines an entire function and we have, for any fixed $R>0$,
\begin{equation}
		\sup_{|z|<R}\sup_{\alpha\in K}|\mathcal Z(z\zeta_\alpha)|<\infty.
\end{equation}
        Moreover, when $m=0$, $\mathcal Z(z\zeta_\alpha)$ is even since $\mathcal M_n(\zeta_\alpha) =0$ if $n$ is odd.

	Let us then turn to the convergence claim. Fix any compact set $K\subset \C^N$ as in the statement of item (iv).
        By Cauchy's integral formula, for any $A\geq 0$ and $\alpha\in K$,
    \begin{align}
      &  
        \left|\mathcal Z(\zeta_{\alpha}(\cdot|\epsilon,m)|\epsilon,m)-\mathcal Z(\zeta_\alpha)\right|
        \nnb 
        &\leq \sum_{n=0}^A \frac{1}{n!}|\mathcal M_n(\zeta_{\alpha}(\cdot|\epsilon,m)|\epsilon,m)-\mathcal M_n(\zeta_\alpha)|\nnb
        &\quad +\sum_{n=A+1}^\infty \left|\oint_{|w|=2} \frac{\mathcal Z(w\zeta_{\alpha}(\cdot|\epsilon,m)|\epsilon,m)}{w^{n+1}}\frac{dw}{2\pi i}\right|\nnb
	&\qquad +\sum_{n=A+1}^\infty \left|\oint_{|w|=2} \frac{\mathcal Z(w\zeta_\alpha)}{w^{n+1}}\frac{dw}{2\pi i}\right|\nnb
	&\leq \sum_{n=0}^A \frac{1}{n!}|\mathcal M_n(\zeta_{\alpha}(\cdot|\epsilon,m)|\epsilon,m)-\mathcal M_n(\zeta_\alpha)|\nnb
	&\qquad +\left(\sup_{|w|=2}|\mathcal Z(w\zeta_{\alpha}(\cdot|\epsilon,m)|\epsilon,m)|+\sup_{|w|=2}|\mathcal Z(w\zeta_\alpha)|\right)\left(\frac{1}{2}\right)^{A}.
	\end{align}
	Uniform convergence now follows readily from our uniform bounds for the partition functions as well as our remark that $\mathcal M_n(\zeta_{\alpha}(\cdot |\epsilon,m)|\epsilon,m)$ converges uniformly. This takes care of statements (i), (ii), and (iv) in Theorem \ref{th:main1}.
    The final positivity claim, claim (iii), follows from Proposition~\ref{pr:pfbound} (iii).
\end{proof}

\subsection{Proof of Lemma~\ref{lem:gffintegrability1}}

Before we go into the proof of Theorem~\ref{thm:cf}, we point out here that Lemma~\ref{le:Vub} can be used to give a proof of Lemma~\ref{lem:gffintegrability1}.

\begin{proof}[Proof of Lemma~\ref{lem:gffintegrability1}]
  Let us begin by noting from multilinearity of the truncated correlation functions as well as the definition \eqref{e:truncatedmoment}, we have for any $f_1,\dots,f_n\in L^\infty_c(\R^2\times \{-1,1\})$,
    \begin{align}
      \int_{(\R^2\times \{-1,1\})^n}&d\xi_1\cdots d\xi_n \, f_1(\xi_1)\cdots f_n(\xi_n)\avga{\prod_{k=1}^n\wick{e^{i\sqrt{\beta}\sigma_k\varphi(x_k)}}_\epsilon}_{\GFF(\epsilon,m)}^T
      \nnb
        &=\avga{\prod_{k=1}^n\wick{e^{i\sqrt{\beta}\sigma_k\varphi}}_\epsilon(f_k)}_{\GFF(\epsilon,m)}^T\nnb
    	&=\left.\frac{\partial^n}{\partial t_1\cdots \partial t_n}\right|_{t=0}\log \avga{e^{\sum_{k=1}^n t_k\wick{e^{i\sqrt{\beta}\sigma_k\varphi}}_\epsilon(f_k)}}_{\GFF(\epsilon,m)}\nnb
    	&=\left.\frac{\partial^n}{\partial t_1\cdots \partial t_n}\right|_{t=0}\log Z\left(\left.-t_1f_1\cdots - t_n f_n\right|\epsilon,m\right),
    \end{align}
    where for $f\in L_c^\infty(\R^2\times \{-1,1\})$, recall that we understand $\wick{e^{i\sqrt{\beta}\sigma\varphi}}_\epsilon(f)$ as shorthand notation for $\int_{\R^2\times \{-1,1\}}d\xi\,  f(\xi)\wick{e^{i\sqrt{\beta}\sigma\varphi(x)}}_\epsilon$.
    For $n\geq 3$, we can write this in terms of the renormalized partition function as
    \begin{align}
    	\int_{(\R^2\times \{-1,1\})^n}&d\xi_1\cdots d\xi_n \, f_1(\xi_1)\cdots f_n(\xi_n)\avga{\prod_{k=1}^n\wick{e^{i\sqrt{\beta}\sigma_k\varphi(x_k)}}_\epsilon}_{\GFF(\epsilon,m)}^T\nnb
    	&=\left.\frac{\partial^n}{\partial t_1\cdots \partial t_n}\right|_{t=0}\log \mathcal Z\left(\left.-t_1f_1\cdots - t_n f_n\right|\epsilon,m\right).
    \end{align}
    On the other hand, by using the expansion (with the notation \eqref{e:Mnzetadef}--\eqref{e:Zcalexpans})
    \begin{equation}
    	\mathcal Z(-t_1f_1-\cdots -t_n f_n|\epsilon,m)=\sum_{i=0}^\infty \frac{(-1)^i}{i!}\mathcal M_i(t_1f_1+\cdots +t_n f_n|\epsilon,m),
    \end{equation}
    we have (for small enough $|t_k|$) that
    \begin{align}
    	&\log \mathcal Z\left(\left.-t_1f_1\cdots - t_n f_n\right|\epsilon,m\right)\\\notag
    	&=\log \left(\sum_{i=0}^\infty \frac{(-1)^i}{i!}\int_{(\R^2\times \{-1,1\})^i}d\xi_1\cdots d\xi_i\prod_{k=1}^i(\sum_{l=1}^n t_lf_l(\xi_k)) \Vsym(\xi_1,\dots,\xi_i|\epsilon,m)\right).
    \end{align}
    Carrying out the $t$-derivatives, we see that 
    \begin{align}
    	\int_{(\R^2\times \{-1,1\})^n}&d\xi_1\cdots d\xi_n \, f_1(\xi_1)\cdots f_n(\xi_n)\avga{\prod_{k=1}^n\wick{e^{i\sqrt{\beta}\sigma_k\varphi(x_k)}}_\epsilon}_{\GFF(\epsilon,m)}^T\nnb
    	&=\int_{(\R^2\times \{-1,1\})^n} d\xi_1\cdots d\xi_n\, f_1(\xi_1)\cdots f_n(\xi_n) \mathcal P(\xi_1,\dots,\xi_n|\epsilon,m),
    \end{align}
    where $\mathcal P$ can be expressed in terms of $\Vsym$, i.e., for some constants $c_P$, we have 
    \begin{equation}
    	\mathcal P(\xi_1,\dots,\xi_n|\epsilon,m)=\sum_{P \in \mathfrak{P}_n} c_P \prod_{j}\Vsym((\xi_l)_{l\in P_j}|\epsilon,m).
    \end{equation}
    As the $f_k$ are arbitrary, this allows us to identify (for $n\geq 3$)
    \begin{equation}
    	\avga{\prod_{k=1}^n\wick{e^{i\sqrt{\beta}\sigma_k\varphi(x_k)}}_\epsilon}_{\GFF(\epsilon,m)}^T =\mathcal P(\xi_1,\dots,\xi_n|\epsilon,m),
    \end{equation}
    and the convergence and local integrability, for $n \geq 3$, follow immediately from  Lemma~\ref{le:Vsymlim} respectively  Lemma~\ref{le:Vub}.
    
    It remains to consider the $n=1$ and $n=2$ cases. The $n=1$ case is trivial,
    while for $n=2$ the statements are straightforward to check due to the assumptions
    that $K$ is supported away from the diagonal respectively that the test functions have disjoint support;
    we omit further details.
\end{proof}

\subsection{Analysis of the sine-Gordon correlation functions} \label{sec:cfconv}

We are finally in a position to prove our main result concerning the sine-Gordon correlation functions,
i.e., prove Theorem~\ref{thm:cf}.

In preparation of the proof,
one  readily checks that since we are dealing with Gaussian random variables and bounded random variables,
for any $\epsilon,m>0$ and $\Lambda\subset\R^2$ compact, the function 
\begin{align}
  \notag &(\mu_1,\dots,\mu_{n},\nu_1,\dots,\nu_{q},\eta_1,\dots,\eta_{q'},z)\mapsto \\
         &\log \avga{\exp\left[\sum_{{k}=1}^{n} \mu_{k} \wick{e^{i\sqrt{\beta}\sigma_{k}\varphi}}_\epsilon(f_{k})+\sum_{{j}=1}^{q} \nu_{j}\partial \varphi(g_{j})+\sum_{{j'}=1}^{q'} \eta_{j'} \bar \partial \varphi(h_{j'})\right] }_{\SG(\beta,z|\epsilon,m,\Lambda)}
\end{align}
is analytic in some neighborhood of the origin
(which may a priori depend on $\epsilon,m,\Lambda$) 
and the correlation function of interest is obtained from it by differentiating once with respect to each $\mu_{k},\nu_{j},\eta_{j'}$, and then setting these parameters to zero.
Our goal is to prove that (after suitable renormalization) this function is actually analytic in a larger domain, that does not depend on $\epsilon,m$, and that it converges uniformly in the relevant parameters. The limiting function will then automatically be analytic in the given domain, and also the relevant derivatives will converge. In particular, as we will eventually see, this will imply the convergence of the correlation functions, and that they are also analytic in $z$.

We begin by applying the Girsanov--Cameron--Martin theorem in a similar way as in Lemma~\ref{le:GFFmixed}.
We now need the following version for Gaussian fields on $\R^2$:
Let $\varphi$ be a smooth Gaussian field on $\R^2$ and let $Y$ be a Gaussian random variable measurable with respect to $(\varphi(x))_{x\in\R^d}$.
Then
\begin{equation}
  \E(F(\varphi)e^{Y-\E Y^2}) = \E(F(\varphi + \E(\varphi Y))),
\end{equation}
where $\E(\varphi Y)$ stands for the function $x\mapsto \E(\varphi(x)Y)$;
see, e.g., \cite[Theorem~2.8]{MR2244975} for a more general setting.
This implies that for real $z$
(the application for complex arguments in the exponential below is justified as in the proof of Lemma~\ref{le:GFFmixed})
\begin{align} \label{e:girsanov1}
  &\avga{\exp\left[\sum_{{k}=1}^{n} \mu_{k} \wick{e^{i\sqrt{\beta}\sigma_{k}\varphi}}_\epsilon(f_{k})+\sum_{{j}=1}^{q} \nu_{j}\partial \varphi(g_{j})+\sum_{{j'}=1}^{q'} \eta_{j'} \bar \partial \varphi(h_{j'})\right] }_{\SG(\beta,z|\epsilon,m,\Lambda)}\nnb
  &=\frac{1}{Z(z|\epsilon,m,\Lambda)} \Bigg\langle \exp\bigg[{-} \int d\xi\, \wick{e^{i\sqrt{\beta}{\sigma}\varphi(x)}}_\epsilon
    \zeta_{\mu,\nu,\eta,z,\Lambda}(\xi|\epsilon,m)\bigg] \Bigg\rangle_{\GFF(\epsilon,m)}
    \nnb
  &\qquad\qquad
    \times
    \exp\Bigg[\frac{1}{2}\avga{\Bigg(\sum_{{j}=1}^{q} \nu_{j}\partial \varphi(g_{j})+\sum_{{j'}=1}^{q'} \eta_{j'} \bar \partial \varphi(h_{j'})\Bigg)^2}_{\GFF(\epsilon,m)}\Bigg]
\end{align}
where we have introduced the notation (recall $\xi=(x,\sigma)$)
\begin{align} \label{e:zetamunudef}
  & 
    \zeta_{\mu,\nu,\eta,z,\Lambda}(\xi | \epsilon,m)
    \nnb & 
           =
  -\nz\mathbf 1_{\Lambda}(x)e^{i\sqrt{\beta}\sigma\sum_{{j}=1}^{q} \nu_{j} \avga{\varphi(x)\partial \varphi(g_{j})}_{\GFF(\epsilon,m)}+i\sqrt{\beta}\sigma\sum_{{j'}=1}^{q'} \eta_{j'} \avga{\varphi(x)\bar\partial \varphi(h_{j'})}_{\GFF(\epsilon,m)}}
           \nnb & 
                  -\sum_{{k}=1}^{n} \mu_{k}f_{k}(x)e^{i\sqrt{\beta}\sigma_{k} \sum_{{j}=1}^{q} \nu_{j} \avga{\varphi(x)\partial \varphi(g_{j})}_{\GFF(\epsilon,m)}+i\sqrt{\beta}\sigma_{k}\sum_{{j'}=1}^{q'} \eta_{j'} \avga{\varphi(x)\bar\partial \varphi(h_{j'})}_{\GFF(\epsilon,m)}}\delta_{\sigma,\sigma_{k}}.
\end{align}
In terms of the renormalized partition function \eqref{e:rpf}, we may write \eqref{e:girsanov1} as
\begin{align} \label{e:girsanov2}
  & 
    \frac{\mathcal Z(\zeta_{\mu,\nu,\eta,z,\Lambda}(\cdot|\epsilon,m)|\epsilon,m,\Lambda)}{\mathcal Z(z|\epsilon,m,\Lambda)}
    \nnb 
  &
    \times  
    \exp\left[\frac{1}{2}\int d\xi_1 \, d\xi_2 \,  \zeta_{\mu,\nu,\eta,z,\Lambda}(\xi_1|\epsilon,m)\zeta_{\mu,\nu,\eta,z,\Lambda}(\xi_2|\epsilon,m)A(\xi_1,\xi_2|\epsilon,m)\right]
    \nnb
  &\times \exp\left[-\frac{z^2}{2}\int d\xi_1 \, d\xi_2 \, \mathbf 1_\Lambda(x_1)\,\mathbf 1_\Lambda(x_2)\, A(\xi_1,\xi_2|\epsilon,m)
    \right]
    \nnb
  &\times
    \exp\left[\frac{1}{2}\avga{\left(\sum_{{j}=1}^{q} \nu_{j}\partial \varphi(g_{j})+\sum_{{j'}=1}^{q'} \eta_{j'} \bar \partial \varphi(h_{j'})\right)^2}_{\GFF(\epsilon,m)}\right].
\end{align}
To study the limit $\epsilon,m\to 0$,         
we define $\zeta_{\mu,\nu,\eta,z,\Lambda}(\xi)$
exactly as $\zeta_{\mu,\nu,\eta,z,\Lambda}(\xi|\epsilon,m)$ in \eqref{e:zetamunudef} with $\GFF(\epsilon,m)$ replaced by $\GFF$
and where 
$\avg{\varphi(x) \partial\varphi(g)}_{\GFF} =\int dy \, g(y) \avg{\varphi(x) \partial\varphi(y)}_{\GFF}$
is given by \eqref{e:cderiv1lim}.
Indeed, Lemma~\ref{lem:cov} implies that, as $\epsilon,m\to 0$, one has 
\begin{equation} \label{e:zetamunu-lim}
  \zeta_{\mu,\nu,\eta,z,\Lambda}(\xi|\epsilon,m) \to
  \zeta_{\mu,\nu,\eta,z,\Lambda}(\xi)
\end{equation}
uniformly on compact sets in $\xi$, and $\mu_{k},\nu_{j},\eta_{j'},z$ (where $\Lambda$ is fixed).
Moreover, since  $\zeta_{\mu,\nu,\eta,z,\Lambda}(\xi|\epsilon,m)$ has uniformly compact support in $\xi$,
the convergence is in fact uniform in $\xi \in \R^2\times \{\pm 1\}$.
We conclude from Theorem \ref{th:main1} item (iv)
and \eqref{e:zetamunu-lim} that
\begin{equation} \label{e:Zzetalim}
  \mathcal Z(\zeta_{\mu,\nu,\eta,z,\Lambda}(\cdot|\epsilon,m)|\epsilon,m)\to \mathcal Z(\zeta_{\mu,\nu,\eta,z,\Lambda}).
\end{equation}
Once again, from the fact that we are dealing with bounded random variables, one readily checks that $\mathcal Z(\zeta_{\mu,\nu,\eta,z,\Lambda}(\cdot|\epsilon,m)|\epsilon,m)$ extends to an entire function of $\mu_{k},\nu_{j},\eta_{j'},z$. Thus our uniform convergence implies that also $\mathcal Z(\zeta_{\mu,\nu,\eta,z,\Lambda})$ extends to an entire function of the variables.

We now consider the cases $n>2$ and $n=0,1,2$ of Theorem~\ref{thm:cf} items (i)--(iii) separately.
The arguments are all very similar.

\begin{proof}[Proof of Theorem~\ref{thm:cf}, (i)--(iii) for ${n}>2$]          
  For ${n}>2$, only $\mathcal Z(\zeta_{\mu,\nu,\eta,z,\Lambda}(\cdot|\epsilon,m)|\epsilon,m)$ in \eqref{e:girsanov2} plays a role
  for the correlation functions
 -- the other terms vanish when we take logarithmic derivatives and set the various parameters to zero. Indeed, 
 \begin{align}
   &\avga{\prod_{{k}=1}^{n}\wick{e^{i\sqrt{\beta}\sigma_{k}\varphi}}_{\epsilon}(f_{k})\prod_{{j}=1}^{q} \partial \varphi(g_{j})\prod_{{j'}=1}^{q'} \bar \partial \varphi(h_{j'})}_{\SG(\beta,z|\epsilon,m,\Lambda)}^T\nnb
   &\qquad =\prod_{{k}=1}^{n}\left.\frac{\partial}{\partial \mu_{k}}\right|_{\mu_{k}=0}\prod_{{j}=1}^{q}\left.\frac{\partial}{\partial \nu_{j}}\right|_{\nu_{j}=0}\prod_{{j'}=1}^{q'}\left.\frac{\partial}{\partial \eta_{j'}}\right|_{\eta_{j'}=0}\log \mathcal Z(\zeta_{\mu,\nu,\eta,z,\Lambda}(\cdot|\epsilon,m)|\epsilon,m).
 \end{align}	
 Now recall from \eqref{e:Zzetalim} that
 \begin{equation}
   \mathcal Z(\zeta_{\mu,\nu,\eta,z,\Lambda}(\cdot|\epsilon,m)|\epsilon,m)\to \mathcal Z(\zeta_{\mu,\nu,\eta,z,\Lambda})
 \end{equation}
 and that the right-hand side is entire in $\mu_{k},\nu_{j},\eta_{j'},z$.
 Moreover, by Theorem \ref{th:main1} item (iii), we know that $\mathcal Z(\zeta_{0,0,0,z,\Lambda})>0$ for $z\in \R$, so we see that there exists some complex neighborhood of the origin $N\subset \C$ and some neighborhood of the real axis $\R\subset N'\subset \C$ such that $\log \mathcal Z(\zeta_{\mu,\nu,\eta,z,\Lambda}(\cdot|\epsilon,m)|\epsilon,m)\to \log \mathcal Z(\zeta_{\mu,\nu,\eta,z,\Lambda})$ uniformly in $\mu_{k},\nu_{j},\eta_{j'}\in N$ and $z$ in a compact subset of $N'$, and that the limit is analytic in this domain. This implies that also the $\mu,\nu,\eta$ derivatives of this logarithm evaluated at zero converge and are analytic in $z\in N'$. 
 We have thus proven item (i) and item (ii) of Theorem~\ref{thm:cf} for ${n}>2$. Let us turn to item (iii).
 
 Again, since we know that  $\log \mathcal Z(\zeta_{\mu,\nu,\eta,z,\Lambda}(\cdot|\epsilon,m)|\epsilon,m)$ converges uniformly (and is analytic in a suitable domain), we know that also its derivatives converge. In particular, going back in our argument, our remaining task is to evaluate
 the $\epsilon\to 0$, $m\to 0$ limit of
 \begin{align}
   &
     \left.\frac{d^{l}}{dz^{l}}\right|_{z=0}\avga{\prod_{{k}=1}^{n} \wick{e^{i\sqrt{\beta}\sigma_{k}\varphi}}_{\epsilon}(f_{k})\prod_{{j}=1}^{q} \partial \varphi(g_{j})\prod_{{j'}=1}^{q'} \bar \partial \varphi(h_{j'})}_{\SG(\beta,z|\epsilon,m,\Lambda)}^T\\
   &=
     \sum_{\tau_1,\dots,\tau_{l}\in\{-1,1\}}\avga{\prod_{{k}=1}^{n} \wick{e^{i\sqrt{\beta}\sigma_{k}\varphi}}_{\epsilon}(f_{k}) \prod_{{j}=1}^{q} \partial \varphi(g_{j})\prod_{{j'}=1}^{q'} \bar \partial \varphi(h_{j'})\prod_{s=1}^{l}
     \wick{e^{i\sqrt{\beta}\tau_s \varphi}}_{\epsilon}(\mathbf 1_\Lambda)}_{\GFF(\epsilon,m)}^T\notag
 \end{align}
 for ${l}\geq 0$. The claim for item (iii) for ${n}>2$ now follows from Lemma \ref{lem:gffintegrability}.
\end{proof}

\begin{proof}[Proof of Theorem \ref{thm:cf}, (i)--(iii) for ${n}=2$] 
  For ${n}=2$, also the
  \begin{equation}
    \frac{1}{2}\int d\xi_1\, d\xi_2 \, \zeta_{\mu,\nu,\eta,z,\Lambda}(\xi_1|\epsilon,m)\zeta_{\mu,\nu,\eta,z,\Lambda}(\xi_2|\epsilon,m)A(\xi_1,\xi_2|\epsilon,m)
  \end{equation}
  term in \eqref{e:girsanov2} and the contribution from the Girsanov transform contribute.
  More precisely, one finds (recalling \eqref{e:ct}) that
  \begin{equation}
    \avga{\prod_{{k}=1}^2 \wick{e^{i\sqrt{\beta}\sigma_{k}\varphi}}_\epsilon(f_{k}) \prod_{{j}=1}^{q} \partial \varphi(g_{j})\prod_{{j'}=1}^{q'} \bar \partial \varphi(h_{j'})}_{\SG(\beta,z|\epsilon,m,\Lambda)}^T
  \end{equation}
  equals
  \begin{align}
  &\prod_{{k}=1}^2\left.\frac{\partial}{\partial \mu_{k}}\right|_{\mu_{k}=0}\prod_{{j}=1}^{q}\left.\frac{\partial}{\partial \nu_{j}}\right|_{\nu_{j}=0}\prod_{{j'}=1}^{q'}\left.\frac{\partial}{\partial \eta_{j'}}\right|_{\eta_{j'}=0}\log \mathcal Z(\zeta_{\mu,\nu,\eta,z,\Lambda}(\cdot|\epsilon,m)|\epsilon,m)\\
    & \quad +\int_{(\R^2)^2}dx_1dx_2\avga{\wick{e^{i\sqrt{\beta}\sigma_1\varphi(x_1)}}_{\epsilon}\wick{e^{i\sqrt{\beta}\sigma_2\varphi(x_2)}}_{\epsilon}}_{\GFF(\epsilon,m)}^Tf_1(x_1)f_2(x_2)\nnb
    &\qquad \qquad \times \prod_{{j}=1}^{q} \left(i\sqrt{\beta}\sigma_1\avga{\varphi(x_1)\partial \varphi(g_{j})}_{\GFF(\epsilon,m)}+i\sqrt{\beta}\sigma_2\avga{\varphi(x_2)\partial \varphi(g_{j})}_{\GFF(\epsilon,m)}\right) \nnb
    &\qquad \qquad \times \prod_{{j'}=1}^{q'} \left(i\sqrt{\beta}\sigma_1\avga{\varphi(x_1)\bar\partial \varphi(h_{j'})}_{\GFF(\epsilon,m)}+i\sqrt{\beta}\sigma_2\avga{\varphi(x_2)\bar \partial \varphi(h_{j'})}_{\GFF(\epsilon,m)}\right). \notag
  \end{align}
  The $\mathcal Z$-term again converges uniformly and gives rise to a function analytic in its parameters. 
  The remaining term on the other hand is readily seen {(as in the proof of Lemma \ref{lem:gffintegrability})} to converge {if $q+q'\geq 1$ or if $f_1,f_2$ have disjoint supports.}
  This reasoning proves items (i) and (ii) of Theorem \ref{thm:cf} for ${n}=2$,
  and item (iii) for ${n}=2$ is verified in the same manner as for ${n}>2$.
\end{proof}

Since ${n}=1$ is relevant also to item (iv), let us next consider $n=0$, and then conclude the proof with the case ${n}=1$.
      
\begin{proof}[Proof of Theorem \ref{thm:cf}, (i)--(iii) for ${n}=0$] 
  The proof for ${n}=0$ is again similar, but now with a further contribution from the third exponential in
  \eqref{e:girsanov2}:
  \begin{equation}
    \avga{\prod_{{j}=1}^{q} \partial \varphi(g_{j})\prod_{{j'}=1}^{q'} \bar \partial \varphi(h_{j'})}_{\SG(\beta,z|\epsilon,m,\Lambda)}^T
  \end{equation}
  equals
  \begin{align}
  &\prod_{{j}=1}^{q}\left.\frac{\partial}{\partial \nu_{j}}\right|_{\nu_{j}=0}\prod_{{j'}=1}^{q'}\left.\frac{\partial}{\partial \eta_{j'}}\right|_{\eta_{j'}=0}\log \mathcal Z(\zeta_{0,\nu,\eta,z,\Lambda}(\cdot|\epsilon,m)|\epsilon,m)\\
    &\quad +\frac{z^2}{2}\sum_{\tau_1,\tau_2\in \{-1,1\}}\int_{\Lambda^2}dx_1dx_2\avga{\wick{e^{i\sqrt{\beta}\tau_1\varphi(x_1)}}_{\epsilon}\wick{e^{i\sqrt{\beta}\tau_2\varphi(x_2)}}_{\epsilon}}_{\GFF(\epsilon,m)}^T\notag\\
    &\qquad \qquad \times \prod_{{j}=1}^{q} \left(i\sqrt{\beta}\tau_1\avga{\varphi(x_1)\partial \varphi(g_{j})}_{\GFF(\epsilon,m)}+i\sqrt{\beta}\tau_2\avga{\varphi(x_2)\partial \varphi(g_{j})}_{\GFF(\epsilon,m)}\right) \notag\\
    &\qquad \qquad \times \prod_{{j'}=1}^{q'} \left(i\sqrt{\beta}\tau_1\avga{\varphi(x_1)\bar\partial \varphi(h_{j'})}_{\GFF(\epsilon,m)}+i\sqrt{\beta}\tau_2\avga{\varphi(x_2)\bar \partial \varphi(h_{j'})}_{\GFF(\epsilon,m)}\right)\notag\\
    &\quad +\delta_{{q},1}\delta_{{q'},1}\avga{\partial \varphi(g_1)\bar{\partial}\varphi(h_1)}_{\GFF(\epsilon,m)}
      \notag\\
    &\quad + \delta_{{q},2}\delta_{{q'},0} \avga{\partial \varphi(g_1)\partial\varphi(g_2)}_{\GFF(\epsilon,m)} +  \delta_{{q},0}\delta_{{q'},2} \avga{\bar\partial \varphi(h_1)\bar{\partial}\varphi(h_2)}_{\GFF(\epsilon,m)}
      .\notag
  \end{align}
  The $\mathcal Z$-term can be treated as before and the last three terms {converge by Lemma \ref{le:gradint}}.
  For the $z^2$-term, we can argue exactly as in the proof of Lemma \ref{lem:gffintegrability} and conclude that also in this case, the correlation functions converge and define analytic functions of $z$.
  Thus we have proven items (i) and (ii) of Theorem \ref{thm:cf} in the case ${n}=0$. Item (iii) is verified in the same way as for ${n}>2$.
  
  Finally, to see that the limit of right-hand side is symmetric in $z$,
  note that $\zeta_{0,\nu,\eta,z}$ is proportional to $z$ (since ${n}=0$)    and that $\mathcal Z$ is even in $\zeta$ in the limit
  $\epsilon\to 0$ and $m\to 0$.
\end{proof}

\begin{proof}[Proof of Theorem \ref{thm:cf}, (i)--(iii) for ${n}=1$] 
  We finally consider ${n}=1$. Let us first look at the situation where ${q+q'}\geq 1$ where we have that
  \begin{align}
    \avga{\wick{e^{i\sqrt{\beta}\sigma_{1}\varphi}}_{\epsilon}(f_{1}) \prod_{{j}=1}^{q} \partial \varphi(g_{j})\prod_{{j'}=1}^{q'} \bar \partial \varphi(h_{j'})}_{\SG(\beta,z|\epsilon,m,\Lambda)}^T
  \end{align}
  equals
  \begin{align}
    &\left.\frac{\partial}{\partial \mu_1}\right|_{\mu_1=0}\prod_{{j}=1}^{q}\left.\frac{\partial}{\partial \nu_{j}}\right|_{\nu_{j}=0}\prod_{{j'}=1}^{q'}\left.\frac{\partial}{\partial \eta_{j'}}\right|_{\eta_{j'}=0}\log \mathcal Z(\zeta_{\mu,\nu,\eta,z,\Lambda}(\cdot|\epsilon,m)|\epsilon,m)\\
    &\quad +\sum_{\tau\in \{-1,1\}}\int_{\R^2\times \Lambda}dx_1\, dx_2\, \avga{\wick{e^{i\sqrt{\beta}\sigma_1\varphi(x_1)}}_{\epsilon}\wick{e^{i\sqrt{\beta}\tau\varphi(x_2)}}_{\epsilon}}_{\GFF(\epsilon,m)}^Tf_1(x_1)\nnb
    &\qquad \qquad \times \prod_{{j}=1}^{q} \left(i\sqrt{\beta}\sigma_1\avga{\varphi(x_1)\partial \varphi(g_{j})}_{\GFF(\epsilon,m)}+i\sqrt{\beta}\tau\avga{\varphi(x_2)\partial \varphi(g_{j})}_{\GFF(\epsilon,m)}\right) \nnb
    &\qquad \qquad \times \prod_{{j'}=1}^{q'} \left(i\sqrt{\beta}\sigma_1\avga{\varphi(x_1)\bar\partial \varphi(h_{j'})}_{\GFF(\epsilon,m)}+i\sqrt{\beta}\tau\avga{\varphi(x_2)\bar \partial \varphi(h_{j'})}_{\GFF(\epsilon,m)}\right). \notag
  \end{align}
  Finiteness and convergence of this quantity is again argued analogously as in the proof of Lemma~\ref{lem:gffintegrability}, so we have the proof of items (i) and (ii) also in the ${n}=1$ case. Item (iii) follows by the same argument as before. 
\end{proof}

\begin{proof}[Proof of Theorem \ref{thm:cf}, (iv)]
  The only thing that remains is thus item (iv). For this, we find with similar reasoning as before (recall we chose $f$ to be supported in $\Lambda$)
  \begin{multline}
    \avga{\wick{e^{i\sqrt{\beta}\sigma_1\varphi}}_\epsilon(f) }_{\SG(\beta,z|\epsilon,m,\Lambda)}^T
    =\left.\frac{\partial}{\partial \mu_1}\right|_{\mu_1=0}\log \mathcal Z(\zeta_{\mu_1,0,0,z,\Lambda}(\cdot|\epsilon,m)|\epsilon,m)
    \\ 
    +\sum_{\tau\in \{-1,1\}}\nz\int_{\Lambda^2}dx_1\, dx_2\, \avga{\wick{e^{i\sqrt{\beta}\sigma_1\varphi(x_1)}}_{\epsilon}\wick{e^{i\sqrt{\beta}\tau\varphi(x_2)}}_{\epsilon}}_{\GFF(\epsilon,m)}^Tf(x_1).
  \end{multline}
  The first term once again has a finite limit as $\epsilon,m\to 0$, but as we now prove, the second term blows up. For the second term, if $\tau=\sigma_1$, then everything is bounded, but for $\tau\neq \sigma_1$, the leading order behavior (in $\epsilon$) is given by (making use of asymptotics e.g. from Lemma \ref{lem:cov})
  \begin{align}
    &\int_{\Lambda^2}dx_1\, dx_2\, \epsilon ^{-\frac{\beta}{2\pi}}e^{-\beta \avga{\varphi(0)^2}_{\GFF(\epsilon,m)}} e^{\beta \avga{\varphi(x_1)\varphi(x_2)}_{\GFF(\epsilon,m)}}f(x_1)\\
    &=(1+o(1))\int_{\Lambda^2}dx_1\, dx_2 \, m^{\frac{\beta}{2\pi}}e^{\frac{\beta}{4\pi}\gamma}  e^{\beta \int_{\epsilon^2}^{\infty} ds \,\frac{e^{-m^2s}}{4\pi s}e^{-\frac{|x_1-x_2|^2}{4s}}}f(x_1)\nnb
    &=(1+o(1)+O(m^2))\int_{\Lambda^2}dx_1\, dx_2 \, e^{\beta \int_{\epsilon^2}^{1} ds \,\frac{e^{-m^2s}-1}{4\pi s}e^{-\frac{|x_1-x_2|^2}{4s}}+\beta \int_{\epsilon^2}^{1} ds \,\frac{1}{4\pi s}e^{-\frac{|x_1-x_2|^2}{4s}}}\nnb
    &\qquad \qquad \times e^{\beta\int_1^\infty ds\, \frac{e^{-m^2s}}{4\pi s}(e^{-\frac{|x_1-x_2|^2}{4s}}-1)}f(x_1),\notag
  \end{align}	
  from which we see that the leading order asymptotics (as $\epsilon,m\to 0$) are given by 
  \begin{align}
    &\int_{\Lambda^2}dx_1\, dx_2\,  e^{\beta \int_{\epsilon^2}^{1}ds \,\frac{1}{4\pi s}e^{-\frac{|x_1-x_2|^2}{4s}}+\beta\int_1^\infty ds\,\frac{1}{4\pi s}(e^{-\frac{|x_1-x_2|^2}{4s}}-1)}f(x_1)\\
    &=e^{-\frac{\gamma \beta}{4\pi}}\int_{\Lambda^2}dx_1\, dx_2\,  |x_1-x_2|^{-\frac{\beta}{2\pi}}e^{-\frac{\beta}{4\pi}\Gamma(0,\frac{|x_1-x_2|^2}{\epsilon^2})}f(x_1)\nnb
    &=e^{-\frac{\gamma \beta}{4\pi}}\int_{\Lambda}dx \, f(x)\int_{x-\epsilon u\in \Lambda}du\, \epsilon ^{2-\frac{\beta}{2\pi}} |u|^{-\frac{\beta}{2\pi}}e^{-\frac{\beta}{4\pi}\Gamma(0,|u|^2)}.\notag
  \end{align}
  For $\beta>4\pi$, $u\mapsto |u|^{-\frac{\beta}{2\pi}}e^{-\frac{\beta}{4\pi}\Gamma(0,|u|^2)}\in L^1(\R^2)$, and we see that 
  \begin{multline}
    e^{-\frac{\gamma \beta}{4\pi}}\int_{\Lambda}dx \, f(x)\int_{x-\epsilon u\in \Lambda}du\,  \epsilon^{2-\frac{\beta}{2\pi}} |u|^{-\frac{\beta}{2\pi}}e^{-\frac{\beta}{4\pi}\Gamma(0,|u|^2)}\ 
    \\
    =(1+o(1))\epsilon^{2-\frac{\beta}{2\pi}}2\pi e^{-\frac{\gamma \beta}{4\pi}}\int_{\Lambda}dx \, f(x)\int_0^\infty dr\, r^{-\frac{\beta}{2\pi}+1}e^{-\frac{\beta}{4\pi}\Gamma(0,r^2)},
  \end{multline}
  which also concludes the proof of item (iv) for $\beta>4\pi$. 
  
  For $\beta=4\pi$, on the other hand, we obtain a logarithmic singularity from the long range behavior of the $u$ integral and one finds for the relevant asymptotics
  \begin{equation} 
    e^{-\gamma}\int_{\Lambda}dx\, f(x)\int_{x-\epsilon u\in \Lambda}du\, |u|^{-2}e^{-\Gamma(0,|u|^2)}
    =(1+o(1))e^{-\gamma}2\pi \log \epsilon ^{-1}\int_{\Lambda}dx\, f(x).
  \end{equation}
  This concludes the proof of item (iv) for $\beta=4\pi$ as well, and also the proof of the theorem.
\end{proof}
      
\subsection{Existence of $\varphi$ field}

Finally, we prove Theorem~\ref{thm:sgphi}. Since the proof is essentially identical to that of Theorem~\ref{thm:cf},
we will be somewhat brief.

\begin{proof}[Proof of Theorem~\ref{thm:sgphi}]
We are interested in the function
\begin{equation}
  (w,z)\mapsto
  \avga{e^{w\varphi(f)} }_{\SG(\beta,z|\epsilon,m,\Lambda)}
\end{equation}
which we may again write using Girsanov's theorem as
\begin{align} \label{e:girsanov2-bis}
  & 
    \frac{\mathcal Z(\zeta_{w,z,\Lambda}(\cdot|\epsilon,m)|\epsilon,m,\Lambda)}{\mathcal Z(z|\epsilon,m,\Lambda)}
    \nnb  
  & \times 
    \exp\left[\frac{1}{2}\int d\xi_1 \, d\xi_2 \,  \zeta_{w,z,\Lambda}(\xi_1|\epsilon,m)\zeta_{w,z,\Lambda}(\xi_2|\epsilon,m)A(\xi_1,\xi_2|\epsilon,m)\right]
    \nnb
  &\times \exp\left[-\frac{z^2}{2}\int d\xi_1 \, d\xi_2 \, \mathbf 1_\Lambda(x_1)\,\mathbf 1_\Lambda(x_2)\, A(\xi_1,\xi_2|\epsilon,m)
    \right]
    \nnb
  &\times
    \exp\left[\frac{w^2}{2}\avga{\varphi(f)^2}_{\GFF(\epsilon,m)}\right]
\end{align}
where now
\begin{equation} \label{e:zetawdef-bis}
  \zeta_{w,z,\Lambda}(\xi | \epsilon,m)
  =
  -\nz\mathbf 1_{\Lambda}(x)e^{i\sqrt{\beta}\sigma w \avga{\varphi(x)\varphi(f)}_{\GFF(\epsilon,m)}}.
\end{equation}
We only consider the limit $\epsilon\to 0$ and $m\to 0$; the argument for $\epsilon\to 0$ with $m>0$ fixed is analogous.
Thus let $f \in L^\infty_c(\R^2)$ with $\int f\, dx =0$. 
Then $\avg{\varphi(f)^2}_{\GFF(\epsilon,m)}$ converges as $\epsilon,m\to 0$ by Lemma~\ref{lem:cov}.
Moreover, again using Lemma~\ref{lem:cov}, we have that
$\zeta_{w,z,\Lambda}(\xi | \epsilon,m) \to   \zeta_{w,z,\Lambda}(\xi)$ uniformly in $\xi \in \R^2 \times \{\pm 1\}$
and uniformly on compact sets of $z,w$, and thus $\mathcal Z(  \zeta_{w,z,\Lambda}(\cdot | \epsilon,m)|\epsilon,m) \to \mathcal Z(\zeta_{w,z,\Lambda})$
and the limit is entire in $z,w$. As in the proof of Theorem~\ref{thm:cf}, the same is true for the other terms.
\end{proof}

\section{Estimates for free fermions with finite volume mass}
\label{sec:ferm}

In this section, we prove Theorem~\ref{thm:ferm}.
Most of our work goes into the construction and analysis of the fundamental solution (Green's function)
of the Dirac operator with a finite volume mass term.
We state these estimates in Section~\ref{sec:Green-statement},
then deduce Theorem~\ref{thm:ferm} in Section~\ref{sec:ferm-pf},
and finally prove the estimates stated in Section~\ref{sec:Green-statement}
in the remainder of Section~\ref{sec:ferm}.

\subsection{Statement of estimates on the Green's function}
\label{sec:Green-statement}

Recall that we are considering $\Lambda_L=\{x=(x_0,x_1)\in \R^2: |x|\leq L\}$,
and that we identify $\R^2$ with $\C$. We are interested in the Dirac operator 
\begin{equation}
  D=D_{\mu,\Lambda_L}
  =\Dirac +\nmu\mathbf 1_{\Lambda_L} :=
  \begin{pmatrix}
    \nmu \mathbf 1_{\Lambda_L} & 2\bar \partial\\
    2\partial & \nmu \mathbf 1_{\Lambda_L}
  \end{pmatrix}
\end{equation}
where $\partial=\frac{1}{2}(-i\partial_0+\partial_1)$ and $\bar \partial=\frac{1}{2}(i\partial_0+\partial_1)$.
For each $y\in \mathrm{int}(\Lambda_L)=\{z\in \C:|z|<L\}$, we are looking for a continuous function $S_{\mu\mathbf 1_{\Lambda_L}}(\cdot,y):\C\setminus \{y\}\to \C^{2\times 2}$ such that 
\begin{equation}
DS_{\mu\mathbf 1_{\Lambda_L}}(\cdot,y)=\delta_{y} \qquad \text{and} \qquad \lim_{|x|\to \infty}S_{\mu\mathbf 1_{\Lambda_L}}(x,y)=0.
\label{e:Green}
\end{equation}
Our results for this function $S_{\mu \mathbf 1_{\Lambda_L}}$ are summarized in the following theorem.
In the statement, we also use $S_0$ and $S_\mu$ to denote the explicit infinite volume Dirac Green's function \eqref{e:Smu}:
\begin{align}
  \label{e:S0-bis}
  S_0(x,y) &=\frac{1}{2\pi}\begin{pmatrix}
    0 & 1/(\bar x-\bar y)\\
    1/(x-y) & 0
  \end{pmatrix},
  \\
  \label{e:Smu-bis}
  S_{\mu}(x,y) &= 
                 -\frac{1}{2\pi} \begin{pmatrix}
                   -\nmu K_0(|\mu||x-y|) & 2\bar\partial_x K_0(|\mu||x-y|)\\
                   2 \partial_xK_0(|\mu||x-y|) & -\nmu K_0(|\mu||x-y|)
                 \end{pmatrix}, \qquad (\nmu\neq 0),
\end{align}
where $K_0$ is the $0$'th modified Bessel function of the second kind.
It is well known that $S_0$ and $S_\mu$ really are the fundamental solutions of $i\Dirac$ and $i\Dirac +\nmu$ on $\R^2$
and it also follows from the well-known asymptotics of $K_0$ that $S_\mu(x,y) \to S_0(x,y)$ as $\mu \to 0$ when $x\neq y$.
For a matrix $S$, we will denote by $|S|$ a submultiplicative matrix norm of $S$.

\begin{theorem}
  For each $L \geq 1$, $y\in \mathrm{int}(\Lambda_L)$ and $\mu \in \R$, 
  there exists a continuous function $S_{\mu\mathbf 1_{\Lambda_L}}(\cdot,y):\C\setminus \{y\}\to \C^{2\times 2}$ that satisfies \eqref{e:Green}
  and has the following properties
  for some polynomial $P= P(L,|\mu|)$ in both variables (which does not depend on any of the arguments below).
	\begin{enumerate}
		\item For all $x,y\in \mathrm{int}(\Lambda_L)$, $x\neq y$, and $L\geq 1$, we have 
		(with $S_0$ as in \eqref{e:S0-bis})
		\begin{equation}
                  |S_{\mu \mathbf 1_{\Lambda_L}}(x,y) - S_0(x,y)| \leq P(L,|\mu|)(1+|\log |x-y||).
                 \label{e:Grbound}
                  \end{equation}
		\item 	For all $x\in\Lambda_L^{\mathsf c}$, $y\in \mathrm{int}(\Lambda_L)$, and $L\geq 1$, we have 
                  \begin{equation}
                  |S_{\mu\mathbf 1_{\Lambda_L}}(x,y)|\leq \frac{P(L,{|\mu|})}{L-|y|}.
              \end{equation}
            \item For each $x,y\in \mathrm{int}(\Lambda_L)$ with $x\neq y$, the function
              $\mu\mapsto S_{\mu\mathbf 1_{\Lambda_L}}(x,y)$
		has an analytic continuation into some $L$-dependent neighborhood of the real axis,
                this analytic continuation also satisfies the estimate \eqref{e:Grbound},
                and
		\begin{equation}
                  \lim_{\mu \to 0}S_{\mu \mathbf 1_{\Lambda_L}}(x,y)
                  =S_0(x,y) =\frac{1}{2\pi}\begin{pmatrix}
                    0 & 1/(\bar x-\bar y)\\
		1/(x-y) & 0
		\end{pmatrix}.
		\end{equation}
		
              \item For any $x,y\in \mathrm{int}(\Lambda_L)$ with $x\neq y$, for $\mu\in \R$ we have
		\begin{equation}
                  \partial_\mu S_{\mu\mathbf 1_{\Lambda_L}}(x,y)= { -}\int_{\Lambda_L}du\,S_{\mu\mathbf 1_{\Lambda_L}}(x,u)S_{\mu\mathbf 1_{\Lambda_L}}(u,y).
		\end{equation}
		
              \item For each fixed $\mu\in \R\setminus \{0\}$, uniformly on compact subsets of $x\neq y \in \R^2$, as $L\to \infty$, 
		\begin{equation}
                  S_{\mu\mathbf 1_{\Lambda_L}}(x,y)\to
                  S_{\mu}(x,y) = 
                  -\frac{1}{2\pi} \begin{pmatrix}
                    -\nmu K_0(|\mu||x-y|) & 2\bar\partial_x K_0(|\mu||x-y|)\\
		2 \partial_xK_0(|\mu||x-y|) & -\nmu K_0(|\mu||x-y|)
		\end{pmatrix}.
		\end{equation}
	\end{enumerate}
	\label{th:gestimate}
\end{theorem}

\subsection{Proof of Theorem \ref{thm:ferm}}
\label{sec:ferm-pf}

For Theorem \ref{thm:ferm}, our function $S_{\mu\mathbf 1_{\Lambda_L}}$ is of course the Green's function of Theorem~\ref{th:gestimate}.
We will denote the components of the $2\times 2$ matrix $S_{\mu\mathbf 1_{\Lambda_L}}(x,y)$
by $S_{\mu\mathbf 1_{\Lambda_L};ij}(x,y)$ where $i,j\in \{1,2\}$.
We also recall the definition of the truncated correlation functions from  \eqref{e:truncated-muLambda},
as well as the truncated two-point functions with singularity subtracted from \eqref{e:trunc2ptsub}.
Let us begin with the proof of item (i) of Theorem \ref{thm:ferm}. We formulate this as the following lemma.

\begin{lemma}  \label{le:f1}
  For $n\geq 3$ and $f_1,\dots,f_n\in L_c^\infty(\Lambda_L)$,
  \begin{equation}
    \mu\mapsto \avga{\prod_{i=1}^n\bar\psi_{\alpha_{i}}\psi_{\beta_{i}}(f_{i})}_{\FF(\mu \mathbf 1_{\Lambda_L})}^T
  \end{equation}
  has an analytic continuation to an $L$-dependent neighborhood of $\R$ (with $S_{\mu\mathbf 1_{\Lambda_L}}=S_0$ for $\mu=0$).
  For $n=2$, the same holds if $f_1$ and $f_2$ have disjoint compact supports or if the truncated
    two-point function is replaced by \eqref{e:trunc2ptsub}.
\end{lemma}

\begin{proof}
From Theorem~\ref{th:gestimate} item~(i), which implies $|S_{\mu \mathbf 1_\Lambda}(x,y)| \leq P(L,|\mu|)/|x-y|$ for $x,y \in \Lambda_L$,
and the representation \eqref{e:truncated-muLambda}, we see that the smeared truncated correlation functions exist for all $\mu\in \R$ and $n \geq 3$ (the claim about what happens at $\mu=0$ following from Theorem \ref{th:gestimate} item (iii)). Here we used that for, any compact $K \subset \R^2$,
\begin{equation}
  \int_K du\, \frac{1}{|x-u||u-y|} \leq C_K (1+|\log |x-y||).
\end{equation}
For $n=2$, in the same way, the truncated two-point function with subtracted singularity  \eqref{e:trunc2ptsub} exists;
or, alternatively, if $f_1$ and $f_2$ have disjoint compact supports, the truncated two-point function also exists trivially.
Moreover, Theorem \ref{th:gestimate} item~(iii) (in particular, the analogue of \eqref{e:Grbound} for complex $\mu$) allows us to construct a candidate for the analytic continuation of the truncated correlation functions (with subtracted singularity for $n=2$).
More precisely, we define the candidate by the formula \eqref{e:truncated-muLambda} though now using the analytic continuation of the Green's function provided by Theorem~\ref{th:gestimate} item~(iii). Using Theorem~\ref{th:gestimate} item~(iii) (or more precisely, the bound analogous to \eqref{e:Grbound} for complex $\mu$), a routine dominated convergence argument shows that this candidate for the analytic continuation is continuous in $\mu$ (in this $L$-dependent neighborhood of the real axis). By Morera's theorem, it remains to prove that for any closed loop $\gamma$ (in our $L$-dependent neighborhood of the real axis), we have 
\begin{equation}
  \oint_{\gamma} d\mu\,\avga{\prod_{i=1}^n\bar\psi_{\alpha_{i}}\psi_{\beta_{i}}(f_{i})}_{\FF(\mu \mathbf 1_{\Lambda_L})}^T =0,
\end{equation}
where we have used the $\avga{\cdot}_{\FF(\mu\mathbf 1_{\Lambda_L})}$-notation for our candidate for the analytic continuation.

Now using the analogue of \eqref{e:Grbound} provided by Theorem \ref{th:gestimate} item~(iii), one can use Fubini to translate this into a contour integral over suitable products of $S_{\mu \mathbf 1_{\Lambda_L}}$ at distinct points. By Theorem~\ref{th:gestimate} item~(iii) and Cauchy's integral theorem, this contour integral vanishes, and we are done.
\end{proof}

We next turn to item (ii) of Theorem~\ref{thm:ferm} which we formulate as the following lemma.
\begin{lemma} \label{le:f2}
  For $l\geq 1$ and $n\geq 3$ and $f_1,\dots,f_n\in L_c^\infty(\Lambda_L)$,
  \begin{equation}   	\label{e:f2n}
    \left.\frac{d^{l}}{d\mu^{l}}\right|_{\mu=0} \avga{\prod_{i=1}^n \bar\psi_{\alpha_i}\psi_{\beta_i}(f_i)}^T_{\FF(\mu \mathbf 1_{\Lambda_L})}
    =
    \avga{\prod_{i=1}^n \bar\psi_{\alpha_i}\psi_{\beta_i}(f_i)\big(\bar\psi_{1}\psi_{1}(\mathbf 1_{\Lambda_L})+\bar\psi_{2}\psi_{2}(\mathbf 1_{\Lambda_L})\big)^{l}}^T_{\FF(0)}.
  \end{equation}
  For $n=2$, the same holds if $f_1$ and $f_2$ have disjoint compact supports or if the truncated
  two-point function on the left-hand side is replaced by \eqref{e:trunc2ptsub}.
\end{lemma}

Before the proof, let us just mention that this derivative is finite by the massless correspondence Corollary \ref{cor:massless-corr} and Lemma \ref{lem:gffintegrability}, which implies that the corresponding bosonic correlation functions are integrable, and thus these smeared correlation functions exist.

\begin{proof}
      First assume that $n \geq 3$.
	  We begin by noting that due to Theorem \ref{th:gestimate} item (iv) (interchanging the order of integration and differentiation follows from a routine Cauchy-integral formula/Fubini argument utilizing Lemma~\ref{le:f1} and Theorem~\ref{th:gestimate})
	  \begin{align}
	  &\frac{d}{d\mu}\avga{\prod_{i=1}^n \bar\psi_{\alpha_i}\psi_{\beta_i}(f_i)}^T_{\FF(\mu{\mathbf 1}_{\Lambda_L})}\nnb
	  &= (-1)^{n+2} \sum_{\pi \in C_n}
	  \sum_{j=1}^{n}
	  \sum_{\alpha_{n+1}=1}^2
	  \sum_{\beta_{n+1}=1}^2 \mathbf 1_{\alpha_{n+1}=\beta_{n+1}}
	   \int_{\Lambda_L^{n+1}}dx_1\cdots  dx_{n+1} \, f_1(x_1)\cdots f_n(x_n) 
	  \nnb
	  &\qquad \times S_{\mu{\mathbf 1}_{\Lambda_L};\alpha_{\pi^j(1)}\beta_{n+1}}(x_{\pi^j(1)},x_{n+1})
	  S_{\mu{\mathbf 1}_{\Lambda_L};\alpha_{n+1}\beta_{\pi^{j+1}(1)}}(x_{n+1},x_{\pi^{j+1}(1)})
            \nnb
            &\qquad\times \prod_{i:i\neq j} S_{\mu{\mathbf 1}_{\Lambda_L};\alpha_{\pi^i(1)}\beta_{\pi^{i+1}(1)}}(x_{\pi^i(1)},x_{\pi^{i+1}(1)})
	  .
	  \end{align}
	  Note that $(\pi,j) \in C_n \times [n]$
	  defines a cyclic permutation $\sigma \in C_{n+1}$  in terms of which
	  the right-hand is
	  \begin{multline}
	  (-1)^{n+2} \sum_{\sigma \in C_{n+1}}
	  \sum_{\alpha_{n+1}=1}^2
	  \sum_{\beta_{n+1}=1}^2 \mathbf 1_{\alpha_{n+1}=\beta_{n+1}}
	  \int_{\Lambda_L^{n+1}}dx_1\cdots  dx_{n+1} \, f_1(x_1)\cdots f_n(x_n) 
	  \\
	  \times \prod_{i=1}^{n+1} S_{\mu \mathbf 1_{\Lambda_L};\alpha_{\sigma^i(1)}\beta_{\sigma^{i+1}(1)}}(x_{\sigma^i(1)},x_{\sigma^{i+1}(1)})
	 	  .
	  \end{multline}
	  In particular, this implies (recalling \eqref{e:fermtruncdef}) that
	  \begin{align}
	  	  &\frac{d}{d\mu}\avga{\prod_{i=1}^n \bar\psi_{\alpha_i}\psi_{\beta_i}(f_i)}^T_{\FF(\mu \mathbf 1_{\Lambda_L})}\nnb
	  	  &=
                    \sum_{\alpha_{n+1}=1}^2
	  	  \sum_{\beta_{n+1}=1}^2 \mathbf 1_{\alpha_{n+1}=\beta_{n+1}}
	  	  \int_{\Lambda_L^{n+1}} dx_1\cdots dx_{n+1} \, f_1(x_1)\cdots f_n(x_i) 
	  	  \avga{\prod_{i=1}^{n+1} \bar\psi_{\alpha_i}\psi_{\beta_i}(x_i)}^T_{\FF(\mu \mathbf 1_{\Lambda_L})}\nnb
	  	  &=
                    \avga{\prod_{i=1}^{n} \bar\psi_{\alpha_i}\psi_{\beta_i}(f_i) \; (\bar\psi_1\psi_1({\mathbf 1}_\Lambda)+\bar\psi_2\psi_2({\mathbf 1}_\Lambda))}^T_{\FF(\mu \mathbf 1_{\Lambda_L})}.
	  \end{align}
	  Setting $\mu=0$ (note that this uses Theorem \ref{th:gestimate} item (iii)), we find that
	  \begin{equation}
            \left.\frac{d}{d\mu}\right|_{\mu=0} \avga{\prod_{i=1}^n \bar\psi_{\alpha_i}\psi_{\beta_i}(f_i)}^T_{\FF(\mu \mathbf 1_{\Lambda_L})}\\
	  =
            \avga{\prod_{i=1}^{n} \bar\psi_{\alpha_i}\psi_{\beta_i}(f_i) \; (\bar\psi_1\psi_1({\mathbf 1}_\Lambda)+\bar\psi_2\psi_2({\mathbf 1}_\Lambda))}^T_{\FF(0)}
	  \end{equation}
	  which is the claim when $l=1$. 
	  The case of general $l$ follows by induction.
          For $n=2$, assuming that the truncated two-point function is replaced by \eqref{e:trunc2ptsub} (or alternatively that $f_1$ and $f_2$
          have disjoint compact supports),
          we note that the argument is completely analoguous. The subtracted singularity ensures the integrability of the left-hand sides,
          but does not contribute to the derivatives.
\end{proof}

The final statement of Theorem~\ref{thm:ferm} is the following lemma.
Recall that, on the right hand side, the correlation functions are given by (smeared versions of) \eqref{e:fermtruncdef} now with the
propagator \eqref{e:Smu} (with infinite volume mass term).

\begin{lemma}
  \label{le:f3}
  For any $\mu\in \R$, $n\geq 3$, $f_1,\dots,f_n\in L_c^\infty(\R^2)$, as $L\to\infty$, 
  \begin{equation}
    \avga{\prod_{i=1}^n \bar\psi_{\alpha_i}\psi_{\beta_i}(f_i)}^T_{\FF(\mu \mathbf 1_{\Lambda_L})}\rightarrow \avga{\prod_{i=1}^n \bar\psi_{\alpha_i}\psi_{\beta_i}(f_i)}^T_{\FF(\mu)}.
  \end{equation}
  For $n=2$, the same holds if $f_1$ and $f_2$ have disjoint compact supports or if the truncated
    two-point function on the left-hand side is replaced by \eqref{e:trunc2ptsub}
    and analogously on the right-hand side.
\end{lemma}
    
\begin{proof}
  This is immediate from the uniform convergence of Theorem \ref{th:gestimate} item (v).
  (The modification for $n=2$ is again only used to guarantee integrability.)
\end{proof}

Combining these lemmas yields the proof of Theorem \ref{thm:ferm}, so we are done.

\subsection{Facts about the Laplacian Green's function and eigenfunctions in a disk}

Our proof of Theorem~\ref{th:gestimate} relies on relating $S_{\mu\mathbf 1_{\Lambda_L}}$ to the
the Green's function of the Laplacian in the disk as well as expansions in terms of the eigenfunctions of this Laplacian.
We begin by collecting some well known facts about these.
First, we recall that 
\begin{equation}
G_{\Lambda_L}(x,y)=\frac{1}{2\pi}\log \frac{1}{|x-y|}-\frac{1}{2\pi}\log \frac{1}{|L-\frac{\overline{x}y}{L}|}
\label{e:GF1}
\end{equation}
is the Green's function for the (positive) Laplacian with zero Dirichlet boundary conditions: 
\begin{align}
  -\Delta_x G_{\Lambda_L} (x,y)&=\delta_y(x) \quad \text{for} \quad x,y\in \mathrm{int}(\Lambda_L) \\
  G_{\Lambda_L}(x,y)&=0 \quad \text{for} \quad x\in \partial \Lambda_L, y\in \mathrm{int}(\Lambda_L).
\label{e:GF2}
\end{align}
In \eqref{e:GF1} we wrote $\overline{x}=x_1-ix_0$ for $x=x_1+ix_0$, while in \eqref{e:GF2} we wrote $\Delta_x$ for the Laplacian acting on the $x$ variable.
We also recall that the eigenfunctions of $-\Delta$ on $\Lambda_L$ (with zero boundary conditions) can be written explicitly in terms of Bessel functions and Fourier modes. More precisely, if for $n\geq 0$, $J_n$ is the $n$'th Bessel function of the first kind and for $k\geq 1$, $j_{n,k}$ is the $k$'th positive zero of $J_n$ (recall that $J_n(0)=0$ for $n>0$, so we do not count this zero), then for $n\in \Z$ and $k\geq 1$
\begin{equation}
e_{n,k}(x)=\frac{1}{\sqrt{\pi}L J_{|n|+1}(j_{|n|,k})}J_{|n|}(j_{|n|,k}\tfrac{r}{L})e^{in\theta}
\label{e:ef}
\end{equation}
are the eigenfunctions of $-\Delta$ on $\Lambda_L$ (with zero boundary conditions),
normalized so that they form an orthonormal basis of $L^2(\Lambda_L)$. Here we have written $x=re^{i\theta}$. 
In particular,
\begin{equation}
\int_{\Lambda_L} dx\, e_{n,k}(x)\overline{e_{m,l}(x)}=\delta_{n,m}\delta_{k,l}.
\label{e:on}
\end{equation}
To simplify notation, we set $j_{n,k} = j_{|n|,k}$ for $n<0$.
The eigenvalue associated to $e_{n,k}$ is then $\frac{j_{n,k}^2}{L^2}$:
\begin{equation}
-\Delta e_{n,k}=\frac{j_{n,k}^2}{L^2}e_{n,k}.
\label{e:ev}
\end{equation}
In terms of the eigenfunctions and eigenvalues, the Laplacian Green's function is
\begin{equation}
G_{\Lambda_L}(x,y)=\sum_{n\in \Z}\sum_{k=1}^\infty \frac{L^2}{j_{n,k}^2}e_{n,k}(x)\overline{e_{n,k}(y)},
\label{e:GF3}
\end{equation}
understood in the sense that for $g\in L^2(\Lambda_L)$, which we can write as $g=\sum_{n\in \Z}\sum_{k=1}^\infty g_{n,k}e_{n,k}$ (with convergence in $L^2(\Lambda_L)$ since $e_{n,k}$ form an orthonormal basis of $L^2(\Lambda_L)$), we have 
\begin{equation}
\int_{\Lambda_L} dy\, G_{\Lambda_L}(x,y)g(y)=\sum_{n\in \Z}\sum_{k=1}^\infty \frac{L^2}{j_{n,k}^2}g_{n,k}e_{n,k}(x),
\label{e:Gint}
\end{equation}
again with convergence in $L^2(\Lambda_L)$.

Since the Dirac Green's function is related to derivatives of the Laplacian Green's function, our construction of the Dirac Green's function also involves
another family of functions (which are also Laplacian eigenfunctions, but with different boundary conditions):
for $n,k\geq 1$, we define
\begin{equation}
f_{n,k}(x)=-2\frac{L}{j_{n-1,k}}\bar \partial e_{n-1,k}(x).
\label{e:ef2}
\end{equation}

The following lemma collects the properties of the $e_{n,k}$ and $f_{n,k}$ we need.
The stated estimates are not all optimal, but sufficient for our purposes.
We write
$\nabla^p g$ for the vector of all combinations of $p$ derivatives of $g$ and
$\|\nabla^p g\|_{L^\infty(K)}$ for the maximum of the $L^\infty(K)$ norm of all combinations of $p$ derivatives of $g$.

\begin{lemma}\label{lem:ebounds}
  For $n \in \Z$ and $k\geq 1$, the eigenvalues (up to the factor $L^2$) satisfy
  \begin{equation}\label{e:evbounds}
    j_{n,k}^2\geq n^2+\left(k-\tfrac{1}{4}\right)^2 \pi^2.
  \end{equation}
  The eigenfunctions satisfy (for some universal constant $C$)
  \begin{alignat}{3} \label{e:efbounds1}
    \|e_{n,k}\|_{L^\infty(\Lambda_L)}&\leq C \frac{j_{n,k}}{L}, &\qquad \|\nabla e_{n,k}\|_{L^\infty(\Lambda_L)}&\leq C \frac{j_{n,k}^2}{L^2},
    &\qquad &(n \in \Z, k \geq 1)
    ,
    \\
      \label{e:efbounds2}
    \|f_{n,k}\|_{L^\infty(\Lambda_L)} &\leq C \frac{j_{n-1,k}}{L}, &\qquad \|\nabla f_{n,k}\|_{L^\infty(\Lambda_L)}&\leq C \frac{j_{n-1,k}^3}{L^3}L
    ,&\qquad &(n,k \geq 1).
  \end{alignat}
  Moreover, for any $p,q\geq 0$, any compact $K \subset \mathrm{int}(\Lambda_L)$,
  and any $f \in C_c^\infty(\mathrm{int}(\Lambda_L))$,
  there are constants $C_{p,K,L}$ and $C_{p,q,f,L}$ such that 
  \begin{gather} \label{e:efboundsn}
    \|\nabla^p e_{n,k}\|_{L^\infty(K)}  + \|\nabla^p f_{n,k}\|_{L^\infty(K)} \leq C_{p,K,L} j_{n,k}^{1+p},
    \\
    \label{e:efdecay}
    \absa{\int dx\, f(x) \, \nabla^p e_{n,k}(x)}+
    \absa{\int dx\, f(x)\, \nabla^p f_{n,k}(x)}
    \leq C_{p,q,f,L} j_{n,k}^{-q}.
  \end{gather}
\end{lemma}

\begin{proof}
  The bounds on the $j_{n,k}$ follow, for example, from \cite[Theorem 3]{MR255909} and the main result of \cite{MR442316}.
  The bounds  \eqref{e:efbounds1} on the eigenfunctions $e_{n,k}$ and their derivatives $\nabla e_{n,k}$ follow, for example, from
  \cite[Theorem~1]{MR1924468} and \cite[Corollary~1.1]{MR3031783} (as well as scaling by $L$).
  The claim for $\|f_{n,k}\|_{L^\infty(\Lambda_L)}$ in \eqref{e:efbounds2} follows directly from the definition of $f_{n,k}$ in \eqref{e:ef2}
  combined with the gradient estimate from \eqref{e:efbounds1}. 

  For the bound on the gradient of $f_{n,k}$ in \eqref{e:efbounds2}, we note that since
  \begin{equation} 
    2\partial f_{n,k}=-\frac{L}{j_{n-1,k}}\Delta e_{n-1,k}=\frac{j_{n-1,k}}{L}e_{n-1,k},
  \end{equation}
  by \eqref{e:efbounds1} we have $\|\partial f_{n,k}\|_{L^\infty(\Lambda_L)}\leq C \frac{j_{n-1,k}^2}{L^2}$, so
  that it suffices to control $\|\bar \partial f_{n,k}\|_{L^\infty(\Lambda_L)}$.	
  For this purpose, using the eigenfunction property \eqref{e:ev}, we see from
   \eqref{e:Gint} that 
  \begin{align}
	f_{n,k}(x)&=-2\frac{j_{n-1,k}}{L}\bar \partial \int_{\Lambda_L} dy\, G_{\Lambda_L}(x,y)e_{n-1,k}(y)\\
	&=\frac{j_{n-1,k}}{L}\int_{\Lambda_L} dy\, \left(\frac{1}{2\pi}\frac{1}{\bar x-\bar y}-\frac{1}{2\pi}\frac{y}{L}\frac{1}{\frac{\bar x y}{L}-L}\right)e_{n-1,k}(y)\nnb
	&=-\frac{j_{n-1,k}}{L}\int_{\Lambda_L}dy\, \left(\frac{1}{\pi}\bar \partial_y\log |x-y|+\frac{1}{\pi}\frac{y^2}{L^2}\partial_y\log |x-L^2 /\overline{y}|\right)e_{n-1,k}(y)\nnb
	&=\frac{j_{n-1,k}}{L}\int_{\Lambda_L}dy\,\frac{1}{\pi}\log |x-y|\bar\partial e_{n-1,k}(y)
   \nnb &\qquad\qquad 
   +\frac{1}{\pi L^2}\frac{j_{n-1,k}}{L}\int_{\Lambda_L}dy\,\log |x-\tfrac{L^2}{\overline{y}}|\partial (y^2 e_{n-1,k}(y)),\notag
  \end{align}
	where in the last step we integrated by parts and made use of the fact that $e_{n-1,k}$ vanishes on the boundary. Thus we find for some universal constant $C>0$
	\begin{align}
	|\bar \partial f_{n,k}(x)|&\leq C \frac{j_{n-1,k}}{L}\|\nabla e_{n-1,k}\|_{L^\infty(\Lambda_L)}\int_{\Lambda_L}dy\,\frac{1}{|x-y|}\\
	&\quad +C\frac{j_{n-1,k}}{L^3}(L\|e_{n-1,k}\|_{L^\infty(\Lambda_L)}+L^2\|\nabla e_{n-1,k}\|_{L^\infty(\Lambda_L)})\int_{\Lambda_L} dy\, \frac{1}{|x-\frac{L^2}{\overline{y}}|}.\notag
	\end{align}
	Since $x\in \Lambda_L$, for the first integral we readily get the bound
	\begin{equation}
	\int_{\Lambda_L}dy\,\frac{1}{|x-y|}\leq \int_{|x-y|\leq 2L} dy\,\frac{1}{|x-y|}\leq CL
	\end{equation}
	for a universal constant $C$. For the second integral, one finds on the other hand by rotational invariance that 
	\begin{equation}
	\int_{\Lambda_L}dy\,\frac{1}{|x-\frac{L^2}{\overline{y}}|}\leq \int_{\Lambda_L}dy\,\frac{1}{|L-\frac{L^2}{\overline{y}}|}=L\int_{|u|\leq 1}\frac{du}{|1-\frac{1}{\overline{u}}|}.
	\end{equation}
	The last integral here is simply some finite constant. Putting everything together and using \eqref{e:efbounds1}
        (and~\eqref{e:evbounds} to deduce that $j_{n-1,k}^2\leq j_{n-1,k}^3$), we see that for some universal constant $C>0$, 
	\begin{equation}
	\|\bar \partial f_{n,k}\|_{L^\infty(\Lambda_L)}\leq C L \frac{j_{n-1,k}^3}{L^3},
	\end{equation}
	which leads to the claim, as we discussed before.

        The bounds  \eqref{e:efboundsn}  on the higher derivatives in the interior are a standard consequence
        of elliptic regularity theory for the Laplace operator.
        For example, one may apply \cite[(4.19)]{MR1814364} iteratively.
  
        To see the decay of \eqref{e:efdecay}, by integrating by parts, it suffices to check this for $p=0$ and for $e_{n,k}$ only.
        In this case, that $e_{n,k}$ is a Laplace eigenfunction and integration by parts show that, for any $q$,
        \begin{equation}
          \absa{\int dx\, f(x) \, e_{n,k}(x)}
          = \left(\frac{L^2}{j_{n,k}^2}\right)^q \absa{\int dx\, (-\Delta)^q f(x) e_{n,k}(x)}
          \leq C \left(\frac{L^2}{j_{n,k}^2}\right)^{q-1/2} \|\Delta^q f\|_{L^\infty},
        \end{equation}
        which gives the claimed bound. 
\end{proof}

\subsection{The building blocks of the  Dirac Green's function -- I}

We next introduce the key building blocks of our construction of the Dirac Green's function with a finite volume mass term.
We begin with the following function which is the projection of the Laplacian Green's function to non-positive
Fourier modes related to the $x$-variable. More precisely, for $x,y\in \Lambda_{L}$, let 
\begin{equation}
E_1(x,y)=\sum_{n=0}^\infty \sum_{k=1}^\infty \frac{L^2}{j_{n,k}^2}e_{-n,k}(x)\overline{e_{-n,k}(y)},
\label{e:E1a}
\end{equation}
where convergence is understood in $L^2(\Lambda_L\times \Lambda_L)$.
We then  define inductively, for $j\geq 1$, the functions
\begin{align}
  \label{e:Ej}
  E_{j+1}(x,y)&=\int_{\Lambda_L}du\,G_{\Lambda_L}(x,u)E_j(u,y),
  \\
  \label{e:Fj}
F_j(x,y)&=4\bar \partial_x \partial_y \overline{E_{j+1}(x,y)}.
\end{align}
That the derivatives indeed exist is a consequence of the explicit formulas we will derive below.
These show that, for $y\in\Lambda_L$, the funtions $E_1(x,y)$ and $F_1(x,y)$ are defined pointwise for $x\neq y$,
and that 
$E_j(x,y)$ and $F_j(x,y)$ with $j>1$ are defined pointwise for all $x,y\in\Lambda_L$.
We also note that, by \eqref{e:GF3},
\begin{align}
  E_j(x,y)&=\sum_{n=0}^\infty \sum_{k=1}^\infty \left(\frac{L^2}{j_{n,k}^2}\right)^je_{-n,k}(x)\overline{e_{-n,k}(y)},
            \label{e:Ej2}
\end{align}
as an element of $L^2(\Lambda_L\times \Lambda_L)$. Based on the definition of $F_j$ and $f_{n,k}$, one then expects that also
\begin{align}
  \label{e:F2}
  F_j(x,y)&=\sum_{n=1}^\infty \sum_{k=1}^\infty \left(\frac{L^2}{j_{n-1,k}^2}\right)^jf_{n,k}(x)\overline{f_{n,k}(y)}.
\end{align}
This is indeed true, and we prove it in Lemma \ref{le:F2} (for $j\geq 3$).

We begin calculating $E_1$ and $F_1$ using \eqref{e:GF1} in the next lemma.

\begin{lemma} 
	\label{le:E1}
	For (almost every) $x,y\in \mathrm{int}(\Lambda_L)$ with $x\neq y$,
	\begin{align}
          \label{e:E1}
          E_1(x,y)&=\begin{cases}
            -\frac{1}{2\pi}\log |y|-\frac{1}{4\pi}\log \big(1-\frac{\bar x}{\bar y}\big)+\frac{1}{2\pi}\log L+\frac{1}{4\pi}\log \big(1-\frac{\bar x y}{L^2}\big), & |x|<|y|\\
            -\frac{1}{2\pi}\log |x|-\frac{1}{4\pi}\log \big(1-\frac{y}{x}\big)+\frac{1}{2\pi}\log L+\frac{1}{4\pi}\log \big(1-\frac{\bar x y}{L^2}\big), & |x|>|y|,
	\end{cases}
          \\
          \label{e:F1}
          F_1(x,y)&=\begin{cases}
          -\frac{1}{4\pi}\log \big(1-\frac{x}{y}\big), & |x|<|y|\\
          	-\frac{1}{4\pi}\log \big(1-\frac{\bar y}{\bar x}\big), & |x|>|y|,
	\end{cases}
      \end{align}
      where the branches of the logarithms are understood to be given by the series expansion of $\log(1+z)$ for $|z|<1$,
      and, for $|x|=|y|$ with $x\neq y$, $E_1$ and $F_1$ are defined by continuity.
      In particular,
      \begin{equation} \label{e:E1plusF1}
        E_1(x,y)+F_1(x,y) = -\frac{1}{2\pi} \log |x-y| + \frac{1}{2\pi} \log (L^2-\bar x y).
      \end{equation}
\end{lemma}

\begin{proof}
	Let us write $\widetilde E_1$ for the right hand side of the claim. Using \eqref{e:GF1}, we see that 
	\begin{equation}
	G_{\Lambda_L}(x,y)-\widetilde E_1(x,y)=\begin{cases}
	-\frac{1}{4\pi}\log\left(1-\frac{x}{y}\right)+\frac{1}{4\pi}\log\left(1-\frac{x\bar y}{L^2}\right), & |x|<|y|\\
	-\frac{1}{4\pi}\log\left(1-\frac{\bar y}{\bar x}\right)+\frac{1}{4\pi}\log\left(1-\frac{x\bar y}{L^2}\right), & |x|>|y|
	\end{cases}. 
	\end{equation}
	Going into polar coordinates, one can readily check from this
        (since there are only strictly positive Fourier modes when one expands the logarithms) that for $n\geq 0$ and $k\geq 1$
	\begin{equation}
	\int_{\Lambda_L}dx\,(G_{\Lambda_L}(x,y)-\widetilde E_1(x,y))\overline{e_{-n,k}(x)}=0.
	\end{equation}
	Similarly one finds in polar coordinates that for $n>0$ and $k\geq 1$ (again since there are only non-positive Fourier modes in the expansion of the logarithms)
	\begin{equation}
	\int_{\Lambda_L}dx\,\widetilde E_1(x,y)\overline{e_{n,k}(x)}=0.
	\end{equation}
	From these two facts, one finds immediately that $\widetilde E_1=E_1$.
      
        For the claim for $F_1$,
        note that from the identity for $E_1$,
        for $y,u\in \mathrm{int}(\Lambda_L)$,
	\begin{equation}
          \partial_y\overline{E_1(u,y)}=
          - \mathbf 1\{|u|<|y|\} \frac{1}{4\pi}\frac{1}{y-u}                                          
          =-\mathbf 1\{|u|<|y|\}\frac{1}{4\pi}\frac{1}{y}\sum_{j=0}^\infty \left(\frac{u}{y}\right)^j.
	\end{equation}
	Moreover, we have for $x,u\in \mathrm{int}(\Lambda_L)$,
	\begin{equation}
	\bar \partial_x G_{\Lambda_L}(x,u)=-\frac{1}{4\pi}\frac{1}{\bar x-\bar u}-\frac{u}{4\pi L^2}\frac{1}{1-\frac{\bar xu}{L^2}}.
	\end{equation}
	In polar coordinates, one readily checks that the second term on the right-hand side is orthogonal to $\partial_y\overline{E_1(u,y)}$ (when integrated over $u$). One then finds
	\begin{align}
	F_1(x,y)&=\int_{|u|<|y|}du\,\frac{1}{4\pi^2}\frac{1}{\bar x-\bar u}\frac{1}{y-u}\nnb
	&=\begin{cases}
	\frac{1}{\bar x y}\frac{1}{4\pi^2}\sum_{j,k=0}^\infty \frac{1}{\bar x^j y^k}\int_{|u|<|y|}du\,\bar u^j u^k, & |x|>|y|\\
	\frac{1}{\bar x y}\frac{1}{4\pi^2}\sum_{j,k=0}^\infty \frac{1}{\bar x^j y^k}\int_{|u|<|x|}du\,\bar u^j u^k
        \\\qquad 
        -\frac{1}{y}\frac{1}{4\pi^2}\sum_{j,k=0}^\infty \frac{\bar x^j}{y^k}\int_{|x|<|u|<|y|}du\,\bar u^{-j-1}u^k, & |x|<|y|
	\end{cases}\nnb
	&=\begin{cases}
	\frac{1}{4\pi}\sum_{j=0}^\infty \frac{1}{\bar x^{j+1}y^{j+1}}\frac{1}{j+1}|y|^{2j+2}, & |y|<|x|\\
	\frac{1}{4\pi}\sum_{j=0}^\infty \frac{1}{\bar x^{j+1}y^{j+1}}\frac{1}{j+1}|x|^{2j+2}, & |x|<|y|
	\end{cases}\nnb
	&=\begin{cases}
	-\frac{1}{4\pi}\log \big(1-\frac{\bar y}{\bar x}\big), & |y|<|x|\\
	-\frac{1}{4\pi}\log \big(1-\frac{x}{y}\big), & |y|>|x|
	\end{cases}.
	\end{align}
        
	The claim that the values of $E_1(x,y)$ and $F_1(x,y)$ for $|x|=|y|$ with $x\neq y$ are given by continuity follows
        by noting that \eqref{e:Ej} and \eqref{e:Fj} are continuous away from the diagonal.
        Finally, \eqref{e:E1plusF1} is a direct computation. This concludes the proof.
\end{proof}

\begin{lemma} \label{le:E12bound}
  There exists a polynomial $P=P(L)$ such that for $L\geq 1$ and all $x,y \in \mathrm{int}(\Lambda_L)$ with $x \neq y$,
  \begin{gather}
    \label{e:E1bound1}
    |E_1(x,y)|+|F_1(x,y)| \leq P(L)(1+|\log|x-y||),
    \\
    \label{e:E1bound2}
    |\partial_xE_1(x,y)|+|\partial_xF_1(x,y)| \leq \frac{P(L)}{|x-y|},
    \\
    \label{e:E2bound}
    |E_2(x,y)|+|F_2(x,y)|+ |\partial_xE_2(x,y)|+|\partial_xF_2(x,y)| \leq P(L).
  \end{gather}
\end{lemma}

\begin{proof}
  We begin with bounding $E_1$.
  By symmetry (up to complex conjugation), we can assume that $|x|<|y|$.
  We start from the elementary inequality
  \begin{equation}
    \left|-\frac{1}{2\pi}\log |y|-\frac{1}{4\pi}\log\left(1-\frac{\bar x}{\bar y}\right)\right|\leq
    C+C|\log |y|| + C |\log|x-y||.
  \end{equation}
  Since $|y|>|x|$, we have $|y|\geq \frac{1}{2}|y-x|$, so we conclude that for some (possibly different) $C>0$ that
  \begin{equation}
    \left|-\frac{1}{2\pi}\log |y|-\frac{1}{4\pi}\log\left(1-\frac{\bar x}{\bar y}\right)\right|\leq
    C+C\log L + C |\log|x-y||.
  \end{equation}
  Similarly,
  $|x|<|y|$ implies $|L^2-x\overline{y}|=|y||\frac{L^2}{|y|^2}y-x|\geq |y||y-x|\geq \frac{1}{2}|y-x|^2$, which leads to 
  \begin{equation}
    \left|\frac{1}{2\pi}\log L+\frac{1}{4\pi}\log\left(1-\frac{\bar xy}{L^2}\right)\right|\leq C+C |\log|L^2-\bar x y||
    \leq C+C\log L + C|\log|x-y||.
  \end{equation}
  Similar reasoning readily proves the analogous bound for $F_1$.
  For $\partial_x E_1$ and $\partial_x F_2$, we obtain the required bound by noting that these derivatives are explicitly given by $1/(4\pi(x-y))$ or $0$
  by differentiating \eqref{e:E1} and \eqref{e:F1}.
       
  In order to bound $E_2$, we start from the definition
  \begin{equation}
    E_2(x,y) = \int_{\Lambda_L}du\, G_{\Lambda_L}(x,u) E_1(u,y).
  \end{equation}
  Using the bound for $E_1$ and that $|G_{\Lambda_L}(x,u)| \leq C(1+|\log |x-u||)$,
  we readily see that $E_2$ is uniformly bounded by a polynomial in $L$.
  For the derivative,
  using that $|\partial_xG_{\Lambda_L}(x,u)| \leq P(L)/|x-u|$,
  it similarly follows from the above bound for $E_1$ that
  \begin{equation}
    |\partial_x E_2(x,y)| \leq P(L) \int_{\Lambda_L} du\,\frac{1}{|x-u|}  E_1(u,y)  \leq P(L), 
  \end{equation}
  where the two polynomials $P(L)$ can be different.
  Again the bounds for $F_2$ are similar.
\end{proof}
  
We next show that the $E_j$ and $F_j$ are bounded for $j \geq 3$.

\begin{lemma}	\label{le:F2}
  For $j\geq 3$, $E_j$ and $F_j$ are given by the series 
  \eqref{e:Ej2} and \eqref{e:F2}, which converge uniformly in $\Lambda_L\times \Lambda_L$.
\end{lemma}

\begin{proof}
  For $E_j$, we already saw that it agrees with the series in an $L^2$-sense.
    Let us now argue that the series converge uniformly and the $E_j$ series can be differentiated termwise (which implies that $F_j$ will be given by the corresponding series.) This follows immediately from applying the bounds \eqref{e:efbounds1}, \eqref{e:efbounds2}, and \eqref{e:evbounds} in the series representations \eqref{e:Ej2} and \eqref{e:F2}.
\end{proof}

Next we note that the $E_j$ and $F_j$ are smooth when tested against a smooth test function that is compactly supported in $\Lambda_L$.
        
\begin{lemma}
  For any $j \geq 1$ and $f\in C_c^\infty(\mathrm{int}(\Lambda_L))$, 
  \begin{equation} \label{e:fEjFfsmooth}
    y\mapsto \int_{\Lambda_L} dx\, f(x)\, E_j(x,y)\in C^\infty(\mathrm{int}(\Lambda_L)),\qquad
    y\mapsto \int_{\Lambda_L} dx\, f(x)\, F_j(x,y)\in C^\infty(\mathrm{int}(\Lambda_L)).
  \end{equation}
\end{lemma}

\begin{proof}
  By \eqref{e:efdecay} and \eqref{e:efboundsn},
  it follows that for any $p>0$, $f\in C_c^\infty(\mathrm{int}(\Lambda_L))$, and $K \subset \mathrm{int}(\Lambda_L)$,
  there are constants $C_{p,f,K,L}$ such that
  \begin{equation} \label{e:feebd}
    \sup_{y \in K} \absa{\nabla_y^p \int dx\, f(x) \, e_{-n,k}(x)\overline{e_{-n,k}(y)} } \leq C_{p,f,K,L} j_{n,k}^{-p},
  \end{equation}
  and analogously for such expressions with $e_{-n,k}$ replaced by $f_{n,k}$. 
  From this, the claim follows again by differentiating the series term by term.
\end{proof}

As a final property of the functions $E_j$ and $F_j$, we record the following recursion properties.

\begin{lemma}	\label{le:ints}
	For $j,k\geq 1$ and $x,y\in \mathrm{int}(\Lambda_L)$, we have
	\begin{align}
	E_{j+k}(x,y)&=\int_{\Lambda_L}du\, E_j(x,u)E_k(u,y),
        \\
	F_{j+k}(x,y)&=\int_{\Lambda_L}du\, F_j(x,u)F_k(u,y),
        \\
	\int_{\Lambda_L} du\, E_j(x,u)F_k(u,y)&=\int_{\Lambda_L} du\, F_k(x,u)E_j(u,y)=0.
	\end{align}
\end{lemma}
\begin{proof}
  The claim for $E_{j+k}$ follows immediately from continuity and the representation \eqref{e:Ej2}
  which implies that the two functions are the same as elements of $L^2(\Lambda_L\times \Lambda_L)$
  (and thus in particular for almost every $x,y\in \mathrm{int}(\Lambda_L)$).
	
  For $F_{j+k}$, we integrate by parts (note that $\bar \partial_x \overline{E_{j+1}(x,u)}$ vanishes for $u\in \partial \Lambda_L$ by \eqref{e:Ej} and the fact that $E_j(v,u)$ vanishes for $u\in \partial \Lambda_L$ -- which follows e.g. from the explicit representation of $E_1$ from Lemma \ref{le:E1} and \eqref{e:Ej}) and find
  \begin{align}
    \int_{\Lambda_L} du\, F_{j}(x,u)F_k(u,y)
    &=4\bar \partial_x \partial_y\int_{\Lambda_L} du\, \overline{E_{j+1}(x,u)}(-\Delta_u)\overline{E_{k+1}(u,y)}
      \nnb
    &=4\bar\partial_x \partial_y \overline{E_{j+k+1}(x,y)}=F_{j+k}(x,y).
  \end{align}
  where we used the fact that $-\Delta_x E_{j+1}(x,y)=E_j(x,y)$ by \eqref{e:Ej} and the first claim of this lemma.
  
  For the final claim, the fact that first integral vanishes follows immediately from the remark that considering $E_j(x,u)$ and $F_k(u,y)$ in polar coordinates for $u$, $E_j$ has only non-positive Fourier modes while $F_k$ has only strictly positive Fourier modes, so the claim follows from Fourier orthogonality. The vanishing of the second integral follows by a similar argument.
\end{proof}

\subsection{The building blocks of the  Dirac Green's function -- II}

The functions $E_j$ and $F_j$ constructed above will turn out to be responsible for the singular behavior in our Dirac Green's function.
To understand the behavior in $\mu$, we introduce the following functions: for $m\geq 1$, let 
\begin{align}
  R_{m;11}(x,y;\mu,L)
  &=(-1)^{m} \sum_{n=0}^\infty \sum_{k=1}^\infty \frac{L^{2(m+1)}\nmu^{2m+1}}{(1+\frac{\mu^2 L^2}{j_{n,k}^2})j_{n,k}^{2(m+1)}}e_{-n,k}(x)\overline{e_{-n,k}(y)}\notag \\
  &\quad +(-1)^{m}\sum_{n=1}^\infty \sum_{k=1}^\infty \frac{L^{2(m+1)}\nmu^{2m+1}}{(1+\frac{\mu^2 L^2}{j_{n-1,k}^2})j_{n-1,k}^{2(m+1)}}f_{n,k}(x)\overline{f_{n,k}(y)}
\label{e:R11}
\end{align}
and 
\begin{align}
  \notag
  R_{m;21}(x,y;\mu,L)&=(-1)^{m+1} \sum_{n=0}^\infty \sum_{k=1}^\infty \frac{L^{2(m+1)}\nmu^{2m}}{(1+\frac{\mu^2 L^2}{j_{n,k}^2})j_{n,k}^{2(m+1)}}(2\partial e_{-n,k})(x)\overline{e_{-n,k}(y)}\\
&\qquad +(-1)^{m+1}\sum_{n=1}^\infty \sum_{k=1}^\infty \frac{L^{2(m+1)}\nmu^{2m}}{(1+\frac{\mu^2 L^2}{j_{n-1,k}^2})j_{n-1,k}^{2(m+1)}}(2\partial f_{n,k})(x)\overline{f_{n,k}(y)}.
\label{e:R21}
\end{align}
A priori, it may not be clear in what sense these series converge, but we now describe the basic facts we will need about these functions -- including regularity.
\begin{lemma} 	\label{le:Rbounds}
	For any $m\geq 3$ and $y\in \mathrm{int}(\Lambda_L)$, the functions $x\mapsto R_{m;11}(x,y)$ and $x\mapsto R_{m;21}(x,y)$ are continuously differentiable in $\mathrm{int}(\Lambda_L)$ and have the following properties:
	\begin{enumerate}
        \item $\frac{1}{\nmu}2\partial_x R_{m;11}(x,y)={-}R_{m;21}(x,y)$ for all $x,y\in \mathrm{int}(\Lambda_L)$.
        \item $-2\bar \partial_x R_{m;21}(x,y) {-}\nmu R_{m;11}(x,y)=(-1)^{m+1} \mu^{2m}(E_m(x,y)+F_m(x,y))$ for all $x,y\in \mathrm{int}(\Lambda_L)$, with $E_m$ as in \eqref{e:Ej} and $F_m$ as in \eqref{e:Fj}.
        \item There exists a polynomial $P_m=P_m(L,|\mu|)$, which does not depend on $x,y$, such that 
          \begin{equation}
            \sup_{x,y\in \mathrm{int}(\Lambda_L)}|R_{m}(x,y)|\leq P_m(L,|\mu|).
          \end{equation}
        \item
          For any $f\in C_c^\infty(\mathrm{int}(\Lambda_L))$, 
          $y\mapsto \int_{\Lambda_L} dx\, f(x)\, R_{m}(x,y)\in C^\infty(\mathrm{int}(\Lambda_L))$.
	\end{enumerate}
      \end{lemma}
      
\begin{proof}
  For continuous differentiability, let us first consider $R_{m;11}$.
  Using \eqref{e:efbounds1} and  \eqref{e:efbounds2}, we find that,
  for some $C(L,\mu)>0$,
	\begin{align}
	&\sum_{n=0}^\infty \sum_{k=1}^\infty \frac{L^{2(m+1)}\mu^{2m}}{(1+\frac{\mu^2 L^2}{j_{n,k}^2})j_{n,k}^{2(m+1)}}\|\nabla e_{-n,k}\|_{L^\infty(\Lambda_L)}\|e_{-n,k}\|_{L^\infty(\Lambda_L)}\\
          &\qquad +\sum_{n=1}^\infty \sum_{k=1}^\infty \frac{L^{2(m+1)}\mu^{2m}}{(1+\frac{\mu^2 L^2}{j_{n-1,k}^2})j_{n-1,k}^{2(m+1)}}\|\nabla f_{n,k}\|_{L^\infty(\Lambda_L)} \|f_{n,k}\|_{L^\infty(\Lambda_L)}
            \notag \\ & 
        \leq C(L,\mu)\sum_{n=0}^\infty \sum_{k=1}^\infty j_{n,k}^{-2(m+1)+4}.\notag
	\end{align}
	By \eqref{e:evbounds}, this series is convergent for $m\geq 3$, so standard results concerning uniform convergent series (involving continuity and differentiability) yields continuous differentiability.
	For $R_{21}$, we point out that, by \eqref{e:ef2}
        and \eqref{e:ev}, 
	\begin{equation}
          2\partial e_{-n,k}=-\frac{j_{n,k}}{L}\overline{f_{n+1,k}},
          \qquad
	2\partial  f_{n,k}=-\frac{L}{j_{n-1,k}}\Delta e_{n-1,k}=\frac{j_{n-1,k}}{L}e_{n-1,k}.
	\end{equation}
	Thus the same argument making use of \eqref{e:efbounds1}, \eqref{e:efbounds2}, and \eqref{e:evbounds} implies the continuous differentiability. (Now we end up with the series $\sum_{n=0}^\infty \sum_{k=1}^\infty j_{n,k}^{-2(m+1)+5}$ which is still convergent for $m\geq 3$.)
	
	We now turn to statement (i). This follows immediately from our preceding argument for continuous differentiability as it allows us to differentiate term by term.
	
	For (ii), we note that again by our continuous differentiability argument, we can differentiate term by term. We find (using \eqref{e:ef2}) that 
	\begin{align}
	-2\bar \partial_x \notag R_{m;21}(x,y)&=(-1)^{m+1} \sum_{n=0}^\infty \sum_{k=1}^\infty \frac{L^{2(m+1)}\mu^{2m}}{(1+\frac{\mu^2 L^2}{j_{n,k}^2})j_{n,k}^{2(m+1)}}(-\Delta e_{-n,k})(x)\overline{e_{-n,k}(y)}\nnb
	&\qquad +(-1)^{m+1}\sum_{n=1}^\infty \sum_{k=1}^\infty \frac{L^{2(m+1)}\mu^{2m}}{(1+\frac{\mu^2 L^2}{j_{n-1,k}^2})j_{n-1,k}^{2(m+1)}}(-\Delta f_{n,k})(x)\overline{f_{n,k}(y)}\nnb
	&=(-1)^{m+1} \sum_{n=0}^\infty \sum_{k=1}^\infty \frac{L^{2m}\mu^{2m}}{(1+\frac{\mu^2 L^2}{j_{n,k}^2})j_{n,k}^{2m}}e_{-n,k}(x)\overline{e_{-n,k}(y)}\nnb
	&\qquad +(-1)^{m+1}\sum_{n=1}^\infty \sum_{k=1}^\infty \frac{L^{2m}\mu^{2m}}{(1+\frac{\mu^2 L^2}{j_{n-1,k}^2})j_{n-1,k}^{2m}}f_{n,k}(x)\overline{f_{n,k}(y)}
	\end{align}
	and 
	\begin{align}
	&-2\bar \partial_x \notag R_{m;21}(x,y){-}\nmu R_{m;11}(x,y)\nnb
	&\quad =(-1)^{m+1} \sum_{n=0}^\infty \sum_{k=1}^\infty \frac{L^{2(m+1)}\mu^{2m}}{(1+\frac{\mu^2 L^2}{j_{n,k}^2})j_{n,k}^{2(m+1)}}\left(\mu^2+\frac{j_{n,k}^2}{L^2}\right)e_{-n,k}(x)\overline{e_{-n,k}(y)}\nnb
	&\qquad +(-1)^{m+1}\sum_{n=1}^\infty \sum_{k=1}^\infty \frac{L^{2(m+1)}\mu^{2m}}{(1+\frac{\mu^2 L^2}{j_{n-1,k}^2})j_{n-1,k}^{2(m+1)}}\left(\mu^2+\frac{j_{n-1,k}^2}{L^2}\right)f_{n,k}(x)\overline{f_{n,k}(y)}\nnb
	&\quad =(-1)^{m+1} \mu^{2m}\sum_{n=0}^\infty \sum_{k=1}^\infty \left(\frac{L^2}{j_{n,k}^2}\right)^m e_{-n,k}(x)\overline{e_{-n,k}(y)}\nnb
	&\qquad +(-1)^{m+1} \mu^{2m}\sum_{n=1}^\infty \sum_{k=1}^\infty \left(\frac{L^2}{j_{n-1,k}^2}\right)^m f_{n,k}(x)\overline{f_{n,k}(y)}.
	\end{align}
	Recalling \eqref{e:Ej2} and Lemma \ref{le:F2}, this concludes the proof of claim (ii).
	
	For (iii), we will prove the claim for $R_{m;21}$ -- the proof for $R_{m;11}$ being similar.
        Using \eqref{e:efbounds1} and \eqref{e:efbounds2}, we have for some constant $C>0$ (independent of $m,\mu,L$) that
\begin{equation}
		\|R_{m;21}\|_{L^\infty(\Lambda_L\times \Lambda_L)}\leq C L^{2m-1}\mu^{2m}\sum_{n=0}^\infty \sum_{k=1}^\infty j_{n,k}^{-2m+2}	
\end{equation}
	and the claim follows from \eqref{e:evbounds}.

          Finally, to prove (iv), let $f\in C_c^\infty(\Lambda_L)$.
          Then \eqref{e:feebd} holds and an analogous bound holds with $e_{-n,k}$ replaced by $f_{n,k}$ or derivatives of these.
          Substituting this into the definition of $R_{m;ij}$, we see that all $y$-derivatives of
          series that defines $\int \, dx f(x) R_{m;ij}(x,y)$ converge, and thus that
          $\int \, dx f(x) R_{m;ij}(x,y)$ is smooth in $y$.

        This concludes the proof.
\end{proof}

\subsection{The proof of Theorem~\ref{th:gestimate}}

In this section, we first define our Dirac Green's function, then prove it satisfies the bounds stated in items (i) and (ii) of Theorem~\ref{th:gestimate}, then prove item (iii) of the theorem, namely analyticity in a neighborhood of the real axis, item (iv) of the theorem -- a kind of resolvent identity -- and finally prove convergence as $L\to\infty$. We split this section into parts where these tasks are carried out. 

\subsubsection{Constructing the Green's function}

We begin by defining the function that will be our $S_{\mu\mathbf 1_{\Lambda_L}}$ in Theorem~\ref{th:gestimate}, and then prove that it satisfies \eqref{e:Green}.
For this definition, recall first the key building blocks $E_j$, $F_j$, $R_{m;11}$ and $R_{m;12}$ from \eqref{e:Ej}, \eqref{e:Fj}, \eqref{e:R11}, and \eqref{e:R21}.

\begin{definition} 	\label{d:Sy}
	For $\mu\in \R$, $L\geq 1$, and $y\in \mathrm{int}(\Lambda_L)$ and $x\in \Lambda_L$, let 
	\begin{equation}\begin{split}
          S_{\mu\mathbf 1_{\Lambda_L};11}(x,y)
          &=\sum_{l=0}^{2}(-1)^{l} \nmu^{2l+1}(E_{l+1}(x,y)+F_{l+1}(x,y))+R_{3;11}(x,y),
            \\
            S_{\mu\mathbf 1_{\Lambda_L};21}(x,y)
          &=2\partial_x \sum_{l=0}^{2}(-1)^{l+1} \mu^{2l}(E_{l+1}(x,y)+F_{l+1}(x,y))+R_{3;21}(x,y),
            \\
          S_{\mu\mathbf 1_{\Lambda_L};12}(x,y)
          &=\overline{S_{\mu\mathbf 1_{\Lambda_L};21}(x,y)},
          \\
          S_{\mu \mathbf 1_{\Lambda_L};22}(x,y)&=\overline{S_{\mu \mathbf 1_{\Lambda_L};11}(x,y)},
        \end{split}\end{equation}
        and for $y\in \mathrm{int}(\Lambda_L)$ but $x=re^{i\theta}\in \Lambda_L^{\mathsf c}$, define 
	\begin{equation}
	S_{\mu \mathbf 1_{\Lambda_L}}(x,y)=\frac{r^2-L^2}{2\pi}\int_0^{2\pi} d\phi\,\frac{1}{r^2+L^2-2rL \cos(\theta-\phi)}S_{\mu \mathbf 1_{\Lambda_L}}(Le^{i\phi},y).
	\end{equation}
\end{definition}

In particular, note that for $x \in \Lambda_L^{\mathsf c}$, 
$S_{\mu \mathbf 1_{\Lambda_L}}(x,y)$ is
the harmonic extension of $S_{\mu \mathbf 1_\Lambda}(\cdot, y)|_{\partial \Lambda_L}$.
We now show that this is indeed a Green's function for the problem we are considering. 

\begin{proposition}
	$S_{\mu \mathbf 1_{\Lambda_L}}$ defined in Definition \ref{d:Sy} satisfies \eqref{e:Green}.
	\label{pr:Greens}
\end{proposition}
\begin{proof}
  Our goal is to show that $S_{\mu\mathbf 1_{\Lambda_L}}(\cdot,y)$ as defined in Definition \ref{d:Sy}
  vanishes at infinity and that for each $f\in C_c^\infty(\R^2)$ and $y\in \mathrm{int}(\Lambda_L)$,
	\begin{multline}	\label{e:goal}
          \nmu \int_{\Lambda_L} dx\, S_{\mu \mathbf 1_{\Lambda_L}}(x,y)f(x)-2\int_{\R^2} dx\,\begin{pmatrix}
	\bar \partial f(x)S_{\mu \mathbf 1_{\Lambda_L};21}(x,y) & \bar \partial f(x)S_{\mu \mathbf 1_{\Lambda_L};22}(x,y)\\
	\partial f(x) S_{\mu \mathbf 1_{\Lambda_L};11}(x,y) & \partial f(x) S_{\mu \mathbf 1_{\Lambda_L};12}(x,y)
      \end{pmatrix}
      \\ 
      =\begin{pmatrix}
	f(y) & 0\\
	0 & f(y)
	\end{pmatrix}.
	\end{multline}
        Writing out $DS_{\mu \mathbf 1_L}$ from \eqref{e:Green} explicitly,
        we have for $x,y\in \mathrm{int}(\Lambda_L)$,
	\begin{multline}
          DS_{\mu \mathbf 1_L}(x,y)
          \\ 
          =\begin{pmatrix}
	2\bar \partial S_{\mu\mathbf 1_{\Lambda_L};21}(x,y) +\nmu S_{\mu\mathbf 1_{\Lambda_L};11}(x,y) & 2\bar \partial \overline{S_{\mu\mathbf 1_{\Lambda_L};11}(x,y)}+\nmu \overline{S_{\mu\mathbf 1_{\Lambda_L};21}(x,y)}\\
	2\partial S_{\mu\mathbf 1_{\Lambda_L};11}(x,y)+\nmu S_{\mu\mathbf 1_{\Lambda_L};21}(x,y) & 2\partial  \overline{S_{\mu\mathbf 1_{\Lambda_L};21}(x,y)}+\nmu \overline{S_{\mu\mathbf 1_{\Lambda_L};11}(x,y)}
	\end{pmatrix}.
	\end{multline}
	We can thus focus on the first column. Let us first consider the $21$-entry.
        Using Definition~\ref{d:Sy} and Lemma~\ref{le:Rbounds}, we find immediately that 
	\begin{equation}
	2\partial_x S_{\mu\mathbf 1_{\Lambda_L};11}(x,y)+\nmu S_{\mu\mathbf 1_{\Lambda_L};21}(x,y)=0 \qquad \text{for} \qquad x\in \mathrm{int}(\Lambda_L).
	\label{e:odin}
	\end{equation}
	For the $11$-entry, we have similarly using Definition~\ref{d:Sy} and Lemma~\ref{le:Rbounds},
	\begin{align}
          2\bar\partial_x S_{\mu\mathbf 1_{\Lambda_L};21}(x,y)+\nmu S_{\mu\mathbf 1_{\Lambda_L};11}(x,y)
          &=\sum_{l=0}^2(-1)^l \mu^{2l}(\mu^2-\Delta_x)(E_{l+1}(x,y)+F_{l+1}(x,y))\nnb
	&\qquad +2\bar\partial_x R_{3;21}(x,y)-\mu R_{3;11}(x,y)\nnb
	&=\sum_{l=0}^2(-1)^l \mu^{2l}(\mu^2-\Delta_x)(E_{l+1}(x,y)+F_{l+1}(x,y))\nnb
	&\qquad -\mu^6 (E_3(x,y)+F_3(x,y)).
	\end{align}
	Note that for $l\geq 1$, we have from \eqref{e:Ej} and \eqref{e:Fj}
	\begin{equation}
	-\Delta_x(E_{l+1}(x,y)+F_{l+1}(x,y))=E_l(x,y)+F_l(x,y).
	\end{equation}
	Thus we see that there are cancellations in the sum and we have in fact
	\begin{equation}
	2\bar\partial_x S_{\mu\mathbf 1_{\Lambda_L};21}(x,y)+\nmu S_{\mu\mathbf 1_{\Lambda_L};11}(x,y)=-\Delta_x(E_1(x,y)+F_1(x,y)).
	\end{equation}
	By Lemma~\ref{le:E1}, we see that 
	\begin{equation}
	E_1(x,y)+F_1(x,y)=\frac{1}{2\pi}\log \frac{1}{|x-y|}+\frac{1}{2\pi}\log L +\frac{1}{4\pi}\log\left(1-\frac{\bar xy}{L^2}\right).
	\end{equation}
	As the latter term here is harmonic (or actually anti-analytic) in $\mathrm{int}(\Lambda_L)$, we have 
	\begin{align}
          2\bar\partial_x S_{\mu\mathbf 1_{\Lambda_L};21}(x,y)+\nmu S_{\mu\mathbf 1_{\Lambda_L};11}(x,y)
          &=-\Delta_x(E_1(x,y)+F_1(x,y))
            \nnb
          &=-\Delta_x \frac{1}{2\pi}\log \frac{1}{|x-y|}=\delta_y(x).
	\label{e:diagin}
	\end{align}
	
	Let us consider now the case $x\in\Lambda_L ^{\mathsf c}$.  Note that for $x\in \Lambda_L^{\mathsf c}$ and $y\in \mathrm{int}(\Lambda_L)$,
        to prove \eqref{e:Green} we need to show that
	\begin{equation}
	\begin{pmatrix}
	2\bar \partial_x S_{\mu\mathbf 1_{\Lambda_L};21}(x,y) & 2\bar\partial_x \overline{S_{\mu\mathbf 1_{\Lambda_L};11}(x,y)}\\
	2\partial_x S_{\mu\mathbf 1_{\Lambda_L};11}(x,y) & 2\partial_x\overline{S_{\mu\mathbf 1_{\Lambda_L};21}(x,y)}
	\end{pmatrix}=0
	\end{equation}
	and that $x\mapsto S_{\mu\mathbf 1_{\Lambda_L}}(x,y)$ vanishes at infinity.
        For $x\in \Lambda_L^{\mathsf c}$, we note that $S_{\mu\mathbf 1_{\Lambda_L}}(x,y)$ is defined as the harmonic extension of $S_{\mu\mathbf 1_{\Lambda_L}}(\cdot ,y)|_{\partial \Lambda_L}$ (where the boundary values are understood as being given by a limit from the interior of $\Lambda_L$).
        Thus our goal is equivalent to $S_{\mu\mathbf 1_{\Lambda_L};21}(\cdot,y)|_{\partial \Lambda_L}$ having only (strictly) negative Fourier modes and $S_{\mu\mathbf 1_{\Lambda_L};11}(\cdot,y)|_{\partial \Lambda_L}$ only (strictly) positive Fourier modes. 
	
	Recalling that $e_{-n,k}$ vanishes on $\partial \Lambda_L$ while $f_{n,k}(L e^{i\theta})$ is proportional to $e^{in\theta}$, we see from \eqref{e:R11}, that $R_{3;11}|_{\partial\Lambda_L}$ has only positive Fourier modes. Similarly $E_j(\cdot,y)$ vanishes on $\partial \Lambda_L$ while $F_j(\cdot,y)|_{\partial \Lambda_L}$ has only positive Fourier modes (this follows from \eqref{e:Fj} since $E_{j+1}$ has only negative Fourier modes), so we see indeed that $S_{\mu\mathbf 1_{\Lambda_L};11}(\cdot,y)|_{\partial \Lambda_L}$ has only positive Fourier modes. Thus $S_{\mu\mathbf 1_{\Lambda_L};11}(\cdot,y)$ is of the correct form. The argument for $S_{\mu\mathbf 1_{\Lambda_L};21}(\cdot,y)$ is similar, but makes use of the fact that $\partial e_{-n+1,k}\propto \overline{f_{n,k}}$ and $\partial f_{n,k}\propto \Delta e_{n-1,k}\propto e_{n-1,k}$ -- we omit the details. We conclude that, for $x\in \Lambda_L^{\mathsf c}$ and $y\in \mathrm{int}(\Lambda_L)$,
	\begin{equation}
	DS_{\mu\mathbf 1_{\Lambda_L}}(x,y)=0.
	\label{e:out}
	\end{equation}
	
	The claim \eqref{e:goal} now follows by splitting the integral over $\R^2$ into that over $\Lambda_L$ and $\Lambda_L ^{\mathsf c}$, integrating by parts -- the boundary terms cancel due to continuity of $x\mapsto S_{\mu\mathbf 1_{\Lambda_L}}(x,y)$ across the boundary -- and combining \eqref{e:odin}, \eqref{e:diagin}, and \eqref{e:out}. Vanishing at infinity follows from the fact that the entries of $S_{\mu \mathbf 1_{\Lambda_L}}(\cdot,y)|_{\partial \Lambda_L}$ had only strictly positive or negative Fourier modes (and the corresponding entries were given by either antiholomorphic or holomorphic continuation of these boundary values) so $S_{\mu \mathbf 1_{\Lambda_L}}(x,y)$ decays at worst like $|x|^{-1}$ as $x\to \infty$.
\end{proof}

We now turn to proving that $S_{\mu\mathbf 1_{\Lambda_L}}$ satisfies the bounds we are after. 

\subsubsection{Proof of Theorem~\ref{th:gestimate} item (i)--(ii): Bounds on the Green's function}

The goal of this section is to prove the following proposition which is precisely item (i) and item (ii) of Theorem \ref{th:gestimate}.

\begin{proposition} 	\label{pr:Gbounds}
  For $L\geq 1$ and $\mu\in \R$, we have for some polynomial $P=P(L,|\mu|)$, which is independent of $x\in \C$ and $y\in \mathrm{int}(\Lambda_L)$,
  the estimates
      \begin{align}
        |S_{\mu \mathbf 1_{\Lambda_L}}(x,y) - S_0(x,y)| &\leq P(L,|\mu|)(1+|\log |x-y||) \qquad \text{for $x \in \Lambda_L$,}
        \\
	|S_{\mu\mathbf 1_{\Lambda_L}}(x,y)| &\leq
	\frac{P(L,|\mu|)}{L-|y|} \qquad \text{for $x\in \Lambda_L^{\mathsf c}$}.
    \end{align}
\end{proposition}
  
\begin{proof}
  For $x,y \in \mathrm{int}(\Lambda_L)$, it follows from that Lemma~\ref{le:E12bound} that
    \begin{gather}
      |E_1(x,y)|+|F_1(x,y)| \leq P(L)(1+|\log|x-y||),
      \\ 
      |\partial_x E_1(x,y)|+|\partial_x F_1(x,y)|
      \leq \frac{P(L)}{|x-y|},
      \\
      |E_2(x,y)|+|F_2(x,y)| +
      |\partial_x E_2(x,y)|+|\partial_x F_2(x,y)|
      \leq P(L),
    \end{gather}
    and note that $S_0$ is given by
    \begin{equation}
      S_{0;11}(x,y) = S_{0;22}(x,y) = 0, \qquad S_{0;21}(x,y) = -2\partial_x (E_1(x,y)+F_1(x,y)),
    \end{equation}
    and complex conjugation for the $12$-entry.
    Hence the singular $\partial_x(E_1+F_1)$ term in the definition of $S_{\mu \mathbf 1_\Lambda;21}$ in Definition~\ref{d:Sy}
    is canceled by $S_{0;21}$ (and analogously for the $21$-entry).
    The above bounds on the $E_1,F_1,E_2,F_2$ 
    together with the bounds from Lemma~\ref{le:Rbounds} for $R_3$ 
    thus readily imply the required bounds for $S_{\mu\mathbf 1_{\Lambda_L}}(x,y)$ when $x,y \in \mathrm{int}(\Lambda_L)$.

    For $x\in \Lambda_L ^{\mathsf c}$, the claim follows from the maximum principle.
    Indeed, $S_{\mathbf 1_{\Lambda_L}}(\cdot,y)$ is harmonic in $\Lambda_L^{\mathsf c}$, vanishes at infinity, and is continuous up to the boundary.
\end{proof}

\subsubsection{Proof of Theorem~\ref{th:gestimate} item (iii): Analyticity in \texorpdfstring{$\mu$}{mu}}

The goal of this section is to prove item (iii) of Theorem \ref{th:gestimate}, which is implied by the following proposition.

\begin{proposition} 	\label{pr:analy}
	For $x,y\in \mathrm{int}(\Lambda_L)$ with $x\neq y$, 
	the function $\mu\mapsto S_{\mu \mathbf 1_{\Lambda_L}}(x,y)$
	has an analytic continuation into some $L$-dependent neighborhood of the real axis.
        In this neighborhood, the estimate \eqref{e:Grbound} continuous to hold. 
        Moreover, in an $L$-dependent neighborhood of the origin, we have 
	\begin{align}
          S_{\mu\mathbf 1_{\Lambda_L};11}(x,y)&=\sum_{l=0}^\infty (-1)^{ l}\nmu^{2l+1}(E_{l+1}(x,y)+F_{l+1}(x,y))
                                             \\
          S_{\mu\mathbf 1_{\Lambda_L};21}(x,y)&=\sum_{l=0}^\infty (-1)^{l+1}\mu^{2l}2\partial_x(E_{l+1}(x,y)+F_{l+1}(x,y))
	\end{align}
	where $E_{l+1}$ and $F_{l+1}$ are as in \eqref{e:Ej} and \eqref{e:Fj}.
	In particular, 
	\begin{equation}
	\lim_{\mu \to 0}S_{\mu \mathbf 1_{\Lambda_L}}(x,y)=\frac{1}{2\pi}\begin{pmatrix}
	0 & 1/(\bar x-\bar y)\\
	1/(x-y) & 0
	\end{pmatrix}.
	\end{equation}
\end{proposition}
\begin{proof}
	By Definition~\ref{d:Sy}, to prove analyticity, it is enough to prove analyticity of $R_{11}$ and $R_{21}$. For this purpose, consider $\mu\in \C$ with $|\mathrm{Im}(\mu)|<\frac{1}{2}\frac{3\pi}{4L}$. For such $\mu$ we have $\frac{(\mathrm{Im} (\mu))^2 L^2}{j_{n,k}^2}<\frac{1}{4}$  by \eqref{e:evbounds} and
	\begin{equation}
	\left|1+\frac{\mu^2 L^2}{j_{n,k}^2}\right|^2=\left(1+\frac{((\mathrm{Re} (\mu))^2-(\mathrm{Im} (\mu))^2)L^2}{j_{n,k}^2}\right)^2+\left(\frac{2(\mathrm{Re}(\mu))(\mathrm{Im} (\mu))L^2}{j_{n,k}^2}\right)^2\geq \frac{9}{16}. 
	\end{equation}	
	Thus retracing our proof of Lemma \ref{le:Rbounds}, we can check that the series defining $R_{11}$ and $R_{21}$ converge uniformly in $\mu$ in such a complex strip. It then follows, for example, by Morera's theorem that $R_{11}$ and $R_{21}$ are analytic functions in $\mu$ on such a strip. 
	
        For the analogue of \eqref{e:Grbound}, we note that the proof of Proposition \ref{pr:Gbounds} works in this setting as well, and we recover our bounds.
	
	The expansion in terms of $E_l$ and $F_l$ in a neighborhood of the origin follows readily from similar arguments and the definition of $R_{11}$ and $R_{21}$ along with \eqref{e:Ej2} and Lemma \ref{le:F2}. 
	
	For the claim about the $\mu\to 0$ limit, we see from the expansions that $\lim_{\mu \to 0}S_{\mu\mathbf 1_{\Lambda_L};11}=0$ while  
	\begin{equation}
	\lim_{\mu \to 0}S_{\mu\mathbf 1_{\Lambda_L};21}(x,y)=-2\partial_x (E_1(x,y)+F_1(x,y))
	\end{equation}
	and the claim for the $21$-entry follows from \eqref{e:E1plusF1}. 
        The claim for the $12$- and $22$-entries follows simply by complex conjugation (recalling Definition~\ref{d:Sy}).
\end{proof}

Our next goal is to establish a type of resolvent identity for $S_{\mathbf 1_{\Lambda_L}}$.

\subsubsection{Proof of Theorem~\ref{th:gestimate} item (iv): A resolvent identity}

The goal of this section is to prove item (iv) in Theorem \ref{th:gestimate}, namely the following result. 

\begin{proposition}
  \label{pr:res}
  For any $L\geq 1$ and $x,y\in \Lambda_L$ with $x\neq y$, we have for $\mu\in \R$,
  \begin{equation}
    \partial_\mu S_{\mu\mathbf 1_{\Lambda_L};ij}(x,y)={-}\sum_{k\in \{1,2\}}\int_{\Lambda_L} du\, S_{\mu\mathbf 1_{\Lambda_L};ik}(x,u)S_{\mu\mathbf 1_{\Lambda_L};kj}(u,y).
  \end{equation}
\end{proposition}

\begin{proof}
	By the definition of $S_{\mu\mathbf 1_{\Lambda_L};22}(x,y)$ and $S_{\mu\mathbf 1_{\Lambda_L};12}(x,y)$ from
	Definition \ref{d:Sy}, it is sufficient to prove the claim for $S_{\mu\mathbf 1_{\Lambda_L};11}(x,y)$ and $S_{\mu\mathbf 1_{\Lambda_L};21}(x,y)$. Moreover, as one readily checks from Proposition \ref{pr:analy} that both sides are analytic functions of $\mu$ in a neighborhood of the real axis, it is enough for us to verify the claim for $\mu$ in the neighborhood of the origin where we can use the series expansion of Proposition \ref{pr:analy}. Our key tool in the proof will be Lemma \ref{le:ints}.
	
	Let us begin with $S_{\mu\mathbf 1_{\Lambda_L};11}$. Using the expansion of Proposition \ref{pr:analy}, we have 
	\begin{equation}
	\partial_\mu S_{\mu\mathbf 1_{\Lambda_L};11}(x,y)=\sum_{l=0}^\infty (2l+1)(-1)^{l} \mu^{2l}(E_{l+1}(x,y)+F_{l+1}(x,y)).
	\end{equation}
	Also, using again Proposition \ref{pr:analy} for the expansions and Lemma \ref{le:ints} to calculate the integrals,
	\begin{align}
          &\sum_{k\in \{1,2\}}\int_{\Lambda_L} du\, S_{\mu\mathbf 1_{\Lambda_L};1k}(x,u)S_{\mu\mathbf 1_{\Lambda_L};k1}(u,y)\\
          &=\sum_{l,m=0}^\infty(-1)^{l+m}\mu^{2l+2m+2}\int_{\Lambda_L}du\,(E_{l+1}(x,u)+F_{l+1}(x,u))(E_{m+1}(u,y)+F_{m+1}(u,y))\nnb
          &\quad \sum_{l,m=0}^\infty(-1)^{l+m}\mu^{2l+2m}\int_{\Lambda_L}du\,(2\bar\partial_x \overline{E_{l+1}(x,u)}+2\bar \partial_x\overline{F_{l+1}(x,u)})
            \nnb  &\qquad\qquad\qquad\qquad\qquad\qquad\qquad\qquad\times 
                    (2\partial_u E_{m+1}(u,y)+2\partial_u F_{m+1}(u,y)). \notag
	\end{align}
	The integrals without the derivatives can be evaluated immediately from Lemma \ref{le:ints}. For the derivative terms, note that 
        \begin{equation}
	\begin{gathered}
            \bar \partial_x \overline{F_{l+1}(x,u)}=\bar\partial_x 4\partial_x\bar \partial_u E_{l+2}(x,u)=-\bar\partial_u E_{l+1}(x,u),
            \\
            \partial_u F_{m+1}(u,y)=-\partial _y\overline{E_{m+1}(u,y)}.
        \end{gathered}
        \end{equation}
	Thus, integrating by parts and recalling that $E_j$ vanishes on $\partial \Lambda_L$ (with respect to either variable), and using Lemma \ref{le:ints}, we find 
	\begin{align}
	&\int_{\Lambda_L}du\,(2\bar\partial_x\overline{E_{l+1}(x,u)}+2\bar\partial_x\overline{F_{l+1}(x,u)})(2\partial_u E_{m+1}(u,y)+2\partial_u F_{m+1}(u,y))\\
	&=-\int_{\Lambda_L}du\,F_l(x,u)E_{m+1}(u,y)-4\bar \partial_x\partial_y\int_{\Lambda_L}du\,\overline{E_{l+1}(x,u)}\overline{E_{m+1}(u,y)}\nnb
          &\qquad -\int_{\Lambda_L}du\,E_{l+1}(x,u)E_m(u,y)-\int_{\Lambda_L}du\,E_{l+1}(x,u)F_m(u,y)
            \nnb &  
            =-F_{l+m+1}(x,y)-E_{l+m+1}(x,y). \notag
	\end{align}
	We conclude that 
	\begin{align}
	\sum_{k\in \{1,2\}}&\int_{\Lambda_L}du\,S_{\mu\mathbf 1_{\Lambda_L};1k}(x,u)S_{\mu\mathbf 1_{\Lambda_L};k1}(u,y)\\ 
	&=\sum_{l,m=0}^\infty(-1)^{l+m}\mu^{2l+2m+2}(E_{l+m+2}(x,y)+F_{l+m+2}(x,y))\nnb
	&\qquad -\sum_{l,m=0}^\infty(-1)^{l+m}\mu^{2l+2m}(E_{l+m+1}(x,y)+F_{l+m+1}(x,y))\nnb
	&=\sum_{n=0}^\infty \left[\sum_{l,m=0}^\infty \mathbf 1\{l+m+1=n\}\right](-1)^{n+1}\mu^{2n}(E_{n+1}(x,y)+F_{n+1}(x,y))\nnb
	&\qquad +\sum_{n=0}^\infty \left[\sum_{l,m=0}^\infty \mathbf 1\{l+m=n\}\right](-1)^{n+1} \mu^{2n}(E_{n+1}(x,y)+F_{n+1}(x,y))\notag
	\end{align}
	Noting that 
	\begin{equation}
	\sum_{l,m=0}^\infty \mathbf 1\{l+m+1=n\}+\sum_{l,m=0}^\infty \mathbf 1\{l+m=n\}=n+(n+1)=2n+1,
	\end{equation}
	we see that 
	\begin{equation}
          {-}\sum_{k\in \{1,2\}}\int_{\Lambda_L}du\,S_{\mu\mathbf 1_{\Lambda_L};1k}(x,u)S_{\mu\mathbf 1_{\Lambda_L};k1}(u,y)
          =	\partial_\mu S_{\mu\mathbf 1_{\Lambda_L};11}(x,y)
	\end{equation}
	as was required. 
	
	We now turn to the $21$-entry. For this, we begin with the remark (from Proposition \ref{pr:analy}) that
	\begin{equation}
	\partial_\mu S_{\mu\mathbf 1_{\Lambda_L};21}(x,y)=\sum_{l=1}^\infty 2l(-1)^{l+1} \mu^{2l-1}2\partial_x(E_{l+1}(x,y)+F_{l+1}(x,y)).
	\end{equation}
	On the other hand, we have 
	\begin{align}
	&\sum_{k\in \{1,2\}}\int_{\Lambda_L}du\,S_{\mu\mathbf 1_{\Lambda_L};2k}(x,u)S_{\mu\mathbf 1_{\Lambda_L};k1}(u,y)\\
	&=\sum_{l,m=0}^\infty (-1)^{l+m+1}\nmu^{2l+2m+1}\int_{\Lambda_L}du\,2\partial_x(E_{l+1}(x,u)+F_{l+1}(x,u))(E_{m+1}(u,y)+F_{m+1}(u,y))\nnb
          & 
            +\sum_{l,m=0}^\infty (-1)^{l+m+1}\nmu^{2l+2m+1}\int_{\Lambda_L}du\,(\overline{E_{l+1}(x,u)}+\overline{F_{l+1}(x,u)})2\partial_u(E_{m+1}(u,y)+F_{m+1}(u,y)). \notag
	\end{align}
	The first integrals can again be evaluated directly by taking the $x$-derivative outside from under the integral and using Lemma \ref{le:ints}. For the second integrals, we treat various terms in different ways: the $\overline{E}$-$E$ term we integrate by parts and note as before that  
	\begin{equation}
	-\partial_u\overline{E_{l+1}(x,u)}=\partial_x F_{l+1}(x,u),
	\end{equation}
	which by Lemma~\ref{le:ints} leads to a term which integrates to zero. 
	
	In the $\overline{F}$-$E$ term we write $\overline{F_{l+1}(x,u)}=4\partial_x\bar\partial_u E_{l+2}(x,u)$ and integrate by parts the $u$-derivative which (by Lemma \ref{le:ints}) leads to 
	\begin{equation}
	2\partial_x\int_{\Lambda_L}du\, E_{l+2}(x,u)(-\Delta_u)E_{m+1}(u,y)=2\partial_xE_{l+m+2}(x,y).
	\end{equation}
	
	For the $\overline{E}$-$F$ and $\overline{F}$-$F$ terms we note that $2\partial_uF_{m+1}(u,y)=-2\partial_y\overline{E_{m+1}(u,y)}$. Thus by Lemma \ref{le:ints} (and a similar argument as before utilizing the definition of $F_j$), the $\overline{F}$-$F$ term integrates to zero while 
	\begin{equation}
	\int_{\Lambda_L}du\,\overline{E_{l+1}(x,u)}2\partial_u F_{m+1}(u,y)=-2\partial_y\overline{E_{l+m+2}(x,y)}=2\partial_xF_{l+m+2}(x,y).
	\end{equation}
	
	Putting everything together, we conclude that both types of integrals have the same total contribution and 
	\begin{align}
	&\sum_{k\in \{1,2\}}\int_{\Lambda_L}du\,S_{\mu\mathbf 1_{\Lambda_L};2k}(x,u)S_{\mu\mathbf 1_{\Lambda_L};k1}(u,y)\\
	&\quad =2\sum_{l,m=0}^\infty (-1)^{l+m +1}\nmu^{2l+2m+1}2\partial_x(E_{l+m+2}(x,y)+F_{l+m+2}(x,y))\nnb
	&\quad =2\sum_{n=0}^\infty \left[\sum_{l,m=0}^\infty \mathbf 1\{l+m=n\}\right](-1)^{n}\mu^{2n+1}2\partial_x(E_{n+2}(x,y)+F_{n+2}(x,y))\nnb
	&\quad =\sum_{n=0}^\infty 2(n+1)(-1)^{n +1}\nmu^{2n+1}2\partial_x(E_{n+2}(x,y)+F_{n+2}(x,y))\nnb
	&\quad =\sum_{n=1}^\infty 2n (-1)^{n}\nmu^{2n-1}2\partial_x(E_{n+1}(x,y)+F_{n+1}(x,y)),\notag
	\end{align}
	which is precisely of the desired form and we are thus done.
\end{proof}

Finally we turn to convergence as $L\to\infty$.

\subsubsection{Proof of Theorem~\ref{th:gestimate} item (v): the \texorpdfstring{$L\to\infty$}{infinite volume} limit}

In this section, we prove item (v) of Theorem~\ref{th:gestimate}. We state this separately as the following proposition.

\begin{proposition}	\label{pr:Linfty}
	For $\mu\neq 0$, as $L\to\infty$, 
	\begin{equation}
	S_{\mu\mathbf 1_{\Lambda_L}}(x,y)\to -\frac{1}{2\pi} \begin{pmatrix}
          -\nmu K_0(|\mu||x-y|) & 2\bar\partial_x K_0(|\mu||x-y|)\\
	2 \partial_xK_0(|\mu||x-y|) & -\nmu K_0(|\mu||x-y|)
	\end{pmatrix}=:S_\mu(x,y)
	\end{equation}
	uniformly in compact subsets of $\{(x,y)\in \C^2:x\neq y\}$.
\end{proposition}

For the proof of Proposition~\ref{pr:Linfty}, we will need the following result which can also be interpreted as a resolvent identity.

\begin{lemma} 	\label{le:res2}
	For $x,y\in \mathrm{int}(\Lambda_L)$, $x \neq y$,
	\begin{equation}\label{e:res2}
	S_{\mu\mathbf 1_{\Lambda_L}}(x,y)-S_\mu(x,y)= \nmu\int_{\Lambda_L ^{\mathsf c}}du\,S_\mu(x,u)S_{\mu\mathbf 1_{\Lambda_L}}(u,y).
	\end{equation}
\end{lemma}
\begin{proof}
  It is of course sufficient for us to prove that for any $f,g\in C_c^\infty (\mathrm{int}(\Lambda_L))$ with disjoint supports,
	\begin{multline}
	\label{e:goalmu}
        \int_{\Lambda_L\times \Lambda_L}dx\,dy\,f(x)g(y)[S_{\mu\mathbf 1_{\Lambda_L}}(x,y)-S_\mu(x,y)]
        \\
          =\nmu\int_{\Lambda_L\times \Lambda_L}dx\,dy\,f(x)g(y)\int_{\Lambda_L ^{\mathsf c}} du\, S_\mu(x,u)S_{\mu\mathbf 1_{\Lambda_L}}(u,y).
	\end{multline}
	Using the disjointness of the supports of $f$ and $g$, Proposition \ref{pr:Gbounds} for $S_{\mu \mathbf 1_{\Lambda_L}}$,
        and  routine asymptotics of Bessel functions for $S_\mu$,
        we see that the integrands here are $L^1$-functions. By Fubini, we can thus perform the integrals in any order we wish.
	We now claim that, on the left-hand side of \eqref{e:goalmu}, 
	\begin{equation}
          \label{e:claim}
          \begin{gathered}
            y\mapsto g(y)\int_{\Lambda_L}dx\,f(x)S_{\mu\mathbf 1_{\Lambda_L}}(x,y)\,  \in C_c^\infty(\mathrm{int}(\Lambda_L)),
        \\
        x\mapsto f(x)\int_{\Lambda_L}dy\,g(y)S_\mu(x,y)\, \in C_c^\infty(\mathrm{int}(\Lambda_L)).
      \end{gathered}
    \end{equation}
    The fact that these functions have compact support follows from $f$ and $g$ having compact support. The smoothness of the $S_{\mu\mathbf 1_{\Lambda_L}}$-term follows from
    Lemma~\ref{le:Rbounds} item (iv) and \eqref{e:fEjFfsmooth}.
    The smoothness of the $S_\mu$-term
    follows immediately from the explicit expression of $S_\mu$ which is smooth off the diagonal. 
        
    By definition of the Green's functions,
    we have 
    $(i\slashed \partial_x+\nmu)S_\mu(x,u)=\delta(x-u)$ and
    $(i\slashed \partial_x+\nmu \mathbf 1_{\Lambda_L}(x))S_{\mu\mathbf 1_{\Lambda_L}}(x,u)=\delta(x-u)$. Since $S_\mu(x,u)$ is a function of $x-u$, also $(-i\slashed \partial_u+\nmu)S_\mu(x,u)=\delta(x-u)$.
    Thus the above smoothness (and integration by parts) implies that
    \begin{align}
      &\int_{\Lambda_L \times \Lambda_L}dx\, dy \, f(x)g(y)[S_{\mu\mathbf 1_{\Lambda_L}}(x,y)-S_\mu(x,y)]
                                         \nnb
      &=\int_{\R^2} du\, \left[(-i\slashed \partial_u-\mu)  \int_{\Lambda_L} dx\, f(x)\, S_\mu(x,u) \right] \left[\int_{\Lambda_L}dy \, g(y)S_{\mu\mathbf 1_{\Lambda_L}}(u,y)\right]\nnb
      &\quad -\int_{\R^2} du\, \left[\int_{\Lambda_L} dx\, f(x)\,S_{\mu}(x,u)\right]  \left[(i\slashed \partial_u+\nmu\mathbf 1_{\Lambda_L}(u))\int_{\Lambda_L}dy\, g(y) \, S_{\mu \mathbf 1_{\Lambda_L}}(u,y)\right]
        \nnb
      &=\int_{\R^2} du \, \left[ \int_{\Lambda_L} dx \, f(x)\, S_\mu(x,u) \right] \left[(i\slashed \partial_u+\nmu) \int_{\Lambda_L}dy \, g(y)S_{\mu\mathbf 1_{\Lambda_L}}(u,y)\right]\nnb
      &\quad -\int_{\R^2} du\, \left[\int_{\Lambda_L} dx\, f(x)\,S_{\mu}(x,u)\right]  \left[(i\slashed \partial_u+\nmu\mathbf 1_{\Lambda_L}(u))\int_{\Lambda_L}dy\, g(y) \, S_{\mu \mathbf 1_{\Lambda_L}}(u,y)\right]
        \nnb
      &=\int_{\R^2} du \left[\int_{\Lambda_L} dx\, f(x)\, S_\mu(x,u)\right](\nmu-\nmu \mathbf 1_{\Lambda_L}(u)) \left[\int_{\Lambda_L}dy \, g(y)S_{\mu\mathbf 1_{\Lambda_L}}(u,y)\right]
    \end{align}
    which is the right-hand side of \eqref{e:res2}.
\end{proof}
      
We now turn to the proof of the final claim of Theorem \ref{th:gestimate}.

\begin{proof}[Proof of Proposition \ref{pr:Linfty}]
	We can assume that $L$ is so large that $x,y\in \mathrm{int}(\Lambda_L)$. By Lemma \ref{le:res2},
	\begin{equation}
	S_{\mu\mathbf 1_{\Lambda_L}}(x,y)-S_\mu(x,y)=\nmu\int_{\Lambda_L ^{\mathsf c}} du\, S_\mu(x,u)S_{\mu\mathbf 1_{\Lambda_L}}(u,y).
	\end{equation}
	Using that for any fixed $a>0$, the Bessel function $K_0$ satisfies, for $|x|\geq a$, $|K_0(|\mu||x|)|\leq C_ae^{-|\mu| |x|}$ for some constant $C_a$ (independent of $\mu,x$) and a similar bound for $\partial K_0(|\mu||x|)$, we find from Proposition~\ref{pr:Gbounds} that for some polynomial $P=P(L,|\mu|)$, 
	\begin{equation}
          |S_\mu(x,y)-S_{\mu\mathbf 1_{\Lambda_L}}(x,y)|\leq \frac{P(L,|\mu|)}{L-|y|}\int_{\Lambda_L^{\mathsf c}} du\, e^{-|\mu||x-u|}.
	\end{equation}
	As we take $x,y$ in a fixed compact subset $B$ of $\C$, $|\int_{\Lambda_L^{\mathsf c}}e^{-|\mu||x-u|}\, du|\leq e^{-\alpha |\mu| L}$ uniformly in $x\in B$ for some $\alpha>0$ depending only on $B$. We thus deduce that given a fixed compact subset $K\subset \{(x,y)\in \C^2:x\neq y\}$ (independent of $L$) and $\mu\neq 0$, 
	\begin{equation}
	\lim_{L\to\infty}\sup_{(x,y)\in K}	|S_\mu(x,y)-S_{\mu\mathbf 1_{\Lambda_L}}(x,y)|=0,
	\end{equation}
	which was the claim.
\end{proof}

Putting together the propositions from this section also concludes our proof of Theorem~\ref{th:gestimate}, and thus that of Theorem~\ref{thm:ferm}.

\appendix
\section{Truncated and free fermion correlations}
\label{app:ferm}

In this appendix, we collect some well-known properties of truncated correlations (joint cumulants)
and free fermion correlations.

\subsection{Truncated correlations}

For arbitrary random variables $A_i$, the truncated correlations are defined by
\begin{equation} \label{e:truncatedmoment}
  \avg{A_1 \cdots A_{n}}^T
  =
  \ddp{^{n}}{t_1\dots \partial t_n}\bigg|_{t=0}
  \log
  \avg{ e^{\sum_{i=1}^n t_i A_i}}
\end{equation}
when the right-hand side exists.
For $N \in \N$ and $t=(t_1,\dots, t_N)$,
it is often convenient to define the tilted measure with expectation $\avg{\cdot}_t$ by
\begin{equation} \label{e:avgt}
    \avg{F}_t = \frac{\avg{F e^{tA}}}{\avg{e^{tA}}}, \qquad e^{tA} = e^{\sum_{i=1}^N t_i A_i},
\end{equation}
when these expressions exist.
For $1 \leq n \leq N-1$, it then follows from \eqref{e:truncatedmoment} that
\begin{equation}  \label{e:avgtnplus1}
    \avg{A_{1} \cdots A_{n+1}}_t^T
    = \ddp{}{t_{n+1}} \avg{A_1 \cdots A_n}_t^T.
\end{equation}

The next lemma shows that the definition \eqref{e:truncatedmoment} is consistent with \eqref{e:truncated}.

\begin{lemma} \label{lem:cumulants}
  Assume that $A_1, \dots, A_n$ are  random variables.
  Then 
  \begin{equation} \label{e:cumulantslogexp}
    \avg{A_1 \cdots A_{n}}^T
    =
    \avg{A_1 \cdots A_n} - \sum_{P \in \mathfrak{P}_n} \prod_j \avg{\prod_{i \in P_j} A_i}^T
  \end{equation}
  assuming all expectations exists.
\end{lemma}

\begin{proof}
  It suffices to show the claim with $\avg{\cdot}$ replaced by $\avg{\cdot}_t$ where $t=(t_1,\dots,t_N)$ and $n \leq N$.
  This is clear for $n=1$.
  To advance the induction, note that
  \begin{align}
    \avg{A_{1} \cdots A_{n+1}}_t^T
    &= \ddp{}{t_{n+1}} \avg{A_1 \cdots A_n}_t^T
      \nnb
    &= \ddp{}{t_{n+1}}\qa{ \avg{A_1 \cdots A_n}_t - \sum_{P \in \mathfrak{P}_n} \prod_j \avg{\prod_{i \in P_j} A_i}^T_t}
      \nnb
    &= \avg{A_1 \cdots A_{n+1}}_t-\avg{A_1 \cdots A_{n}}_t\avg{A_{n+1}}_t
      \nnb & \qquad\qquad 
      - \sum_{P \in \mathfrak{P}_n} \sum_k
      \avg{\prod_{i \in P_k \cup \{n+1\}} A_i}^T_t
      \prod_{j \neq k} \avg{\prod_{i \in P_j} A_i}^T_t
      \nnb
    &= \avg{A_1 \cdots A_{n+1}}_t
      - \sum_{P \in \mathfrak{P}_{n+1}} 
      \prod_{j} \avg{\prod_{i \in P_j} A_i}^T_t
  \end{align}
  as needed.
\end{proof}

\subsection{Grassmann integrals}

Let $\wedge^{2N}$ be the exterior algebra (Grassmann algebra)
on $2N$ generators $\bar\psi_1,\psi_1, \dots, \bar\psi_N,\psi_N$ over $\C$.
The bars only have notational meaning here and for notational simplicity we drop the $\wedge$ from the product notation, e.g.,
$\bar\psi_i \wedge \psi_j \equiv \bar\psi_i\psi_j$.
Thus elements $F \in \wedge^{2N}$ are noncommutative polynomials in the generators
of degree at most $2N$.
An element $F \in \wedge^{2N}$ is called even if it is a linear combination of
even monomials (i.e., ones with an even number of factors of the generators).
Let $\partial_{\bar\psi_j}$ and $\partial_{\psi_j}$ be the antiderivations on $\wedge^{2N}$ defined by
\begin{equation}
  \partial_{\bar\psi_j} (\bar\psi_jF) = F, \qquad \partial_{\bar\psi_j} F = 0
\end{equation}
for any (noncommutative) monomial $F \in \wedge^{2N}$ that does not contain a factor $\bar\psi_j$,
and analogously for the $\partial_{\psi_j}$.
For any $F \in \wedge^{2N}$ the Grassmann integral of $F$ is then defined by
\begin{equation}
  \int d_\psi d_{\bar\psi} \, F
  := \partial_\psi \partial_{\bar\psi} \, F
  := \partial_{\psi_N} \partial_{\bar\psi_N} \cdots \partial_{\psi_1} \partial_{\bar\psi_1}\, F.
\end{equation}
Note that the right-hand side is a scalar.
For any even elements $A_1, \dots, A_n$ of $\wedge^{2N}$ and any smooth function
$g \in C^\infty(\R^{n})$, we define an element $g(A_1, \dots, A_n) \in \wedge^{2N}$
by the truncation of the formal Taylor expansion of $g$ of at order $2N$.
For example, using the above definitions,
we write, for any $N \times N$ matrix $M$,
\begin{equation}
  e^{-\psi M \bar\psi}
  = \sum_{n=0}^N \frac{(-1)^n}{n!} \pa{\sum_{i,j=1}^N \psi_i M_{ij} \bar\psi_j}^n,
\end{equation}
and then have
\begin{equation}
  \int d_\psi d_{\bar\psi} \, e^{-\psi M \bar\psi}
  = \frac{(-1)^N}{N!} \int d_\psi d_{\bar\psi} \, (\psi M \bar\psi)^N
  = \frac{(-1)^N}{N!} \partial_\psi \partial_{\bar\psi} \, (\psi M \bar\psi)^N
  = \det M,
\end{equation}
by the anticommutativity of the generators and the definition of the determinant.

The following lemma is a variant of Wick's theorem for Grassmann integrals.

\begin{lemma} \label{le:fcor}
Let $K$ be an invertible $N\times N$ matrix. Then
\begin{equation} \label{e:fcor}
  \det K \int d_\psi d_{\bar\psi} \,\prod_{i=1}^n \bar\psi_{\alpha_i}\psi_{\beta_i} e^{ -\psi K^{-1} \bar\psi}
  = \det(K_{\alpha_i\beta_j})_{i,j=1}^n
  .
\end{equation}
\end{lemma}

\begin{remark} \label{rk:KS}
  Note that the Grassmann integral representation of the determinant, \eqref{e:fcor},
  can be used in the context of \eqref{e:fermdef} and \eqref{e:fermdef-bis}. For finitely many points $x_1, \dots, x_n, y_1, \dots, y_n$ one may indeed
apply this lemma to the matrix defined by
$K_{ij} = S_{\alpha_i\beta_j}(x_i,y_j)$
for $i\neq j$ and $K_{ii}=C$ for a sufficiently large constant $C$ such that $K$ is invertible.
\end{remark}

\begin{proof}
For an invertible $N\times N$ matrix $K$, the fermionic Gaussian integration by parts formula holds:
\begin{equation} \label{e:ibp}
  \int d_\psi d_{\bar\psi} \, \bar\psi_i F \, e^{-\psi K^{-1} \bar\psi} 
  =
  \sum_{j} K_{ij}
  \int d_\psi d_{\bar\psi} \, (\partial_{\psi_j}F) \, e^{-\psi K^{-1} \bar\psi} 
  .
\end{equation}
Indeed, it follows from the definitions that
\begin{equation}
  \partial_{\psi_j} e^{- \psi K^{-1} \bar\psi} 
  = - \sum_{i} (K^{-1})_{ji} \bar\psi_i  e^{-\psi K^{-1} \bar\psi} 
  ,
\end{equation}
and hence
\begin{equation}
  \bar\psi_i  e^{-\psi K^{-1} \bar\psi} 
  = - \sum_{j} K_{ij}    \partial_{\psi_j} e^{-\psi K^{-1} \bar\psi} 
  .
\end{equation}
Note that we may assume that $F$ is odd in \eqref{e:ibp} as otherwise both sides are $0$.
Therefore
\begin{equation}
  \int d_\psi d_{\bar\psi} \, \bar\psi_i F \, e^{-\psi K^{-1} \bar\psi} 
  =
  - \int d_\psi d_{\bar\psi} \, F \bar\psi_i\, e^{-\psi K^{-1} \bar\psi} 
  =
  \sum_{j} K_{ij} \int d_\psi d_{\bar\psi} \, F    \partial_{\psi_j} e^{-\psi K^{-1} \bar\psi} 
  .
\end{equation}
The claim now follows from the fact that, since $F$ is odd, for any $G$ one has
\begin{equation}
  0=\int d_\psi d_{\bar\psi} \, \partial_{\psi_j}(FG) =
  \int d_\psi d_{\bar\psi} \,\qa{ (\partial_{\psi_j}F)G - F(\partial_{\psi_j}G) }.
\end{equation}
Note that  for any monomial $F \in \wedge^{2N}$ with $n$ factors of $\bar\psi_1,\dots, \bar\psi_N$ we have $F = \frac{1}{n}\sum_i \bar\psi_i \partial_{\bar\psi_i}F$.
Thus if $F$ has degree $2n$ then
\begin{align}
  \int d_\psi d_{\bar\psi} \, F \, e^{-\psi K^{-1} \bar\psi} 
  &=
  \frac{1}{n} \sum_{i,j} K_{ij}
  \int d_\psi d_{\bar\psi} \, (\partial_{\psi_j}\partial_{\bar\psi_i} F) \, e^{-\psi K^{-1} \bar\psi} 
  \nnb 
  &=
  \frac{1}{n}
  \int d_\psi d_{\bar\psi} \, (\Delta_K F) \, e^{-\psi K^{-1} \bar\psi} 
\end{align}
where
\begin{equation}
   \Delta_K F = \sum_{i,j} K_{ij}
  \partial_{\psi_j}\partial_{\bar\psi_i} F.
\end{equation}
Iterating this, for $F$ of degree $2n$ thus
\begin{equation}
  \int d_\psi d_{\bar\psi} \, F \, e^{- \psi K^{-1} \bar\psi} 
  =
  \frac{1}{n!} \Delta_K^n F   \int d_\psi d_{\bar\psi} \, e^{-\psi K^{-1} \bar\psi} 
  .
\end{equation}
In particular,
\begin{equation}
  \det K \int d_\psi d_{\bar\psi} \, \bar\psi_i \psi_j \, e^{-\psi K^{-1} \bar\psi} 
  = K_{ij}
\end{equation}
and repeated application gives
\begin{equation}
  \det K \int d_\psi d_{\bar\psi} \,\prod_{i=1}^n \bar\psi_{\alpha_i}\psi_{\beta_i} e^{-\psi K^{-1} \bar\psi} 
  = \det(K_{\alpha_i\beta_j})_{i,j=1}^n
\end{equation}
as claimed.
\end{proof}

Given an invertible $N\times N$ matrix $K$, we now write
\begin{align} \label{e:Ffermdef}
  \avg{F} = \det K \int d_\psi d_{\bar\psi} \, e^{- \psi K^{-1} \bar\psi} 
  F.
\end{align}
From this representation, it is also easy to deduce the following properties of the fermionic correlation functions. Using Remark~\ref{rk:KS}, we make use of the properties in Section \ref{sec:massless-ferm}.

\begin{lemma} \label{le:fcor2}
For any $\sigma \in S_n$,
\begin{equation} \label{e:fermantisym}
  \avga{\prod_{k=1}^n \bar\psi_{i_k}\psi_{j_k}}
  =
  (-1)^\sigma \avga{\prod_{k=1}^n \bar\psi_{i_k}\psi_{j_{\sigma(k)}}}
\end{equation}
Moreover, if $F, G \in \wedge^{2N}$ are monomials such that
for every factor $\bar\psi_i$ in $F$ and every factor $\psi_j$ in $G$ one has $K_{ij}=0$
and for every factor $\psi_i$ in $F$ and every factor $\bar\psi_j$ in $G$ one also has $K_{ij} = 0$ then
\begin{equation} \label{e:fermfactorize}
  \avg{FG} = \avg{F} \avg{G}.
\end{equation}
\end{lemma}

As in \eqref{e:truncatedmoment}, for even elements $A_i \in \wedge^{2N}$ the truncated correlations are defined by
\begin{equation} \label{e:truncatedmoment-ferm}
  \avg{A_1 \cdots A_n}^T
  =
  \ddp{^{n}}{t_1\dots \partial t_n}\Big|_{t=0}
  \log
  \avg{e^{\sum_{i=1}^n t_i A_i}}.
\end{equation}
The next lemma gives equivalent characterizations of the truncated correlation functions.

\begin{lemma}   \label{le:fcor3}
  Assume that $A_1, \dots, A_n$ are even elements of $\wedge^{2N}$. Then
  \begin{equation} \label{e:fermcumulants1}
    \avg{A_{1} \cdots A_{n}}^T
    =
    \avg{A_1 \cdots A_n} - \sum_{P \in \mathfrak{P}_n} \prod_j \avg{\prod_{i \in P_j} A_i}^T
    .
  \end{equation}
  Moreover, if the $A_i$ are of the form $A_i = \bar\psi_{\alpha_i} \psi_{\beta_i}$, then
  \begin{equation} \label{e:fermcumulants2}
    \avg{A_1 \cdots A_n}^T = (-1)^{n+1} \sum_{\pi \in C_n} \prod_{i=1}^n K_{\alpha_{\pi^{i}(1)}\beta_{\pi^{i+1}(1)}}.
  \end{equation}
\end{lemma}

\begin{proof}
  The proof of \eqref{e:fermcumulants1}
  is identical to that of Lemma~\ref{lem:cumulants}.
  To see \eqref{e:fermcumulants2},
  we assume by induction that the identity holds for every invertible matrix $K$.
  If $n=1$, this claim is
  \begin{equation}
    \avg{A_1} = \avg{\bar\psi_{\alpha_1}\psi_{\beta_1}} = K_{\alpha_1 \beta_1} 
  \end{equation}
  which is true by Lemma~\ref{le:fcor}.
  To advance the induction,   for $t$ sufficiently small, set
  $K(t) = ({\mathbf 1}+\sum_i t_i K {\bf 1}_{\beta_i\alpha_i})^{-1}K$
  where ${\bf 1}_{ij}$ is the matrix with value $1$ in entry $ij$ and $0$ in all other entries,
  and define $\avg{\cdot}_t$ as in \eqref{e:Ffermdef} with $K(t)$ instead of $K$.
  Since
  \begin{equation}
    K(t)^{-1}= K^{-1}(\mathbf 1+\sum_i t_i K \mathbf 1_{\beta_i\alpha_i}) = K^{-1} + \sum_i t_i \mathbf 1_{\beta_i\alpha_i}
  \end{equation}
  this definition is consistent with $\avg{\cdot}_t$ is defined as in \eqref{e:avgt}, i.e.,
  \begin{equation}
    \avg{F}_t
    = \frac{\avg{F e^{-\sum t_i \psi_{\beta_i}\bar\psi_{\alpha_i}}}}{\avg{e^{-\sum t_i {\psi_{\beta_i}\bar\psi_{\alpha_i}}}}}
    = \frac{\avg{F e^{\sum t_i \bar\psi_{\alpha_i}\psi_{\beta_i}}}}{\avg{e^{\sum t_i {\bar\psi_{\alpha_i}\psi_{\beta_i}}}}}.
  \end{equation}
  Also note that
  \begin{equation} \label{e:dK}
    \ddp{}{t_j} K(t)  = -K(t) {\bf 1}_{\beta_j\alpha_j} K(t)
  \end{equation}
  as follows from
  \begin{equation}
    \ddp{}{t_j} (\mathbf 1 + \sum_i t_i K \mathbf 1_{\beta_i\alpha_i})^{-1}
    =
    -(\mathbf 1 + \sum_i t_i K \mathbf 1_{\beta_i\alpha_i})^{-1}
    K\mathbf 1_{\beta_j\alpha_j}
    (\mathbf 1 + \sum_i t_i K \mathbf 1_{\beta_i\alpha_i})^{-1}.
  \end{equation}
  By the induction hypothesis, now
  \begin{equation}
    \avg{A_1 \cdots A_n}^T_t = (-1)^{n+1} \sum_{\pi \in C_n} \prod_{i=1}^n K_{\alpha_{\pi^{i}(1)}\beta_{\pi^{i+1}(1)}}(t),
  \end{equation}
  and the claim follows from \eqref{e:avgtnplus1} and \eqref{e:dK}.
\end{proof}

\section*{Errata}

The published version of the paper has a sign error in the bosonization identity in which $-i\bar\partial\varphi$ should have been $+i\bar\partial\varphi$.
The signs are corrected in this arXiv version.
The error occured in the second case of Lemma~\ref{lem:fermtruncid}.
Precisely, the corrections compared to the published version are:
\begin{itemize}
\item Replacement of $-i\bar\partial\varphi$ by $+i\bar\partial\varphi$ in (1.12), (1.13), (1.40), (2.70), and (2.72).
\item Change of sign of $g$  in (1.47) and (1.50) for the conjectured Coleman correspondence.
\item The previous arXiv version also had an incorrect sign in (1.15) and in front of the first term on the right-hand side of (1.21); the second term in (1.21) was correct.
	These sign errors were already been corrected in the published version, but the corresponding changes in the equations mentioned above were overlooked.
\end{itemize}

\begin{ack}
We thank J.~Fr\"ohlich and V.~Mastropietro for helpful correspondence,
and also a referee for useful comments that improved the presentation.
\end{ack}

\begin{funding}
RB gratefully acknowledges funding from the European Research Council (ERC)
under the European Union's Horizon 2020 research and innovation programme
(grant agreement 851682 SPINRG).
CW wishes to thank the Academy of Finland for support through the grant 308123.
\end{funding}

\bibliography{all}

\begin{thebibliography}{10}

\bibitem{MR2319516}
R.J. Adler and J.E. Taylor.
\newblock {\em Random fields and geometry}.
\newblock Springer Monographs in Mathematics. Springer, New York, 2007.

\bibitem{MR4236060}
D.~Bahns, K.~Fredenhagen, and K.~Rejzner.
\newblock Local nets of von {N}eumann algebras in the sine-{G}ordon model.
\newblock {\em Commun. Math. Phys.}, 383(1):1--33, 2021.

\bibitem{MR4303014}
R.~Bauerschmidt and T.~Bodineau.
\newblock Log-{S}obolev inequality for the continuum sine-{G}ordon model.
\newblock {\em Comm. Pure Appl. Math.}, 74(10):2064--2113, 2021.

\bibitem{2009.09664}
R.~Bauerschmidt and M.~Hofstetter.
\newblock Maximum and coupling of the sine-{G}ordon field.
\newblock {\em Ann. Probab.}
\newblock to appear.

\bibitem{MR814849}
G.~Benfatto.
\newblock An iterated {M}ayer expansion for the {Y}ukawa gas.
\newblock {\em J. Statist. Phys.}, 41(3-4):671--684, 1985.

\bibitem{MR2308750}
G.~Benfatto, P.~Falco, and V.~Mastropietro.
\newblock Functional integral construction of the massive {T}hirring model:
  verification of axioms and massless limit.
\newblock {\em Commun. Math. Phys.}, 273(1):67--118, 2007.

\bibitem{MR2461991}
G.~Benfatto, P.~Falco, and V.~Mastropietro.
\newblock Massless sine-{G}ordon and massive {T}hirring models: proof of
  {C}oleman's equivalence.
\newblock {\em Commun. Math. Phys.}, 285(2):713--762, 2009.

\bibitem{MR649810}
G.~Benfatto, G.~Gallavotti, and F.~Nicol\`o.
\newblock On the massive sine-{G}ordon equation in the first few regions of
  collapse.
\newblock {\em Commun. Math. Phys.}, 83(3):387--410, 1982.

\bibitem{MR1947693}
G.~Benfatto and V.~Mastropietro.
\newblock On the density-density critical indices in interacting {F}ermi
  systems.
\newblock {\em Commun. Math. Phys.}, 231(1):97--134, 2002.

\bibitem{MR2070092}
G.~Benfatto and V.~Mastropietro.
\newblock Ward identities and vanishing of the beta function for {$d=1$}
  interacting {F}ermi systems.
\newblock {\em J. Statist. Phys.}, 115(1-2):143--184, 2004.

\bibitem{MR1297289}
D.~Bernard and A.~LeClair.
\newblock Differential equations for sine-{G}ordon correlation functions at the
  free fermion point.
\newblock {\em Nuclear Phys. B}, 426(3):534--558, 1994.

\bibitem{MR574172}
D.C. Brydges and P.~Federbush.
\newblock Debye screening.
\newblock {\em Commun. Math. Phys.}, 73(3):197--246, 1980.

\bibitem{MR914427}
D.C. Brydges and T.~Kennedy.
\newblock Mayer expansions and the {H}amilton-{J}acobi equation.
\newblock {\em J. Statist. Phys.}, 48(1-2):19--49, 1987.

\bibitem{MR3305999}
F.~Camia, C.~Garban, and C.M. Newman.
\newblock Planar {I}sing magnetization field {I}. {U}niqueness of the critical
  scaling limit.
\newblock {\em Ann. Probab.}, 43(2):528--571, 2015.

\bibitem{MR3296821}
D.~Chelkak, C.~Hongler, and K.~Izyurov.
\newblock Conformal invariance of spin correlations in the planar {I}sing
  model.
\newblock {\em Ann. of Math. (2)}, 181(3):1087--1138, 2015.

\bibitem{PhysRevD.11.2088}
S.~Coleman.
\newblock Quantum sine-{G}ordon equation as the massive {T}hirring model.
\newblock {\em Phys. Rev. D}, 11:2088--2097, Apr 1975.

\bibitem{MR923851}
F.~Cornu and B.~Jancovici.
\newblock On the two-dimensional {C}oulomb gas.
\newblock {\em J. Statist. Phys.}, 49(1-2):33--56, 1987.

\bibitem{MR2244975}
G.~Da~Prato.
\newblock {\em An introduction to infinite-dimensional analysis}.
\newblock Universitext. Springer-Verlag, Berlin, 2006.
\newblock Revised and extended from the 2001 original by Da Prato.

\bibitem{PhysRevD.11.3424}
R.F. Dashen, B.~Hasslacher, and A.~Neveu.
\newblock Particle spectrum in model field theories from semiclassical
  functional integral techniques.
\newblock {\em Phys. Rev. D}, 11:3424--3450, Jun 1975.

\bibitem{MR1488586}
C.~Destri and H.J. de~Vega.
\newblock Non-linear integral equation and excited-states scaling functions in
  the sine-{G}ordon model.
\newblock {\em Nuclear Phys. B}, 504(3):621--664, 1997.

\bibitem{MR1672504}
J.~Dimock.
\newblock Bosonization of massive fermions.
\newblock {\em Commun. Math. Phys.}, 198(2):247--281, 1998.

\bibitem{MR1240586}
J.~Dimock and T.R. Hurd.
\newblock Construction of the two-dimensional sine-{G}ordon model for
  {$\beta<8\pi$}.
\newblock {\em Commun. Math. Phys.}, 156(3):547--580, 1993.

\bibitem{MR1777310}
J.~Dimock and T.R. Hurd.
\newblock Sine-{G}ordon revisited.
\newblock {\em Ann. Henri Poincar\'e}, 1(3):499--541, 2000.

\bibitem{MR3369909}
J.~Dub\'{e}dat.
\newblock Dimers and families of {C}auchy-{R}iemann operators {I}.
\newblock {\em J. Amer. Math. Soc.}, 28(4):1063--1167, 2015.

\bibitem{MR1672130}
V.~Fateev, D.~Fradkin, S.~Lukyanov, A.~Zamolodchikov, and A.~Zamolodchikov.
\newblock Expectation values of descendent fields in the sine-{G}ordon model.
\newblock {\em Nuclear Phys. B}, 540(3):587--609, 1999.

\bibitem{MR818828}
P.~Federbush and T.~Kennedy.
\newblock Surface effects in {D}ebye screening.
\newblock {\em Commun. Math. Phys.}, 102(3):361--423, 1985.

\bibitem{MR0421409}
J.~Fr\"{o}hlich.
\newblock Quantized ``sine-{G}ordon'' equation with a nonvanishing mass term in
  two space-time dimensions.
\newblock {\em Phys. Rev. Lett.}, 34:833--836, 1975.

\bibitem{MR0434278}
J.~Fr\"{o}hlich.
\newblock Classical and quantum statistical mechanics in one and two
  dimensions: two-component {Y}ukawa- and {C}oulomb systems.
\newblock {\em Commun. Math. Phys.}, 47(3):233--268, 1976.

\bibitem{MR0523016}
J.~Fr\"{o}hlich.
\newblock Quantum sine-{G}ordon equation and quantum solitons in two space-time
  dimensions.
\newblock pages 371--414. NATO Advanced Study Inst. Series C: Math. and Phys.
  Sci., Vol. 23, 1976.

\bibitem{MR937363}
J.~Fr\"{o}hlich and P.~Marchetti.
\newblock Bosonization, topological solitons and fractional charges in
  two-dimensional quantum field theory.
\newblock {\em Commun. Math. Phys.}, 116(1):127--173, 1988.

\bibitem{MR0456220}
J.~Fr\"{o}hlich and Y.M. Park.
\newblock Remarks on exponential interactions and the quantum sine-{G}ordon
  equation in two space-time dimensions.
\newblock {\em Helv. Phys. Acta}, 50(3):315--329, 1977.

\bibitem{MR496191}
J.~Fr\"{o}hlich and Y.M. Park.
\newblock Correlation inequalities and the thermodynamic limit for classical
  and quantum continuous systems.
\newblock {\em Commun. Math. Phys.}, 59(3):235--266, 1978.

\bibitem{MR0443693}
J.~Fr\"{o}hlich and E.~Seiler.
\newblock The massive {T}hirring-{S}chwinger model ({QED$_2$}): convergence of
  perturbation theory and particle structure.
\newblock {\em Helv. Phys. Acta}, 49(6):889--924, 1976.

\bibitem{MR799499}
M.~Gaudin.
\newblock L'isotherme critique d'un plasma sur r\'{e}seau
  {$(\beta=2,\;d=2,\;n=2)$}.
\newblock {\em J. Physique}, 46(7):1027--1042, 1985.

\bibitem{MR1814364}
D.~Gilbarg and N.S. Trudinger.
\newblock {\em Elliptic partial differential equations of second order}.
\newblock Classics in Mathematics. Springer-Verlag, Berlin, 2001.
\newblock Reprint of the 1998 edition.

\bibitem{MR3606736}
A.~Giuliani, V.~Mastropietro, and F.L. Toninelli.
\newblock Height fluctuations in interacting dimers.
\newblock {\em Ann. Inst. Henri Poincar\'{e} Probab. Stat.}, 53(1):98--168,
  2017.

\bibitem{MR4121614}
A.~Giuliani, V.~Mastropietro, and F.L. Toninelli.
\newblock Non-integrable dimers: universal fluctuations of tilted height
  profiles.
\newblock {\em Commun. Math. Phys.}, 377(3):1883--1959, 2020.

\bibitem{MR1924468}
D.~Grieser.
\newblock Uniform bounds for eigenfunctions of the {L}aplacian on manifolds
  with boundary.
\newblock {\em Comm. Partial Differential Equations}, 27(7-8):1283--1299, 2002.

\bibitem{2005.11530}
C.~Guillarmou, A.~Kupiainen, R.~Rhodes, and V.~Vargas.
\newblock Conformal bootstrap in {L}iouville {T}heory.
\newblock Preprint, arXiv:2005.11530.

\bibitem{MR255909}
H.W. Hethcote.
\newblock Bounds for zeros of some special functions.
\newblock {\em Proc. Amer. Math. Soc.}, 25:72--74, 1970.

\bibitem{1307.4104}
C.~Hongler, F.J. Viklund, and K.~Kyt\"ol\"a.
\newblock {C}onformal {F}ield {T}heory at the {L}attice {L}evel: {D}iscrete
  {C}omplex {A}nalysis and {V}irasoro {S}tructure.
\newblock Preprint, arXiv:1307.4104.

\bibitem{MR1175177}
C.~Itzykson and J.-M. Drouffe.
\newblock {\em Statistical field theory. {V}ol. 2}.
\newblock Cambridge Monographs on Mathematical Physics. Cambridge University
  Press, Cambridge, 1989.
\newblock Strong coupling, Monte Carlo methods, conformal field theory, and
  random systems.

\bibitem{MR4149524}
J.~Junnila, E.~Saksman, and C.~Webb.
\newblock Imaginary multiplicative chaos: {M}oments, regularity and connections
  to the {I}sing model.
\newblock {\em Ann. Appl. Probab.}, 30(5):2099--2164, 2020.

\bibitem{MR4010781}
W.~Kroschinsky and D.H.U. Marchetti.
\newblock On the {M}ayer series of two-dimensional {Y}ukawa gas at inverse
  temperature in the interval of collapse.
\newblock {\em J. Stat. Phys.}, 177(2):324--364, 2019.

\bibitem{MR4060417}
A.~Kupiainen, R.~Rhodes, and V.~Vargas.
\newblock Integrability of {L}iouville theory: proof of the {DOZZ} formula.
\newblock {\em Ann. of Math. (2)}, 191(1):81--166, 2020.

\bibitem{1903.01394}
H.~Lacoin, R.~Rhodes, and V.~Vargas.
\newblock A probabilistic approach of ultraviolet renormalisation in the
  boundary {S}ine-{G}ordon model.
\newblock Preprint, arXiv:1903.01394.

\bibitem{MR3339158}
H.~Lacoin, R.~Rhodes, and V.~Vargas.
\newblock Complex {G}aussian multiplicative chaos.
\newblock {\em Commun. Math. Phys.}, 337(2):569--632, 2015.

\bibitem{MR1453266}
S.~Lukyanov and A.~Zamolodchikov.
\newblock Exact expectation values of local fields in the quantum sine-{G}ordon
  model.
\newblock {\em Nuclear Phys. B}, 493(3):571--587, 1997.

\bibitem{MR1850795}
S.~Lukyanov and A.~Zamolodchikov.
\newblock Form factors of soliton-creating operators in the sine-{G}ordon
  model.
\newblock {\em Nuclear Phys. B}, 607(3):437--455, 2001.

\bibitem{MR172638}
D.C. Mattis and E.H. Lieb.
\newblock Exact solution of a many-fermion system and its associated {B}oson
  field.
\newblock {\em J. Mathematical Phys.}, 6:304--312, 1965.

\bibitem{MR442316}
R.C. McCann.
\newblock Lower bounds for the zeros of {B}essel functions.
\newblock {\em Proc. Amer. Math. Soc.}, 64(1):101--103, 1977.

\bibitem{MR849210}
F.~Nicol\`o, J.~Renn, and A.~Steinmann.
\newblock On the massive sine-{G}ordon equation in all regions of collapse.
\newblock {\em Commun. Math. Phys.}, 105(2):291--326, 1986.

\bibitem{MR0475380}
Y.M. Park.
\newblock Massless quantum sine-{G}ordon equation in two space-time dimensions:
  correlation inequalities and infinite volume limit.
\newblock {\em J. Mathematical Phys.}, 18(12):2423--2426, 1977.

\bibitem{MR0493420}
M.~Reed and B.~Simon.
\newblock {\em Methods of modern mathematical physics. {II}. {F}ourier
  analysis, self-adjointness}.
\newblock Academic Press [Harcourt Brace Jovanovich Publishers], New York,
  1975.

\bibitem{MR446210}
R.~Seiler and D.A. Uhlenbrock.
\newblock On the massive {T}hirring model.
\newblock {\em Ann. Physics}, 105(1):81--110, 1977.

\bibitem{MR3031783}
Y.~Shi and B.~Xu.
\newblock Gradient estimate of a {D}irichlet eigenfunction on a compact
  manifold with boundary.
\newblock {\em Forum Math.}, 25(2):229--240, 2013.

\bibitem{MR0489552}
B.~Simon.
\newblock {\em The {$P(\phi )_{2}$} {E}uclidean (quantum) field theory}.
\newblock Princeton University Press, Princeton, N.J., 1974.
\newblock Princeton Series in Physics.

\bibitem{10.1142/2436}
M.~Stone.
\newblock {\em Bosonization}.
\newblock World Scientific, 1994.

\bibitem{MR2021021}
A.M. Tsvelik.
\newblock {\em Quantum field theory in condensed matter physics}.
\newblock Cambridge University Press, Cambridge, second edition, 2003.

\bibitem{MR923850}
W.-S. Yang.
\newblock Debye screening for two-dimensional {C}oulomb systems at high
  temperatures.
\newblock {\em J. Statist. Phys.}, 49(1-2):1--32, 1987.

\bibitem{10.1142/S0217751X9500053X}
A.B. Zamolodchikov.
\newblock Mass scale in the sine-{G}ordon model and its reductions.
\newblock {\em International Journal of Modern Physics A}, 10(08):1125--1150,
  1995.

\end{thebibliography}
\bibliographystyle{plain}

\end{document}